\documentclass[reqno]{amsart}
\usepackage{amssymb}
\usepackage{amsfonts}

\usepackage{mathrsfs}
\usepackage{caption}
\usepackage{mathtools}
\usepackage{enumitem}
\usepackage{pifont}
\usepackage[numbers,sort&compress]{natbib}
\usepackage{multirow}
\makeatletter
\@addtoreset{equation}{section}
\makeatother

\usepackage{microtype}

\usepackage{quoting}

\DeclareMathOperator{\tr}{Tr}

\DeclareMathOperator{\sgn}{sgn}

\newtheorem{theorem}{Theorem}[section]

\newtheorem{proposition}[theorem]{Proposition}
\newtheorem{lemma}[theorem]{Lemma}

\newtheorem{remark}[theorem]{Remark}

\newtheorem{notation}[theorem]{Notation}

\usepackage{hyperref}
\allowdisplaybreaks[4]

\usepackage{amsthm}

\setlist[enumerate]{leftmargin=*} 

\begin{document}
\title[Positive Temperature Huang-Yang Formula]{The second order Huang-Yang approximation to the Fermi thermodynamic pressure}
\date{\today}

\author[X. Chen]{Xuwen Chen}
\address{Department of Mathematics, University of Rochester, Rochester, NY 14627, USA}

\email{xuwenmath@gmail.com}

\author[J. Wu]{Jiahao Wu}
\address{School of Mathematical Sciences, Peking University, Beijing, 100871, China}
\email{wjh00043@gmail.com}

\author[Z. Zhang]{Zhifei Zhang}
\address{School of Mathematical Sciences, Peking University, Beijing, 100871, China}

\email{zfzhang@math.pku.edu.cn}

\begin{abstract}
  We consider a dilute Fermi gas in the thermodynamic limit with interaction potential scattering length $\mathfrak{a}_0$ at temperature $T>0$. We prove the 2nd order Huang-Yang approximation for the Fermi pressure of the system, in which there is a 2nd order term carrying the positive temperature efffect.Our formula is valid up to the temperature $T<\rho^{\frac{2}{3}+\frac{1}{6}}$,  which is, by scaling, also necessary for the Huang-Yang formula to hold. Here, $T_F\sim\rho^{\frac{2}{3}}$ is the Fermi temperature. We also establish during the course of the proof, a conjecture regarding the second order approximation of density $\rho$ by R. Seiringer \cite{FermithermoTpositive}. Our proof uses frequency localization techniques from the analysis of nonlinear PDEs and does not involve spatial localization or Bosonization. In particular, our method covers the classical Huang-Yang formula at zero temperature.
\end{abstract}

\maketitle
\tableofcontents

\section{Introduction}
\par In 1957 \cite{huangyang,LHY106,TDLEECNYANG112}, K.Huang, T.D.Lee and C.N.Yang wrote down two second order formulas regarding the precise energy density of interacting quantum many-body Boson and Fermion systems in the thermodynamics limits at zero temperature. At the time, the formulations of such quantum-mechanical system and their physical interpretations were famous problems among physicists. Not only the now so called Bosonic Lee-Huang-Yang and Fermionic Huang-Yang formulas provide accurate (up to 2nd order) predictions to the ground state energy of such complicated interacting systems, they state that, up to the 2nd order, the ground state energy of such complicated interacting systems is fully determined by the scattering length $a_0$ of the microscopic interparticle interaction.
\par For a $N$ spin-$1/2$ Fermions system inside a box with edge length $L$, we define the energy density $e(\rho)$ as the $N$-particle ground state energy $E_N$ over volume $L^3$ and the thermodynamics limit as the limit where firstly, $N\to\infty$ with the particle density $\rho=N/L^3$ fixed, and then $\rho\to0$. If half particles spin up and the other half spin down and the two-body interaction potential scattering length is $a_0$, the Huang-Yang formula \cite{huangyang} reads that,
  \begin{equation}\label{no1}
    e(\rho)=\frac{3}{5}(3\pi^2)^{\frac{2}{3}}\rho^{\frac{5}{3}}+2\pi a_0\rho^2
    +\frac{12}{35}(11-2\ln2)3^{\frac{1}{3}}\pi^{\frac{2}{3}}a_0^2\rho^{\frac{7}{3}}
    +o(\rho^{\frac{7}{3}}).
  \end{equation}
in the thermodynamic limit.
On the other hand, for Boson systems, with particle density given by $\rho$ and scattering length given by $a_0$, the Lee-Huang-Yang formula \cite{huangyang,LHY106,TDLEECNYANG112} is,
 \begin{equation}\label{no2}
     e(\rho)
     =4\pi\rho^2 a_0\left(1+\frac{128}{15\sqrt{\pi}}(\rho a_0^3)^{1/2}+
     o((\rho a_0^3)^{1/2})\right).
\end{equation}
in the thermodynamic limit.

\par Proving (\ref {no1}) and (\ref {no2}) mathematically have been open problems in mathematical physics for decades. For example, the 1st order term in (\ref {no2}) as an upper bound was proved by Dyson in 1957 \cite{Dyson1957}, but was only proved as a lower bound in 1998, almost 40 years later, by Lieb and Yngvason \cite{LiebYng1998}. Many more important and interesting rigorous studies of Bosons were collected and produced by Lieb, Seiringer, Solovej and Yngvason in \cite{lieb2005mathematics}. For Fermions, \cite{2005fermiphy} gave a rigorous proof to the first order of Huang-Yang formula. Moreover, one can see \cite{workshopfermi} for many significant and pioneering references. This paper is devoted to prove (\ref{no1}) in full and beyond.

\par For (\ref{no2}), the second order Lee-Huang-Yang formula in the thermodynamics limit had been proved in \cite{bose2ndthermo1,bose2ndthermo2} by S.Fournais and J.P.Solovej, with delicate spatial localization and treatment of the boundary condition. Previous to the proof of Lee-Huang-Yang formula in the thermodynamic limit, the second order Lee-Huang-Yang formula for Boson systems in the Gross-Pitaevskii regime has been mathematically first proved in the work \cite{2018Bogoliubov} by C.Boccato, C.Brennecke, S.Cenatiempo, and B.Schlein, after many preparations like in \cite{boccatoBrenCena2020optimal,BEC2018,BEC2002}. For the low temperature $T>0$ case, \cite{haberberger2023,haberberger2024} gave the Lee-Huang-Yang formula, with localization of the volume. For (\ref{no1}), the Fermion case, for zero temperature energy approximation to Fermion systems in the thermodynamics limit, \cite{dilutefermiBog,fermiupper,giacomelli2024} gave the result up to the first order via Bogoliubov theory, then recently, an upper bound towards the second order Huang-Yang formula has been obtained in \cite{2024huangyangformulalowdensityfermi}. Particularly, the ground state energy expansion for spin-polarized Fermi gas in the thermodynamic regime can be seen in \cite{polarizedupper}. For the low temperature $T>0$ case, R. Seiringer gave the first order approximation to the Fermi pressure for the order of temperature not greater than the Fermi temperature $T_F\sim\rho^{\frac{2}{3}}$ in \cite{FermithermoTpositive}. A thorough analysis to the spin-polarized ($\mathbf{q}=1$) dilute Fermi gas was carried out in \cite{spin1lower,spin1upper}. In \cite{WJH}, we have proved (\ref{no1}) in the Gross-Pitaevskii regime as a catalyst of the thermodynamics limit and a test ground of our new techniques more based on frequency localizations from nonlinear PDEs without using Bosonization. In this paper, not only we advance our frequency localization techniques above the Gross-Pitaevskii limit to the thermodynamics limit and hence proving (\ref {no1}) in full without spatial localization, we go beyond zero temperature and have established the Huang-Yang formula up to the scaling-critical temperature, which might be a phase transition temperature, and we resolve a conjecture by R. Seiringer \cite[p.732]{FermithermoTpositive}\footnote{3 lines below \cite[Corollary 1]{FermithermoTpositive}} during the course.

\par  Another important approximation problem is the mean-field limit. For Boson systems, \cite{meanfield2011seiR} first gave a proof of Bogoliubov theory in the mean field regime. Later on, \cite{spectrum} predicts the lowest part of the spectrum of the many-body Hamiltonian of Bosons in the mean-ﬁeld regime, covering a large class of interacting Bose gases. For Fermion systems, \cite{meanfieldfermi2019,meanfieldfermi2021} extracted the correlation energy of weakly interacting Fermi gases in high density regime, via the method of Bosonization\footnote{We are not using Bosonization in this paper.} developed recently. More recent works include, for example \cite{Seiringerfermi,Phanfermimeanfield,christiansen2024correlationenergyelectrongas}.

\subsection{The Model \& Temperature Threshold for the Huang-Yang Formula}\label{tem thre}
\par In this paper, we mainly consider fermionc system of indefinite particle number, which may be referred to \textit{grand canonical ensemble} in physics. This model is the choice of physicists when dealing with fermions at positive temperature, as it is considered to be more convenient to handle with and more accurate (see for example \cite[8.6]{huang2008statistical}). We define the grand canonical ensemble of the interacting dilute spin $(\mathbf{q}-1)/2$ Fermi gas below. The system is set at the positive temperature $T>0$, and we set $\beta=1/T$. Let the spin $(\mathbf{q}-1)/2$ be given by $\mathbf{q}\in\mathbb{N}_{>0}$. Denote the set of $\mathbf{q}$ spin states by $\mathcal{S}_\mathbf{q}=\{\varsigma_i,i=1,\dots,\mathbf{q}\}$. The space-spin variable for single particle has the form $z=(x,\sigma)$. Let $N\in\mathbb{N}_{\geq0}$ be a non-negative integer. For $N$ Fermions of spin $(\mathbf{q}-1)/2$ , confined to a large 3D box $\Lambda_L=[-L/2,L/2]^3$ equipped with periodic boundary condition, the wave function that describes these $N$ Fermions should be in the Hilbert space $L_a^2\big((\Lambda_L\times\mathcal{S}_\mathbf{q})^N\big)\eqqcolon \mathcal{H}^{\wedge N}$ consisting of functions that are anti-symmetric with respect to permutations of the $N$ particles.

\par The Hilbert space for the grand canonical ensemble of Fermions is given by the Fock space $\mathcal{F}=\bigoplus_{N=0}^\infty\mathcal{H}^{\wedge N}$, where we denote $\mathcal{H}^{\wedge 0}=\mathbb{C}$. The Hamiltonian is correspondingly $H=\bigoplus_{N=0}^\infty H_N$, where $H_0=0$, $H_1=-\Delta_{x}$ and
\begin{equation}\label{Hamilton n}
 H_{N}=\sum_{j=1}^{N}-\Delta_{{x}_j}+\sum_{1\leq i<j\leq N}v({x}_i-{x}_j)
\end{equation}
for $N\geq2$. We require the interaction potential $v$ to be non-negative, radially-symmetric and compactly supported in some 3D ball. Moreover, we assume $v$ has scattering length $\mathfrak{a}_0\geq0$ (For a precise discussion of the scattering length one can see, for example, \cite[Appendix C]{lieb2005mathematics}). For further usage, we also denote the non-interacting Hamiltonian by $\mathcal{K}=\bigoplus_{N=0}^\infty\mathcal{K}_N$, where $\mathcal{K}_0=0$ and
\begin{equation}\label{Hamilton K_n non-inter}
 \mathcal{K}_{N}=\sum_{j=1}^{N}-\Delta_{{x}_j}
\end{equation}
for $N\geq1$. The interaction operator is given by $\mathcal{V}=\bigoplus_{N=0}^\infty\mathcal{V}_N$, where $\mathcal{V}_0=\mathcal{V}_1=0$ and
\begin{equation}\label{Interaction op V_n}
  \mathcal{V}_N=\sum_{1\leq i<j\leq N}v({x}_i-{x}_j)
\end{equation}
for $N\geq2$. Moreover, the particle number operator on the Fock space is defined by $\mathcal{N}=\bigoplus_{N=0}^\infty N$.

\par An operator $\Gamma$ on $\mathcal{F}$ is a density matrix (also referred to a state) if it is a non-negative trace class operator with $\tr_{\mathcal{F}}\Gamma=1$. We say $\Gamma$ preserves particle number if it commutes with the particle number operator $\mathcal{N}$. It is straight-forward to verify that, for a density matrix $\Gamma$ which preserves particle number, it can be written as $\Gamma=\bigoplus_{N=0}^\infty\Gamma_N$, where $\Gamma_N$ is a non-negative trace class operator on $\mathcal{H}^{\wedge N}$.

\par Since $\Gamma_N$ is a trace class operator, it can be characterized by an integral kernel, which we also denote it by $\Gamma_N$. Therefore, for a density matrix $\Gamma$ that preserves particle number, the normalized one-particle density matrix $\gamma$, which is an operator on $\mathcal{H}$ defined through the quadratic form\footnote{We define in Section \ref{Fock space} the creation and annihilation operators $a^*(f)$ and $a(g)$.}, for $f,g\in\mathcal{H}$:
\begin{equation}\label{1pdm define}
  \langle\gamma f,g\rangle=\tr_\mathcal{F}\big(a^*(f)a(g)\Gamma\big).
\end{equation}
 By definition, its integral kernel is given by
\begin{equation}\label{gamma 1pdm}
  \gamma(z,z^\prime)=\sum_{N=1}^\infty
  N\int_{(\Lambda_L\times\mathcal{S}_{\mathbf{q}})^N}\Gamma_N(z,Z_{N-1};z^\prime,Z_{N-1})
  dZ_{N-1}.
\end{equation}
We can easily verify by definition that
\begin{equation}\label{tr gamma=tr NGamma}
  \tr_{\mathcal{H}}\gamma=\tr_{\mathcal{F}}(\mathcal{N}\Gamma).
\end{equation}
We will only consider $\Gamma$ such that $\tr (\mathcal{N}\Gamma)$ is finite, i.e. $\gamma$ is a positive trace class operator on $\mathcal{H}$.

\par Let $Z_N=(z_1,\dots,z_N)\in (\Lambda_L\times\mathcal{S}_{\mathbf{q}})^N$ and $h\in\mathbb{R}^3$, we apply the short-hand notation
\begin{equation}\label{short hand notation Z_N+h}
  Z_N+h\coloneqq(x_1+h,\sigma_1,\dots,x_N+h,\sigma_N).
\end{equation}
For $\Psi=(\psi_0,\psi_1,\dots,\psi_N,\dots)\in\mathcal{F}$, we define the translation operator $\mathcal{T}_h$ on $\mathcal{F}$ by
\begin{equation}\label{translation op}
  \mathcal{T}_h\Psi=
  (\psi_0,\psi_1(Z_1+h),\dots,\psi_N(Z_N+h),\dots).
\end{equation}
Obviously $\mathcal{T}_h^{-1}=\mathcal{T}_h^*=\mathcal{T}_{-h}$. An operator $\Gamma$ on $\mathcal{F}$ is said to be translation-invariant, if $\mathcal{T}_{h}^*\Gamma\mathcal{T}_h=\Gamma$. For a density matrix $\Gamma$ which preserves particle number and possesses translation-invariance, its one-particle density matrix can be written in the form
\begin{equation}\label{1-pdm translation invariance}
  \gamma(z,z^\prime)=\gamma(x-x^\prime,\sigma,\sigma^\prime).
\end{equation}
In particular, we denote $\gamma(x,\sigma,\sigma)=\gamma(x,\sigma)$ for short. We then use (\ref{1-pdm translation invariance}) to define for $k\in(2\pi/L)\mathbb{Z}^3$ and $\nu\in\mathcal{S}_{\mathbf{q}}$
\begin{equation}\label{gamma_k,sigma}
\begin{aligned}
  \hat{\gamma}(k,\nu)&\coloneqq \frac{1}{L^3}
  \int_{(\Lambda_L\times\mathcal{S}_{\mathbf{q}})^2}
  \gamma(z,z^\prime)e^{-ikx/L}e^{ikx^\prime/L}\chi_{\sigma=\sigma^\prime=\nu}dzdz^\prime\\
  &=\int_{\Lambda_L}\gamma(x,\nu)e^{-ikx/L}dx,
\end{aligned}
\end{equation}
which satisfies $0\leq\hat{\gamma}(k,\nu)\leq1$ since $0\leq\gamma\leq1$ by (\ref{1pdm define}) and (\ref{anticommutator}).

\par For a density matrix $\Gamma$, where we demand $\Gamma$ is a non-negative trace class operator on $\mathcal{F}$ with $\tr_\mathcal{F}\Gamma=1$. The thermodynamic pressure functional for Hamiltonian $H$ is defined by
\begin{equation}\label{pressure functional interacting}
  -L^3P^L[\Gamma]\coloneqq \tr\big((H-\mu\mathcal{N})\Gamma\big)-\frac{1}{\beta}S[\Gamma],
\end{equation}
where $S[\Gamma]\coloneqq -\tr(\Gamma\ln\Gamma)$ is the von Neumann entropy, and $\mu>0$ denotes the chemical potential. We can always consider $\Gamma$ such that all the eigenfunctions of $\Gamma$ lies in the quadratic form domain of both $\Gamma$ and $\mathcal{N}$, otherwise we can simply take $\tr\big((H-\mu\mathcal{N})\Gamma\big)=+\infty$.

\par Let $P^L(\beta,\mu)$ denote the supremum of the thermodynamic pressure functional $P^L[\Gamma]$ over all density matrices. $P^L(\beta,\mu)$ is uniquely attained by the Gibbs density matrix $G$ (see, for example, \cite[Theorem 2.8]{GenH_F})
\begin{equation}\label{Gibbs state interacting}
  G=Z^{-1}e^{-\beta(H-\mu\mathcal{N})},\quad Z=\tr e^{-\beta(H-\mu\mathcal{N})}.
\end{equation}
Since $H$ and $\mathcal{N}$ both preserve particle number and possess translation-invariance, so is the Gibbs state $G$. In fact, we can calculate directly
\begin{equation}\label{P^L[G]}
  P^L(\beta,\mu)=P^L[G]=\frac{1}{L^3\beta}\ln Z.
\end{equation}
Our main concern is the limit pressure of the grand canonical system
\begin{equation}\label{pressure lim}
  P(\beta,\mu)\coloneqq\lim_{L\to\infty}P^L(\beta,\mu)=\lim_{L\to\infty}P^L[G].
\end{equation}
One can, on the other hand, consider the free energy for system with definite particle number. In physics terminology, the free energy and the thermodynamic pressure can inter-convert via Legendre transformation. However, as mentioned above, the thermodynamic pressure is more convenient for direct analysis, according to \cite[8.6]{huang2008statistical}.

\par The above discussion can be applied correspondingly to the non-interacting Hamiltonian $\mathcal{K}$. The Gibbs state $G_0$ to the non-interacting pressure functional $P^L_0[\Gamma]$ is given by
\begin{equation}\label{Gibbs state non-inter}
  G_0=Z_0^{-1}e^{-\beta(\mathcal{K}-\mu\mathcal{N})},\quad Z_0=\tr e^{-\beta(\mathcal{K}-\mu\mathcal{N})}.
\end{equation}
$G_0$ is translation-invariant and preserves particle number. Similar to (\ref{P^L[G]}), we have
\begin{equation}\label{P^L_0[G_0]}
  P^L_0(\beta,\mu)=P^L_0[G_0]=\frac{1}{L^3\beta}\ln Z_0.
\end{equation}
Let $\gamma_0$ be the one-particle density matrix to $G_0$. From \cite[(8.65)]{huang2008statistical}, we can calculate the diagonal Fourier coefficients (\ref{gamma_k,sigma}) to $\gamma_0$ by
\begin{equation}\label{gamma_0}
  \hat{\gamma}_0(k,\sigma)=\frac{1}{1+\mathfrak{z}^{-1}e^{\beta\vert k\vert^2}},
\end{equation}
where $\mathfrak{z}=e^{\beta\mu}$ is the fugacity. From \cite[(8.63)]{huang2008statistical}, we can write the limit pressure $P_0(\beta,\mu)$ by
\begin{equation}\label{pressure non_interacting}
  P_0(\beta,\mu)=\frac{\mathbf{q}}{(2\pi)^3\beta}\int_{\mathbb{R}^3}
  \ln\big(1+\mathfrak{z}e^{-\beta\vert x\vert^2}\big)dx.
\end{equation}
We can calculate directly the corresponding density $\rho_0=\partial P_0/\partial\mu$
\begin{equation}\label{density non interacting}
  \rho_0(\beta,\mu)=\frac{\mathbf{q}}{(2\pi)^3}\int_{\mathbb{R}^3}
  \frac{dx}{1+\mathfrak{z}^{-1}e^{\beta\vert x\vert^2}},
\end{equation}
or we can write it by (see \cite[25.12]{NIST})
\begin{equation}\label{density non intercating Fermi Dirac dis}
  \rho_{0}(\beta,\mu)=-\frac{\mathbf{q}}{(4\pi\beta)^{\frac{3}{2}}}L_{\frac{3}{2}}(-\mathfrak{z})
  =\frac{\mathbf{q}}{(4\pi\beta)^{\frac{3}{2}}}F_{\frac{1}{2}}(\beta\mu),
\end{equation}
where $L_s$ is the polylogarithm, and $F_s$ is the Fermi-Dirac integral. When $\beta\mu\gg1$ (see for example \cite{dingle1973asymptotic}), we have
\begin{equation}\label{asymptotic rho_0}
  \rho_{0}(\beta,\mu)=\frac{\mathbf{q}}{6\pi^2\beta^{\frac{3}{2}}}
  \Big((\beta\mu)^{\frac{3}{2}}+O\big((\beta\mu)^{-\frac{1}{2}}\big)\Big).
\end{equation}

\par We are interested in the low temperature limit $\beta\mu\gg1$ and low density limit, which means the density $\rho$ is defined by $\rho=\partial P/\partial\mu$ should be small enough. In fact, $\rho$ may not exist point-wisely. However, since $P(\beta,\mu)$ is convex in $\mu$, at least the right and left derivatives of $P$ with respect of $\mu$, i.e. $\rho_\pm$ exist. It has been proven in \cite[Corollary 1]{FermithermoTpositive} that, when $\rho_0$ is small enough, $\rho_0\sim\rho_\pm$. Therefore, we might assume $\rho_0$ small enough instead.

\par For Huang-Yang formula to hold in the positive temperature, we need to further assume
\begin{equation}\label{restriction on betamu}
  \beta\mu\gtrsim\rho_0^{-\alpha_1}
\end{equation}
for some $\alpha_1\geq0$. We demand $\alpha_1>\frac{1}{6}$. This condition is critical for Huang-Yang formula for $T>0$. In fact, we can check this fact by combining (\ref{asymptotic rho_0}) and (\ref{restriction on betamu}), we have $\rho_0\sim \mu^{\frac{3}{2}}$, and moreover, we expect
\begin{equation}\label{critical power 1/6}
  \rho_0^{\frac{7}{3}}\gg\rho_0^{2+2\alpha_1}\gtrsim
  \rho_0\big(\rho_0-\mathbf{q}(6\pi^2)^{-1}\mu^{\frac{3}{2}}\big).
\end{equation}
Therefore, we obtain the temperature threshold for Huang-Yang formula: $\alpha_1>\frac{1}{6}$, or in other words, $T\lesssim\rho_0^{\frac{2}{3}+\frac{1}{6}+}$ since $\rho_0^{\frac{2}{3}}\sim\mu$.

\par For $R>0$, we denote the number of lattice points in $\{k\in(2\pi/L)\mathbb{Z}^3\,\vert\,\vert k\vert\leq R\}$ by $N(R)$:
\begin{equation}\label{N(R)}
  N(R)=\#\{k\in(2\pi/L)\mathbb{Z}^3\,\vert\,\vert k\vert\leq R\}.
\end{equation}
When $RL$ is large enough, the asymptotic behavior of $N(R)$ is given by (see for example \cite{HeathBrown+1999+883+892})
\begin{equation}\label{asy N(R)}
  N(R)=\frac{R^3L^3}{6\pi^2}+O\big((RL)^{\frac{21}{16}+\varepsilon}\big)
\end{equation}
for any $\varepsilon>0$ small enough when $RL\to\infty$.

\par Let $\kappa>0$ be a universal constant, independent of $L$, $\beta$ and $\mu$. The exact value of $\kappa$ will be determined later in (\ref{choose kappa}). We choose $k_F=\tilde{\mu}^{\frac{1}{2}}>0$, such that
\begin{equation}\label{mu tilde}
  \tilde{\mu}=\mu-\kappa\tilde{\mu}^{\frac{3}{2}}.
\end{equation}
Obviously there is a unique $\tilde{\mu}>0$ that solves (\ref{mu tilde}) and additionally satisfies $\mu\sim\tilde{\mu}$ and $0<\mu-\tilde{\mu}<\kappa\mu^{\frac{3}{2}}$. We remark that this is in fact a retrospective definition of $\tilde{\mu}$, we only know this is the right number after the calculation in Section \ref{cub}.

\par We also let the Fermi ball $B_F=\{k\in(2\pi/L)\mathbb{Z}^3\,\vert\,\vert k\vert\leq k_F\}$ and by (\ref{asy N(R)}), we let
\begin{equation}\label{average particle number}
  \bar{N}_0=N(k_F)=\# B_F=\frac{\tilde{\mu}^{\frac{3}{2}}L^3}{6\pi^2}
  +O\Big((\tilde{\mu}^\frac{1}{2}L)^{\frac{21}{16}+\varepsilon}\Big)
\end{equation}
for any $\varepsilon>0$ small enough when $k_FL\to\infty$.

\par In the following context, we apply $\rho_0=\rho_0(\beta,\mu)$ and $\tilde{\rho}_0=
\rho_0(\beta,\tilde{\mu})$ for short. Similarly, the tilde version of (\ref{Gibbs state non-inter})-(\ref{pressure non_interacting}) can be defined by replacing $\mu$ by $\tilde{\mu}$. Also, we let $\tilde{\mathfrak{z}}=e^{\beta\tilde{\mu}}$.


\subsection{Main Theorem}

\begin{theorem}[Huang-Yang formula for $T>0$]
\label{core}
  Let $v$ be non-negative, radially symmetric and compactly supported.
  Assume further the potential $v$ is smooth enough. In the low density limit $\rho_0\to0$, and the low temperature condition given by (\ref{restriction on betamu}). Let $\alpha_1>\frac{1}{6}$, that is $T\lesssim\rho_0^{\frac{2}{3}+\frac{1}{6}+}$, which is the necessary threshold. We set $d=\min\{\alpha_1-\frac{1}{6},\frac{1}{9}\}$, then
  \begin{equation}\label{core pressure}
  \begin{aligned}
    P(\beta,\mu)=&P_0(\beta,\mu)
    -4\pi\mathfrak{a}_0\Big(1-\frac{1}{\mathbf{q}}\Big)\rho_0^2\\
    &-\frac{12}{35}(11-2\ln 2)3^{\frac{1}{3}}\pi^{\frac{2}{3}}\mathfrak{a}_0^2
    2^{\frac{4}{3}}\mathbf{q}^{-\frac{1}{3}}\Big(1-\frac{1}{\mathbf{q}}\Big)
    \rho_0^\frac{7}{3}\\
    &+\frac{1}{2}(8\pi\mathfrak{a}_0)^2\Big(1-\frac{1}{\mathbf{q}}\Big)^2\rho_0^2\frac{\partial \rho_0}
    {\partial\mu}+O\big(\rho_0^{\frac{7}{3}+\frac{d}{400}}\big),
    \end{aligned}
  \end{equation}
  and
  \begin{equation}\label{core density}
    \rho_{\pm}(\beta,\mu)=\rho_0(\beta,\mu)-
    8\pi\mathfrak{a}_0\Big(1-\frac{1}{\mathbf{q}}\Big)\rho_0\frac{\partial \rho_0}
    {\partial\mu}+O\big(\rho_0^{\frac{4}{3}+\frac{d}{800}}\big).
  \end{equation}
\end{theorem}

\begin{remark}
  \par (\ref{core density}) resolves the conjecture raised by R.Seiringer \cite[p.732]{FermithermoTpositive}. We prove (\ref{core density}) as a corollary of (\ref{core pressure}). In fact, if one does know (\ref{mu tilde}) together with (\ref{core density}), one can indeed sort of reach (\ref{core pressure}) by plugging (\ref{mu tilde}) and (\ref{core density}) into the zero temperature Huang-Yang formula (\ref{core ground state energy}) plus the free particle pressure $P_0(\beta,\mu)$. However, as we have mentioned before, (\ref{mu tilde}) is actually retrospectively defined after the calculations in Section \ref{cub}, at which point we have already known (\ref{core pressure}).
\end{remark}

\begin{remark}
  \par The fourth term on the right -hand side of (\ref{core density}) is in fact of the order $\rho_0^{\frac{7}{3}}$, since $\partial\rho_0/\partial\mu\sim\rho_0^{\frac{1}{3}}$. It acts as a positive temperature correction effect to the Huang-Yang formula, arises directly from the modification of $\mu$ in (\ref{mu tilde}). It carries the information of the deviation of the interacting Gibbs state versus the non-interacting Fermi-Dirac statistic due to the positive temperature. When we consider the ground state energy density for zero temperature (see (\ref{core ground state energy}) below), this term vanishes.
\end{remark}

\begin{remark}
  \par If the threshold $\alpha_1>\frac{1}{6}$ is removed in Theorem \ref{core}, the exact coefficients for each term of (\ref{core pressure}) might no longer correspond with the original Huang-Yang formula for zero temperature (see (\ref{core ground state energy}) below), and it can be predicted that the coefficients may become much more complicated, and even depend on the fugacity $\mathfrak{z}=e^{\beta\mu}$ and thus the temperature $T$, as presented in \cite[Remark 1.6]{spin1lower} for the spin-polarized case.
\end{remark}

  \par The method in this paper can also be applied to the zero temperature case, i.e. the proof of the Huang-Yang formula (\ref{no1}), with an modification to the Fermi radius $k_F$. Then we need to replace the operator $\mathcal{N}_{re}$ defined in (\ref{N_relaxation}) by $\mathcal{N}_{ex}$, the excitation particle number operator, and the effect of relative entropy $\beta^{-1}S(\Gamma,\tilde{G}_0)$ defined in (\ref{relative entropy}) is taken place by $\mathcal{K}_s$, the kinetic energy operator of excitations. Notice that there will be no correction (\ref{mu tilde}) of $\mu$ (when dealing with the ground state energy, we define $\mu^{\frac{1}{2}}=k_F$) since the particle number should by fixed, and we need not to consider the threshold $\alpha_1$ since the temperature is zero. We can deduce the following theorem parallel to Theorem \ref{core}.
  \begin{theorem}[Huang-Yang formula for $T=0$\footnote{Notice that, in \cite{giacomelli2025huangyangconjecturelowdensityfermi} which was posted within 24 hours with previous version of this paper, a proof using method analogous to the cubic renormalization from \cite{WJH} has also been provided by E.Giacomelli, C.Hainzl, P.T.Nam and R.Seiringer.}]
  \label{core T=0}
  Let $v$ be non-negative, radially symmetric and compactly supported.
  Assume further the potential $v$ is smooth enough. In the low density limit $\rho_0\to0$,
    \begin{equation}\label{core ground state energy}
  \begin{aligned}
    e(\rho_0)=&\frac{3}{5}\Big(\frac{6\pi^2}{\mathbf{q}}\Big)^{\frac{2}{3}}\rho_0^{\frac{5}{3}}
    +4\pi\mathfrak{a}_0\Big(1-\frac{1}{\mathbf{q}}\Big)\rho_0^2\\
    &+\frac{12}{35}(11-2\ln 2)3^{\frac{1}{3}}\pi^{\frac{2}{3}}\mathfrak{a}_0^2
    2^{\frac{4}{3}}\mathbf{q}^{-\frac{1}{3}}\Big(1-\frac{1}{\mathbf{q}}\Big)
    \rho_0^\frac{7}{3}+o\big(\rho_0^{\frac{7}{3}}\big).
    \end{aligned}
  \end{equation}
  \end{theorem}

\subsection{Outline of the Proof \& Accuracy of Littlewood-Paley}
\
\par Our proof of Theorem \ref{core} largely follows the steps in our work
proving Theorem \ref{core T=0} in the Gross-Pitaevskii regime \cite{WJH}
which develops the renormalization method from \cite{2018Bogoliubov,hainzlSchleinTriay2022bogoliubov}. That
is because we have found an original and structurally crucial 
Fermionic cubic renormalization\footnote{%
This is one of the reasons we wrote in \cite{WJH} that the proof of the zero
temperature case, Theorem \ref{core T=0}, is not too far away.} (inspired by
\cite{me}) in \cite{WJH}. The two challenges we tackle in this paper are the
positive temperature effects and the $L\to \infty $ large box limit.

\par Formula (\ref{core pressure}) has not been computed nor
predicted before this paper. Though (\ref{core density}) was anticipated, it actually comes out as a corollary of (\ref{core pressure}). Despite the clean form $%
\frac{1}{2}\left( 8\pi \mathfrak{a}_{0}\right) ^{2}\Big(1-\frac{1}{\mathbf{q}%
}\Big)^{2}\rho _{0}^{2}\frac{\partial \rho _{0}}{\partial \mu }$ we have
found, we did not know the format nor the whereabouts of this positive
temperature effects term at all beforehand. Needless to say, proving the
formula up to its scaling-threshold temperature is also very difficult and
complex. The only path forward is much more advanced and precise computation
and estimates.

\par As claimed in \cite{WJH}, since we can derive a-priori excitations estimates
to start in the thermodynamic limit, it is possible to calculate
the second order approximations without using localization of the volume,
but with the frequency localization as in \cite{WJH}.\footnote{%
Some other versions and ideas of frequency localization were also used in
\cite{dilutefermiBog,2024huangyangformulalowdensityfermi,
giacomelli2025huangyangconjecturelowdensityfermi}.} In the thermodynamic limit, to manage the length scale $L\to\infty$, we need much more delicate frequency cut-off in the choice of renormalization operators. We tackle these difficulties via a more nonlinear-PDE-like version of the
simple Littlewood-Paley frequency localization analysis, which was
originally devised to deal with octave band/out-of-band noise in Fourier analysis. In
the end, not only we achieve our goal to prove the Huang-Yang formula up to
the scaling-threshold temperature, we precisely identify all the freqency
relations and nonlinear interactions in all the renormalization components
and renormalization operators, and thus provide accurate and robust
information for future work.

\par To deal with the non-linearity of the entropy $S[\Gamma ]$\footnote{%
Such an entropy were also dealt with, by introducing the relative entropy in
the work \cite{spin1lower,spin1upper} on the thermodynamic pressure of
spin-polarized fermions.}, we are going to choose a suitable unitary operator $%
\mathcal{U}=e^{B}e^{B^{\prime }}e^{\tilde{B}}$ on the Fock space $\mathcal{F}
$, such that after the renormalizations, we have the precise formula up to the
scaling-threshold temperature that
\begin{equation*}
  \mathcal{U}^*(H-\mu\mathcal{N})\mathcal{U}=E_0+(\mathcal{K}-\tilde{\mu}\mathcal{N})+
  e^{-\tilde{B}}\mathcal{V}_{4,4h}e^{\tilde{B}}+\mathcal{E},
\end{equation*}
where $\mathcal{V}_{4,4h}\geq 0$ is the suitable all-high frequency (extreme
far away from the Fermi ball $B_{F}$) component of the interaction operator $%
\mathcal{V}$, and $\mathcal{E}$ is the residue, bounded by
\begin{equation*}
  \vert\tr\mathcal{E}\Gamma\vert\leq C\tilde{\rho}_0^{\frac{7}{3}+}L^3+c\big(
  \beta^{-1}S(\Gamma,\tilde{G}_0)+
  \tr e^{-\tilde{B}}\mathcal{V}_{4,4h}e^{\tilde{B}}\Gamma\big)
\end{equation*}
for any translation-invariant state $\Gamma$, and some constant $0<c<1$. Here, we observe $E_{0}$ and hence the clean form of the positive
temperature effects term after the conjugation of $e^{B}e^{B^{\prime }}$. Then
\begin{equation*}
  0\leq\frac{1}{\beta}S(\Gamma,\tilde{G}_0)\coloneqq
  L^3\big(\tilde{P}^L_0[\tilde{G}_0]-\tilde{P}^L_0[\Gamma]\big)
\end{equation*}
is the relative entropy with respect to the free Gibbs state $\tilde{G}_0$ defined in Section \ref{tem thre}. Therefore, by letting $\Gamma^\mathcal{U}=\mathcal{U}^*\Gamma\mathcal{U}$, we have (since $\tilde{P}_0[\tilde{G}_0]=\tilde{P}_0^L(\beta,\tilde{\mu})$)
\begin{equation*}
  -L^3P^L[\Gamma]=E_0-\tilde{P}_0^L(\beta,\tilde{\mu})+\beta^{-1}S(\Gamma^\mathcal{U},\tilde{G}_0)+
  \tr e^{-\tilde{B}}\mathcal{V}_{4,4h}e^{\tilde{B}}\Gamma^\mathcal{U}
  +\tr\mathcal{E}\Gamma^\mathcal{U}.
\end{equation*}
That is, by choosing $\Gamma =G$, we can deduce the lower bound for $%
-L^{3}P^{L}(\beta ,\mu )$, while we can choose $\Gamma ^{\mathcal{U}}=\tilde{%
G}_{0}$ to reach the upper bound. Thus the proof of Theorem \ref{core} is
concluded. The leftover is the choice and handling of the suitable unitary
operator $\mathcal{U}=e^{B}e^{B^{\prime }}e^{\tilde{B}}.$
\par In the thermodynamic limit, we make a careful analysis in the choice of the renormalization operators with fine frequency localizations to manage the length scale $L\to\infty$. These frequency localizations provide detailed
information to precisely identify the particle interactions that contains main
energy contribution. Based on our findings, it is not always the low frequency nor always the high frequency being the main information carrier.

\par For the first one, the quadratic renormalization, we
set
\begin{align*}
B=\frac{1}{2}\sum_{k,p,q,\sigma,\nu}
  \eta_k\phi^+(k)a^*_{p-k,\sigma}a^*_{q+k,\nu}a_{q,\nu}a_{p,\sigma}\chi_{p-k,q+k\notin
  B_F}\chi_{p,q\in B_F}-h.c.
\end{align*}
where $\{\eta_k\}$ are defined through the 3D s-wave scattering equation with Neumann boundary condition in Section \ref{scattering eqn sec}:
\begin{equation*}
  \left\{\begin{aligned}
  &(-\Delta_{x}+\frac{1}{2}v)f_\ell=\lambda_\ell f_\ell,\quad \vert
  x\vert\leq \ell L,\\
  &\left.\frac{\partial f_\ell}{\partial \mathbf{n}}\right\vert_{\vert
   x\vert=\ell L}=0,\quad \left.f_\ell\right\vert_{\vert
   x\vert=\ell L}=1.
  \end{aligned}\right.
\end{equation*}
with
\begin{equation}\label{lL1}
  \ell L=\tilde{\rho}_0^{-\frac{1}{3}+\alpha_3}.
\end{equation}
for some $0<\alpha_3<\frac{1}{3}$. $\phi^+$ is the smooth cut-off that discard frequencies below the frequency level $\tilde{\rho}_0^{\frac{1}{3}-\alpha_2}$ for some $\alpha_2>0$. The renormalization $e^{B}$
extracts the information of $\mathcal{V}_{21}$ as defined in (\ref{split V detailed}).
The $B$ operator is bosonic (particles grouped in pairs) and its similar
version has been used in many aforementioned work.

\par For the second one, the cubic renormalization, which is crucial to reaching the second order approximation, we let
\begin{equation*}
  B^\prime=\sum_{k,p,q,\sigma,\nu}
  \eta_k\phi^+(k)\zeta^-(k)a^*_{p-k,\sigma}a^*_{q+k,\nu}a_{q,\nu}a_{p,\sigma}\chi_{p-k,q+k\notin B_F}\chi_{q\in A_{F,\delta_4}}
  \chi_{p\in B_F}-h.c.
\end{equation*}
where
\begin{equation*}
  A_{F,\delta_4}=\{k\in(2\pi/L)\mathbb{Z}^3,\,k_F<\vert k\vert\leq k_F+\tilde{\mu}^{\frac{1}{2}}\tilde{\rho}_0^{\delta_4}\}
\end{equation*}
for $\frac{1}{3}>\delta_4>0$. Notice here $\zeta^-$ is also a smooth cut-off that discard frequencies higher than the frequency level $\tilde{\rho}_0^{\frac{1}{3}-\beta_1}$ for some $\beta_1>\frac{1}{3}$. In our
definition, $\phi ^{+}$ removes the \textquotedblleft
infrared\textquotedblright $\,$ frequency, and $\zeta ^{-}$ remove the
\textquotedblleft ultra-violet\textquotedblright $\,$ frequency, to avoid
both infrared and ultra-violet catastrophes cause by the large scale $L$. The
mathematical effect of $e^{B^{\prime }}$ is to exact the energy from the
relative low frequency part of $\mathcal{V}_{3}$ as defined in (\ref{split V detailed}%
).

\par Here, notice, in the definition of $B^{\prime }$ that, one of the four
particles is located inside the Fermi ball, and three are outside, while
one of the three outside ones is very close to the surface of the Fermi ball.
On the one hand, this operator is purely fermionic and cannot be bosonized.
On the other hand, it describes a very special type of energy or
correlation. We 1st produced this renormalization in \cite{WJH}.



\par Finally, the third one $e^{\tilde{B}}$ also known as a form of the
Bogoliubov transformation is defined by
\begin{equation*}
  \tilde{B}=\frac{1}{2}\sum_{k,p,q,\sigma,\nu}
  \xi_{k,q,p}^{\nu,\sigma}a^*_{p-k,\sigma}a^*_{q+k,\nu}a_{q,\nu}a_{p,\sigma}\chi_{p-k,q+k\notin B_F}\chi_{p,q\in B_F}-h.c.
\end{equation*}
with
\begin{equation*}
  \xi_{k,q,p}^{\nu,\sigma}=\frac{-\big(L^{-3}W_k\tilde{\zeta}^-(k)
  +\eta_kk(q-p)\phi^+(k)\tilde{\zeta}^-(k)\big)}
  {\frac{1}{2}\big(\vert q+k\vert^2+\vert p-k\vert^2-\vert q\vert^2-\vert p\vert^2\big)
  +\epsilon_0}
  \chi_{p-k,q+k\notin B_F}\chi_{p,q\in B_F}.
\end{equation*}
Here $W_k$ stands for the residue of the scattering cancellation, and $\eta_kk(q-p)$ is the kinetic residue excluded for fermions, that helps neutralize the second order to the correct form. We further more use the cut-off $\tilde{\zeta}^-$ to leave out frequencies higher than the frequency level $\tilde{\rho}_0^{\frac{1}{3}-\alpha_6}$ for some $\alpha_6>0$. $\epsilon_0$ is a small but positive (when $L$ tends to infinity) gap to avoid logarithmic growth in $L$, and we let $\epsilon_0=\tilde{\rho}_0^2$.\footnote{%
Similar choice of $\epsilon _{0}$ can also be seen in \cite%
{2024huangyangformulalowdensityfermi}.} The operator $e^{\tilde{B}}$ deals
with the low frequency part of $\mathcal{V}_{21}^{\prime }$ which is a residue of $%
e^{B}$ like usual, and also takes care of an energy at a suitable intermediate
frequency, defined in (\ref{define V_21' and Omega tilde}).

\par We layout this proof as follows. In Section \ref{coeff}, we define precisely the coefficients used in the choice of renormalizations, and provide useful estimates concerning these coefficients. In Section \ref{Fock space}, we give a brief introduction to the Fock space and collect useful tools and inequalities used in our proof. In Section \ref{a-prior}, we give some a-priori estimates concerning the relative entropy, and have an aforehand analysis to the Hamiltonian. In Section \ref{renormal}, we collect successively the result of the three renormalizations, and we will prove them in Sections \ref{qua}-\ref{bog}. In Section \ref{main}, we prove our main theorem, Theorem \ref{core}.

\section{Coefficients of Renormalizations}\label{coeff}
\subsection{Scattering Equation}\label{scattering eqn sec}
\
\par We choose
\begin{equation}\label{lL2}
  \ell L=\tilde{\rho}_0^{-\frac{1}{3}+\alpha_3}
\end{equation}
for some $0<\alpha_3<\frac{1}{3}$. Since we are considering the dilute limit, we can always assume $\ell L>C$ for some universal constant $C$. Consider the following ground state energy equation with Neumann boundary
condition for some parameter $\ell\in(0,\frac{1}{2})$
\begin{equation}\label{asymptotic energy pde on the ball}
  \left\{\begin{aligned}
  &(-\Delta_{x}+\frac{1}{2}v)f_\ell=\lambda_\ell f_\ell,\quad \vert
  x\vert\leq \ell L,\\
  &\left.\frac{\partial f_\ell}{\partial \mathbf{n}}\right\vert_{\vert
   x\vert=\ell L}=0,\quad \left.f_\ell\right\vert_{\vert
   x\vert=\ell L}=1.
  \end{aligned}\right.
\end{equation}
Equation (\ref{asymptotic energy pde on the ball}) has been thoroughly
analyzed by many people in many works, and one can consult
\cite{dy2006,2018Bogoliubov,boccatoBrenCena2020optimal} for details. We define
$w_{\ell}=1-f_\ell$, and make constant extensions to both of $f_\ell$
and $w_\ell$ outside of the 3D closed ball $\overline{B}_{\ell L}$ such
that $f_\ell\in H^2_{loc}(\mathbb{R}^3)$ and $w_\ell\in H^2(\mathbb{R}^3)$.

Regarding ${w}_\ell$ as a periodic function on the torus $\Lambda_L$, we
observe that it satisfies the equation
\begin{equation}\label{asymptotic energy pde on the torus}
  \Big(-\Delta_{x}+\frac{1}{2}v(x)\Big)
 {w}_\ell(x)
  =\frac{1}{2}v(x)
-\lambda_\ell\big(1-{w}_\ell(x)\big)
  \chi_{\ell L}(x),\quad x\in\Lambda_L.
\end{equation}
Here $\chi_{\ell L}$ is the characteristic function of the closed 3D ball
$\overline{B}_{\ell L}$, and we choose suitable $\ell\in(0,\frac{1}{2})$ so that
$\overline{B}_{\ell L}\subset\Lambda_{L}$. Standard elliptic equation theory grants the
uniqueness of solution to equation (\ref{asymptotic energy pde on the torus}). By Fourier
transform, (\ref{asymptotic energy pde on the torus}) is equivalent to its discrete version
 \begin{equation}\label{discrete asymptotic energy pde on the torus}
\begin{aligned}
\left\vert p\right\vert^2{w}_{\ell,p}+\frac{1}{2L^3}
  \sum_{q\in(2\pi/L)\mathbb{Z}^3}\hat{v}_{p-q}{w}_{\ell,q}
=\frac{1}{2L^{\frac{3}{2}}}
  \hat{v}_p+{\lambda_\ell}
  {w}_{\ell,p}-\frac{\lambda_\ell}{L^{\frac{3}{2}}}\widehat{\chi}_{\ell L}(p),
\end{aligned}
\end{equation}
where $p$ is an arbitrary 3D vector in $(2\pi/L)\mathbb{Z}^3$, $\hat{v}_p$ is defined by
\begin{equation}\label{define v_k scatt}
  \hat{v}_p=\int_{\mathbb{R}^3}v(x)e^{-ipx}dx,
\end{equation}
 and
\begin{equation*}
  {w}_{\ell,p}=\frac{1}{L^{\frac{3}{2}}}\int_{\Lambda_L}{w}_{\ell}(x)
  e^{-ipx}dx,\quad
\widehat{\chi}_{\ell L}
  (p)
=\int_{\Lambda_L}\chi_{\ell L}(x)
   e^{-ipx}dx.
\end{equation*}
The required properties of $f_\ell$ and $w_\ell$ are collected in the next lemma.

\begin{lemma}[\cite{me}, Lemma 3.1]
\label{fundamental est of v,w,lambda}
Let $v$ be a smooth, radially-symmetric, compactly supported and non-negative
function with scattering length $\mathfrak{a}_0$. Let $f_\ell$, $\lambda_\ell$ and $w_\ell$ be defined as above. Then for parameter $\ell\in
(0,\frac{1}{2})$ satisfying $\ell L>C$ for a large universal constant C, there exist some universal constants, also denoted as C, independent of $L$ and $\ell$, such that following estimates hold true for $\ell L$ large enough.
\begin{enumerate}[label=(\arabic*)]
  \item The asymptotic estimate of ground state energy $\lambda_\ell$ is
  \begin{equation}\label{est of lambda_l}
          \left\vert\lambda_\ell-\frac{3\mathfrak{a}_0}{(\ell L)^3}
          \left(1+\frac{9}{5}\frac{\mathfrak{a}_0}{\ell L}\right)\right\vert\leq
          \frac{C\mathfrak{a}_0^3}{(\ell L)^5}.
  \end{equation}
  \item $f_\ell$ is radially symmetric and smooth away from the boundary of
      $B_{{\ell L}}$ and there is a certain constant $0<c<1$ independent of
      $a$ and $\ell$ such that
  \begin{equation}\label{est of f_l}
    0<c\leq f_\ell(x)\leq1.
  \end{equation}
  Moreover, for any integer $0\leq k\leq3$
  \begin{equation}\label{est of w_l and grad w_l}
        \vert D_{x}^kw_\ell(x)\vert\leq\frac{C}{1+\vert
         x\vert^{k+1}}.
  \end{equation}
  \item We have
  \begin{equation}\label{est of int vf_l}
          \left\vert\int_{\mathbb{R}^3}v(x)f_\ell(x)dx
          -8\pi\mathfrak{a}_0\left(1+\frac{3}{2}
          \frac{\mathfrak{a}_0}{\ell L}\right)\right\vert\leq
          \frac{C\mathfrak{a}_0^3}{(\ell L)^2},
  \end{equation}
  and
        \begin{equation}\label{est of int w_l}
          \left\vert\frac{1}{(\ell L)^2}\int_{\mathbb{R}^3}
          w_\ell(x)dx-\frac{2}{5}\pi\mathfrak{a}_0\right\vert\leq
          \frac{C\mathfrak{a}_0^2}{\ell L}.
        \end{equation}
  \item For all $p\in(2\pi/L)\mathbb{Z}^3\backslash\{0\}$
  \begin{equation}\label{est of w_l,p}
          \vert{w}_{\ell,p}\vert\leq\frac{C}{L^{\frac{3}{2}}\vert p\vert^2}
  \end{equation}
\end{enumerate}
\end{lemma}
\begin{remark}\label{remark 3d scat eqn}
The construction of $w_\ell$ can not ensure smoothness on the boundary of
$B_{{\ell}{L}}$, but we still use the notation $D_{x}^kw_\ell$ to
represent the $k$-th derivative of $w_\ell$ away from the boundary of
$B_{{\ell}{L}}$. Moreover, since $w_\ell$ is supported on
$B_{{\ell}{L}}$, the integrals concerning $D^k_{x}w_\ell$ always
mean integrating inside of $B_{{\ell}{L}}$ unless otherwise specified.
\end{remark}

\par With Lemma \ref{fundamental est of v,w,lambda}, we thereafter define for all
$p\in(2\pi/L^3)\mathbb{Z}^3$
\begin{equation}\label{eta_p}
  \eta_p= -\frac{{w}_{\ell,p}}{L^{\frac{3}{2}}}.
\end{equation}
(\ref{discrete asymptotic energy pde on the torus}) then reads
\begin{equation}\label{eqn of eta_p}
\begin{aligned}
  \left\vert p\right\vert^2\eta_p+\frac{1}{2L^3}
  \sum_{q\in(2\pi/L)\mathbb{Z}^3}\hat{v}_{p-q}\eta_q
=-\frac{1}{2L^3}
  \hat{v}_p+{\lambda_\ell}
 \eta_p+\frac{\lambda_\ell}{L^3}\widehat{\chi}_{\ell L}(p).
\end{aligned}
\end{equation}
Since $w_\ell$ is real-valued and radially symmetric, we have
$\eta_p=\eta_{-p}=\overline{\eta_p}$. Moreover, we let $\eta\in L^2(\Lambda_L)$ be the periodic function with Fourier coefficients $\eta_p$, that is
\begin{equation}\label{eta(x)}
  \eta(x)=\frac{1}{L^{\frac{3}{2}}}\sum_{p\in (2\pi/L)\mathbb{Z}^3}\eta_p
  e^{ip\cdot x}=-\frac{1}{L^{\frac{3}{2}}}w_{\ell}(x).
\end{equation}
Then with (\ref{est of w_l and grad w_l}), we deduce
\begin{equation}\label{est of eta and eta_perp}
\begin{aligned}
  \Vert\eta\Vert_2^2=\frac{1}{L^3}
  \int_{\vert x\vert\leq{\ell L}}\vert
  w_{\ell}(x)\vert^2dx
\leq
  \frac{1}{L^3}\int_{\vert x\vert\leq{\ell L}}\frac{C}{\vert
  x\vert^2} dx
  =C\ell L^{-2}.
\end{aligned}
\end{equation}
We can bound $\eta_p$ for all $p\in(2\pi/L)\mathbb{Z}^3$ in the same way
\begin{equation}\label{est of eta_0}
  \vert\eta_p\vert\leq\frac{1}{L^{3}}\int_{\vert x\vert\leq
  \ell L}{w}_{\ell}(x)dx
  \leq\frac{1}{L^{3}}\int_{\vert x\vert\leq
  \ell L}\frac{C}{\vert
  x\vert} dx
  \leq C\ell^2L^{-1},
\end{equation}
and similarly,
\begin{equation}\label{est of eta L1}
  \Vert\eta\Vert_1\leq C\ell^2L^{\frac{1}{2}}.
\end{equation}
Moreover,
\begin{equation}\label{est of Deta D^2eta}
  \Vert\nabla\eta\Vert_1\leq C\ell L^{-\frac{1}{2}},\quad
  \Vert\nabla\eta\Vert_2,\Vert\Delta\eta\Vert_2,\Vert\Delta\eta\Vert_1\leq CL^{-\frac{3}{2}}.
\end{equation}
Noticing that $f_{\ell}=1-w_\ell$, we deduce, via Plancherel's equality,
\begin{equation}\label{discrete int vf}
  \int_{\mathbb{R}^3}v(x)f_{\ell}(x)dx=\hat{v}_0+\sum_{p\in(2\pi/L)\mathbb{Z}^3}\hat{v}_p\eta_p.
\end{equation}

\par We also introduce the notation for $x\in\Lambda_L$
\begin{equation}\label{define W}
  W(x)=\lambda_\ell\chi_{\ell L}(x)\big(1-w_\ell(x)\big),
\end{equation}
and for $p\in(2\pi/L)\mathbb{Z}^3$
\begin{equation}\label{define W_p}
  W_p=\lambda_\ell\big(\hat{\chi}_{\ell L}(p)+L^3\eta_p\big),
\end{equation}
such that
\begin{equation}\label{relation W and W_p}
  W_p=\int_{\mathbb{R}^3}W(x)e^{-ipx}dx.
\end{equation}
Combining (\ref{define W_p}) with (\ref{eqn of eta_p}), we find that
\begin{equation}\label{Wdiscrete asymptotic energy pde on the torus}
\begin{aligned}
  \left\vert p\right\vert^2\eta_p+\frac{1}{2L^3}
  \sum_{q\in(2\pi/L)\mathbb{Z}^3}\hat{v}_{p-q}\eta_q
=-\frac{1}{2L^3}
  \hat{v}_p+\frac{1}{L^3}W_p.
\end{aligned}
\end{equation}
Since $\hat{\chi}_{\ell L}(p)\leq C\ell L\vert p\vert^{-2}$, then combining this estimate with (\ref{est of lambda_l}), (\ref{est of w_l,p}) and (\ref{eta_p}), we have
\begin{equation}\label{est of W_p p^2}
  \vert W_p\vert\leq\frac{C}{(\ell L)^2\vert p\vert^2}.
\end{equation}
It is easy to check that
\begin{equation}\label{L1 L2 est W}
  \vert W_p\vert,\Vert W\Vert_1\leq C,\quad \Vert W\Vert_2\leq C(\ell L)^{-\frac{3}{2}}.
\end{equation}

\par We introduce a  radially-symmetric cut-off function $\vartheta\in C^\infty_c(\mathbb{R}^3;\mathbb{R})$, such that $0\leq\vartheta\leq1$, and
\begin{equation}\label{cutoff theta}
  \vartheta(x)=\left\{
  \begin{aligned}
  &1,\quad\vert x\vert\leq \frac{1}{2}\tilde{\rho}_0^{-\frac{1}{3}}\\
  &0,\quad\vert x\vert\geq \tilde{\rho}_0^{-\frac{1}{3}}
  \end{aligned}\right.
\end{equation}
Since $W$ is supported inside the closed ball $B_{\ell L}$, with
$\ell L=\tilde{\rho}_0^{-\frac{1}{3}+\alpha_3}$ for $\alpha_3>0$, we know that
$W(x)=W(x)\vartheta(x)$. If we set
\begin{equation}\label{vartheta_k}
  \vartheta_k=\frac{1}{L^{\frac{3}{2}}}\int_{\Lambda_L}\vartheta(x)e^{-ikx}dx.
\end{equation}
Then for any $k\in(2\pi/L)\mathbb{Z}^3$
\begin{equation}\label{Wtheta}
  W_k=\frac{1}{L^{\frac{3}{2}}}\sum_{m\in(2\pi/L)\mathbb{Z}^3}W_{k-m}\vartheta_m.
\end{equation}
We can always require
\begin{equation}\label{theta norm}
  \begin{aligned}
  &\Vert\vartheta\Vert_\infty\leq 1,\quad\Vert\vartheta\Vert_2\leq C\tilde{\rho}_0^
  {-\frac{1}{2}},\quad\Vert\vartheta\Vert_1\leq C\tilde{\rho}_0^{-1}\\
  &\Vert\nabla\vartheta\Vert_\infty\leq C\tilde{\rho}_0^{\frac{1}{3}},\quad
  \Vert\nabla\vartheta\Vert_2\leq C\tilde{\rho}_0^{-\frac{1}{6}}
  \end{aligned}
\end{equation}
We define for $x\in\Lambda_L$
\begin{equation}\label{define U}
  U(x)=\frac{1}{L^3}\sum_{k\neq0}\frac{W_kk}{\vert k\vert^2}e^{ikx}.
\end{equation}
Then using (\ref{est of W_p p^2}) and (\ref{L1 L2 est W}), we deduce
\begin{equation}\label{U L2}
  \Vert U\Vert_2\leq\Big( \frac{1}{L^3}\sum_{k\neq0}\frac{\vert W_k\vert^2}{\vert k\vert^2}
  \Big)^{\frac{1}{2}}
  \leq C\tilde{\rho}_0^{\frac{1}{6}-\frac{\alpha_3}{2}}.
\end{equation}

\subsection{Smooth Cut-off Functions}\label{smooth cutoff sec}
\
\par We introduce the radially-symmetric cut-off function $\phi^-\in C^\infty_c(\mathbb{R}^3;\mathbb{R})$, such that $0\leq\phi^-\leq1$, and for some constant $\alpha_2>0$ to be determined,
\begin{equation}\label{cutoff}
  \phi^-(p)=\left\{
  \begin{aligned}
  &1,\quad \vert p\vert<\frac{3}{2}\tilde{\mu}^{\frac{1}{2}}\tilde{\rho}_0^{-\alpha_2}\\
  &0,\quad \vert p\vert>2\tilde{\mu}^{\frac{1}{2}}\tilde{\rho}_0^{-\alpha_2}
  \end{aligned}
  \right.
\end{equation}
We can always assume for integer $k\geq0$,
\begin{equation}\label{derivatives bound cutoff phi^-}
  \big\vert\partial^{k}_r\phi^-(r)\big\vert\leq C
  \tilde{\rho}_0^{-k(\frac{1}{3}-\alpha_2)}
\end{equation}
for some universal constant $C$. We also let $\phi^+=1-\phi^-$.

\par For $x\in\mathbb{R}^3$, we let
\begin{equation}\label{define f}
  f(x)=\int_{\mathbb{R}^3}\phi^-(2\pi y)e^{2\pi ixy}dy=\frac{1}{(2\pi)^3}
  \big(\mathfrak{F}^{-1}\phi^-\big)\Big(\frac{x}{2\pi}\Big),
\end{equation}
where the Fourier transform $\mathfrak{F}$ is given by
\begin{equation}\label{Fourier}
  \mathfrak{F}f(\xi)=\int_{\mathbb{R}^3}f(x)e^{-2\pi x\xi}dx.
\end{equation}
It is easy to check that for $p\in (2\pi/L)\mathbb{Z}^3$, $\mathfrak{F}f\big(p/(2\pi)\big)
=\phi^-(p)$. Moreover, since $\phi^-$ is compactly supported, we have for any positive integer $n$,
\begin{equation}\label{bound f}
  \vert f(x)\vert\leq\frac{C(n,\alpha_2)\tilde{\rho}_0^d}{1+\vert x\vert^n}
\end{equation}
for some constant $C(n,\alpha_2)$ depends only on $n$ and $\alpha_2$ and $d\in\mathbb{R}$ is some power of $\tilde{\rho}_0$. In particular, since
\begin{equation}\label{fourier radius}
  \mathfrak{F}\phi^-(\xi)=\frac{2\pi}{\vert\xi\vert}\int_{0}^{\infty}\phi^-(r)\sin\big(
  2\pi r\vert\xi\vert\big)rdr,
\end{equation}
by (\ref{derivatives bound cutoff phi^-}), (\ref{define f}) and integration by parts, it can be verified that for $x\neq0$,
\begin{equation}\label{bound f detailed}
  \vert f(x)\vert\leq\frac{C\tilde{\rho}_0^{\frac{1}{3}-\alpha_2}}{\vert x\vert^2},
  \quad \vert f(x)\vert\leq\frac{C\tilde{\rho}_0^{\alpha_2-\frac{1}{3}}}{\vert x\vert^4}
\end{equation}
for some universal constant $C$. Therefore, it is east to find
\begin{equation}\label{f L1}
  \Vert f\Vert_{L^1(\mathbb{R}^3)}\leq C.
\end{equation}
We also let
\begin{equation}\label{defind F}
  F(x)=\sum_{k\in L\mathbb{Z}^3}f(x+k).
\end{equation}
Obviously $F$ is a smooth periodic function on $\Lambda_L$. It is known that
\begin{equation}\label{cal F}
  F(x)=\frac{1}{L^3}\sum_{p\in (2\pi/L)\mathbb{Z}^3}\phi^-(p)e^{ipx},
\end{equation}
and we can using (\ref{f L1}) to bound 
begin{equation}\label{L1 F}
\begin{equation}
\begin{aligned}
  \Vert F\Vert_{L^1(\Lambda_L)}\leq 
  \Vert f\Vert_{L^1(\mathbb{R}^3)}\leq C.
\end{aligned}
\end{equation}

\par We choose for $x\in\Lambda_L$,
\begin{equation}\label{etatilde}
  {\eta}_{\phi^+}(x) =\frac{1}{L^{\frac{3}{2}}}\sum_{p\in (2\pi/L)\mathbb{Z}^3}\eta_p\phi^+(p)
  e^{ipx}.
\end{equation}
Therefore,
\begin{equation}\label{eta-etatilde}
 {\eta}^{\phi^-}\coloneqq \eta(x)-{\eta}_{\phi^+}(x)=(\eta\ast F) (x).
\end{equation}
By Young's inequality, we know that estimates (\ref{est of eta and eta_perp}), (\ref{est of eta L1}) and (\ref{est of Deta D^2eta}) also hold when replacing $\eta$ by ${\eta}_{\phi^+}$ or ${\eta}^{\phi^-}$.
Moreover, using (\ref{est of eta_0}), we have the following more subtle bound
\begin{equation}\label{bound eta-etatilde Linfty}
  \Vert\eta^{\phi^-}\Vert_{\infty}\leq
  L^{-\frac{3}{2}}
  \sum_{\vert p\vert\leq \tilde{\rho}_0^{1/3-\alpha_2}}\vert\eta_p\vert\leq CL^{-\frac{3}{2}}
  \tilde{\rho}_0^{\frac{1}{3}-3\alpha_2+2\alpha_3},
\end{equation}
provided $\ell L=\tilde{\rho}_{0}^{-\frac{1}{3}+\alpha_3}$ for some $0<\alpha_3<\frac{1}{3}$.

\par The above analysis can be extended to the $\alpha_2\leq0$ case, and thus can be applied to a more general case:
\begin{lemma}\label{cutoff lemma}
  Let $\phi_1^-,\phi_2^-\in C^{\infty}_c(\mathbb{R}^3;[0,1])$ be two smooth radially-symmetric cut-off functions, such that for some universal constant $C_j>c_j>0$, $j=1,2$ and $+\infty\geq\kappa_1>\kappa_2\geq-\infty$, they satisfy
  \begin{equation}\label{cutofflemmaphi1and2}
  \phi^-_1(p)=\left\{
  \begin{aligned}
  &1,\quad \vert p\vert<c_1\tilde{\rho}_0^{\kappa_1}\\
  &0,\quad \vert p\vert>C_1\tilde{\rho}_0^{\kappa_1}
  \end{aligned}
  \right.\quad
  \phi^-_2(p)=\left\{
  \begin{aligned}
  &1,\quad \vert p\vert<c_2\tilde{\rho}_0^{\kappa_2}\\
  &0,\quad \vert p\vert>C_2\tilde{\rho}_0^{\kappa_2}
  \end{aligned}
  \right.
\end{equation}
We also let $\phi_j^+=1-\phi_j^-$ for $j=1,2$. When $\kappa_1=+\infty$ we let $\phi^-_1\equiv0$, and when $\kappa_2=-\infty$ we let $\phi_2^-\equiv 1$. if $+\infty\geq\kappa_1>\kappa_2>-\infty$, we let
\begin{equation}\label{Flemma}
  F_{\phi_1^+}^{\phi_2^-}(x)=\frac{1}{L^3}\sum_{p\in(2\pi/L)\mathbb{Z}}
  \phi^+_1(p)\phi^-_2(p)e^{ipx}.
\end{equation}
Then when $\tilde{\rho}_0$ is small enough (independent of $L$), and $L$ is large enough (depending on $\tilde{\rho}_0$),
\begin{equation}\label{L1norm Flemma}
  \Vert F_{\phi_1^+}^{\phi_2^-}\Vert_1\leq C.
\end{equation}
Moreover, for any $\psi\in L^2(\Lambda_L)$, such that
\begin{equation*}
  \psi_p=\frac{1}{L^{\frac{3}{2}}}\int_{\Lambda_L}\psi(x)e^{-ipx}dx,
\end{equation*}
 and $+\infty\geq\kappa_1>\kappa_2\geq-\infty$, we take
\begin{equation}\label{psilemma}
  \psi_{\phi_1^+}^{\phi_2^-}=\frac{1}{L^{\frac{3}{2}}}\sum_p \psi_p\phi^+_1(p)\phi^-_2(p)e^{ipx},
\end{equation}
then we have for $1\leq m\leq\infty$,
\begin{equation}\label{psi lemma L1 and L2 norm}
  \Vert \psi_{\phi_1^+}^{\phi_2^-}\Vert_{L^m}\leq C\Vert \psi\Vert_{L^m}.
\end{equation}
In particular, we let for $t\in(0,\infty)$:
\begin{equation}\label{E^tlemma}
  E^t(x)=\frac{1}{L^3}\sum_{p\in(2\pi/L)\mathbb{Z}}
  \phi^+_1(p)\phi^-_2(p)e^{-p^2t}e^{ipx}.
\end{equation}
we have
\begin{equation}\label{E^tbound}
  \sup_{t\in(0,\infty)}\Vert E^t\Vert_{L^1(\Lambda_L)}\leq C t^{-\frac{3}{2}}\Vert 
  \exp(-x^2/ct)\Vert_{L^1_x(\mathbb{R}^3)}\leq C.
\end{equation}
\end{lemma}
\begin{notation}
  For notations like (\ref{psilemma}), they will be repeatedly used throughout our analysis without further specification.
\end{notation}

\par Similar to (\ref{cutoff}), we also introduce another two radially-symmetric cut-off functions $\zeta^-,\tilde{\zeta}^-\in C^\infty_c(\mathbb{R}^3;\mathbb{R})$, such that $0\leq\zeta^-,\tilde{\zeta}^-\leq1$, and for some constant $\beta_1=\frac{1}{3}+\alpha_5>\frac{1}{3}$ and $\alpha_6>2\alpha_3$ to be determined later,
\begin{equation}\label{cutoffzeta}
  \zeta^-(p)=\left\{
  \begin{aligned}
  &1,\quad \vert p\vert<\frac{3}{2}\tilde{\mu}^{\frac{1}{2}}\tilde{\rho}_0^{-\beta_1}\\
  &0,\quad \vert p\vert>2\tilde{\mu}^{\frac{1}{2}}\tilde{\rho}_0^{-\beta_1}
  \end{aligned}
  \right.\quad\quad
  \tilde{\zeta}^-(p)=\left\{
  \begin{aligned}
  &1,\quad \vert p\vert<\frac{3}{2}\tilde{\mu}^{\frac{1}{2}}\tilde{\rho}_0^{-\alpha_6}\\
  &0,\quad \vert p\vert>2\tilde{\mu}^{\frac{1}{2}}\tilde{\rho}_0^{-\alpha_6}
  \end{aligned}
  \right.
\end{equation}
We also let $\zeta^+=1-\zeta^-$ and $\tilde{\zeta}^+=1-\tilde{\zeta}^-$. In particular, using (\ref{est of w_l,p}) and (\ref{eta_p}), we can bound
\begin{equation}\label{special eta L2 norm}
  \Vert\eta_{\zeta^+}\Vert^2_2\leq C\sum_{k}\frac{\vert\zeta^+(k)\vert^2}{L^6\vert k\vert^4}
  \leq CL^{-3}\tilde{\rho}_0^{\alpha_5},
\end{equation}
given that $\beta_1=\frac{1}{3}+\alpha_5>\frac{1}{3}$. Similarly, by (\ref{est of W_p p^2}), we have
\begin{equation}\label{special W L2 norm}
  \Vert W_{\zeta^+}\Vert^2_2\leq C(\ell L)^{-4}\tilde{\rho}_0^{\alpha_5}.
\end{equation}
Also, we have
\begin{equation}\label{special zeta tilde eta W L2 norm}
  \Vert\eta_{\tilde{\zeta}^+}\Vert^2\leq CL^{-3}\tilde{\rho}_0^{\alpha_6-\frac{1}{3}},
  \quad
  \Vert W_{\tilde{\zeta}^+}\Vert^2_2\leq C(\ell L)^{-4}\tilde{\rho}_0^{\alpha_6-\frac{1}{3}}
\end{equation}
Besides, we have the following useful estimates:
\begin{lemma}\label{eta sum lemma}
  Let $\eta_k$ be defined in (\ref{eta_p}), we have
  \begin{equation}\label{eta sum}
    \sum_{k}\vert\eta_k\vert\leq C.
  \end{equation}
  For $v$ smooth enough, we also have
  \begin{equation}\label{eta cutoff sum}
    \sum_{k}\vert\eta_k\zeta^+(k)\vert\leq C\tilde{\rho}_0^{\frac{2}{3}-2\alpha_3+\alpha_5}.
  \end{equation}
  Therefore
  \begin{equation}\label{special grad eta L2}
    \Vert\nabla\eta_{\zeta^+}\Vert_2\leq CL^{-\frac{3}{2}}\tilde{\rho}_0^{\frac{1}{3}-\alpha_3+\frac{\alpha_5}{2}},
  \end{equation}
  and
  \begin{equation}\label{special eta Linf}
    \Vert\eta_{\zeta^+}\Vert_{\infty}\leq CL^{-\frac{3}{2}}\tilde{\rho}_0^
    {\frac{2}{3}-2\alpha_3+\alpha_5}.
  \end{equation}
\end{lemma}
\begin{proof}
  \par We use (\ref{Wdiscrete asymptotic energy pde on the torus}) to rewrite for $k\in(2\pi/L)\mathbb{Z}^3\backslash\{0\}$:
  \begin{equation}\label{rewrite eta_k}
    \eta_k=\frac{1}{L^3\vert k\vert^2}W_k-\frac{1}{2L^3\vert k\vert^2}\Big(
    \hat{v}_k+\sum_{l}\hat{v}_{k-l}\eta_l\Big).
  \end{equation}
  Using (\ref{est of W_p p^2}) and (\ref{L1 L2 est W}), we have
  \begin{equation}\label{W_k/k^2 sum}
    \sum_{k\neq0}\frac{W_k}{L^3\vert k\vert^2}\leq C\tilde{\rho}_0^{\frac{1}{3}-\alpha_3},
  \end{equation}
  for $0<\alpha_3<\frac{1}{3}$. On the other hand, since
  \begin{equation}\label{Z sum function}
    \frac{1}{L^3}\sum_{k}\Big(
    \hat{v}_k+\sum_{l}\hat{v}_{k-l}\eta_l\Big)e^{ikx}=v(x)f_\ell(x),
  \end{equation}
  we have
  \begin{equation}\label{L1 L2 Z}
    \Big\vert\hat{v}_k+\sum_{l}\hat{v}_{k-l}\eta_l\Big\vert\leq C,\quad
    \frac{1}{L^3}\sum_k\Big\vert\hat{v}_k+\sum_{l}\hat{v}_{k-l}\eta_l\Big\vert^2\leq C.
  \end{equation}
  Therefore, for some $d\in\mathbb{R}$,
  \begin{equation}\label{uidsjfhd}
    \sum_{0<\vert k\vert\leq\tilde{\rho}_0^{d}} \frac{1}{L^3\vert k\vert^2}\Big\vert\hat{v}_k+\sum_{l}\hat{v}_{k-l}\eta_l\Big\vert\leq C\tilde{\rho}_0^{d},
    \quad
    \sum_{\vert k\vert>\tilde{\rho}_0^{d}} \frac{1}{L^3\vert k\vert^2}\Big\vert\hat{v}_k+\sum_{l}\hat{v}_{k-l}\eta_l\Big\vert\leq C\tilde{\rho}_0^{-\frac{d}{2}}.
  \end{equation}
  By choosing $d=0$, we get to
  \begin{equation}\label{Z sum}
    \sum_{k\neq0} \frac{1}{L^3\vert k\vert^2}\Big\vert\hat{v}_k+\sum_{l}\hat{v}_{k-l}\eta_l\Big\vert\leq C.
  \end{equation}
  Combining (\ref{est of eta_0}), (\ref{rewrite eta_k}), (\ref{W_k/k^2 sum}) and (\ref{Z sum}) we reach (\ref{eta sum}).
  
  \par For (\ref{eta cutoff sum}), we notice that using (\ref{est of W_p p^2}), 
  \begin{equation}\label{W cutoff sum}
    \sum_{k}\frac{W_k\zeta^+(k)}{L^3\vert k\vert^2}\leq C\tilde{\rho}_0^{\frac{2}{3}-2\alpha_3+\alpha_5}.
  \end{equation}
  On the other hand, if $v$ is smooth enough, so is $vf_\ell$ inside of the support of $v$. Also, assume $v$ is supported in the closed ball $B_{R_0}$, then similar to the proof of (\ref{est of w_l and grad w_l}), we have for some large positive integer $n$, 
  \begin{equation}\label{est of D^n  f_l}
    \Vert D^{(n+1)}f_\ell\Vert_{L^\infty(B_{2R_0})}\leq C.
  \end{equation}
  Thus, by (\ref{Z sum function}), we have for some large positive integer $n$,
  \begin{equation}\label{dusijfnv}
    \vert k\vert^{(n+1)}\Big\vert\hat{v}_k+\sum_{l}\hat{v}_{k-l}\eta_l\Big\vert\leq C,
  \end{equation}
  which leads to
  \begin{equation}\label{Z cutoff sum}
    \sum_{k} \frac{\zeta^+(k)}{L^3\vert k\vert^2}\Big\vert\hat{v}_k+\sum_{l}\hat{v}_{k-l}\eta_l\Big\vert\leq C\tilde{\rho}_0^{n\alpha_5},
  \end{equation}
  where $\beta_1=\frac{1}{3}+\alpha_5>\frac{1}{3}$. Combining (\ref{rewrite eta_k}), (\ref{W cutoff sum}) and (\ref{Z cutoff sum}) we have (\ref{eta cutoff sum}).
  
  \par (\ref{special eta Linf}) is a direct consequence of (\ref{eta cutoff sum}), and (\ref{special grad eta L2}) is reached by combining (\ref{est of w_l,p}), (\ref{eta_p}) and (\ref{eta cutoff sum}).
\end{proof}

\par Finally, we choose another radially-symmetric cut-off function $\gamma^-\in C^\infty_c(\mathbb{R}^3;\mathbb{R})$, such that $0\leq\gamma^-\leq1$, and for some constant $\delta_3>\alpha_3/2$,
\begin{equation}\label{cutoffgamma}
  \gamma^-(p)=\left\{
  \begin{aligned}
  &1,\quad \vert p\vert<\frac{3}{2}\tilde{\mu}^{\frac{1}{2}}\tilde{\rho}_0^{-\delta_3}\\
  &0,\quad \vert p\vert>2\tilde{\mu}^{\frac{1}{2}}\tilde{\rho}_0^{-\delta_3}
  \end{aligned}
  \right.
\end{equation}
Similarly we let $\gamma^+=1-\gamma^-$.

\par In summary, $\phi^+$ is used in the quadratic and cubic renormalizations, explicitly defined in Sections \ref{qua sec} and \ref{cub sec}, $\zeta^-$ is used in the cubic renormalization given in Section \ref{cub sec}, and $\tilde{\zeta}^-$ is used in the definition of Bogoliubov transformation in Section \ref{bog sec}. $\gamma^-$ is used in the cut-off of $\mathcal{V}_3$ in Lemma \ref{lemma qua conj V_3}.

\subsection{Coefficients of the Bogoliubov Transformation}\label{Bog coeff sec}
\
\par Recall that we choose $\ell L=\tilde{\rho}_0^{-\frac{1}{3}+\alpha_3}$ for some $0<\alpha_3<\frac{1}{3}$. Using the discussion in Section \ref{scattering eqn sec}, the coefficients of the Bogoliubov transformation is defined by
\begin{equation}\label{define xi_k,q,p,nu,sigma0}
  \xi_{k,q,p}^{\nu,\sigma}=\frac{-\big(L^{-3}W_k\tilde{\zeta}^-(k)
  +\eta_kk(q-p)\phi^+(k)\tilde{\zeta}^-(k)\big)}
  {\frac{1}{2}\big(\vert q+k\vert^2+\vert p-k\vert^2-\vert q\vert^2-\vert p\vert^2\big)
  +\epsilon_0}
  \chi_{p-k,q+k\notin B_F}\chi_{p,q\in B_F}.
\end{equation}
For simplicity, we let $\xi_{k,q,p}^{\nu,\sigma}=\xi_{k,q,p}^{\nu,\sigma,1}+\xi_{k,q,p}^{\nu,\sigma,2}$, where
\begin{equation}\label{xi1and 20}
  \begin{aligned}
  &\xi_{k,q,p}^{\nu,\sigma,1}=\frac{-L^{-3}W_k\tilde{\zeta}^-(k)}
  {\frac{1}{2}\big(\vert q+k\vert^2+\vert p-k\vert^2-\vert q\vert^2-\vert p\vert^2\big)
  +\epsilon_0}
  \chi_{p-k,q+k\notin B_F}\chi_{p,q\in B_F}\\
  &\xi_{k,q,p}^{\nu,\sigma,2}=\frac{-\eta_kk(q-p)\phi^+(k)\tilde{\zeta}^-(k)}
  {\frac{1}{2}\big(\vert q+k\vert^2+\vert p-k\vert^2-\vert q\vert^2-\vert p\vert^2\big)
  +\epsilon_0}
  \chi_{p-k,q+k\notin B_F}\chi_{p,q\in B_F}.
  \end{aligned}
\end{equation}
Here $\epsilon_0>0$ is a small but positive (when $L$ tends to infinity) gap to avoid logarithmic growth in $L$\footnote{This idea can also be seen in the Bogoliubov transformation dealing with the fermionic ground state energy in \cite{2024huangyangformulalowdensityfermi}.}. We let $\epsilon_0=\tilde{\rho}_0^2$. The collisional energy certainly satisfies the usual hyper plane and parabola parametrization:
\begin{equation}\label{useful formula}
  \frac{1}{2}\big(\vert q+k\vert^2+\vert p-k\vert^2-\vert q\vert^2-\vert p\vert^2\big)
  =\vert k\vert^2+k\cdot(q-p)
  =\Big\vert k+\frac{q-p}{2}\Big\vert^2-\Big\vert\frac{q-p}{2}\Big\vert^2.
\end{equation}
We have the following useful estimates:
\begin{lemma}\label{est xi/ p^2 lemma}
  For $\epsilon_0=\tilde{\rho}_0^{\frac{2}{3}+\iota}$ for some $\iota>0$, and $p\notin B_F$ and $q\in B_F$, we have
  \begin{align}
  \frac{1}{L^3}\sum_{p\notin B_F}\frac{\vert W_{p-q}\vert}{p^2-k_F^2+\epsilon_0}
  \leq &C\tilde{\rho}_0^{\frac{1}{3}-\alpha_3},\label{est Wk/p^2}\\
  \frac{\tilde{\rho}_0^{\frac{1}{3}}}{L^3}\sum_{p\notin B_F}\frac{\vert \eta_{p-q}\vert\vert p-q\vert L^3}{p^2-k_F^2+\epsilon_0}
  \leq &C\tilde{\rho}_0^{\frac{1}{3}}\vert\ln\tilde{\rho}_0\vert.\label{est etakk/p^2}
  \end{align}
\end{lemma}
\begin{proof}
  \par For (\ref{est Wk/p^2}), we divide it into three parts. The first part, we use (\ref{L1 L2 est W}) to get to
  \begin{equation}\label{1st W}
    \frac{1}{L^3}\sum_{k_F<\vert p\vert<2k_F}\frac{\vert W_{p-q}\vert}{p^2-k_F^2+\epsilon_0}
  \leq C\int_{k_F}^{2k_F}\frac{r^2dr}{r^2-k_F^2+\epsilon_0}\leq C\tilde{\rho}_0^{\frac{1}{3}}
  \vert\ln\tilde{\rho}_0\vert.
  \end{equation}
  For the second part, we take some $d<\frac{1}{3}$ to reach
  \begin{equation}\label{2nd W}
    \frac{1}{L^3}\sum_{2k_F\leq\vert p\vert<\tilde{\rho}_0^d}\frac{\vert W_{p-q}\vert}{p^2-k_F^2+\epsilon_0}
  \leq C\int_{2k_F}^{\tilde{\rho}_0^d}\frac{r^2dr}{r^2-k_F^2+\epsilon_0}\leq C\tilde{\rho}_0^{d}.
  \end{equation}
  For the third part, we also use (\ref{L1 L2 est W}) to deduce
  \begin{equation}\label{3rd W}
    \frac{1}{L^3}\sum_{\tilde{\rho}_0^d\leq\vert p\vert}\frac{\vert W_{p-q}\vert}{p^2-k_F^2+\epsilon_0}
  \leq \Vert W\Vert_2\Big(\int_{\tilde{\rho}^d}^{\infty}\frac{dr}{r^2}\Big)^{\frac{1}{2}} \leq C\tilde{\rho}_0^{\frac{1}{2}-\frac{3}{2}\alpha_3-\frac{d}{2}}.
  \end{equation}
  By choosing $d=\frac{1}{3}-\alpha_3$, we reach (\ref{est Wk/p^2}).
  
  \par For (\ref{est etakk/p^2}), we notice that from (\ref{est of w_l,p}),(\ref{eta_p}) and (\ref{est of eta_0}), we have
  \begin{equation}\label{useful bound eta}
    L^3\vert p\vert^2\vert\eta_p\vert\leq C,\quad L^3\vert p\vert \vert\eta_p\vert
    \leq\tilde{\rho}_0^{-\frac{1}{3}+\alpha_3}.
  \end{equation}
  Therefore, we have for the first part
  \begin{equation}\label{1st eta}
    \tilde{\rho}_0^{\frac{1}{3}}\sum_{k_F<\vert p\vert<2k_F}\frac{\vert \eta_{p-q}\vert\vert p-q\vert}{p^2-k_F^2+\epsilon_0}
  \leq C\tilde{\rho}_0^{\alpha_3}\int_{k_F}^{2k_F}\frac{r^2dr}{r^2-k_F^2+\epsilon_0}\leq C\tilde{\rho}_0^{\frac{1}{3}+\alpha_3}
  \vert\ln\tilde{\rho}_0\vert.
  \end{equation}
  For the second part, we take some $d<\frac{1}{3}$ to reach
  \begin{equation}\label{2nd eta}
   \tilde{\rho}_0^{\frac{1}{3}}\sum_{2k_F\leq\vert p\vert<\tilde{\rho}_0^d}\frac{\vert \eta_{p-q}\vert\vert p-q\vert}{p^2-k_F^2+\epsilon_0}
  \leq C\tilde{\rho}_0^{\frac{1}{3}}\int_{2k_F}^{\tilde{\rho}_0^d}
  \frac{dr}{r}\leq C\tilde{\rho}_0^{\frac{1}{3}}
  \vert\ln\tilde{\rho}_0\vert.
  \end{equation}
  For the third part, we use (\ref{est of Deta D^2eta}) to deduce 
  \begin{equation}\label{3rd eta}
   \tilde{\rho}_0^{\frac{1}{3}}\sum_{\tilde{\rho}_0^d\leq\vert p\vert}\frac{\vert \eta_{p-q}\vert\vert p-q\vert}{p^2-k_F^2+\epsilon_0}
  \leq \tilde{\rho}_0^{\frac{1}{3}}L^{\frac{3}{2}}\Vert \nabla\eta
  \Vert_2\Big(\int_{\tilde{\rho}^d}^{\infty}\frac{dr}{r^2}\Big)^{\frac{1}{2}} \leq C\tilde{\rho}_0^{\frac{1}{3}-\frac{d}{2}}.
  \end{equation}
  By choosing $d=0$, we reach (\ref{est etakk/p^2}).
\end{proof}

\section{Fock Space and Quadratic Operators}\label{Fock space}
\
\par The fermionic Fock space is given by
\begin{equation}\label{Fock space definition}
  \mathcal{F}=\bigoplus_{N=0}^\infty \mathcal{H}^{\wedge N},
\end{equation}
where $\mathcal{H}^{\wedge 0}=\mathbb{C}$,
\begin{equation}\label{H define}
  \mathcal{H}=L^2(\Lambda_L\times\mathcal{S}_\mathbf{q})\cong L^2(\Lambda_L;\mathbb{C}^q),
\end{equation}
and $\mathcal{H}^{\wedge N}=L^2_a\big((\Lambda_L\times\mathcal{S}_{\mathbf{q}})^N\big)$ is the space of $N$-particle fermionic wave functions that are anti-symmetric with respect to permutations of particle coordinates.

\par For $k\in(2\pi/L)\mathbb{Z}^3$ and $\sigma\in\mathcal{S}_\mathbf{q}$, we choose the orthonormal basis of $\mathcal{H})$ by
\begin{equation}\label{orthonormal basis H}
  f_{k,\sigma}(x,\nu)=f_k(x)\chi_{\nu=\sigma}(\nu)=\frac{\chi_{\nu=\sigma}(\nu)}{L^{\frac{3}{2}}}
  e^{ik\cdot x},\quad z=(x,\nu)=\Lambda_L\times\mathcal{S}_\mathbf{q}.
\end{equation}
For non-negative integers $n,m$, and $\psi\in L^2\big((\Lambda_L\times\mathcal{S}_\mathbf{q})^n\big)$, $\varphi\in L^2\big((\Lambda_L\times\mathcal{S}_\mathbf{q})^m\big)$, we define the anti-symmetric wedge product between $\psi$ and $\varphi$
\begin{equation*}\label{wedge product}
\begin{aligned}
  \psi\wedge\varphi(z_1,\dots,z_{n+m})=& \frac{1}{\sqrt{n!m!(n+m)!}}\\
&\times\sum_{\varrho\in S_{n+m}}\sgn{\varrho}\,\psi(z_{\varrho(1)},\dots,z_{\varrho(n)})
  \varphi(z_{\varrho(n+1)},\dots,z_{\varrho(n+m)}).
\end{aligned}
\end{equation*}
Under these definitions, it is straight-forward to verify $\psi\wedge\varphi\in L^2_a\big((\Lambda_L\times\mathcal{S}_\mathbf{q})^{n+m}\big)$, and the wedge product is associative and anti-commutative, i.e.
\begin{equation}\label{anti-commutative}
  \psi\wedge\varphi=(-1)^{nm}\varphi\wedge\psi.
\end{equation}
Under the definition of the wedge product, the notation $L_a^2\big((\Lambda_L\times\mathcal{S}_\mathbf{q})^N\big)= \mathcal{H}^{\wedge N}$ holds rigorously and we find an orthonormal basis of $\mathcal{H}^{\wedge N}$ given by
\begin{equation}\label{orthonormal basis H wedge N}
  \bigwedge_{i=1}^{N} f_{k_i,\sigma_i}=f_{k_1,\sigma_1}\wedge\cdots\wedge f_{k_N,\sigma_N}.
\end{equation}

\par For any $N\in\mathbb{N}_{\geq0}$, We can define the creation and annihilation operators on $\mathcal{H}^{\wedge N}$ by
\begin{equation}\label{define a*(f) & a(g)}
\left.\begin{aligned}
  (a^*(f)\psi)({z}_1,\dots,{z}_{N+1}) &= \frac{1}{\sqrt{N+1}}\sum_{i=1}^{N+1}(-1)^{i+1}f(z_i)
\psi(z_1,\dots,\hat{z_i},\dots,z_{N+1}) \\
  (a(g)\psi)(z_1,\dots,z_{N-1}) &= \sqrt{N}\int_{\Lambda_{L}\times\mathcal{S}_\mathbf{q}}\overline{g}(z)
\psi(z,z_1,\dots,z_{N-1})dz, \quad N\geq 1
\end{aligned}\right.
\end{equation}
We additionally define $a(g)=0$ for $N=0$. It is then also reasonable to treat $a^*$ and $a$ as operators defined on the fermionic Fock space $\mathcal{F}$. We can check the following anti-commutator formula:
\begin{equation}\label{anticommutator}
\begin{aligned}
  &\{a(g),a^*(f)\}=a(g)a^*(f)+a^*(f)a(g)=\langle f,g\rangle,\\
  &\{a(g),a(f)\}=\{a^*(g),a^*(f)\}=0.
\end{aligned}
\end{equation}
Using (\ref{anticommutator}), we can bound the norm of creation and annihilation operators by
\begin{equation}\label{norm bound  of creation and annihilation}
  \Vert a(g)\Vert\leq\Vert g\Vert_2,\quad\Vert a^*(f)\Vert\leq \Vert f\Vert_2.
\end{equation}
Using (\ref{orthonormal basis H}), we denote
\begin{equation}\label{define a_k,sigma a*_k,sigma}
  a_{k,\sigma}=a(f_{k,\sigma}),\quad a^*_{k,\sigma}=a^*(f_{k,\sigma}).
\end{equation}

\subsection{Particle Number Operators}\label{particle}
\
\par For fixed $N\in\mathbb{N}_{\geq0}$, it is easy to check that
\begin{equation}\label{N}
  N=\sum_{\sigma,k\in(2\pi/L)\mathbb{Z}^3}a_{k,\sigma}^*a_{k,\sigma}.
\end{equation}
Recall the Fermi ball defined around (\ref{average particle number}), we define the relaxation particle number operators $\mathcal{N}_{re}$ and $\tilde{\mathcal{N}}_{re}$ (Here we borrow the terminology \textquotedblleft relaxation\textquotedblright from dynamic system, that implies the perturbed system evolves into an thermodynamic equilibrium state):
\begin{equation}\label{N_relaxation}
  \mathcal{N}_{re}\coloneqq\sum_{\sigma,k\notin B_F}a_{k,\sigma}^*a_{k,\sigma},
  \quad
  \tilde{\mathcal{N}}_{re}\coloneqq\sum_{\sigma,k\in B_F}a_{k,\sigma}a_{k,\sigma}^*
\end{equation}
By (\ref{average particle number}) and (\ref{N}), we can also check that
\begin{equation}\label{N_re and N_retilde}
  \mathcal{N}_{re}-\tilde{\mathcal{N}}_{re}=N-\mathbf{q}\bar{N}_0.
\end{equation}
For $x\in\Lambda_L$ and $\sigma\in\mathcal{S}_{\mathbf{q}}$, we define for
$z\in\Lambda_L\times\mathcal{S}_{\mathbf{q}}$:
\begin{equation}\label{define g,h}
  g_{x,\sigma}(z)=\sum_{k\in B_F}\frac{e^{ikx}}{L^{\frac{3}{2}}}f_{k,\sigma}(z),
  \quad
  h_{x,\sigma}(z)=\sum_{k\notin B_F}\frac{e^{ikx}}{L^{\frac{3}{2}}}f_{k,\sigma}(z)
\end{equation}
and therefore
\begin{equation}\label{cal N_re and N_re tilde}
  \int_{\Lambda_L}\sum_{\sigma}a^*(h_{x,\sigma})a(h_{x,\sigma})dx=\mathcal{N}_{re},
  \quad
  \int_{\Lambda_L}\sum_{\sigma}a(g_{x,\sigma})a^*(g_{x,\sigma})dx=\tilde{\mathcal{N}}_{re}
\end{equation}

\par Due to technical reasons, for $\delta>0$ and $d\in\mathbb{R}$, we introduce the following notations (recall that $k_F=\tilde{\mu}$):
\begin{equation}\label{partition 3D space}
  \begin{aligned}
  P_{F,d}&=\{k\in(2\pi/L)\mathbb{Z}^3,\,\vert k\vert>k_F+\tilde{\mu}
  ^{\frac{1}{2}}\tilde{\rho}_0^d\}\\
  A_{F,d}&=\{k\in(2\pi/L)\mathbb{Z}^3,\,k_F<\vert k\vert\leq
  k_F+\tilde{\mu}^{\frac{1}{2}}\tilde{\rho}_0^d\}\\
  \underline{A}_{F,\delta}&=\{k\in(2\pi/L)\mathbb{Z}^3,\,r<\vert k\vert\leq
  k_F-\tilde{\mu}^{\frac{1}{2}}\tilde{\rho}_0^\delta\}\\
  \underline{B}_{F,\delta}&=\{k\in(2\pi/L)\mathbb{Z}^3,\,\vert k\vert\leq
  k_F-\tilde{\mu}^{\frac{1}{2}}\tilde{\rho}_0^\delta\}
  \end{aligned}
\end{equation}
With the notations in (\ref{partition 3D space}), we define the sharp momentum projections
\begin{equation}\label{tools1}
  \begin{aligned}
  &H_{x,\sigma}[d](z)=\sum_{k\in P_{F,d}}\frac{e^{ikx}}{L^{\frac{3}{2}}}f_{k,\sigma}(z),\quad
  L_{x,\sigma}[d](z)=\sum_{k\in A_{F,d}}\frac{e^{ikx}}{L^{\frac{3}{2}}}f_{k,\sigma}(z)\\
  &S_{x,\sigma}[\delta](z)=\sum_{k\in \underline{A}_{F,\delta}}\frac{e^{ikx}}{L^{\frac{3}{2}}}
  f_{k,\sigma}(z),\quad
  I_{x,\sigma}[\delta](z)=\sum_{k\in \underline{B}_{F,\delta}}\frac{e^{ikx}}{L^{\frac{3}{2}}}
  f_{k,\sigma}(z)
  \end{aligned}
\end{equation}
and the corresponding particle number operators
\begin{equation}\label{tools2}
  \begin{aligned}
  \mathcal{N}_{high}[d]&\coloneqq \sum_{\sigma,k\in P_{F,d}}a^*_{k,\sigma}a_{k,\sigma}
  =\int_{\Lambda_L}\sum_{\sigma}a^*(H_{x,\sigma}[d])a(H_{x,\sigma}[d])dx\\
  \mathcal{N}_{low}[d]&\coloneqq \sum_{\sigma,k\in A_{F,d}}a^*_{k,\sigma}a_{k,\sigma}
  =\int_{\Lambda_L}\sum_{\sigma}a^*(L_{x,\sigma}[d])a(L_{x,\sigma}[d])dx\\
   \mathcal{N}_{surface}[\delta]&\coloneqq \sum_{\sigma,k\in \underline{A}_{F,\delta}}a_{k,\sigma}a^*_{k,\sigma}
  =\int_{\Lambda_L}\sum_{\sigma}a(S_{x,\sigma}[\delta])a^*(S_{x,\sigma}[\delta])dx\\
  \mathcal{N}_{inner}[\delta]&\coloneqq \sum_{\sigma,k\in \underline{B}_{F,\delta}}a_{k,\sigma}a^*_{k,\sigma}
  =\int_{\Lambda_L}\sum_{\sigma}a(I_{x,\sigma}[\delta])a^*(I_{x,\sigma}[\delta])dx
  \end{aligned}
\end{equation}
For the sake of succinctness, we may write $\mathcal{N}_h=\mathcal{N}_{high}[d]$ and so on. For $\sharp=h,l,s,i$, the $\mathcal{N}_{\sharp,\sigma}$ version of (\ref{tools2}) can be defined similarly. 
We can easily check that
\begin{equation}\label{tools3}
  \begin{aligned}
    &h_{x,\sigma}=H_{x,\sigma}+L_{x,\sigma},\quad g_{x,\sigma}=S_{x,\sigma}+I_{x,\sigma}\\
    &\mathcal{N}_{re}=\mathcal{N}_h+\mathcal{N}_l,\quad\quad\,\,\,\,
    \tilde{\mathcal{N}}_{re}=\mathcal{N}_{s}+\mathcal{N}_i.
  \end{aligned}
\end{equation}

\par We also point out that, using the simple facts that
\begin{equation}\label{simple facts}
  \tilde{\mathcal{N}}_{re}\leq\mathbf{q}\bar{N}_0,\quad\mathcal{N}_{s}[\delta]\leq C\tilde{\rho}_0^{1+\delta}L^3,
\end{equation}
we can deduce
\begin{equation}\label{est of Ntildere^2}
  \tilde{\mathcal{N}}_{re}^2\leq C\tilde{\rho}_0L^3\mathcal{N}_i[\delta]+
  C\tilde{\rho}_0^{1+\delta}L^3\tilde{\mathcal{N}}_{re}.
\end{equation}

\subsection{Other Useful Quadratic Operators}\label{other}
\
\par As in \cite{dilutefermiBog}, we define for $\varphi\in L^2(\Lambda_L)$
\begin{equation}\label{define b}
  b^*_{y,\sigma}(\varphi)=\int_{\Lambda_L}\varphi(x-y)
  a^*(h_{x,\sigma})a(g_{x,\sigma})dx.
\end{equation}
From \cite[Lemma 5.3]{dilutefermiBog}, we have the following Lemma \ref{b^* bound by b}:
\begin{lemma}\label{b^* bound by b}
For any $N\in\mathbb{N}_{\geq0}$ and $\psi\in \mathcal{H}^{\wedge N}$,
\begin{equation}\label{ineq b 1}
  \Vert b^*_{y,\sigma}(\varphi)\psi\Vert\leq \Vert b_{y,\sigma}(\varphi)\psi\Vert
  +C\tilde{\rho}_0^{\frac{1}{2}}\Vert\varphi\Vert_2\Vert\psi\Vert_2.
\end{equation}
Using the notation (\ref{psilemma}), we have
\begin{equation}\label{ineq b 2}
  \Vert b^*_{y,\sigma}({\eta}_{\phi^+})\Vert,\Vert b^*_{y,\sigma}({\eta}_{\phi^+}
  ^{\zeta^-})\Vert\leq
  CL^{-\frac{3}{2}}\tilde{\rho}_0^{\frac{1}{3}+\frac{\alpha_2}{2}}.
\end{equation}
\end{lemma}
The proof of (\ref{ineq b 2}) can be found in \cite[Lemma 4.8]{fermiupper} and \cite[Lemma 3.1]{2024huangyangformulalowdensityfermi}.

\par In fact, we can define a more general operator $b^*$ (despite there might be some ambiguity in notations):
\begin{equation}\label{define btilde}
  {b}^*_{y,\sigma}(\varphi)=\int_{\Lambda_L}\varphi(x-y)
  a^*(\mathbf{h}_{x,\sigma})a(\partial^{\mathbf{n}} \mathbf{g}_{x,\sigma})dx,
\end{equation}
where for $n\in\mathbb{N}_{\geq0}$, $\partial^{\mathbf{n}}=\partial_{x_{j_1}}\cdots\partial
_{x_{j_n}}$, with $j_m=1,2,3$ for $1\leq m\leq n$ (when $n=0$ it is simply the identity). Besides, for some subset $M_g\subset (2\pi/L)\mathbb{Z}^3$ and $M_h\subset (2\pi/L)\mathbb{Z}^3\backslash B_F$, such that for some closed ball $B_r\subset (2\pi/L)\mathbb{Z}^3$ with $r>0$, $M_g\subset B_r$, and
\begin{equation}\label{bold f g}
  \mathbf{g}_{x,\sigma}(z)=\sum_{k\in M_g}\frac{e^{ikx}}{L^{\frac{3}{2}}}f_{k,\sigma}(z),
  \quad
  \mathbf{h}_{x,\sigma}(z)=\sum_{k\in M_h}\frac{e^{ikx}}{L^{\frac{3}{2}}}f_{k,\sigma}(z)
\end{equation}
Similar to Lemma \ref{b^* bound by b}, we have the following Lemma \ref{b^* bound by b general}
\begin{lemma}\label{b^* bound by b general}
  Let $b^*$ be defined in (\ref{define btilde}), then for any $N\in\mathbb{N}_{\geq0}$ and $\psi\in \mathcal{H}^{\wedge N}$,
  \begin{equation}\label{ineq b1 general}
    \Vert b^*_{y,\sigma}(\varphi)\psi\Vert\leq \Vert b_{y,\sigma}(\varphi)\psi\Vert
  +Cr^n\mathfrak{c}^{\frac{1}{2}}\Vert\varphi\Vert_2\Vert\psi\Vert_2,
  \end{equation}
  where
  \begin{equation}\label{define frak c}
    \mathfrak{c}=\min\{{\# M_g},\# M_h\}L^{-3}< \infty.
  \end{equation}
  Here, $\#$ represents the cardinal of a countable set. Moreover, for $r\ll\tilde{\rho}_0^{\frac{1}{3}-\alpha_2}$, we also have
  \begin{equation}\label{ineq b 2 general}
  \Vert b^*_{y,\sigma}({\eta}_{\phi^+})\Vert,\Vert b^*_{y,\sigma}({\eta}_{\phi^+}
  ^{\zeta^-})\Vert\leq
  CL^{-\frac{3}{2}}r^n\tilde{\rho}_0^{\frac{\alpha_2}{2}-\frac{1}{6}}\Big(
  \frac{\# M_g}{L^3}\Big)^{\frac{1}{2}}.
\end{equation}
\end{lemma}
The proof of Lemma \ref{b^* bound by b general} does not differ much from the proof of \cite[Lemma 5.3]{dilutefermiBog}, we thus omit further details.


\par We also let (also with some abuse of notations)
\begin{equation}\label{define A_k}
  A_{\sigma}(k)=\int_{\Lambda_L}\frac{e^{-ikx}}{L^{\frac{3}{2}}}
  a^*(\mathbf{h}_{x,\sigma})a(\mathbf{g}_{x,\sigma})dx,
\end{equation}
with $\mathbf{h}$ and $\mathbf{g}$ given in (\ref{bold f g}). Then we have
\begin{lemma}\label{lemma A_k B_k}
  Let $A^{(1)}_{\sigma}(k)$ and $A^{(2)}_{\nu}(k)$ both have the form (\ref{define A_k}), and let $d_k$ be a family of constant , such that $\sup_k\vert d_k\vert=\mathfrak{d}<+\infty$. Then for any $N\in\mathbb{N}_{\geq0}$ and $\psi\in\mathcal{H}^{\wedge N}$,
  \begin{equation}\label{A_k B_k ineq}
  \begin{aligned}
    \Big\vert\sum_{k,\sigma,\nu}d_k
    \langle A^{(1)}_\sigma(-k)A^{(2)}_\nu(k)\psi,\psi \rangle\Big\vert
    \leq& C\mathfrak{d}\Big(\sum_{\sigma,\nu}\int_{\Lambda_L}
    \Vert a^*(\mathbf{h}^{(2)}_{y,\nu})a(\mathbf{g}^{(2)}_{y,\nu})
    \psi\Vert^2dy\Big)^{\frac{1}{2}}\\
    &\quad\quad\times
    \Big(\sum_{\sigma,\nu}\int_{\Lambda_L}
    \Vert a^*(\mathbf{g}^{(1)}_{x,\sigma})a(\mathbf{h}^{(1)}_{x,\sigma})
    \psi\Vert^2dx\Big)^{\frac{1}{2}}.
  \end{aligned}
  \end{equation}
\end{lemma}
\begin{proof}
  \par We let
  \begin{equation}\label{ffff}
    f(x,y)=\langle a^*(\mathbf{h}^{(2)}_{y,\nu})a(\mathbf{g}^{(2)}_{y,\nu})
    \psi,a^*(\mathbf{g}^{(1)}_{x,\sigma})a(\mathbf{h}^{(1)}_{x,\sigma})
    \psi\rangle.
  \end{equation}
  Then
  \begin{equation}\label{re A^1 A^2}
    \langle A^{(1)}_\sigma(-k)A^{(2)}_\nu(k)\psi,\psi \rangle
    =\int_{\Lambda_L^2}L^{-3}e^{ik(x-y)}f(x,y)dxdy.
  \end{equation}
  Obviously, for $j=1,2$:
  \begin{equation}\label{re ffff 1}
    a^*(\mathbf{h}^{(j)}_{y,\nu})a(\mathbf{g}^{(j)}_{y,\nu})
    =\sum_{m}\frac{e^{imy}}{L^{\frac{3}{2}}}B^{(j)}_\nu(m),
  \end{equation}
  for some operators $B^{(j)}_\nu(m)$ independent of $y$. Therefore,
  \begin{equation}\label{re A^1 A^2 new}
    \big\vert \langle A^{(1)}_\sigma(-k)A^{(2)}_\nu(k)\psi,\psi \big\rangle
    \leq \Vert B^{(2)}_\nu(k)\psi\Vert \cdot\Vert B^{(1)*}_\sigma(-k)\psi\Vert.
  \end{equation}
  Notice that by (\ref{re ffff 1}), we have
  \begin{equation}\label{re ffff2}
    \begin{aligned}
    \int_{\Lambda_L}
    \Vert a^*(\mathbf{h}^{(j)}_{y,\nu})a(\mathbf{g}^{(j)}_{y,\nu})
    \psi\Vert^2dy
    =&\int_{\Lambda_L}\sum_{m,m^\prime}\frac{e^{i(m-m^\prime)y}}{L^3}
    \langle B^{(j)}_\nu(m)\psi,B^{(j)}_\nu(m^\prime)\psi\rangle dx\\
    =&\sum_m\Vert B^{(j)}_\nu(m)\psi\Vert^2.
    \end{aligned}
  \end{equation}
  Combining (\ref{re A^1 A^2 new}) and (\ref{re ffff2}), then (\ref{A_k B_k ineq}) can be obtained by using Cauchy-Schwartz inequality.
\end{proof}

\par By the fact for $p-k,q+k\notin B_F$ and $p,q\in B_F$:
\begin{equation}\label{int e}
  \frac{1}{\vert q+k\vert^2+\vert p-k\vert^2-\vert q\vert^2-\vert p\vert^2+2\epsilon_0}=\int_{0}^{\infty}e^{-\big(\vert q+k\vert^2+\vert p-k\vert^2-\vert q\vert^2-\vert p\vert^2+2\epsilon_0\big)t}dt,
\end{equation}
to deal with the Bogoliubov coefficients defined in Section \ref{Bog coeff sec} in the thermodynamic limit, we need additionally define the $t-$version of (\ref{define g,h}) (like in \cite{2024huangyangformulalowdensityfermi}):
\begin{equation}\label{htgt0}
    h_{x,\sigma}^t(z)=\sum_{p\notin
    B_F}\frac{e^{ipx}}{L^{\frac{3}{2}}}e^{-p^2t}f_{p,\sigma}(z),
    \quad
    g_{x,\sigma}^t(z)=\sum_{q\in B_F}\frac{e^{iqx}}{L^{\frac{3}{2}}}e^{q^2t}f_{q,\sigma}(z)
  \end{equation}
  for $(x,\sigma)\in\Lambda_L\times\mathcal{S}_{\mathbf{q}}$, and $t\in(0,\infty)$. The $t-$version of (\ref{tools1}) can be defined similarly. It is then easy to verify some basic properties of (\ref{htgt0}) collected in the lemma below
  \begin{lemma}\label{property htgt}
    For any $N\in\mathbb{N}_{\geq0}$, we have
    \begin{equation}\label{property gt norm}
      \Vert a(g_{x,\sigma}^t)\Vert\leq \frac{1}{L^3}\sum_{q\in B_F}e^{2q^2t},
    \end{equation}
    and
    \begin{equation}\label{property htgt N}
      \begin{aligned}
      &\sum_{\sigma}\int_{\Lambda_L}a^*(h_{x,\sigma}^t)a(h_{x,\sigma}^t)dx\leq e^{-2k_F^2t}
      \mathcal{N}_{re},\\
      &\sum_{\sigma}\int_{\Lambda_L}a(g_{x,\sigma}^t)a^*(g_{x,\sigma}^t)dx\leq e^{2k_F^2t}
      \tilde{\mathcal{N}}_{re}.
      \end{aligned}
    \end{equation}
    Moreover, for $\delta>0$ and $d\in\mathbb{R}$,
    \begin{equation}\label{property htgt N_h and i}
      \begin{aligned}
      &\sum_{\sigma}\int_{\Lambda_L}a^*(H_{x,\sigma}^t[d])a(H_{x,\sigma}^t[d])dx\leq e^{-2(k_F+\tilde{\mu}^{\frac{1}{2}}\tilde{\rho}_0^d)^2t}
      \mathcal{N}_{h}[d],\\
      &\sum_{\sigma}\int_{\Lambda_L}a(I_{x,\sigma}^t[\delta])a^*(I_{x,\sigma}^t[\delta])dx\leq e^{2(k_F-\tilde{\mu}^{\frac{1}{2}}\tilde{\rho}_0^\delta)^2t}
      {\mathcal{N}}_{i}[\delta],
      \end{aligned}
    \end{equation}
    and
    \begin{equation}\label{property L norm}
      \Vert a(L_{x,\sigma}^t[d])\Vert\leq \frac{1}{L^3}\sum_{p\in A_{F,d}}e^{-2p^2t}.
    \end{equation}
  \end{lemma}
  It is also useful to have the following estimates. Their proofs however, are rather simple using the integral to bound the Riemann sum, as in the thermodynamic limit we have $L\to\infty$. Thus we omit further details.
  \begin{lemma}\label{est sum htgt}
    For $\epsilon_0=\tilde{\rho}_0^{\frac{2}{3}+\iota}$ for $\iota>0$, we have
    \begin{equation}\label{est sum}
      \frac{1}{L^3}\sum_{q\in B_F}\frac{1}{k_F^2-q^2+\epsilon_0}\leq C\tilde{\rho}_0^{\frac{1}{3}}
      \vert\ln\tilde{\rho}_0\vert\leq C\tilde{\rho}_0^{\frac{1}{3}-\varepsilon},
    \end{equation}
    for some $\varepsilon>0$ fixed but arbitrarily small. On the other hand, for $d\in\mathbb{R}$, we have
    \begin{equation}\label{est sum d}
      \frac{1}{L^3}\sum_{q\in B_F}\frac{1}{(k_F+\tilde{\mu}^{\frac{1}{2}}\tilde{\rho}_0^d)^2-q^2+\epsilon_0}\leq
      \left\{\begin{aligned}
      &C\tilde{\rho}_0^{\frac{1}{3}}
      \vert\ln\tilde{\rho}_0\vert\leq C\tilde{\rho}_0^{\frac{1}{3}-\varepsilon},
      \quad d>0\\
      &C\tilde{\rho}_0^{\frac{1}{3}-2d},
      \quad d\leq0
      \end{aligned}\right.
    \end{equation}
    and for $d>0$,
    \begin{equation}\label{est sum A_F d}
     \frac{1}{L^3} \sum_{p\in A_{F,d}}\frac{1}{k_F^2-q^2+\epsilon_0}\leq C\tilde{\rho}_0^{\frac{1}{3}}
      \vert\ln\tilde{\rho}_0\vert\leq C\tilde{\rho}_0^{\frac{1}{3}-\varepsilon},
    \end{equation}
    and for $d\leq0$
    \begin{equation}\label{est sum A_F d <=0}
      \frac{1}{L^3}\sum_{p\in A_{F,d}}\frac{1}{k_F^2-q^2+\epsilon_0}\leq  C\big(\tilde{\rho}_0^{\frac{1}{3}+d}+\tilde{\rho}_0^{\frac{1}{3}-\varepsilon}\big),
    \end{equation}
  \end{lemma}

  \par For $\varphi\in L^2(\Lambda_L)$, we let
  \begin{equation}\label{ct0}
    c^{t*}_{x,\sigma}(\varphi)=\int_{\Lambda_L}\varphi(x-y)a^*(h^t_{y,\sigma})a(g^t_{y,\sigma})
    dy.
  \end{equation}
 Similar to the proof of Lemma \ref{b^* bound by b}, we have
 \begin{lemma}\label{lemma bog htgt}
  For any $N\in\mathbb{N}_{\geq0}$ and $(x,\sigma)\in\Lambda_L\times\mathcal{S}_{\mathbf{q}}$, $\varphi\in  L^2(\Lambda_L)$, and $\psi\in \mathcal{H}^{\wedge N}$, we have
  \begin{equation}\label{ineq ct}
    \Vert c^{t*}_{x,\sigma}(\varphi)\Vert^2\leq \Vert c^{t}_{x,\sigma}(\varphi)\Vert^2
    +\Big(\frac{1}{L^3}\sum_{\substack{p\notin B_F}\\
    q\in B_F}\vert f_{p-q}\vert^2e^{-2(p^2-q^2)t}\Big)\Vert \psi\Vert^2.
  \end{equation}
\end{lemma}
The proof of Lemma \ref{lemma bog htgt} is the same as Lemma \ref{b^* bound by b}, we thus omit further details. Besides, similar to the definition (\ref{define btilde}), we can define a general version of (\ref{ct0}), and similarly reach a result corresponding to Lemma \ref{b^* bound by b general}. The proof, however, is also trivial, and we omit the lengthy explanation.

\section{A-Priori Estimates}\label{a-prior}

\subsection{Relative Entropy}
\
\par For a translation-invariant density matrix $\Gamma$ with $\gamma$ denoting its one-particle density matrix, we define the relative entropy\footnote{Similar definition of relative entropy can also be found in \cite{spin1upper}.}
\begin{equation}\label{relative entropy}
  \frac{1}{\beta}S(\Gamma,\tilde{G}_0)\coloneqq
  L^3\big(\tilde{P}^L_0[\tilde{G}_0]-\tilde{P}^L_0[\Gamma]\big).
\end{equation}
For $k\in (2\pi/L)\mathbb{Z}^3$ and $\sigma\in\mathcal{S}_{\mathbf{q}}$, the 1-pdm entropy is given by
\begin{equation}\label{entropy 1pdm}
  S(\gamma)\coloneqq-\sum_{k,\sigma}\big\{\hat{\gamma}(k,\sigma)\ln\hat{\gamma}(k,\sigma)
  +\big(1-\hat{\gamma}(k,\sigma)\big)\ln\big(1-\hat{\gamma}(k,\sigma)\big)\big\}.
\end{equation}
We have $S(\Gamma)\leq S(\gamma)$ (see \cite{4pages}\footnote{This reasonable inequality has been used in the field more than once without source. We find it useful to have the origin}). Since $\tilde{G}_0$ is quasi-free (see \cite{wickformula}) and preserves particle number, $S(\tilde{\gamma}_0)=S(\tilde{G}_0)$. The relative 1-pdm entropy is given by,
\begin{equation}\label{relative entropy 1pdm}
  S(\gamma,\tilde{\gamma}_0)\coloneqq
  \sum_{k,\sigma}\Big(\hat{\gamma}(k,\sigma)\ln\frac{\hat{\gamma}(k,\sigma)}
  {\hat{\tilde{\gamma}}_0(k,\sigma)}+\big(1-\hat{\gamma}(k,\sigma)\big)
  \ln\frac{1-\hat{\gamma}(k,\sigma)}{1-\hat{\tilde{\gamma}}_0(k,\sigma)}\Big).
\end{equation}
 Since $0\leq\hat{\gamma}(k,\sigma),\hat{\tilde{\gamma}}_0(k,\sigma)\leq1$, we have $S(\gamma,\tilde{\gamma}_0)\geq0$. By definition, we have
 \begin{equation}\label{trace diff}
   \tr\big((\mathcal{K}-\tilde{\mu}\mathcal{N})(\Gamma-\tilde{G}_0)\big)=\sum_{k,\sigma}
   \big(\hat{\gamma}(k,\sigma)-\hat{\tilde{\gamma}}_0(k,\sigma)\big)
   \big(\vert k\vert^2-\tilde{\mu}\big),
 \end{equation}
 we can use the $\tilde{\mu}$ version of (\ref{gamma_0}), (\ref{relative entropy}) and (\ref{trace diff}) to rewrite (\ref{relative entropy 1pdm}) into
 \begin{equation}\label{relative entropy 1pdm rewrt}
   \frac{1}{\beta}S(\gamma,\tilde{\gamma}_0)=
   \tr\big((\mathcal{K}-\tilde{\mu}\mathcal{N})(\Gamma-\tilde{G}_0)\big)
   -\frac{1}{\beta}S(\gamma)+\frac{1}{\beta}S(\tilde{\gamma}_0)\leq\frac{1}{\beta}
   S(\Gamma,\tilde{G}_0).
 \end{equation}

\par As in \cite[Lemma 3.3]{spin1upper}, we define for $0\leq t\leq 1$ and $0<t^\prime<1$
\begin{equation}\label{s function}
  s(t)=-t\ln t-(1-t)\ln(1-t),\quad s(t,t^\prime)=
  t\ln\frac{t}{t^\prime}+(1-t)\ln\frac{1-t}{1-t^\prime}.
\end{equation}
We have for $h>0$, $t_0=(1+e^h)^{-1}$ and $0\leq t\leq1$ (details see \cite[Lemma 3.3]{spin1upper}),
\begin{equation}\label{claim ABL and RS}
  ht\leq 2s(t,t_0)+\frac{h}{1+e^{h/2}}.
\end{equation}
We then apply (\ref{claim ABL and RS}) for $h=\beta(\vert k\vert^2-\tilde{\mu})$ and $t=\hat{\gamma}(k,\sigma)$, such that $t_0=\hat{\tilde{\gamma}}_0(k,\sigma)$, we deduce by summing up in $\sigma\in \mathcal{S}_{\mathbf{q}}$ and $\vert k\vert>k_F+\tilde{\mu}^{\frac{1}{2}}\tilde{\rho}_0^\delta$:
\begin{equation}\label{apply claim ABL and RS}
  \sum_{\sigma,k\in P_{F,\delta}}\beta(\vert k\vert^2-\tilde{\mu})\hat{\gamma}
(k,\sigma)\leq 2S(\gamma,\tilde{\gamma}_0)+\sum_{\sigma,k\in P_{F,\delta}}
\frac{\beta(\vert k\vert^2-\tilde{\mu})}{1+e^{\beta(\vert k\vert^2-\tilde{\mu})/2}}.
\end{equation}
From \cite[Appendix A]{spin1upper}, we know that
\begin{equation}\label{Riemann Sum 1}
  \sum_{\sigma,k\in P_{F,\delta}}
\frac{\beta(\vert k\vert^2-\tilde{\mu})}{1+e^{\beta(\vert k\vert^2-\tilde{\mu})/2}}\leq
 CL^3\beta^{-1}\tilde{\mu}^{\frac{1}{2}}e^{-\beta\tilde{\mu}\tilde{\rho}_0^{\delta}}+CL^{2}.
\end{equation}
Therefore, using (\ref{restriction on betamu}), we have
\begin{equation}\label{aproribound number of particle outside}
\begin{aligned}
  \tr\mathcal{N}_{h}[\delta]\Gamma=&\sum_{\sigma,k\in P_{F,\delta}}\hat{\gamma}(k,\sigma)
  \leq(\beta\tilde{\mu})^{-1}\tilde{\rho}_0^{-\delta}
  \sum_{\sigma,k\in P_{F,\delta}}\beta(\vert k\vert^2-\tilde{\mu})\hat{\gamma}
(k,\sigma)\\
\leq&C\tilde{\rho}_0^{-\frac{2}{3}-\delta}\beta^{-1}S(\gamma,\tilde{\gamma}_0)+
CL^3\tilde{\rho}_0^{1+2\alpha_1-\delta}
e^{-\tilde{\rho}_0^{-\alpha_1+\delta}}+CL^{2}\tilde{\rho}_0^{\alpha_1-\delta}.
\end{aligned}
\end{equation}
 We also apply (\ref{claim ABL and RS}) for $h=\beta(\tilde{\mu}-\vert k\vert^2)$ and $t=1-\hat{\gamma}
(k,\sigma)$, notice that this time $t_0=1-\hat{\tilde{\gamma}}_0(k,\sigma)$. Also notice that $s(1-t,1-t_0)=s(t,t_0)$. We deduce by summing up in $\sigma\in \mathcal{S}_{\mathbf{q}}$ and $\vert k\vert\leq k_F-\tilde{\mu}^{\frac{1}{2}}\tilde{\rho}_0^\delta$:
\begin{equation}\label{apply claim ABL and RS2}
  \sum_{\sigma,k\in \underline{B}_{F,\delta}}\beta(\tilde{\mu}-\vert k\vert^2)\big(1-\hat{\gamma}
(k,\sigma)\big)\leq 2S(\gamma,\tilde{\gamma}_0)+\sum_{\sigma,k\in \underline{B}_{F,\delta}}
\frac{\beta(\tilde{\mu}-\vert k\vert^2)}{1+e^{\beta(\tilde{\mu}-\vert k\vert^2)/2}}.
\end{equation}
Also, similar to the calculation in \cite[Appendix A]{spin1upper}(using the error function $erf(z)$ rather than the complementary error function $erfc(z)$), we have 
\begin{equation}\label{Riemann Sum 2}
  \sum_{\sigma,k\in \underline{B}_{F,\delta}}
\frac{\beta(\tilde{\mu}-\vert k\vert^2)}{1+e^{\beta(\tilde{\mu}-\vert k\vert^2)/2}}\leq
 CL^3\beta^{-1}\tilde{\mu}^{\frac{1}{2}}e^{-\beta\tilde{\mu}\tilde{\rho}_0^{\delta}}+CL^{2}.
\end{equation}
Similar to (\ref{aproribound number of particle outside}), we have
\begin{equation}\label{aproribound number of particle inside2}
\begin{aligned}
  \tr\mathcal{N}_{i}[\delta]\Gamma=&\sum_{\sigma,k\in \underline{B}_{F,\delta}}\big(1-\hat{\gamma}(k,\sigma)\big)
  \leq(\beta\tilde{\mu})^{-1}\tilde{\rho}_0^{-\delta}
  \sum_{\sigma,k\in \underline{B}_{F,\delta}}
  \beta(\tilde{\mu}-\vert k\vert^2)\big(1-\hat{\gamma}
(k,\sigma)\big)\\
\leq&C\tilde{\rho}_0^{-\frac{2}{3}-\delta}\beta^{-1}S(\gamma,\tilde{\gamma}_0)+
CL^3\tilde{\rho}_0^{1+2\alpha_1-\delta}
e^{-\tilde{\rho}_0^{-\alpha_1+\delta}}+CL^{2}\tilde{\rho}_0^{\alpha_1-\delta}.
\end{aligned}
\end{equation}
On the other hand, for a density matrix $\Gamma$ and $\delta>0$
\begin{equation}\label{est N_s and N_l}
  \tr\mathcal{N}_s[\delta]\Gamma,\tr\mathcal{N}_{l}[\delta]\Gamma
  \leq C\tilde{\rho}_0^{1+\delta}L^3.
\end{equation}
Combining (\ref{est N_s and N_l}) and (\ref{aproribound number of particle outside}), we have for an approximate Gibbs state $\Gamma$,
\begin{equation}\label{est trNreGamma}
  \tr\mathcal{N}_{re}\Gamma\leq C\tilde{\rho}_0^{-\frac{2}{3}-\delta}\beta^{-1}S(\gamma,\tilde{\gamma}_0)+
CL^3\big(\tilde{\rho}_0^{1+\delta}+\tilde{\rho}_0^{1+2\alpha_1-\delta}
e^{-\tilde{\rho}_0^{-\alpha_1+\delta}}\big)
+CL^{2}\tilde{\rho}_0^{\alpha_1-\delta}
\end{equation}
for some $\delta>0$. Similarly, using (\ref{est N_s and N_l}) and (\ref{aproribound number of particle inside2}) we deduce
\begin{equation}\label{est trNtildereGamma}
  \tr\tilde{\mathcal{N}}_{re}\Gamma\leq C\tilde{\rho}_0^{-\frac{2}{3}-\delta}\beta^{-1}S(\gamma,\tilde{\gamma}_0)+
CL^3\big(\tilde{\rho}_0^{1+\delta}
+\tilde{\rho}_0^{1+2\alpha_1-\delta}e^{-\tilde{\rho}_0^{-\alpha_1+\delta}}\big)
+CL^{2}\tilde{\rho}_0^{\alpha_1-\delta}.
\end{equation}

\par We also define
\begin{equation}\label{define K_s}
  \mathcal{K}_s=\sum_{\sigma,k\notin B_F}\big(\vert k\vert^2-k_F^2\big)a_{k,\sigma}^*a_{k,\sigma}+
  \sum_{\sigma,k\in B_F}\big(k_F^2-\vert k\vert^2\big)a_{k,\sigma}a_{k,\sigma}^*.
\end{equation}
It is easy to check
\begin{equation}\label{K_s and K}
  \mathcal{K}_s=\mathcal{K}-\tilde{\mu}N+\mathbf{q}\tilde{\mu}\bar{N}_0
  -\sum_{\sigma,k\in B_F}\vert k\vert^2
\end{equation}
Using (\ref{apply claim ABL and RS}), (\ref{Riemann Sum 1}), (\ref{apply claim ABL and RS2}) and (\ref{Riemann Sum 2}), we have
\begin{equation}\label{est trK_sGamma}
  \begin{aligned}
  \tr\mathcal{K}_s\Gamma&=\sum_{\sigma,k\notin B_F}
  \big(\vert k\vert^2-k_F^2\big)\hat{\gamma}(k,\sigma)
  +\sum_{\sigma,k\in B_F}\big(k_F^2-\vert k\vert^2\big)\big(1-\hat{\gamma}(k,\sigma)\big)\\
  &\leq C\beta^{-1}S(\gamma,\gamma_0)
  +CL^3\tilde{\rho}_0^{\frac{5}{3}+2\alpha_1}+CL^2\beta^{-1}.
  \end{aligned}
\end{equation}
For $\delta<0$, We also define
\begin{equation}\label{K_h[delta]}
  \mathcal{K}_h[\delta]=\sum_{\sigma,k\in P_{F,\delta}}\vert k\vert^2a_{k,\sigma}^*a_{k,\sigma}
\end{equation}
Similarly, we have
\begin{equation}\label{est trK_h[delta]Gamma}
  \tr\mathcal{K}_{h}[\delta]\Gamma\leq C\beta^{-1}S(\gamma,\gamma_0)
  +CL^3\tilde{\rho}_0^{\frac{5}{3}+2\alpha_1}e^{-\tilde{\rho}_0^
  {-\alpha_1+\delta}}+CL^2\beta^{-1}.
\end{equation}

\subsection{A-Priori Estimates on the Hamiltonian}
\
\par For any $N\in\mathbb{N}_{\geq0}$, $k\in(2\pi/L)\mathbb{Z}^3$ and $\sigma\in\mathcal{S}_{\mathbf{q}}$, we can rewrite $H_N$ in the form of creation and annihilation operators:
\begin{equation}\label{second quantization H_N}
  H_N=\mathcal{K}_N+\mathcal{V}_N=\sum_{k,\sigma}\vert
  k\vert^2a_{k,\sigma}^*a_{k,\sigma}
  +\frac{1}{2L^3}\sum_{k,p,q,\sigma,\nu}\hat{v}_k
  a^*_{p-k,\sigma}a^*_{q+k,\nu}a_{q,\nu}a_{p,\sigma},
\end{equation}
where
\begin{equation}\label{define v_k}
  \hat{v}_k
  =\int_{\mathbb{R}^3}v(x)e^{-ik\cdot x}dx.
\end{equation}
We decompose the potential part $\mathcal{V}_N$  as
\begin{equation}\label{split V}
  \mathcal{V}_N=\mathcal{V}_0+\mathcal{V}_1+\mathcal{V}_{21}+\mathcal{V}_{22}+\mathcal{V}_{23}
  +\mathcal{V}_3+\mathcal{V}_4,
\end{equation}
with
\begin{equation}\label{split V detailed}
  \begin{aligned}
  &\mathcal{V}_0=\frac{1}{2L^3}\sum_{k,p,q,\sigma,\nu}
 \hat{v}_ka^*_{p-k,\sigma}a^*_{q+k,\nu}a_{q,\nu}a_{p,\sigma}\chi_{p,p-k\in B_F}\chi_{q,q+k\in B_F}\\
  &\mathcal{V}_1=\frac{1}{L^3}\sum_{k,p,q,\sigma,\nu}
  \hat{v}_k(a^*_{p-k,\sigma}a^*_{q+k,\nu}a_{q,\nu}a_{p,\sigma}+h.c.)\chi_{p-k\notin B_F}\chi_{p\in B_F}\chi_{q,q+k\in B_F}\\
  &\mathcal{V}_{21}=\frac{1}{2L^3}\sum_{k,p,q,\sigma,\nu}
  \hat{v}_k(a^*_{p-k,\sigma}a^*_{q+k,\nu}a_{q,\nu}a_{p,\sigma}+h.c.)\chi_{p-k\notin B_F}\chi_{q+k\notin B_F}\chi_{p\in B_F}\chi_{q\in B_F}\\
  &\mathcal{V}_{22}=\frac{1}{2L^3}\sum_{k,p,q,\sigma,\nu}
  \hat{v}_ka^*_{p-k,\sigma}a^*_{q+k,\nu}a_{q,\nu}a_{p,\sigma}\\
  &\quad\quad\quad\quad\quad\quad\quad\quad
  \times(\chi_{p,p-k\notin B_F}\chi_{q,q+k\in B_F}+\chi_{p,p-k\in B_F}\chi_{q,q+k\notin B_F})\\
  &\mathcal{V}_{23}=\frac{1}{2L^3}\sum_{k,p,q,\sigma,\nu}
  \hat{v}_k(a^*_{p-k,\sigma}a^*_{q+k,\nu}a_{q,\nu}a_{p,\sigma}+h.c.)\chi_{p-k\notin B_F}\chi_{q+k\in B_F}\chi_{p\in B_F}\chi_{q\notin B_F}\\
  &\mathcal{V}_3=\frac{1}{L^3}\sum_{k,p,q,\sigma,\nu}
  \hat{v}_k(a^*_{p-k,\sigma}a^*_{q+k,\nu}a_{q,\nu}a_{p,\sigma}+h.c.)\chi_{p-k\notin B_F}\chi_{p\in B_F}\chi_{q,q+k\notin B_F}\\
  &\mathcal{V}_4=\frac{1}{2L^3}\sum_{k,p,q,\sigma,\nu}
  \hat{v}_ka^*_{p-k,\sigma}a^*_{q+k,\nu}a_{q,\nu}a_{p,\sigma}\chi_{p,p-k\notin B_F}\chi_{q,q+k\notin B_F}
  \end{aligned}
\end{equation}

\par  Before we get down to the evaluations of three renormalizations, we establish
 some useful estimates about the terms in (\ref{split V detailed}). The proofs are highly related with lemmas in \cite[Section 3]{WJH}, thus we only give brief explanations to Lemmas \ref{lemma V_0}-\ref{lemma V_22,23}.
\begin{lemma}\label{lemma V_0}
 For any $N\in\mathbb{N}_{\geq0}$, we can write $\mathcal{V}_0$ defined in (\ref{split V detailed}) by
\begin{equation}\label{est V_0}
  \mathcal{V}_0=\frac{\hat{v}_0}{2L^3}\sum_{\sigma\neq\nu}
  (\bar{N}_0-\tilde{\mathcal{N}}_{re,\sigma})
(\bar{N}_0-\tilde{\mathcal{N}}_{re,\nu})+\mathcal{E}_{\mathcal{V}_0},
\end{equation}
where
\begin{equation}\label{est E_V_0}
  \pm\mathcal{E}_{\mathcal{V}_0}\leq C\tilde{\rho}_0^{\frac{8}{3}}L^{3}+C\tilde{\rho}_0\mathcal{N}_i[\delta_1]+
  C\tilde{\rho}_0^{1+\delta_1}\tilde{\mathcal{N}}_{re},
\end{equation}
for $\delta_1>0$.
\end{lemma}
\begin{proof}
  \par We rewrite
\begin{equation*}
  \mathcal{V}_0=\mathcal{V}_{01}+\mathcal{V}_{02}+\mathcal{V}_{03},
\end{equation*}
where, for
\begin{equation*}
  \begin{aligned}
  &\mathcal{V}_{01}=\frac{1}{2L^3}\sum_{k,p,q,\sigma,\nu}
  \hat{v}_ka^*_{p-k,\sigma}a^*_{q+k,\nu}a_{q,\nu}a_{p,\sigma}\chi_{p,p-k\in B_F}\chi_{q,q+k\in B_F}\chi_{k=0}\\
  &\mathcal{V}_{02}=\frac{1}{2L^3}\sum_{k,p,q,\sigma,\nu}
  \hat{v}_ka^*_{p-k,\sigma}a^*_{q+k,\nu}a_{q,\nu}a_{p,\sigma}\chi_{p,p-k\in B_F}\chi_{q,q+k\in B_F}\chi_{k\neq0}\chi_{\sigma=\nu}\chi_{p=q+k}\\
  &\mathcal{V}_{03}=\frac{1}{2L^3}\sum_{k,p,q,\sigma,\nu}
  \hat{v}_ka^*_{p-k,\sigma}a^*_{q+k,\nu}a_{q,\nu}a_{p,\sigma}\chi_{p,p-k\in B_F}\chi_{q,q+k\in B_F}\chi_{k\neq0}\\
  &\quad\quad\quad\quad\quad\quad\quad\quad
  \times(1-\chi_{\sigma=\nu}\chi_{p=q+k})
  \end{aligned}
\end{equation*}
\par For $\mathcal{V}_{01}$, we have
\begin{equation}\label{cal V_01}
  \begin{aligned}
  \mathcal{V}_{01}&=\frac{1}{2L^3}\sum_{p,q,\sigma,\nu}\hat{v}_0a^*_{p,\sigma}a^*_{q,\nu}
  a_{q,\nu}a_{p,\sigma}\chi_{p\in B_F}\chi_{q\in B_F}\\
  &=\frac{\hat{v}_0}{2L^3}\sum_{\sigma,\nu}(\bar{N}_{0}-\tilde{\mathcal{N}}_{re,\sigma})
  (\bar{N}_{0}-\delta_{\sigma,\nu}-\tilde{\mathcal{N}}_{re,\nu})
  \end{aligned}
\end{equation}

\par For $\mathcal{V}_{02}$, we have
\begin{equation}\label{cal V_02 1}
\begin{aligned}
  \mathcal{V}_{02}
  =\frac{1}{2L^3}\sum_{p,q,\sigma}
  \hat{v}_{p-q}a^*_{q,\sigma}a^*_{p,\sigma}a_{q,\sigma}a_{p,\sigma}\chi_{p,q\in B_F}\chi_{p-q\neq0}.
\end{aligned}
\end{equation}
Since for $\theta\in (0,1)$,
\begin{equation}\label{v_k-v_0}
  \begin{aligned}
  \big\vert\hat{v}_k-\hat{v}_0\big\vert
  &=\Big\vert \int_{\mathbb{R}^3}v(x)\big(e^{-ikx}-1\big)dx\Big\vert\\
  &=\Big\vert \int_{\mathbb{R}^3}v(x)\big(-ik\cdot x-
  \theta^2e^{-ik\cdot\theta x} k\otimes k(x,x) \big)dx\Big\vert\\
  &\leq C\vert k\vert^2\int_{\mathbb{R}^3}\vert x\vert^2v(x)dx\leq C\vert k\vert^2.
  \end{aligned}
\end{equation}
Then for $\psi\in\mathcal{H}^{\wedge N}$,
\begin{equation}\label{cal V_02 2}
\begin{aligned}
  &\Big\vert\sum_{p,q,\sigma}
  \frac{1}{L^3}\big(\hat{v}_{p-q}-\hat{v}_0\big)\chi_{p,q\in B_F}\chi_{p-q\neq0}
  \langle a^*_{q,\sigma}a^*_{p,\sigma}a_{q,\sigma}a_{p,\sigma}\psi,\psi\rangle\Big\vert\\
  &\leq  \frac{C}{L^3}\sum_{p,q,\sigma}
  \vert p-q\vert^2\chi_{p,q\in B_F}\chi_{p-q\neq0}
  \Vert a_{q,\sigma}a_{p,\sigma}\psi\Vert^2
  \leq \tilde{\rho}_0^{\frac{8}{3}}L^3\Vert\psi\Vert^2.
\end{aligned}
\end{equation}
Since $a_{p,\sigma}a_{q,\sigma}=0$ when $p=q$, we have,
\begin{equation}\label{cal V_02}
  \mathcal{V}_{02}=-\frac{\hat{v}_0}{2L^3}\sum_{\sigma}(\bar{N}_0-\tilde{\mathcal{N}}_{re,\sigma})
  (\bar{N}_0-1-\tilde{\mathcal{N}}_{re,\sigma})+\mathcal{E}_{\mathcal{V}_{02}},
\end{equation}
where
\begin{equation}\label{est E_V_02}
  \pm\mathcal{E}_{\mathcal{V}_{02}}\leq C\tilde{\rho}_0^{\frac{8}{3}}L^3.
\end{equation}

\par For $\mathcal{V}_{03}$,
\begin{equation}\label{cal V_03 5}
  \mathcal{V}_{03}=\mathcal{V}_{03,r}+\mathcal{E}_{\mathcal{V}_{03}}.
\end{equation}
where
\begin{equation}\label{cal V_03 5.5}
  \mathcal{V}_{03,r}=\frac{1}{2L^3}\sum_{k,p,q,\sigma,\nu}
  \hat{v}_ka_{q,\nu}a_{p,\sigma}a^*_{p-k,\sigma}a^*_{q+k,\nu}
  \chi_{p,p-k\in B_F}\chi_{q,q+k\in B_F},
\end{equation}
and, by (\ref{est of Ntildere^2}),
\begin{equation}\label{cal V_03 6}
\begin{aligned}
  \pm\mathcal{E}_{\mathcal{V}_{03}}&\leq  CL^{-3}\mathcal{N}_{re}^2
  +\tilde{\rho}_0^{\frac{8}{3}}L^3\\
  &\leq C\tilde{\rho}_0\mathcal{N}_i[\delta_1]+
  C\tilde{\rho}_0^{1+\delta_1}\tilde{\mathcal{N}}_{re}+\tilde{\rho}_0^{\frac{8}{3}}L^3,
\end{aligned}
\end{equation}
for $\delta_1>0$.
\begin{equation}\label{cal V_03 7}
  \mathcal{V}_{03,r}=\frac{1}{2}\sum_{\sigma,\nu}
  \int_{\Lambda_L^2}v(x-y)
  a(g_{y,\nu})a(g_{x,\sigma})a^*(g_{x,\sigma})a^*(g_{y,\nu})dxdy.
\end{equation}
Using
\begin{equation*}
  \begin{aligned}
  a(g_{x,\sigma})a^*(g_{x,\sigma})&=
  a(S_{x,\sigma}[\delta_1])a^*(S_{x,\sigma}[\delta_1])
  +a(I_{x,\sigma}[\delta_1])a^*(S_{x,\sigma}[\delta_1])\\
  &+a(S_{x,\sigma}[\delta_1])a^*(I_{x,\sigma}[\delta_1])
  +a(I_{x,\sigma}[\delta_1])a^*(I_{x,\sigma}[\delta_1]),
  \end{aligned}
\end{equation*}
We can rewrite $\mathcal{V}_{03,r}=\sum_{j=1}^{4}\mathcal{V}_{03,rj}$. It is straight-forward to bound them by
\begin{align}
\pm\mathcal{V}_{03,r1}\leq& C\tilde{\rho}_0^{1+\delta_1}\tilde{\mathcal{N}}_{re}
\label{V_03,r1}\\
\pm\mathcal{V}_{03,r2},\mathcal{V}_{03,r3}\leq& C\tilde{\rho}_0^{1+\delta_1}\tilde{\mathcal{N}}_{re}
+C\tilde{\rho}_0\mathcal{N}_i[\delta_1]\label{V_03,r2,r3}\\
\pm\mathcal{V}_{03,r4}\leq& C\tilde{\rho}_0\mathcal{N}_i[\delta_1]\label{V_03,r4}
\end{align}
Combining all the results above we finish the proof of Lemma \ref{lemma V_0}.

\end{proof}

\begin{lemma}\label{lemma V_1}
For any $N\in\mathbb{N}_{\geq0}$, $\mathcal{V}_1$ defined in (\ref{split V detailed}) can be bounded by
\begin{equation}\label{est V_1}
  \pm\mathcal{V}_1\leq C\tilde{\rho}_0\mathcal{N}_i[\delta_1]
  +C\tilde{\rho}_0^{1-\delta_2}\mathcal{N}_h[\delta_1]+C\tilde{\rho}_0^{1+\delta_2}\tilde{\mathcal{N}}_{re},
\end{equation}
for $\delta_1\geq\delta_2>0$.
\end{lemma}
\begin{proof}
We can write
  \begin{equation}\label{cal V_1 1}
    \mathcal{V}_{1}=\sum_{\sigma,\nu}
  \int_{\Lambda_L^2}v(x-y)
  a^*(h_{x,\sigma})a(g_{y,\nu})a(g_{x,\sigma})a^*(g_{y,\nu})dxdy+h.c.
  \end{equation}
  Using
  \begin{equation*}
    \begin{aligned}
    a^*(h_{x,\sigma})a^*(g_{y,\nu})&=
    a^*(H_{x,\sigma}[\delta_1])a^*(I_{y,\nu}[\delta_1])
    +a^*(H_{x,\sigma}[\delta_1])a^*(S_{y,\nu}[\delta_1])\\
    &+a^*(L_{x,\sigma}[\delta_1])a^*(I_{y,\nu}[\delta_1])
    +a^*(L_{x,\sigma}[\delta_1])a^*(S_{y,\nu}[\delta_1]),
    \end{aligned}
  \end{equation*}
  we rewrite $\mathcal{V}_{1}=\sum_{j=1}^{4}\mathcal{V}_{1,j}$. It is straight-forward to bound them by
  \begin{align}
  \pm\mathcal{V}_{1,1}\leq&
  C\tilde{\rho}_0\mathcal{N}_h[\delta_1]+C\tilde{\rho}_0\mathcal{N}_i[\delta_1]
  \label{V_1,1}\\
  \pm\mathcal{V}_{1,2}\leq&C\tilde{\rho}_0^{1+\delta_2}\tilde{\mathcal{N}}_{re}
  +C\tilde{\rho}_0^{1-\delta_2}\mathcal{N}_h[\delta_1]\label{V_1,2}\\
  \pm\mathcal{V}_{1,3}\leq&C\tilde{\rho}_0^{1+\delta_1}\tilde{\mathcal{N}}_{re}+
  C\tilde{\rho}_0\mathcal{N}_i[\delta_1]\label{V_1,3}
  \end{align}
  For $\mathcal{V}_{1,4}$, we furthermore use
  \begin{equation*}
    a(g_{y,\nu})=a(I_{y,\nu}[\delta_1])+a(S_{y,\nu}[\delta_1])
  \end{equation*}
  to rewrite $\mathcal{V}_{1,4}=\mathcal{V}_{1,41}+\mathcal{V}_{1,42}$. Then we have
  \begin{align}
  \pm\mathcal{V}_{1,41}\leq&
  C\tilde{\rho}_0^{1+\delta_1}\tilde{\mathcal{N}}_{re}+
  C\tilde{\rho}_0\mathcal{N}_i[\delta_1]\label{V_1,41}\\
  \pm\mathcal{V}_{1,42}\leq&C\tilde{\rho}_0^{1+\delta_1}\tilde{\mathcal{N}}_{re}
  \label{V_1,42}
  \end{align}
  Combining all the results above we finish the proof of Lemma \ref{lemma V_1}.
\end{proof}

\begin{lemma}\label{lemma V_22,23}
For any $N\in\mathbb{N}_{\geq0}$, $\mathcal{V}_{22}$ and $\mathcal{V}_{23}$ defined in (\ref{split V detailed}) can be rewrote by
\begin{equation}\label{cal V_22}
  \mathcal{V}_{22}=\mathbf{q}\bar{N}_0L^{-3}\hat{v}_0\mathcal{N}_{re}+
  \mathcal{E}_{\mathcal{V}_{22}},
\end{equation}
where
\begin{equation}\label{est E_V_22}
  \pm\mathcal{E}_{\mathcal{V}_{22}}\leq C\tilde{\rho}_0\mathcal{N}_h[\delta_1]+C\tilde{\rho}_0^{1+\delta_1}
  \tilde{\mathcal{N}}_{re},
\end{equation}
for some $\delta_1>0$. On the other hand,
\begin{equation}\label{cal V_23}
  \mathcal{V}_{23}=-\bar{N}_0L^{-3}\hat{v}_0\mathcal{N}_{re}+
  \mathcal{E}_{\mathcal{V}_{23}},
\end{equation}
where
\begin{equation}\label{est E_V_23}
  \pm\mathcal{E}_{\mathcal{V}_{23}}\leq C\tilde{\rho}_0\mathcal{N}_h[\delta_1]
  +C\tilde{\rho}_0^{1+\delta_1}\tilde{\mathcal{N}}_{re}+C\tilde{\rho}_0^{\frac{5}{3}}
  \mathcal{N}_{re}+C\tilde{\rho}_0\mathcal{K}_s
\end{equation}

\end{lemma}
\begin{proof}
For $\mathcal{V}_{22}$, we write
  \begin{equation}\label{cal V_22 1}
    \mathcal{V}_{22}=\sum_{\sigma,\nu}
  \int_{\Lambda_L^2}v(x-y)
  a^*(h_{x,\sigma})a^*(g_{y,\nu})a(g_{y,\nu})a(h_{x,\sigma})dxdy.
  \end{equation}
  Using
  \begin{equation*}
  a(g_{y,\nu})a^*(g_{y,\nu})+a^*(g_{y,\nu})a(g_{y,\nu})=\Vert g_{y,\nu}\Vert^2=
  \frac{\bar{N}_0}{L^3},
\end{equation*}
we reach (\ref{cal V_22}) with
\begin{equation}\label{cal V_22 2}
  \mathcal{E}_{\mathcal{V}_{22}}
  =-\sum_{\sigma,\nu}
  \int_{\Lambda_L^2}v(x-y)
  a^*(h_{x,\sigma})a(g_{y,\nu})a^*(g_{y,\nu})a(h_{x,\sigma})dxdy.
\end{equation}
For $\delta_1>0$, we use
\begin{equation*}
  \begin{aligned}
   a^*(h_{x,\sigma})a(h_{x,\sigma})&=
   a^*(H_{x,\sigma}[\delta_1])a(h_{H,\sigma}[\delta_1])
   +a^*(L_{x,\sigma}[\delta_1])a(h_{H,\sigma}[\delta_1])\\
   &+a^*(H_{x,\sigma}[\delta_1])a(h_{L,\sigma}[\delta_1])
   +a^*(L_{x,\sigma}[\delta_1])a(h_{L,\sigma}[\delta_1]),
  \end{aligned}
\end{equation*}
we rewrite $\mathcal{E}_{\mathcal{V}_{22}}=\sum_{j=1}^{4}\mathcal{E}_{\mathcal{V}_{22},j}$. It is straight-forward to bound them by
\begin{align}
  \pm\mathcal{E}_{\mathcal{V}_{22},1}\leq&C\tilde{\rho}_0\mathcal{N}_h[\delta_1]
  \label{V_22,1}\\
  \pm\mathcal{E}_{\mathcal{V}_{22},2},
  \mathcal{E}_{\mathcal{V}_{22},3}\leq&
  C\tilde{\rho}_0^{1+\delta_1}\tilde{\mathcal{N}}_{re}+
  C\tilde{\rho}_0\mathcal{N}_h[\delta_1]
  \label{V_22,2,3}\\
  \pm\mathcal{E}_{\mathcal{V}_{22},4}\leq&
  C\tilde{\rho}_0^{1+\delta_1}\tilde{\mathcal{N}}_{re}
  \label{V_22,4}
\end{align}

\par For $\mathcal{V}_{23}$, we rewrite $\mathcal{V}_{23}=
\mathcal{V}_{23,1}+\mathcal{V}_{23,2}$ with
\begin{align*}
  \mathcal{V}_{23,1}=&-\frac{1}{2L^3}\sum_{k,p,q,\sigma}\delta_{p,q+k}
  \hat{v}_k(a^*_{p-k,\sigma}a_{q,\sigma}+h.c.)\chi_{p-k\notin B_F}\chi_{q+k\in B_F}\chi_{p\in B_F}\chi_{q\notin B_F}\\
  \mathcal{V}_{23,2}=&\frac{1}{2L^3}\sum_{k,p,q,\sigma,\nu}
  \hat{v}_k(a^*_{p-k,\sigma}a_{p,\sigma}a^*_{q+k,\nu}a_{q,\nu}+h.c.)\chi_{p-k\notin B_F}\chi_{q+k\in B_F}\chi_{p\in B_F}\chi_{q\notin B_F}
\end{align*}
For $\mathcal{V}_{23,2}$, we have
\begin{equation}\label{cal V_23 1}
  \mathcal{V}_{23,2}
  =\sum_{\sigma,\nu}
  \int_{\Lambda_L^2}v(x-y)
  a^*(h_{x,\sigma})a(g_{x,\sigma})a^*(g_{y,\nu})a(h_{y,\nu})dxdy.
\end{equation}
Similar to the bound of $\mathcal{E}_{\mathcal{V}_{22}}$ (\ref{cal V_22 2}) above, we have
\begin{equation}\label{cal V 23 2}
  \pm\mathcal{V}_{23,2}\leq  C\tilde{\rho}_0^{1+\delta_1}\tilde{\mathcal{N}}_{re}+
  C\tilde{\rho}_0\mathcal{N}_h[\delta_1].
\end{equation}
For $\mathcal{V}_{23,1}$, we have
\begin{equation}\label{cal V_23 3}
  \begin{aligned}
  \mathcal{V}_{23,1}=&-\frac{1}{L^3}\sum_{p,q,\sigma}\hat{v}_{p-q}a_{q,\sigma}^*a_{q,\sigma}
  \chi_{p\in B_F}\chi_{q\notin B_F}\\
  =&-\bar{N}_0L^{-3}\hat{v}_0\mathcal{N}_{re}+
  -\frac{1}{L^3}\sum_{p,q,\sigma}(\hat{v}_{p-q}-\hat{v}_0)a_{q,\sigma}^*a_{q,\sigma}
  \chi_{p\in B_F}\chi_{q\notin B_F}.
  \end{aligned}
\end{equation}
  By (\ref{v_k-v_0}), we know that
  \begin{equation}\label{cal V_23 4}
    \vert \hat{v}_{p-q}-\hat{v}_0\vert\leq C\vert p-q\vert^2
    \leq C\big(\vert q\vert^2-k_F^2+k_F^2+\vert p\vert^2\big).
  \end{equation}
  Thus we have
  \begin{equation}\label{cal V_23 5}
    \pm\frac{1}{L^3}\sum_{p,q,\sigma}(\hat{v}_{p-q}-\hat{v}_0)a_{q,\sigma}^*a_{q,\sigma}
  \chi_{p\in B_F}\chi_{q\notin B_F}\leq
  \tilde{\rho}_0\mathcal{K}_s+\tilde{\rho}_0^{\frac{5}{3}}\mathcal{N}_{re}.
  \end{equation}
  Combining all the results above we finish the proof of Lemma \ref{lemma V_22,23}.
\end{proof}

\section{Renormalized Hamiltonian and Main Propositions}\label{renormal}
\subsection{Quadratic Renormalization}\label{qua sec}
\
\par We define
\begin{equation}\label{define B0}
  B=\frac{1}{2}(A-A^*)
\end{equation}
where
\begin{equation}\label{define A0}
  A=\sum_{k,p,q,\sigma,\nu}
  \eta_k\phi^+(k)a^*_{p-k,\sigma}a^*_{q+k,\nu}a_{q,\nu}a_{p,\sigma}\chi_{p-k,q+k\notin
  B_F}\chi_{p,q\in B_F}.
\end{equation}
Recall that $\eta_k$ has been defined in (\ref{eta_p}), and $\phi^+$ is the smooth high frequency cut-off function defined below (\ref{cutoff}). We let
\begin{equation}\label{G_N}
  \mathcal{G}_N=e^{-B}H_Ne^B.
\end{equation}
Notice that $e^B$ can be regarded as an unitary operator on the fermionic Fock space $\mathcal{F}$, we thus also denote $\mathcal{G}=e^{-B}He^B=\bigoplus_{N=0}^\infty\mathcal{G}_N$. In the up-coming Sections \ref{cub sec} and \ref{bog sec}, we also apply similar notations without further specifications.

\par Then we record the effect of this quadratic renormalization in Proposition \ref{qua prop}.
\begin{proposition}\label{qua prop}
For any $N\in\mathbb{N}_{\geq0}$, and  $\frac{1}{24}>\frac{1}{2}\delta_3>\frac{1}{4}\alpha_3>\alpha_2>2\alpha_4>0$, $\frac{1}{3}>\delta_1\geq\delta_2>0$, $\beta_1=\frac{1}{3}+\alpha_5$ and $2\alpha_3>\alpha_5>2\alpha_4$, we have
\begin{equation}\label{G_N qua prop}
  \mathcal{G}_N=\mathcal{K}+\mathcal{V}_4+\mathcal{V}_{21}^\prime+\Omega
  +\mathcal{V}_{3,l}+C_{\mathcal{G}_N}+Q_{\mathcal{G}_N}+\mathcal{E}_{\mathcal{G}_N},
\end{equation}
where
\begin{equation}\label{qua prop V_21' and Omega and V_3,l}
  \begin{aligned}
  \mathcal{V}_{21}^\prime&=\frac{1}{L^3}\sum_{k,p,q,\sigma,\nu}
  W_k\zeta^{-}(k)(a^*_{p-k,\sigma}a^*_{q+k,\nu}a_{q,\nu}a_{p,\sigma}+h.c.)
  \chi_{p-k,q+k\notin
  B_F}\chi_{p,q\in B_F}\\
  \Omega&=\sum_{k,p,q,\sigma,\nu}
  \eta_k\phi^+(k)\zeta^{-}(k)k(q-p)
  \\
  &\quad\quad\quad\quad
  \times(a^*_{p-k,\sigma}a^*_{q+k,\nu}a_{q,\nu}a_{p,\sigma}+h.c.)\chi_{p-k,q+k\notin B_F}
  \chi_{p,q\in B_F},\\
  \mathcal{V}_{3,l}&=\frac{1}{L^3}\sum_{k,p,q,\sigma,\nu}\hat{v}_k\gamma^-(q)
(a_{p-k,\sigma}^*a_{q+k,\nu}^*a_{q,\nu}a_{p,\sigma}+h.c.)
\chi_{p-k,q+k\notin B_F}\chi_{q\notin B_F}\chi_{p\in B_F},
  \end{aligned}
\end{equation}
and
\begin{equation}\label{qua prop C and Q}
  \begin{aligned}
  C_{\mathcal{G}_N}&=\frac{1}{2L^3}\Big(\hat{v}_0+\sum_{k\in(2\pi/L)\mathbb{Z}^3}
  \hat{v}_k\eta_k\Big)\mathbf{q}(\mathbf{q}-1)\bar{N}_0^2
  +\frac{\hat{v}_0}{L^3}(\mathbf{q}-1)\bar{N}_0(N-\mathbf{q}\bar{N}_0)\\
  &+\frac{1}{L^3}\sum_{k,p,q,\sigma,\nu}W_k\eta_k\phi^+(k)\zeta^-(k)
  \chi_{p-k,q+k\notin B_F}\chi_{p,q\in B_F}\\
  &-\frac{1}{L^3}\sum_{k,p,q,\sigma}W_k\eta_{k+q-p}\zeta^-(k)
  \phi^+(k+q-p)\chi_{p-k,q+k\notin B_F}\chi_{p,q\in B_F}\chi_{p\neq q}\\
  &-\sum_{k,p,q,\sigma}\eta_k\eta_{k+q-p}\zeta^-(k)\phi^+(k)
  \phi^+(k+q-p)k(q-p)
  \chi_{p-k,q+k\notin B_F}\chi_{p,q\in B_F}\chi_{p\neq q}\\
  Q_{\mathcal{G}_N}&=-\frac{1}{L^3}\sum_{k\in(2\pi/L)\mathbb{Z}^3}
  \hat{v}_k\eta_k(\mathbf{q}-1)\bar{N}_0\tilde{\mathcal{N}}_{re}
  \end{aligned}
\end{equation}
and the error is bounded by
\begin{equation}\label{qua prop E}
  \begin{aligned}
  \pm\mathcal{E}_{\mathcal{G}_N}\leq&
  C\big(\tilde{\rho}_0^{1+\delta_2}+\tilde{\rho}_0^{\frac{4}{3}+\alpha_2-3\delta_3-\alpha_4}
  +\tilde{\rho}_0^{\frac{3}{2}+\frac{\alpha_2}{2}+4\alpha_3-\frac{3}{2}\alpha_5}\big)
  \mathcal{N}_{re}+C\big(\tilde{\rho}_0^{1+\delta_2}
  +\tilde{\rho}_0^{\frac{4}{3}-\alpha_3}\big)\tilde{\mathcal{N}}_{re}\\
  &+C\tilde{\rho}_0\mathcal{N}_i[\delta_1]+C\tilde{\rho}_0^{1-\delta_2}\mathcal{N}_h[\delta_1]
  +C\tilde{\rho}_0^{\frac{2}{3}-\alpha_4}\mathcal{N}_h[-\delta_3]\\
  &+C\tilde{\rho}_0^{\frac{5}{6}+\frac{\alpha_2}{2}-\frac{3}{2}\alpha_5}
  \mathcal{N}_h[-\alpha_2]
  +C\tilde{\rho}_0^{1+\alpha_5
  -4\alpha_3-\alpha_4}\mathcal{N}_h[-\beta_1]\\
  &+C\big(\tilde{\rho}_0^{\frac{1}{3}+2\alpha_3}+\tilde{\rho}_0^{\frac{1}{3}+
  \alpha_5-\alpha_4}\big)\mathcal{K}_s+C\tilde{\rho}_0^{\frac{1}{3}+
  \alpha_4}\mathcal{V}_4\\
  &+C\big(\tilde{\rho}_0^{\frac{7}{3}+\alpha_4}+\tilde{\rho}_0^{2+4\alpha_3-7\alpha_2}
  +\tilde{\rho}_0^{\frac{5}{2}+\frac{\alpha_2}{2}+\alpha_3-\frac{3}{2}\alpha_5}
  +\tilde{\rho}_0^{\frac{8}{3}+\alpha_5-3\alpha_3-\alpha_4}\big)L^3.
  \end{aligned}
\end{equation}

\end{proposition}

\par We leave the proof of Proposition \ref{qua prop} to Section \ref{qua}.

\subsection{Cubic Renormalization}\label{cub sec}
\
\par We define
\begin{equation}\label{define B'0}
  B^\prime=A^\prime-{A^\prime}^*
\end{equation}
where
\begin{equation}\label{define A'0}
  A^\prime=\sum_{k,p,q,\sigma,\nu}
  \eta_k\phi^+(k)\zeta^-(k)a^*_{p-k,\sigma}a^*_{q+k,\nu}a_{q,\nu}a_{p,\sigma}\chi_{p-k,q+k\notin B_F}\chi_{q\in A_{F,\delta_4}}
  \chi_{p\in B_F},
\end{equation}
and
\begin{equation}\label{A^nu_F,kappa0}
  A_{F,\delta_4}=\{k\in(2\pi/L)\mathbb{Z}^3,\,k_F<\vert k\vert\leq k_F+\tilde{\mu}^{\frac{1}{2}}\tilde{\rho}_0^{\delta_4}\}
\end{equation}
for $\frac{1}{3}>\delta_4\geq\delta_1$. The cutoff function $\zeta^-$ has been defined in (\ref{cutoffzeta}). Notice that (\ref{define A'0}) is pure fermionic, and first shown up in \cite{WJH}.

 \par We define
\begin{equation}\label{J_N}
  \mathcal{J}_N=e^{-B^\prime}\mathcal{G}_Ne^{B^\prime}.
\end{equation}
Then we collect the result of this cubic renormalization in Proposition \ref{cub prop} below.

\begin{proposition}\label{cub prop}
For any $N\in\mathbb{N}_{\geq0}$, and $\frac{1}{24}>\frac{1}{2}\delta_3>\frac{1}{4}\alpha_3>\alpha_2>2\alpha_4>0$, $\frac{1}{3}>\delta_4\geq\delta_1\geq\delta_2>\frac{1}{12}$, $\delta_4>2\alpha_3+2\alpha_4$, $\beta_1=\frac{1}{3}+\alpha_5$ and $2\alpha_3>\alpha_5>2\alpha_4$, we have
\begin{equation}\label{J_N cub prop}
  \mathcal{J}_N=\mathcal{K}+\mathcal{V}_{4,4h}+\mathcal{V}_{21}^\prime+\Omega
  +C_{\mathcal{J}_N}+\mathcal{E}_{\mathcal{J}_N},
\end{equation}
where
\begin{equation}\label{cub C_J_N}
  \begin{aligned}
  C_{\mathcal{J}_N}&=\frac{1}{2L^3}\Big(\hat{v}_0+\sum_{k\in(2\pi/L)\mathbb{Z}^3}
  \hat{v}_k\eta_k\Big)\mathbf{q}(\mathbf{q}-1)\bar{N}_0^2
  +\frac{4}{3\pi}
  \mathfrak{a}_0\tilde{\mu}^{\frac{3}{2}}(\mathbf{q}-1)(N-\mathbf{q}\bar{N}_0)\\
  &+\frac{1}{L^3}\sum_{k,p,q,\sigma,\nu}W_k\eta_k\phi^+(k)\zeta^-(k)
  \chi_{p-k,q+k\notin B_F}\chi_{p,q\in B_F}\\
  &-\frac{1}{L^3}\sum_{k,p,q,\sigma}W_k\eta_{k+q-p}\zeta^-(k)\phi^+(k)
  \phi^+(k+q-p)\chi_{p-k,q+k\notin B_F}\chi_{p,q\in B_F}\chi_{p\neq q}\\
  &-\sum_{k,p,q,\sigma}\eta_k\eta_{k+q-p}\zeta^-(k)
  \phi^+(k+q-p)k(q-p)
  \chi_{p-k,q+k\notin B_F}\chi_{p,q\in B_F}\chi_{p\neq q}
  \end{aligned}
\end{equation}
and
\begin{equation}\label{cub V_4,4h}
  \mathcal{V}_{4,4h}=\frac{1}{2L^3}\sum_{k,p,q,\sigma,\nu}
  \hat{v}_ka^*_{p-k,\sigma}a^*_{q+k,\nu}a_{q,\nu}a_{p,\sigma}\chi_{p,p-k\in P_{F,\delta_4}}\chi_{q,q+k\in P_{F,\delta_4}}.
\end{equation}
Moreover, the error is bounded by
\begin{equation}\label{cub Error}
   \begin{aligned}
  &\pm \mathcal{E}_{\mathcal{J}_N}\\
  \leq&
  C\big(\tilde{\rho}_0^{1+\delta_2}+\tilde{\rho}_0^{\frac{4}{3}+\alpha_2-3\delta_3-\alpha_4}\big)
  \mathcal{N}_{re}+C\big(\tilde{\rho}_0^{1+\delta_2}
  +\tilde{\rho}_0^{\frac{4}{3}-\alpha_3}\big)\tilde{\mathcal{N}}_{re}\\
  &+C\tilde{\rho}_0^{1-3\alpha_2}\mathcal{N}_i[\delta_4]
  +C\tilde{\rho}_0^{1-\delta_2-3\alpha_2}\mathcal{N}_h[\delta_4]
  +C\tilde{\rho}_0^{\frac{2}{3}-\alpha_4}\mathcal{N}_h[-\delta_3]\\
  &+C\big(\tilde{\rho}_0^{\frac{2}{3}+\delta_4-\alpha_2-\alpha_3}
  +\tilde{\rho}_0^{\frac{4}{3}+\alpha_2-\alpha_4-3\alpha_3}+
  \tilde{\rho}_0^{\frac{5}{6}+\frac{\alpha_2}{2}-\frac{3}{2}\alpha_5}\big)
  \mathcal{N}_h[-\alpha_2]\\
  &+C\tilde{\rho}_0^{1+\alpha_5
  -4\alpha_3-\alpha_4}\mathcal{N}_h[-\beta_1]
  +C\big(\tilde{\rho}_0^{\frac{1}{3}+2\alpha_3}+\tilde{\rho}_0^{\frac{1}{3}+
  \alpha_5-\alpha_4}\big)\mathcal{K}_s
  +C\tilde{\rho}_0^{\delta_4-\alpha_3-\alpha_4}\mathcal{K}_h[-\alpha_2]\\
  &+C\big(\tilde{\rho}_0^{\frac{7}{3}+\alpha_4}+\tilde{\rho}_0^{2+4\alpha_3-7\alpha_2}
  +\tilde{\rho}_0^{\frac{5}{2}+\frac{\alpha_2}{2}+\alpha_3-\frac{3}{2}\alpha_5}
  +\tilde{\rho}_0^{\frac{8}{3}+\alpha_5-3\alpha_3-\alpha_4}
  +\tilde{\rho}_0^{2+2\delta_4-\alpha_4}\big)L^3\\
  &+C\tilde{\rho}_0^{\alpha_4}\mathcal{V}_{4,4h}.
  \end{aligned}
\end{equation}

\end{proposition}

\par We leave the proof of Proposition \ref{cub prop} to Section \ref{cub}

\subsection{Bogoliubov Transformation}\label{bog sec}
\
\par The Bogoliubov transformation is defined through
\begin{equation}\label{define B tilde1}
  \tilde{B}=\frac{1}{2}(\tilde{A}-\tilde{A}^*)
\end{equation}
with
\begin{equation}\label{define A tilde1}
  \tilde{A}=\sum_{k,p,q,\sigma,\nu}
  \xi_{k,q,p}^{\nu,\sigma}a^*_{p-k,\sigma}a^*_{q+k,\nu}a_{q,\nu}a_{p,\sigma}\chi_{p-k,q+k\notin B_F}\chi_{p,q\in B_F}.
\end{equation}
Recall in Section \ref{Bog coeff sec}, we choose
\begin{equation}\label{define xi_k,q,p,nu,sigma1}
  \xi_{k,q,p}^{\nu,\sigma}=\frac{-\big(L^{-3}W_k\tilde{\zeta}^-(k)
  +\eta_kk(q-p)\phi^+(k)\tilde{\zeta}^-(k)\big)}
  {\frac{1}{2}\big(\vert q+k\vert^2+\vert p-k\vert^2-\vert q\vert^2-\vert p\vert^2\big)
  +\epsilon_0}
  \chi_{p-k,q+k\notin B_F}\chi_{p,q\in B_F},
\end{equation}
with $\epsilon_0=\tilde{\rho}_0^2$. Notice that in this Bogoliubov transformation, we renormalizes an extra term  $\Omega$ (compared with the bosonic case), together with the residue of the scattering equation cancellation $\mathcal{V}_{21}^\prime$, defined in (\ref{qua prop V_21' and Omega and V_3,l}). This observation also first came from \cite{WJH}.

\par We define
\begin{equation}\label{Z_N}
  \mathcal{Z}_N= e^{-\tilde{B}}\mathcal{J}_Ne^{\tilde{B}}
\end{equation}
We record the effect of the Bogoliubov transformation .
\begin{proposition}\label{bog prop}
  For any $N\in\mathbb{N}_{\geq0}$, and $\frac{1}{24}>\frac{1}{2}\delta_3>\frac{1}{4}\alpha_3>\alpha_2>\alpha_5>2\alpha_4>2\varepsilon >0$, $\frac{1}{3}>\delta_4\geq\delta_1\geq\delta_2>\frac{1}{12}$, $\delta_4>\alpha_6+\varepsilon>2\alpha_3+2\alpha_4+3\varepsilon$, $\beta_1=\frac{1}{3}+\alpha_5$, we have
\begin{equation}\label{Z_N bog prop}
  \mathcal{Z}_N=C_{\mathcal{Z}_N}+\mathcal{K}+e^{-\tilde{B}}\mathcal{V}_{4,4h}e^{\tilde{B}}
  +\mathcal{E}_{\mathcal{Z}_N},
\end{equation}
where
\begin{equation}\label{C_Z_N bog prop}
  \begin{aligned}
  &{C}_{\mathcal{Z}_N}=\\
  &{C}_{\mathcal{J}_N}+
  \sum_{k,p,q,\sigma,\nu}\big(\frac{W_k}{L^3}\tilde{\zeta}^-(k)
  +\eta_kk(q-p)\phi^+(k)\tilde{\zeta}^-(k)\big)\xi_{k,q,p}^{\nu,\sigma}
  \chi_{p-k,q+k\notin B_F}
  \chi_{p,q\in B_F}\\
  &-\sum_{k,p,q,\sigma}\big(\frac{W_k}{L^3}\tilde{\zeta}^-(k)
  +\eta_kk(q-p)\phi^+(k)\tilde{\zeta}^-(k)\big)\xi_{(k+q-p),p,q}^{\sigma,\sigma}
  \chi_{p-k,q+k\notin B_F}\chi_{p,q\in B_F}\chi_{p\neq q}
  \end{aligned}
\end{equation}
and
\begin{align}
   \pm&\mathcal{E}_{\mathcal{Z}_N}\nonumber\\
   &\leq
   C\big(\tilde{\rho}_0^{1+\delta_2}+\tilde{\rho}_0^{\frac{4}{3}+\alpha_2-3\delta_3-\alpha_4}
   +\tilde{\rho}_0^{\frac{5}{3}+\alpha_5-4\alpha_3-\alpha_4-2\varepsilon}
   +\tilde{\rho}_0^{\frac{4}{3}-4\alpha_3+\alpha_6-\alpha_4-2\varepsilon}\big)
   \mathcal{N}_{re}\nonumber\\
  & +C\tilde{\rho}_0^{\frac{7}{3}-4\alpha_6-\delta_2-3\varepsilon}\mathcal{N}_{re}
  +C\big(\tilde{\rho}_0^{1+\delta_2}
  +\tilde{\rho}_0^{\frac{4}{3}-\alpha_3}\big)\tilde{\mathcal{N}}_{re}
  +C\tilde{\rho}_0^{1-3\alpha_2}\mathcal{N}_i[\delta_4]\nonumber\\
  &+C\big(\tilde{\rho}_0^{1-\delta_2-3\alpha_2}+
  \tilde{\rho}_0^{\frac{5}{3}-4\alpha_6-\delta_2-\varepsilon}
  \big)\mathcal{N}_h[\delta_4]
  +C\tilde{\rho}_0^{\frac{2}{3}-\alpha_4}\mathcal{N}_h[-\delta_3]\nonumber\\
  & +C\big(\tilde{\rho}_0^{\frac{2}{3}+\delta_4-\alpha_2-\alpha_3}
  +\tilde{\rho}_0^{\frac{4}{3}+\alpha_2-\alpha_4-3\alpha_3}+
  \tilde{\rho}_0^{\frac{5}{6}+\frac{\alpha_2}{2}-\frac{3}{2}\alpha_5}\big)
  \mathcal{N}_h[-\alpha_2]\nonumber\\
  &+C\tilde{\rho}_0^{\frac{2}{3}-4\alpha_3+\alpha_6-\alpha_4}\mathcal{N}_h
    [-\alpha_6]+C\tilde{\rho}_0^{1+\alpha_5
  -4\alpha_3-\alpha_4}\mathcal{N}_h[-\beta_1]\nonumber\\
  &+C\big(\tilde{\rho}_0^{\frac{1}{3}+2\alpha_3}+\tilde{\rho}_0^{\frac{1}{3}+
  \alpha_5-\alpha_4}\big)\mathcal{K}_s
  +C\big(\tilde{\rho}_0^{\delta_4-\alpha_3-\alpha_4}+\tilde{\rho}_0^
  {\alpha_6-\alpha_4}\big)\mathcal{K}_h[-\alpha_2]
  \nonumber\\
  &+C\big(\tilde{\rho}_0^{\frac{7}{3}+\alpha_4}+\tilde{\rho}_0^{2+4\alpha_3-7\alpha_2}
  +\tilde{\rho}_0^{\frac{5}{2}+\frac{\alpha_2}{2}+\alpha_3-\frac{3}{2}\alpha_5}
  +\tilde{\rho}_0^{\frac{8}{3}+\alpha_5-3\alpha_3-\alpha_4}\big)L^3\nonumber\\
  &+C\big(\tilde{\rho}_0^{\frac{10}{3}-4\alpha_6-\alpha_3-\delta_2-\varepsilon}
  +\tilde{\rho}_0^{2+2\delta_4-\alpha_4}+
  \tilde{\rho}_0^{\frac{8}{3}-\alpha_3-\delta_2-3\alpha_2-\varepsilon}
  +\tilde{\rho}_0^{\frac{7}{3}+2\delta_3-\alpha_3-\alpha_4}\big)L^3\nonumber\\
  &+Ce^{-\tilde{B}}\tilde{\rho}_0^{\alpha_4}\mathcal{V}_{4,4h}e^{\tilde{B}}.\label{prop bog control error}
   \end{align}
\end{proposition}
\par We leave the proof of Proposition \ref{bog prop} to Section \ref{bog}.

\section{Proof of the Main Theorem}\label{main}
\subsection{More about (\ref{C_Z_N bog prop})}
\begin{lemma}\label{lemma cal C_Z_N}
  For any $N\in\mathbb{N}_{\geq0}$, and $\frac{1}{24}>\frac{1}{2}\delta_3>\frac{1}{4}\alpha_3>\alpha_2>\alpha_5>2\alpha_4>2\varepsilon >0$, $\frac{1}{3}>\delta_4\geq\delta_1\geq\delta_2>\frac{1}{12}$, $\delta_4>\alpha_6+\varepsilon>2\alpha_3+2\alpha_4+3\varepsilon$, $\beta_1=\frac{1}{3}+\alpha_5$, we have
  \begin{equation}\label{cal C_Z_N}
    \begin{aligned}
    {C_{\mathcal{Z}_N}}&=\Big(4\pi\mathfrak{a}_0\Big(1-\frac{1}{\mathbf{q}}\Big)
    \tilde{\rho}_0^2
    +\frac{12}{35}(11-2\ln 2)3^{\frac{1}{2}}\pi^{\frac{2}{3}}\mathfrak{a}_0^2
    2^{\frac{4}{3}}\mathbf{q}^{-\frac{1}{3}}\Big(1-\frac{1}{\mathbf{q}}\Big)
    \tilde{\rho}_0^\frac{7}{3}\Big)L^3\\
    &+\frac{4}{3\pi}
  \mathfrak{a}_0\tilde{\mu}^{\frac{3}{2}}(\mathbf{q}-1)(N-\mathbf{q}\bar{N}_0)
    +O\big(\tilde{\rho}_0^{2+2\alpha_1}L^3+\tilde{\rho}_0^{\frac{8}{3}-2\alpha_3}L^3
    +\tilde{\rho}_0^{\frac{7}{3}+2\alpha_4}L^3\big).
    \end{aligned}
  \end{equation}
  For brevity, we might place the error in (\ref{cal C_Z_N}) to $\mathcal{E}_{\mathcal{Z}_N}$,  and denote the constant independent of $N$ by
  \begin{equation}\label{df E_0}
    \begin{aligned}
    E_0&=\Big(4\pi\mathfrak{a}_0\Big(1-\frac{1}{\mathbf{q}}\Big)
    \tilde{\rho}_0^2
    +\frac{12}{35}(11-2\ln 2)3^{\frac{1}{2}}\pi^{\frac{2}{3}}\mathfrak{a}_0^2
    2^{\frac{4}{3}}\mathbf{q}^{-\frac{1}{3}}\Big(1-\frac{1}{\mathbf{q}}\Big)
    \tilde{\rho}_0^\frac{7}{3}\Big)L^3\\
    &-\frac{4}{3\pi}
  \mathfrak{a}_0\tilde{\mu}^{\frac{3}{2}}(\mathbf{q}-1)\mathbf{q}\bar{N}_0.
    \end{aligned}
  \end{equation}
\end{lemma}
\begin{proof}
  \par We follow exactly as the proof of \cite[Lemma 6.1]{WJH} to get to (recall that we choose $\epsilon_0=\tilde{\rho}_0^2$):
  \begin{equation}\label{cal C_Z_N penult}
     \begin{aligned}
    {C_{\mathcal{Z}_N}}&=4\pi\mathfrak{a}_0\Big(1-\frac{1}{\mathbf{q}}\Big)
    \tilde{\rho}_0^2L^3
    +\frac{4}{3\pi}
  \mathfrak{a}_0\tilde{\mu}^{\frac{3}{2}}(\mathbf{q}-1)(N-\mathbf{q}\bar{N}_0)\\
    &+\sum_{\substack{\sigma\neq\nu\\k,p,q}}\frac{(4\pi\mathfrak{a}_0)^2}{L^6}
  \Big(\frac{1}{\vert k\vert^2}-\frac{\chi_{p-k,q+k\notin B_F}}{\vert k\vert^2+k(q-p)+\epsilon_0}\Big)\chi_{p,q\in B_F}\chi_{\vert k\vert>4k_F}\\
  &+\sum_{\substack{\sigma\neq\nu\\k,p,q}}\frac{(4\pi\mathfrak{a}_0)^2}{2L^6}
  \Big(\frac{2}{\vert k\vert^2}-\frac{1}{\vert k\vert^2+k(q-p)+\epsilon_0}\\
  &\quad\quad\quad\quad\quad\quad
  -\frac{1}{\vert k\vert^2-k(q-p)+\epsilon_0}\Big)\chi_{p,q\in B_F}\chi_{0<\vert k\vert\leq4k_F}\\
    &+O\big(\tilde{\rho}_0^{2+2\alpha_1}L^3+\tilde{\rho}_0^{\frac{8}{3}-2\alpha_3}L^3
    +\tilde{\rho}_0^{\frac{7}{3}+2\alpha_4}L^3\big).
    \end{aligned}
  \end{equation}
  The summations on the second and the third line of the right-hand side of (\ref{cal C_Z_N penult}) converge absolutely. Thus by dominated convergence theorem, in the large $L$ limit, they become
  \begin{equation}\label{int lim L}
    \begin{aligned}
    &\frac{(4\pi\mathfrak{a}_0)^2}{(2\pi)^9}\mathbf{q}(\mathbf{q}-1)IL^3
    +o_{L\to\infty}(1)L^3,
    \end{aligned}
  \end{equation}
  where $\lim_{L\to\infty}o_{L\to\infty}(1)=0$, and
  \begin{equation}\label{int lim L int}
   \begin{aligned}
    I&=\int_
    {\vert q\vert<k_F}\int_{\vert p\vert<k_F}\int_{k\in\mathbb{R}^3}
    \Big(\frac{1}{\vert k\vert^2}-\frac{2\chi_{p-k,q+k\notin B_F}}{
    \vert p-k\vert^2+\vert q+k\vert^2-\vert p\vert^2-\vert q\vert^2+2\epsilon_0}\Big)\\
    &=\int_
    {\vert q\vert<k_F}\int_{\vert p\vert<k_F}\int_{k\in\mathbb{R}^3}
    \Big(\frac{1}{\vert k\vert^2}-\frac{2\chi_{p-k,q+k\notin B_F}}{
    \vert p-k\vert^2+\vert q+k\vert^2-\vert p\vert^2-\vert q\vert^2}\Big)+O(\epsilon_0
    \tilde{\rho}_0^{\frac{5}{3}}).
   \end{aligned}
  \end{equation}
  The second equality comes from \cite[Lemma 6.1]{2024huangyangformulalowdensityfermi}. From, for example, \cite[Appendix A]{WJH} and \cite[Appendix]{soviet}, together with the expansion of $\tilde{\rho}_0$ in (\ref{asymptotic rho_0})(we have $k_F=\tilde{\mu}^{\frac{1}{2}}$), we can calculate the exact value of the integral on the second line of (\ref{int lim L int}), and thus combining this result with (\ref{cal C_Z_N penult}) we reach (\ref{cal C_Z_N}).

  \par We sketch the idea in proving (\ref{cal C_Z_N penult}). We write
  \begin{align*}
    &C_1=\frac{1}{2L^3}\Big(\hat{v}_0+\sum_{k\in(2\pi/L)\mathbb{Z}^3}
  \hat{v}_k\eta_k\Big)\mathbf{q}(\mathbf{q}-1)\bar{N}_0^2\\
  &C_2=
  \frac{1}{L^3}\sum_{k,p,q,\sigma,\nu}W_k\eta_k\phi^+(k)\zeta^-(k)
  \chi_{p-k,q+k\notin B_F}\chi_{p,q\in B_F}\\
  &-\frac{1}{L^3}\sum_{k,p,q,\sigma}W_k\eta_{k+q-p}\zeta^-(k)\phi^+(k)
  \phi^+(k+q-p)\chi_{p-k,q+k\notin B_F}\chi_{p,q\in B_F}\chi_{p\neq q}\\
  &-\sum_{k,p,q,\sigma}\eta_k\eta_{k+q-p}\zeta^-(k)
  \phi^+(k+q-p)k(q-p)
  \chi_{p-k,q+k\notin B_F}\chi_{p,q\in B_F}\chi_{p\neq q}\\
  &C_3=\sum_{k,p,q,\sigma,\nu}\big(\frac{W_k}{L^3}\tilde{\zeta}^-(k)
  +\eta_kk(q-p)\phi^+(k)\tilde{\zeta}^-(k)\big)\xi_{k,q,p}^{\nu,\sigma}
  \chi_{p-k,q+k\notin B_F}
  \chi_{p,q\in B_F}\\
  &-\sum_{k,p,q,\sigma}\big(\frac{W_k}{L^3}\tilde{\zeta}^-(k)
  +\eta_kk(q-p)\phi^+(k)\tilde{\zeta}^-(k)\big)\xi_{(k+q-p),p,q}^{\sigma,\sigma}
  \chi_{p-k,q+k\notin B_F}\chi_{p,q\in B_F}\chi_{p\neq q}
  \end{align*}

  \par For $C_1$, using (\ref{asymptotic rho_0}), (\ref{average particle number}), (\ref{est of int vf_l}) and (\ref{discrete int vf}), we deduce
  \begin{equation}\label{C_1}
    \begin{aligned}
    C_1&=\big(4\pi\mathfrak{a}_0+6\pi\mathfrak{a}_0^2(\ell L)^{-1}\big)
    \Big(1-\frac{1}{\mathbf{q}}\Big)\tilde{\rho}_0^2L^3+
    O\big(\tilde{\rho}_0^{\frac{8}{3}-2\alpha_3}L^3+\tilde{\rho}_0^{2+2\alpha_1}L^3\big).
    \end{aligned}
  \end{equation}
  \par For $C_3$, to discard the cutoff $\tilde{\zeta}^-$, we use the bounds
  \begin{equation}\label{bound C_3 zeta^- 1}
    \begin{aligned}
    \frac{\vert W_k\vert}{L^3}\leq \frac{C}{L^3},
    \quad
    \vert\eta_k\vert\vert k\vert^2\leq \frac{C}{L^3},
    \quad\sum_{k}\frac{\tilde{\zeta}^+}{L^3\vert k\vert^2}\leq \tilde{\rho}_0^{\alpha_6-\frac{1}{3}}
    \end{aligned}
  \end{equation}
  and estimates (\ref{special zeta tilde eta W L2 norm}) to control the residues by
  \begin{equation}\label{bound C_3 zeta^- 2}
    \tilde{\rho}_0^{\frac{7}{3}+\alpha_6-2\alpha_3}L^3.
  \end{equation}
  To discard the cutoff ${\phi}^+$, we use the bounds
  \begin{equation}\label{bound C_3 phi^+ 1}
    \begin{aligned}
    \frac{\vert W_k\vert}{L^3}\leq \frac{C}{L^3},
    \quad
    \vert\eta_k\vert\vert k\vert\leq \frac{C\tilde{\rho}_0^{-\frac{1}{3}+\alpha_3}}{L^3}
    \end{aligned}
  \end{equation}
  and the estimate in the form of (\ref{est sum A_F d <=0}) to control the residues by
  \begin{equation}\label{bound C_3 phi^+ 2}
    \tilde{\rho}_0^{\frac{7}{3}+\alpha_3-\alpha_2}L^3.
  \end{equation}
  Thus
  \begin{equation}\label{C_3}
    \begin{aligned}
    C_3&=-\sum_{\substack{\sigma,\nu\\k,p,q}}\frac{\big(
    L^{-3}W_k+\eta_kk(q-p)\big)^2}{\vert k\vert^2+k(q-p)+\epsilon_0}
    \chi_{p-k,q+k\notin B_F}\chi_{p,q\in B_F}\\
    &+\sum_{\substack{\sigma,\nu\\k,p,q}}\frac{\big(
    L^{-3}W_k+\eta_kk(q-p)\big)\big(L^{-3}W_{k+q-p}
    \eta_{k+q-p}(k+q-p)(p-q)\big)}{\vert k\vert^2+k(q-p)+\epsilon_0}\\
    &\quad\quad\quad\quad\quad\quad\quad\times
    \chi_{p-k,q+k\notin B_F}\chi_{p,q\in B_F}
    +O(\tilde{\rho}_0^{\frac{7}{3}+2\alpha_4}L^3)
    \end{aligned}
  \end{equation}

  \par For $C_2$, we first notice that from (\ref{cutoff effect})
  \begin{equation}\label{cutoff effect0}
  \begin{aligned}
    \sum_{k,p,q,\sigma,\nu}\eta^2_kk(q-p)\phi^+(k)\zeta^-(k)\chi_{p-k,q+k\notin B_F}
    \chi_{p,q\in B_F}=0.
  \end{aligned}
  \end{equation}
  To discard the cutoff $\phi^+$ and $\zeta^-$, we have the following estimates, for $p,q\in B_F$:
  \begin{equation}\label{C_2 discard1}
    \begin{aligned}
    \frac{1}{L^3}\sum_{k,p,q}\vert W_k\eta_k\phi^-(k)\vert
    +\vert W_k\eta_{k+q-p}\phi^-(k+q-p)\vert&\leq C\tilde{\rho}_0^{\frac{7}{3}
    +2\alpha_3-3\alpha_2}L^3\\
    \frac{1}{L^3}\sum_{k,p,q}\vert W_k\eta_k\zeta^+(k)\vert
    +\vert W_k\eta_{k+q-p}\zeta^+(k)\vert&\leq C\tilde{\rho}_0^{\frac{8}{3}
    -2\alpha_3+\alpha_5}L^3
    \end{aligned}
  \end{equation}
  and
  \begin{equation}\label{C_2 discard 2}
  \begin{aligned}
    \frac{1}{L^3}\sum_{k,p,q}\Big(\vert \eta_k^2k(q-p)\phi^-(k)\vert
    +\vert \eta_k\eta_{k+q-p}k(q-p)&\phi^-(k+q-p)\vert\\
    +\vert \eta_k\eta_{k+q-p}k(q-p)\phi^-(k)\vert\Big)&\leq C\tilde{\rho}_0^{\frac{7}{3}
    +3\alpha_3-3\alpha_2}L^3\\
    \frac{1}{L^3}\sum_{k,p,q}\vert \eta_k^2k(q-p)\zeta^+(k)\vert
    +\vert \eta_k\eta_{k+q-p}k(q-p)\zeta^+(k)\vert&\leq C\tilde{\rho}_0^{\frac{8}{3}
    -\alpha_3+\alpha_5}L^3
  \end{aligned}
  \end{equation}
  Moreover, we have for $p,q\in B_F$ (comparing with \cite[6.8]{WJH}):
  \begin{equation}\label{C_2 pneq q}
    \frac{1}{L^3}\sum_{k,p,q}\vert W_k\eta_{k+q-p}\vert\chi_{p=q}\leq C(\tilde{\rho}_0)\ll L^3.
  \end{equation}
  We write (corresponds to \cite[(6.9)]{WJH})
  \begin{equation}\label{Z_k}
    \frac{1}{L^3}Z_k=\frac{1}{2L^3}\sum_{m}\hat{v}_{k-m}\eta_m+\frac{1}{2L^3}\hat{v}_k.
  \end{equation}
  Using (\ref{eqn of eta_p}), we have (corresponds to \cite[(6.11)]{WJH})
  \begin{equation}\label{Z_k eq}
   \begin{aligned}
   \big(\vert k\vert^2+k(q-p)+\epsilon_0\big)\eta_k
   &=\big(\frac{W_k}{L^3}+k(q-p)\eta_k+\epsilon_0\eta_k\big)-\frac{Z_k}{L^3}\\
   \big(\vert k\vert^2+k(q-p)+\epsilon_0\big)\eta_{k+q-p}
   &=\big(\frac{W_{k+q-p}}{L^3}+(k+q-p)(p-q)\eta_{k+q-p}\\
   &\quad\quad+
   \epsilon_0\eta_{k+q-p}\big)-\frac{Z_{k+q-p}}{L^3}
   \end{aligned}
 \end{equation}
 Using estimates (corresponds to \cite[(6.18)]{WJH})
 \begin{equation}\label{Z_k-Z_0}
   \frac{1}{L^3}\vert Z_{k+q-p}-Z_k\vert\leq \frac{C}{L^3}\tilde{\rho}^{\frac{1}{3},}
 \end{equation}
 and (corresponds to \cite[(6.26)]{WJH})
 \begin{equation}\label{Z W}
   \begin{aligned}
   {\vert Z_k\vert}\leq{C},&\quad
   {\vert W_k\vert}\leq{C}\\
   \vert Z_k-Z_0\vert\leq C\vert k\vert^2,&\quad
   \vert W_k-W_0\vert\leq C\vert k\vert^2\tilde{\rho}_0^{-\frac{2}{3}+2\alpha_3}\\
   \vert Z_0-4\pi\mathfrak{a}_0\vert\leq \tilde{\rho}_0^{\frac{1}{3}-\alpha_3},&\quad
   \vert W_0-4\pi\mathfrak{a}_0\vert\leq \tilde{\rho}_0^{\frac{1}{3}-\alpha_3}
   \end{aligned}
 \end{equation}
 Then following the course in \cite[Lemma 6.1]{WJH}, we can deduce
 \begin{equation}\label{C_2}
   \begin{aligned}
   C_2&=\sum_{\substack{\sigma\neq\nu\\k,p,q}}\frac{(4\pi\mathfrak{a}_0)^2}{L^6}
  \Big(\frac{1}{\vert k\vert^2}-\frac{\chi_{p-k,q+k\notin B_F}}{\vert k\vert^2+k(q-p)+\epsilon_0}\Big)\chi_{p,q\in B_F}\chi_{\vert k\vert>4k_F}\\
  &+\sum_{\substack{\sigma\neq\nu\\k,p,q}}\frac{(4\pi\mathfrak{a}_0)^2}{2L^6}
  \Big(\frac{2}{\vert k\vert^2}-\frac{1}{\vert k\vert^2+k(q-p)+\epsilon_0}\\
  &\quad\quad\quad\quad\quad\quad
  -\frac{1}{\vert k\vert^2-k(q-p)+\epsilon_0}\Big)\chi_{p,q\in B_F}\chi_{0<\vert k\vert\leq4k_F}\\
  &-\sum_{\substack{\sigma\neq\nu\\k\neq0,p,q}}\frac{W_kZ_k}{L^6\vert k\vert^2}
  \chi_{p,q\in B_F}
  -C_3+O(\tilde{\rho}_0^{\frac{7}{3}+2\alpha_4}L^3).
   \end{aligned}
 \end{equation}

 \par Notice that
  \begin{equation}\label{correction}
 6\pi\mathfrak{a}_0^2(\ell L)^{-1}
    \Big(1-\frac{1}{\mathbf{q}}\Big)\tilde{\rho}_0^2L^3
    -\sum_{\substack{\sigma\neq\nu\\k\neq0,p,q}}\frac{W_kZ_k}{L^6\vert k\vert^2}
  \chi_{p,q\in B_F}=O(\tilde{\rho}_0^{\frac{8}{3}-2\alpha_3}L^3).
  \end{equation}
  For detailed proof of (\ref{correction}), one can check \cite[Lemma 5.4]{2018Bogoliubov} or \cite[Lemma 5.1]{me}

 \par Then combining (\ref{C_Z_N bog prop}), (\ref{C_1}), (\ref{C_2}), (\ref{C_3}) and (\ref{correction}), we reach (\ref{cal C_Z_N penult}).
\end{proof}

\subsection{Proof of (\ref{core pressure})}
\
\par Since $\alpha_1>\frac{1}{6}$, we set $d=\min\{\alpha_1-\frac{1}{6},\frac{1}{9}\}$, and we choose
\begin{equation}\label{choose delta}
  \delta_4=\delta_1=\frac{1}{6}+\frac{d}{10}<\min\{\alpha_1,\frac{1}{3}\},
  \quad\delta_2=\frac{1}{6}-\frac{3}{4}d>\frac{1}{12}.
\end{equation}
We also set
\begin{equation}\label{choose alpha}
\begin{aligned}
  \delta_3=\frac{1}{24}+\frac{3}{160}d<\frac{1}{12},
  \quad\alpha_3=\frac{1}{12}+\frac{d}{40}<2\delta_3,
  \quad\alpha_6=\frac{1}{6}+\frac{3}{50}d,
\end{aligned}
\end{equation}
and
\begin{equation}\label{choose alpha small}
  \alpha_2=\frac{d}{160},\quad\alpha_5=\frac{d}{180},\quad
  \alpha_4=\frac{d}{400},\quad\varepsilon=\frac{d}{800}.
\end{equation}
We choose
\begin{equation}\label{choose kappa}
  \kappa=\frac{4}{3\pi}\mathfrak{a}_0(\mathbf{q}-1).
\end{equation}
Therefore, from (\ref{mu tilde}), we have
\begin{equation}\label{choose mu tilde}
  \tilde{\mu}=\mu-\kappa\tilde{\mu}^{\frac{3}{2}}.
\end{equation}
We let, for any $N\in\mathbb{N}_{\geq0}$,
\begin{equation}\label{define U unitary}
  \mathcal{U}=e^{B}e^{B^\prime}e^{\tilde{B}}.
\end{equation}
Notice that $\mathcal{U}$ can also be regarded as an unitary operator on $\mathcal{F}$. From Proposition \ref{bog prop} and Lemma \ref{lemma cal C_Z_N}, we have
\begin{equation}\label{UHU}
  \mathcal{U}^*(H-\mu\mathcal{N})\mathcal{U}=\mathcal{Z}=E_0+(\mathcal{K}-\tilde{\mu}\mathcal{N})
  +e^{-\tilde{B}}\mathcal{V}_{4,4h}e^{\tilde{B}}+\mathcal{E}_{\mathcal{Z}}.
\end{equation}
We also let $\Gamma^{\mathcal{U}}=\mathcal{U}^*\Gamma\mathcal{U}$ for any state $\Gamma$, and

\par Now we are going to bound $-L^3P^L(\beta,\mu)$, using Proposition \ref{bog prop} and Lemma \ref{lemma cal C_Z_N}, and a-priori estimates in Section \ref{a-prior}.
\begin{proof}[Lower bound]
for some constant $0<c<1$:
  \begin{align}
  &-L^3P^L(\beta,\mu)= -L^3P^L[G]=\tr(H-\mu\mathcal{N})G-\frac{1}{\beta}S[G]\nonumber\\
  &=\tr(\mathcal{Z}-\mu\mathcal{N})G^{\mathcal{U}}-\frac{1}{\beta}S[G^{\mathcal{U}}]\nonumber\\
  &=E_0+\frac{1}{\beta}S(G^{\mathcal{U}},\tilde{G}_0)-L^3\tilde{P}^L_0[\tilde{G}_0]
  +\tr e^{-\tilde{B}}\mathcal{V}_{4,4h}e^{\tilde{B}}G^\mathcal{U}
  +\tr\mathcal{E}_{\mathcal{Z}} G^\mathcal{U}\nonumber\\
  &\geq E_0-\tilde{P}_0^L(\beta,\tilde{\mu})+(1-c)\big(
  \beta^{-1}S(G^{\mathcal{U}},\tilde{G}_0)+
  \tr e^{-\tilde{B}}\mathcal{V}_{4,4h}e^{\tilde{B}}G^\mathcal{U}\big)
  +C\tilde{\rho}_0^{\frac{7}{3}+\frac{d}{400}}L^3\nonumber\\
  &\geq E_0-\tilde{P}_0^L(\beta,\tilde{\mu})
  +C\tilde{\rho}_0^{\frac{7}{3}+\frac{d}{400}}L^3
  \label{lower bound}
  \end{align}
\end{proof}

\begin{proof}[Upper bound]
We choose the test state $\Gamma_0$ such that $\Gamma_0^{\mathcal{U}}=\tilde{G}_0$.
  \begin{align}
  &-L^3P^L(\beta,\mu)\leq -L^3P^L[\Gamma_0]
  =\tr(H-\mu\mathcal{N})\Gamma_0-\frac{1}{\beta}S[\Gamma_0]\nonumber\\
  &=\tr(\mathcal{Z}-\mu\mathcal{N})
  \Gamma_0^{\mathcal{U}}-\frac{1}{\beta}S[\Gamma_0^{\mathcal{U}}]\nonumber\\
  &=E_0-\tilde{P}_0^L(\beta,\tilde{\mu})+C\tilde{\rho}_0^{\frac{7}{3}+\frac{d}{400}}L^3
  +C\tr e^{-\tilde{B}}\mathcal{V}_{4,4h}e^{\tilde{B}}\tilde{G}_0\nonumber\\
  &\leq
  E_0-\tilde{P}_0^L(\beta,\tilde{\mu})+C\tilde{\rho}_0^{\frac{7}{3}+\frac{d}{400}}L^3
  \label{upper bound}
  \end{align}
\end{proof}

\par Recall that from (\ref{density non intercating Fermi Dirac dis}), (\ref{asymptotic rho_0}) and \cite{dingle1973asymptotic}, we have
\begin{equation}\label{sim}
   \rho_0^{\frac{2}{3}}\sim\mu\sim\tilde{\mu}, \quad \frac{\partial\rho_0}{\partial\mu}\sim\mu^{\frac{1}{2}}\sim\rho_0^{\frac{1}{3}},\quad
   \frac{\partial^2\rho_0}{\partial\mu^2}\sim\mu^{-\frac{1}{2}}\sim\rho_0^{-\frac{1}{3}}.
\end{equation}
We then use (\ref{choose mu tilde}) and Taylor expansion to obtain
\begin{equation}\label{mu expand}
  \tilde{\mu}^{\frac{3}{2}}=\mu^{\frac{3}{2}}-\frac{3}{2}\kappa\mu^2+O(\rho_0^{\frac{5}{3}}).
\end{equation}
Also using (\ref{choose mu tilde}) and Taylor expansion, we can deduce
\begin{equation}\label{1st expand}
  4\pi\mathfrak{a}_0\Big(1-\frac{1}{\mathbf{q}}\Big)
    \tilde{\rho}_0^2=4\pi\mathfrak{a}_0\Big(1-\frac{1}{\mathbf{q}}\Big)
    {\rho}_0^2-(8\pi\mathfrak{a}_0)^2\Big(1-\frac{1}{\mathbf{q}}\Big)^2
    {\rho}_0^2\frac{\partial\rho_0}{\partial\mu}
    +O(\rho_0^{\frac{8}{3}}),
\end{equation}
and
\begin{equation}\label{3rd expand}
  \begin{aligned}
  \lim_{L\to\infty}-\tilde{P}_0^L(\beta,\tilde{\mu})=&
  -P_0(\beta,\tilde{\mu})\\
  =&-P_0(\beta,\mu)+8\pi\mathfrak{a}_0\Big(1-\frac{1}{\mathbf{q}}\Big){\rho}_0^2
  -\frac{3}{2}\kappa^2\mu^2\rho_0\\
  &-\frac{1}{2}(8\pi\mathfrak{a}_0)^2\Big(1-\frac{1}{\mathbf{q}}\Big)^2
    {\rho}_0^2\frac{\partial\rho_0}{\partial\mu}
    +O(\rho_0^{\frac{8}{3}}+\rho_0^{2+2\alpha_1}).
  \end{aligned}
\end{equation}
On the other hand, (\ref{asymptotic rho_0}), (\ref{average particle number}), (\ref{mu expand}) and (\ref{1st expand}) together yield
\begin{equation}\label{2nd expand}
  \begin{aligned}
  &\lim_{L\to\infty}-\frac{4}{3\pi}
  \mathfrak{a}_0\tilde{\mu}^{\frac{3}{2}}(\mathbf{q}-1)\frac{\mathbf{q}\bar{N}_0}{L^3}\\
  =&-8\pi\mathfrak{a}_0\Big(1-\frac{1}{\mathbf{q}}\Big){\rho}_0^2
  +\frac{3}{2}\kappa^2\mu^2\rho_0
  +(8\pi\mathfrak{a}_0)^2\Big(1-\frac{1}{\mathbf{q}}\Big)^2
    {\rho}_0^2\frac{\partial\rho_0}{\partial\mu}
    +O(\rho_0^{\frac{8}{3}}+\rho_0^{2+2\alpha_1}).
  \end{aligned}
\end{equation}
Therefore,
\begin{equation}\label{final}
  \begin{aligned}
  P(\beta,\mu)=\lim_{L\to\infty}P^L(\beta,\mu)=
  &P_0(\beta,\mu)
    -4\pi\mathfrak{a}_0\Big(1-\frac{1}{\mathbf{q}}\Big)\rho_0^2\\
    &-\frac{12}{35}(11-2\ln 2)3^{\frac{1}{2}}\pi^{\frac{2}{3}}\mathfrak{a}_0^2
    2^{\frac{4}{3}}\mathbf{q}^{-\frac{1}{3}}\Big(1-\frac{1}{\mathbf{q}}\Big)
    \rho_0^\frac{7}{3}\\
    &+\frac{1}{2}(8\pi\mathfrak{a}_0)^2\Big(1-\frac{1}{\mathbf{q}}\Big)^2\rho_0^2\frac{\partial \rho_0}
    {\partial\mu}+O\big(\rho_0^{\frac{7}{3}+\frac{d}{400}}\big)
  \end{aligned}
\end{equation}
which is exactly (\ref{core pressure}).

\subsection{Proof of (\ref{core density})}
\
\par (\ref{core density}) is a direct corollary of (\ref{core pressure}), due to the convexity of $P(\beta,\mu)$ with respect to $\mu$. See \cite[Section 5]{FermithermoTpositive}. More precisely, we use the convexity of $P(\beta,\mu)$ to obtain
\begin{equation}\label{convexity}
  \frac{P(\beta,\mu)-P(\beta,\mu-\delta)}{\delta}\leq\rho_-(\beta,\mu)\leq
  \rho_+(\beta,\mu)\leq\frac{P(\beta,\mu+\delta)-P(\beta,\mu)}{\delta}
\end{equation}
with $\delta=\rho_0(\beta,\mu)^{1+\frac{d}{800}}$. With this choice of $\delta$, it is easy to find
\begin{equation*}
  \rho_0(\beta,\mu\pm\delta)\sim\rho_0(\beta,\mu)\Rightarrow
  \beta\mu\gtrsim\rho_0(\beta,\mu\pm\delta)^{-\alpha_1},
\end{equation*}
and therefore $P(\beta,\mu\pm\delta)$ can be expanded using (\ref{core pressure}). The rest is using estimates similar to (\ref{sim}), (\ref{1st expand}) and (\ref{3rd expand}) to reach (\ref{core density}).

\section{Quadratic Renormalization}\label{qua}
\
\par In this section, we analyze the excitation Hamiltonian $\mathcal{G}_N$ defined in (\ref{G_N}), and prove Proposition \ref{qua prop}. Recall that we have chosen
\begin{equation}\label{choose ell qua}
  \ell L=\tilde{\rho}_0^{-\frac{1}{3}+\alpha_3}
\end{equation}
for some $0<\alpha_3<\frac{1}{3}$ to be determined. Also, we define
\begin{equation}\label{define B}
  B=\frac{1}{2}(A-A^*)
\end{equation}
and
\begin{equation}\label{define A}
  A=\sum_{k,p,q,\sigma,\nu}
  \eta_k\phi^+(k)a^*_{p-k,\sigma}a^*_{q+k,\nu}a_{q,\nu}a_{p,\sigma}\chi_{p-k,q+k\notin
  B_F}\chi_{p,q\in B_F}.
\end{equation}
Recall that $\phi^+$ is the smooth high frequency cut-off function defined below (\ref{derivatives bound cutoff phi^-}).

\par Using (\ref{second quantization H_N}), (\ref{split V}) and Newton-Leibniz formula,  we rewrite $\mathcal{G}_N$ by
\begin{equation}\label{define G_N}
  \begin{aligned}
  \mathcal{G}_N\coloneqq e^{-B}H_Ne^B=&\mathcal{K}+\mathcal{V}_{4}+\mathcal{V}_{21}^\prime
  +\Omega+e^{-B}\big(\mathcal{V}_{0}
  +\mathcal{V}_{1}+\mathcal{V}_{22}+\mathcal{V}_{23}+\mathcal{V}_{3}\big)e^B\\
  &+\int_{0}^{1}e^{-tB}\Gamma e^{tB}dt
  +\int_{0}^{1}\int_{0}^{t}e^{-sB}[\mathcal{V}_{21}^\prime+\Omega,B]e^{sB}dsdt\\
  &+\int_{0}^{1}\int_{t}^{1}e^{-sB}[\mathcal{V}_{21},B]e^{sB}dsdt
  \end{aligned}
\end{equation}
with
\begin{equation}\label{define V_21' and Omega}
  \begin{aligned}
  \mathcal{V}_{21}^\prime&=\frac{1}{L^3}\sum_{k,p,q,\sigma,\nu}
  W_k\zeta^{-}(k)(a^*_{p-k,\sigma}a^*_{q+k,\nu}a_{q,\nu}a_{p,\sigma}+h.c.)
  \chi_{p-k,q+k\notin
  B_F}\chi_{p,q\in B_F}\\
  \Omega&=\sum_{k,p,q,\sigma,\nu}
  \eta_k\phi^+(k)\zeta^{-}(k)k(q-p)
  \\
  &\quad\quad\quad\quad
  \times(a^*_{p-k,\sigma}a^*_{q+k,\nu}a_{q,\nu}a_{p,\sigma}+h.c.)\chi_{p-k,q+k\notin B_F}
  \chi_{p,q\in B_F}
  \end{aligned}
\end{equation}
and
\begin{equation}\label{define Gamma}
  \Gamma=[\mathcal{K}+\mathcal{V}_4,B]+\mathcal{V}_{21}-\mathcal{V}_{21}^\prime-\Omega.
\end{equation}

\par To prove Proposition \ref{qua prop}, we are going to analyze each term on the right-hand side of (\ref{define G_N}), in Lemmas \ref{lemma qua conj V_0 1 22 23}-\ref{lemma qua Gamma}. Before we rigorously calculate each term of $\mathcal{G}_N$, we need to bound the action of quadratic renormalization on some special operators, such as $\mathcal{N}_{re}, \mathcal{K}_s$ and $\mathcal{V}_4$. These results are collected respectively in Lemmas \ref{lemma qua control N_re}-\ref{lemma qua control V_4}.

\par Lemma \ref{lemma qua control N_re} controls the action of $e^B$ on several particle number operators defined in Section \ref{particle}.

\begin{lemma}\label{lemma qua control N_re}
For any $N\in\mathbb{N}_{\geq0}$ and $\vert t\vert\leq 1$,
\begin{align}
 e^{-tB}\mathcal{N}_{re}e^{tB}&\leq C\mathcal{N}_{re}+C\tilde{\rho}_0^{\frac{5}{3}+\alpha_3}L^3,
 \label{qua control N_re}\\
 e^{-tB}\tilde{\mathcal{N}}_{re}e^{tB}&\leq C\tilde{\mathcal{N}}_{re}+ C\mathcal{N}_{re}+C\tilde{\rho}_0^{\frac{5}{3}+\alpha_3}L^3,
 \label{qua control N_retilde}\\
 \pm\big(e^{-tB}\mathcal{N}_{re}e^{tB}-\mathcal{N}_{re}\big)&\leq
 C\tilde{\rho}_0^{\frac{1}{3}}\mathcal{N}_{re}+C\tilde{\rho}_0^{\frac{4}{3}+\alpha_3}L^3.
 \label{qua control diff N_re}\\
  \pm\big(e^{-tB}\tilde{\mathcal{N}}_{re}e^{tB}-\tilde{\mathcal{N}}_{re}\big)&\leq
 C\tilde{\rho}_0^{\frac{1}{3}}\mathcal{N}_{re}+C\tilde{\rho}_0^{\frac{4}{3}+\alpha_3}L^3.
 \label{qua control diff N_retilde}
\end{align}
Furthermore, for any $\delta>0$, we have
\begin{align}
  e^{-tB}\mathcal{N}_{h}[\delta]e^{tB}&\leq C\mathcal{N}_{h}[\delta]
  +C\tilde{\rho}_0^{\frac{2}{3}+4\alpha_3}\mathcal{N}_{re}
  +C\tilde{\rho}_0^{\frac{5}{3}+\alpha_3}L^3,
 \label{qua control N_h}\\
 e^{-tB}{\mathcal{N}}_{i}[\delta]e^{tB}&\leq C{\mathcal{N}}_{i}[\delta]
 +C\mathcal{N}_h[\delta]+
  C\tilde{\rho}_0^{\frac{2}{3}+4\alpha_3}\mathcal{N}_{re}
  +C\tilde{\rho}_0^{\frac{5}{3}+\alpha_3}L^3,
 \label{qua control N_i}
\end{align}
(\ref{qua control N_h}) also holds when $\delta\leq0$.
\end{lemma}
\begin{proof}
  \par We have $[\mathcal{N}_{re},A]=2A$ and therefore $[\mathcal{N}_{re},B]=A+A^*$. For
  $\psi\in\mathcal{H}^{\wedge N}$,
  \begin{equation}\label{bound A}
    \begin{aligned}
    \vert\langle A\psi,\psi\rangle\vert
    &=L^{\frac{3}{2}}\Big\vert\sum_{\sigma,\nu}\int_{\Lambda_L^2}\eta_{\phi^+}(x-y)\langle a^*(h_{x,\sigma})a^*(h_{y,\nu})
    a(g_{y,\nu})a(g_{x,\sigma})\psi,\psi\rangle dxdy\Big\vert\\
    &=L^{\frac{3}{2}}\Big\vert\sum_{\sigma,\nu}\int_{\Lambda_L}\langle
    b^*_{y,\sigma}(\eta_{\phi^+})\psi,a^*(g_{y,\nu})a(h_{y,\nu})\psi\rangle dy\Big\vert\\
    &\leq C\tilde{\rho}_0L^{\frac{3}{2}}\Vert\eta_{\phi^+}\Vert_1
    \langle\mathcal{N}_{re}\psi,\psi\rangle
    +C\tilde{\rho}_0L^{\frac{3}{2}}\Vert\eta_{\phi^+}\Vert_2
    \langle\mathcal{N}_{re}\psi,\psi\rangle^{\frac{1}{2}}
    \langle\psi,\psi\rangle^{\frac{1}{2}},
    \end{aligned}
  \end{equation}
  where we have used Lemma \ref{b^* bound by b} in the last inequality. Therefore, for $\theta>0$.
  \begin{equation}\label{bound A 2}
    \pm A\leq C\theta\mathcal{N}_{re}+C
    \theta^{-1}\tilde{\rho}_0^{\frac{5}{3}+\alpha_3}L^3,
  \end{equation}
   Then (\ref{qua control N_re}) is obtained by choosing $\theta=1$ and Gronwall's inequality.

  \par For (\ref{qua control N_retilde}), we have $[\tilde{\mathcal{N}}_{re},A]=2A$, and
  \begin{equation}\label{qua Nretilde 1}
    e^{-tB}\tilde{\mathcal{N}}_{re}e^{tB}=
    \tilde{\mathcal{N}}_{re}+\int_{0}^{t}e^{-sB}[\tilde{\mathcal{N}}_{re},B]e^{sB}ds.
  \end{equation}
  We can then reach (\ref{qua control N_retilde}) using (\ref{qua control N_re}), (\ref{qua Nretilde 1}) and (\ref{bound A 2}) with $\theta=1$. On the other hand, we can reach (\ref{qua control diff N_retilde}) by choosing $\theta=\tilde{\rho}_0^{\frac{1}{3}}$. (\ref{qua control diff N_re}) can be deduced similarly.

  \par For (\ref{qua control N_h}), we can calculate directly for $\delta\in\mathbb{R}$,
  \begin{equation}\label{com N_h B}
    \begin{aligned}
    [\mathcal{N}_h[\delta],B]&=\sum_{\substack{k,p,q\\\sigma,\nu}}
  \eta_k\phi^+(k)(a^*_{p-k,\sigma}a^*_{q+k,\nu}a_{q,\nu}a_{p,\sigma}+h.c.)
  \chi_{p-k\in A_{F,\delta}}\chi_{q+k\notin B_F}\chi_{p,q\in B_F}\\
  &= L^{\frac{3}{2}}\sum_{\sigma,\nu}\int_{\Lambda_L^2}\eta_{\phi^+}(x-y) a^*(H_{x,\sigma}[\delta])a^*(h_{y,\nu})
    a(g_{y,\nu})a(g_{x,\sigma}) dxdy+h.c.\\
    &= L^{\frac{3}{2}}\sum_{\sigma,\nu}\int_{\Lambda_L} a^*(H_{x,\sigma}[\delta])a(g_{x,\sigma})b^*_{x,\nu}(\eta_{\phi^+}) dx+h.c.
    \end{aligned}
  \end{equation}
  Similar to (\ref{bound A}), we have, for $\psi\in\mathcal{H}^{\wedge N}$, that
  \begin{equation}\label{bound com N_h B 1st}
    \begin{aligned}
    \big\vert\langle [\mathcal{N}_h[\delta],B]\psi,\psi\rangle \big\vert
    \leq &C\tilde{\rho}_0L^{\frac{3}{2}}\Vert\eta_{\phi^+}\Vert_1
    \langle\mathcal{N}_{re}\psi,\psi\rangle^\frac{1}{2}
    \langle\mathcal{N}_{h}[\delta]\psi,\psi\rangle^\frac{1}{2}\\
    &+C\tilde{\rho}_0L^{\frac{3}{2}}\Vert\eta_{\phi^+}\Vert_2
    \langle\mathcal{N}_{h}[\delta]\psi,\psi\rangle^{\frac{1}{2}}
    \langle\psi,\psi\rangle^{\frac{1}{2}}.
    \end{aligned}
  \end{equation}
  Therefore,
  \begin{equation}\label{bound com N_h B}
    \pm[\mathcal{N}_h[\delta],B]\leq C\mathcal{N}_h[\delta]
    +C\tilde{\rho}_0^{\frac{2}{3}+4\alpha_3}\mathcal{N}_{re}
    +C\tilde{\rho}_0^{\frac{5}{3}+\alpha_3}L^3.
  \end{equation}
  Then we deduce (\ref{qua control N_h}) via Gronwall's inequality.

  \par For (\ref{qua control N_i}), we calculate directly for $\delta>0$, that
  \begin{equation}\label{com N_i B}
    \begin{aligned}
    [\mathcal{N}_i[\delta],B]&=\sum_{\substack{k,p,q\\\sigma,\nu}}
  \eta_k\phi^+(k)(a^*_{p-k,\sigma}a^*_{q+k,\nu}a_{q,\nu}a_{p,\sigma}+h.c.)
  \chi_{p-k,q+k\notin B_F}\chi_{q\in B_F}\chi_{p\in\underline{B}_{F,\delta}}\\
  &= L^{\frac{3}{2}}\sum_{\sigma,\nu}\int_{\Lambda_L^2}\eta_{\phi^+}(x-y) a^*(h_{x,\sigma})a^*(h_{y,\nu})
    a(g_{y,\nu})a(I_{x,\sigma}[\delta]) dxdy+h.c.
    \end{aligned}
  \end{equation}
  Using $a^*(h_{x,\sigma})=a^*(H_{x,\sigma}[\delta])+a^*(L_{x,\sigma}[\delta])$, we rewrite
  $[\mathcal{N}_i[\delta],B]=\mathrm{I}+\mathrm{II}$. Similar to (\ref{com N_h B}),
  \begin{equation}\label{com N_i B I}
    \pm\mathrm{I}\leq C\mathcal{N}_h[\delta]
    +C\tilde{\rho}_0^{\frac{2}{3}+4\alpha_3}\mathcal{N}_{re}
    +C\tilde{\rho}_0^{\frac{5}{3}+\alpha_3}L^3.
  \end{equation}
  On the other hand, for $\psi\in\mathcal{H}^{\wedge N}$
   \begin{equation}\label{bound com N_i delta II}
    \begin{aligned}
    \vert\langle \mathrm{II}\psi,\psi\rangle\vert
    &=L^{\frac{3}{2}}\Big\vert\sum_{\sigma,\nu}\int_{\Lambda_L}\langle
    b^*_{x,\nu}(\eta_{\phi^+})\psi,a(L_{x,\sigma}[\delta])
    a^*(I_{x,\sigma}[\delta])\psi\rangle dy\Big\vert\\
    &\leq C\tilde{\rho}_0^{1+\frac{\delta}{2}}L^{\frac{3}{2}}\Vert\eta_{\phi^+}\Vert_1
    \langle\mathcal{N}_{re}\psi,\psi\rangle^\frac{1}{2}
    \langle\mathcal{N}_{i}[\delta]\psi,\psi\rangle^\frac{1}{2}\\
    &+C\tilde{\rho}_0^{1+\frac{\delta}{2}}L^{\frac{3}{2}}\Vert\eta_{\phi^+}\Vert_2
    \langle\mathcal{N}_{i}[\delta]\psi,\psi\rangle^{\frac{1}{2}}
    \langle\psi,\psi\rangle^{\frac{1}{2}}.
    \end{aligned}
  \end{equation}
  Therefore,
  \begin{equation}\label{com N_i B II}
    \pm\mathrm{II}\leq C\mathcal{N}_i[\delta]
    +C\tilde{\rho}_0^{\frac{2}{3}+4\alpha_3+\delta}\mathcal{N}_{re}
    +C\tilde{\rho}_0^{\frac{5}{3}+\alpha_3+\delta}L^3.
  \end{equation}
  Combing (\ref{com N_i B I}) and (\ref{com N_i B II}), we deduce (\ref{qua control N_i}) via Gronwall's inequality.

\end{proof}

\par Lemma \ref{lemma qua cal [K,B]} collects estimates about the commutator of $\mathcal{K}$ and $B$, and controls the action of $e^B$ on $\mathcal{K}_s$.

\begin{lemma}\label{lemma qua cal [K,B]}
  For any $N\in\mathbb{N}_{\geq0}$ and $0<\delta<\frac{1}{3}$,
  \begin{align}
    \pm\Omega\leq& C\tilde{\rho}_0^{\frac{1}{3}}\mathcal{K}_s+C\tilde{\rho}_0\mathcal{N}_{re}+C\tilde{\rho}_0^{2+\alpha_3}L^3,
    \label{est qua Omega}\\
    \pm[\mathcal{K},B]\leq& C\mathcal{K}_s+C\tilde{\rho}_0^
    {\frac{2}{3}}\mathcal{N}_h[-\alpha_2]
    +C\tilde{\rho}_0\mathcal{N}_{re}
    +C\tilde{\rho}_0^2L^3.\label{est qua [K,B]}
  \end{align}
  Therefore, for $\vert t\vert\leq 1$,
  \begin{equation}\label{qua control K}
    e^{-tB}\mathcal{K}_se^{tB}\leq C\mathcal{K}_s+C\tilde{\rho}_0^{\frac{2}{3}}\mathcal{N}_h[-\alpha_2]
    +C\tilde{\rho}_0
    \mathcal{N}_{re}
    +C\tilde{\rho}^2L^3.
  \end{equation}
  Furthermore, we have
  \begin{equation}\label{[K,B] detailed}
    [\mathcal{K},B]=\Omega_m-\Omega_{m,l}+\Omega+\Omega_h
  \end{equation}
  where
  \begin{equation}\label{Omega_m}
    \begin{aligned}
    \Omega_m=\sum_{k,p,q,\sigma,\nu}
  \eta_k\vert k\vert^2
  (a^*_{p-k,\sigma}a^*_{q+k,\nu}a_{q,\nu}a_{p,\sigma}+h.c.)
  \times\chi_{p-k,q+k\notin B_F}\chi_{p,q\in B_F}
    \end{aligned}
  \end{equation}
  and
  \begin{equation}\label{Omega m,l and Omega h bound}
    \begin{aligned}
    \pm\Omega_{m,l}\leq&C\tilde{\rho}_0^{1+\delta_1}\mathcal{N}_{re}
  +C\tilde{\rho}_0\mathcal{N}_h[\delta_1]+C\tilde{\rho}_0^{2+4\alpha_3-7\alpha_2}L^3\\
    \pm\Omega_{h}\leq &C(\tilde{\rho}_0^{1+2\alpha_3}+\tilde{\rho}_0^{1+\alpha_5-\alpha_4})
  \mathcal{N}_h[-\beta_1]
  +C(\tilde{\rho}_0^{\frac{1}{3}+2\alpha_3}+\tilde{\rho}_0^{\frac{1}{3}+\alpha_5-\alpha_4})
  \mathcal{K}_s\\
  &+C\tilde{\rho}_0^{\frac{7}{3}+\alpha_4}L^3
    \end{aligned}
  \end{equation}
  for $\frac{1}{3}>\delta_1>0$, $\beta_1=\frac{1}{3}+\alpha_5$ and $2\alpha_3>\alpha_5>2\alpha_4$.
\end{lemma}
\begin{proof}
\par For $\Omega$ defined in (\ref{define V_21' and Omega}), let $C$ denote any universal constant, we have
\begin{equation}\label{Omega1st}
  \begin{aligned}
  &\Omega=CL^{\frac{3}{2}}\sum_{\sigma,\nu}\int_{\Lambda_L^2}\nabla_x\eta_{\phi^+}^
  {\zeta^-}(x-y)a^*(h_{x,\sigma})a^*(h_{y,\nu})a(g_{y,\nu})a(\nabla_xg_{x,\sigma})dxdy+h.c.\\
  =&-CL^{\frac{3}{2}}\sum_{\sigma,\nu}\int_{\Lambda_L^2}\eta_{\phi^+}^
  {\zeta^-}(x-y)a^*(\nabla_xh_{x,\sigma})a^*(h_{y,\nu})a(g_{y,\nu})a(\nabla_xg_{x,\sigma})dxdy
  +h.c.\\
  &-CL^{\frac{3}{2}}\sum_{\sigma,\nu}\int_{\Lambda_L^2}\eta_{\phi^+}^
  {\zeta^-}(x-y)a^*(h_{x,\sigma})a^*(h_{y,\nu})a(g_{y,\nu})a(\Delta_xg_{x,\sigma})dxdy+h.c.\\
  \eqqcolon&\Omega_1+\Omega_2.
  \end{aligned}
\end{equation}
We have for $\psi\in\mathcal{H}^{\wedge N}$,
\begin{equation}\label{Omega_1 bound 1}
  \begin{aligned}
  \big\vert\langle\Omega_1\psi,\psi\rangle\big\vert=&
  CL^{\frac{3}{2}}\Big\vert\sum_{\sigma,\nu}\int_{\Lambda_L}
  \langle b^*_{x,\nu}(\eta_{\phi^+}^{\zeta^-})\psi,
  a^*(\nabla_xg_{x,\sigma})a(\nabla_xh_{x,\sigma})\psi\rangle dx\Big\vert\\
  \leq& C\tilde{\rho}_0^{\frac{4}{3}}L^{\frac{3}{2}}\Vert \eta_{\phi^+}^{\zeta^-}\Vert_1
  \langle\mathcal{N}_{re}\psi,\psi\rangle^{\frac{1}{2}}
  \langle (\mathcal{K}_s+\tilde{\mu}\mathcal{N}_{re})\psi,\psi\rangle^{\frac{1}{2}}\\
  +&C\tilde{\rho}_0^{\frac{4}{3}}L^{\frac{3}{2}}\Vert \eta_{\phi^+}^{\zeta^-}\Vert_2
  \langle\psi,\psi\rangle^{\frac{1}{2}}
  \langle (\mathcal{K}_s+\tilde{\mu}\mathcal{N}_{re})\psi,\psi\rangle^{\frac{1}{2}},
  \end{aligned}
\end{equation}
where we have used
\begin{equation}\label{K_s bound}
  0\leq\sum_{\sigma}\int_{\Lambda_L} a^*(\nabla_x h_{x,\sigma})a(\nabla_x h_{x,\sigma})dx
  \leq\mathcal{K}_s+\tilde{\mu}\mathcal{N}_{re}.
\end{equation}
Since $\tilde{\mu}\sim\tilde{\rho}_0^{\frac{2}{3}}$, we have
\begin{equation}\label{Omega 1}
  \pm\Omega_1\leq C\tilde{\rho}_0^{\frac{1}{3}}\mathcal{K}_s+C\tilde{\rho}_0\mathcal{N}_{re}+
  C\tilde{\rho}_0^{2+\alpha_3}L^3.
\end{equation}
Similarly,
\begin{equation}\label{Omega 2}
  \pm\Omega_2\leq C\tilde{\rho}_0\mathcal{N}_{re}+
  C\tilde{\rho}_0^{2+\alpha_3}L^3.
\end{equation}
Combining (\ref{Omega 1}) and (\ref{Omega 2}) we reach (\ref{est qua Omega}).

\par For (\ref{est qua [K,B]}), we have
\begin{equation}\label{[K,B]}
  \begin{aligned}
   [\mathcal{K},B]&=\sum_{k,p,q,\sigma,\nu}
  \eta_k\big(\vert k\vert^2+k(q-p)\big)
  \phi^+(k)
  (a^*_{p-k,\sigma}a^*_{q+k,\nu}a_{q,\nu}a_{p,\sigma}+h.c.)\\
  &\quad\quad\quad\quad\quad\quad\quad\quad\quad
  \times\chi_{p-k\notin B_F}\chi_{q+k\notin B_F}\chi_{p\in B_F}\chi_{q\in B_F}\\
  &\eqqcolon\Omega_{m,h}+\Omega_{a}.
 \end{aligned}
\end{equation}
Similar to the bound of $\Omega$ in (\ref{est qua Omega}), we have
\begin{equation}\label{Omega_a bound}
  \pm\Omega_a\leq C\tilde{\rho}_0^{\frac{1}{3}}
  \mathcal{K}_s+C\tilde{\rho}_0\mathcal{N}_{re}+
  C\tilde{\rho}_0^{2+\alpha_3}L^3.
\end{equation}
For $\Omega_{m,h}$, we notice that $\phi^+(k)>0$ additionally sets $\vert k\vert\geq\frac{3}{2}
\tilde{\mu}^{\frac{1}{2}}\tilde{\rho}_0^{-\alpha_2}$. Since in $\Omega_{m,h}$, $\vert p\vert,\vert q\vert\leq k_F=\tilde{\mu}^{\frac{1}{2}}$, we have $p-k,q+k\in P_{F,-\alpha_2}$, when $\tilde{\rho}_0$ and therefore $\rho_0$ are small enough (independent of $L$). Therefore, we can write
\begin{equation}\label{Omegam}
   \begin{aligned}
  \Omega_{m,h}
  =&CL^{\frac{3}{2}}\sum_{\sigma,\nu}\int_{\Lambda_L^2}\Delta_x\eta_{\phi^+}
  (x-y)a^*(H_{x,\sigma}[-\alpha_2])a^*(H_{y,\nu}[-\alpha_2])\\
  &\quad\quad\quad\quad\quad\quad\times
  a(g_{y,\nu})a(g_{x,\sigma})dxdy+h.c.
  \end{aligned}
\end{equation}
Integrating by parts yields
\begin{align}
-\Omega_{m,h}
  =&CL^{\frac{3}{2}}\sum_{\sigma,\nu}\int_{\Lambda_L^2}\nabla_x\eta_{\phi^+}
  (x-y)a^*(\nabla_xH_{x,\sigma}[-\alpha_2])a^*(H_{y,\nu}[-\alpha_2])\nonumber\\
  &\quad \quad \quad \quad\quad \quad  \times
  a(g_{y,\nu})a(g_{x,\sigma})dxdy+h.c.\nonumber\\
  &+CL^{\frac{3}{2}}\sum_{\sigma,\nu}\int_{\Lambda_L^2}\nabla_x\eta_{\phi^+}
  (x-y)a^*(H_{x,\sigma}[-\alpha_2])a^*(H_{y,\nu}[-\alpha_2])\nonumber\\
  &\quad \quad \quad \quad\quad \quad  \times
  a(g_{y,\nu})a(\nabla_xg_{x,\sigma})dxdy+h.c.\label{Omega m h}
\end{align}
Similar to the bound of $\Omega$, we have
\begin{equation}\label{Omega bound}
  \pm\Omega_{m,h}\leq C\mathcal{K}_s+C\tilde{\rho}_0^{\frac{2}{3}}\mathcal{N}_{h}[-\alpha_2]+
  C\tilde{\rho}_0^{2}L^3.
\end{equation}
Combining (\ref{Omega_a bound}) and (\ref{Omega bound}), we reach (\ref{est qua [K,B]}).

\par For (\ref{qua control K}), by (\ref{K_s and K}), we have $[\mathcal{K}_s,B]
=[\mathcal{K},B]$. Therefore, (\ref{qua control K}) is deduced via (\ref{est qua [K,B]}), Gronwall's inequality and Lemma \ref{lemma qua control N_re}.

\par Let $\Omega_{m,l}=\Omega_m-\Omega_{m,h}$, and $\Omega_h=\Omega_a-\Omega$, then for (\ref{[K,B] detailed}), it can be easily verified using (\ref{Omega_m}) and (\ref{[K,B]}).

\par  For $\Omega_{m,l}$:
\begin{equation}\label{Omega_m,l}
  \Omega_{m,l}=\sum_{k,p,q,\sigma,\nu}
  \eta_k\vert k\vert^2\phi^-(k)
  (a^*_{p-k,\sigma}a^*_{q+k,\nu}a_{q,\nu}a_{p,\sigma}+h.c.)
  \times\chi_{p-k,q+k\notin B_F}\chi_{p,q\in B_F}
\end{equation}
As $\phi^-(k)\neq0$ implies $\vert k\vert\leq 2\tilde{\mu}^{\frac{1}{2}}\tilde{\rho}_0
^{-\alpha_2}$, therefore $\vert p-k\vert,\vert q+k\vert\leq
3\tilde{\mu}^{\frac{1}{2}}\tilde{\rho}_0^{-\alpha_2}$. Hence, we rewrite it by
\begin{equation}\label{Omega_m,l 2}
  \begin{aligned}
  \Omega_{m,l}=&\sum_{k,\sigma,\nu}L^3\eta_k\vert k\vert^2\phi^-(k)
  \Big(\int_{\Lambda_L}\frac{e^{ikx}}{L^{\frac{3}{2}}}
  a^*(\mathbf{h}_{x,\sigma})a(g_{x,\sigma})dx\Big)\\
  &\quad\quad\quad\quad\quad\quad\times
  \Big(\int_{\Lambda_L}\frac{e^{-iky}}{L^{\frac{3}{2}}}
  a^*(\mathbf{h}_{y,\nu})a(g_{y,\nu})dx\Big)+h.c.
  \end{aligned}
\end{equation}
In (\ref{Omega_m,l 2}), we have temporarily taken
\begin{equation*}
  \mathbf{h}_{x,\sigma}(z)=\sum_{k_F<\vert k\vert\leq 3\tilde{\mu}^{\frac{1}{2}}\tilde{\rho}
  _0^{-\alpha_2}}\frac{e^{ikx}}{L^{\frac{3}{2}}}f_{k,\sigma}(z).
\end{equation*}
Using $\mathbf{h}_{x,\sigma}=L_{x,\sigma}[\delta_1]+(\mathbf{h}_{x,\sigma}-L_{x,\sigma}[\delta_1])$ for $\frac{1}{3}>\delta_1>0$, we rewrite $\Omega_{m,l}=\Omega_{m,l1}+\Omega_{m,l2}$. Since we have
\begin{equation*}
  \big\vert\eta_k\vert k\vert^2\phi^-(k)L^3\big\vert\leq C\tilde{\rho}_0^{2\alpha_3-2\alpha_2},
\end{equation*}
then by Lemma \ref{lemma A_k B_k}, we can bound
\begin{equation}\label{Omega_m,l1 and 2 bound}
  \begin{aligned}
  \pm\Omega_{m,l1}\leq& C\tilde{\rho}_0\mathcal{N}_h[\delta_1]+C\tilde{\rho}_0^{2+4\alpha_3-7\alpha_2}L^3\\
  \pm\Omega_{m,l2}\leq& C\tilde{\rho}_0^{1+\delta_1}\mathcal{N}_{re}
  +C\tilde{\rho}_0^{2+4\alpha_3-7\alpha_2}L^3
  \end{aligned}
\end{equation}
Therefore
\begin{equation}\label{Omega_m,l bound}
  \pm\Omega_{m,l}\leq C\tilde{\rho}_0^{1+\delta_1}\mathcal{N}_{re}
  +C\tilde{\rho}_0\mathcal{N}_h[\delta_1]+C\tilde{\rho}_0^{2+4\alpha_3-7\alpha_2}L^3.
\end{equation}

\par For $\Omega_h$:
\begin{equation}\label{Omega_h}
  \Omega_h=\sum_{k,p,q,\sigma,\nu}
  \eta_kk(q-p)
  \zeta^+(k)
  (a^*_{p-k,\sigma}a^*_{q+k,\nu}a_{q,\nu}a_{p,\sigma}+h.c.)
  \chi_{p-k,q+k\notin B_F}\chi_{p,q\in B_F}
\end{equation}
Notice that $\zeta^+(k)>0$ requires $p-k,q+k\in P_{F,-\beta_1}$, then similar to the bound of $\Omega$, with furthermore the bound (\ref{special eta L2 norm}), we have
\begin{equation}\label{Omega_h bound}
  \pm\Omega_h\leq C(\tilde{\rho}_0^{1+2\alpha_3}+\tilde{\rho}_0^{1+\alpha_5-\alpha_4})
  \mathcal{N}_h[-\beta_1]
  +C(\tilde{\rho}_0^{\frac{1}{3}+2\alpha_3}+\tilde{\rho}_0^{\frac{1}{3}+\alpha_5-\alpha_4})
  \mathcal{K}_s+C\tilde{\rho}_0^{\frac{7}{3}+\alpha_4}L^3,
\end{equation}
with $\beta_1=\frac{1}{3}+\alpha_5$ and $2\alpha_3>\alpha_5>2\alpha_4$.
\end{proof}

\par Lemma \ref{lemma qua V_21'} bounds the operator $\mathcal{V}_{21}^\prime$, which is useful in Section \ref{bog}.

\begin{lemma}\label{lemma qua V_21'}
  For any $N\in\mathbb{N}_{\geq0}$, we have
  \begin{equation}\label{qua V_21'}
    \pm\mathcal{V}_{21}^\prime\leq C\mathcal{K}_s+C\tilde{\rho}_0^{1-\alpha_3}\mathcal{N}_{re}
    +C\tilde{\rho}_0^2L^3,
  \end{equation}
  If we further define
  \begin{equation}\label{define V_21'_h}
    \mathcal{V}_{21,h}^\prime=\frac{1}{L^3}\sum_{k,p,q,\sigma,\nu}
  W_k\zeta^{+}(k)(a^*_{p-k,\sigma}a^*_{q+k,\nu}a_{q,\nu}a_{p,\sigma}+h.c.)
  \chi_{p-k,q+k\notin
  B_F}\chi_{p,q\in B_F}
  \end{equation}
  we have
  \begin{equation}\label{qua V_21,h'}
    \pm\mathcal{V}_{21,h}^\prime\leq C\big(\tilde{\rho}_0
    +\tilde{\rho}_0^{1+\alpha_5-4\alpha_3-\alpha_4}\big)\mathcal{N}_h[-\beta_1]
    +C\tilde{\rho}_0^{\frac{7}{3}+\alpha_4}L^3,
  \end{equation}
    for $\frac{1}{24}>\frac{1}{4}\alpha_3>\alpha_2>2\alpha_4>0$, with $\beta_1=\frac{1}{3}+\alpha_5$ and $2\alpha_3>\alpha_5>2\alpha_4$.
\end{lemma}
\begin{proof}
  \par For $\mathcal{V}_{21,h}^\prime$, since $\zeta^+(k)>0$ implies $p-k,q+k\in P_{F,-\beta_1}$, then
  \begin{equation}\label{V_21,h'}
    \begin{aligned}
    \mathcal{V}_{21,h}^\prime=\sum_{\sigma,\nu}
    \int_{\Lambda_L^2}W_{\zeta^+}(x-y)a^*(H_{x,\sigma}[-\beta_1])a^*(H_{y,\nu}[-\beta_1])
    a(g_{y,\nu})a(g_{x,\sigma})dxdy+h.c.
    \end{aligned}
  \end{equation}
  Using estimate (\ref{special W L2 norm}) and Lemma \ref{b^* bound by b general}, we have
  \begin{equation}\label{V_21,h' bound}
    \pm\mathcal{V}_{21,h}^\prime\leq C\big(\tilde{\rho}_0
    +\tilde{\rho}_0^{1+\alpha_5-4\alpha_3-\alpha_4}\big)\mathcal{N}_h[-\beta_1]
    +C\tilde{\rho}_0^{\frac{7}{3}+\alpha_4}L^3,
  \end{equation}

  \par For $\mathcal{V}_{21}^\prime$, we first use (\ref{Wtheta}) to rewrite
  \begin{equation}\label{split W}
  \begin{aligned}
    W_k=\Big(\frac{1}{L^{\frac{3}{2}}}
    \sum_{m\neq k}\frac{W_{k-m}(k-m)}{\vert k-m\vert^2}\vartheta_m\Big)
    \big(p-(p-k)-m\big)+\frac{1}{L^{\frac{3}{2}}}W_0\vartheta_k,
  \end{aligned}
  \end{equation}
  and then to rewrite it by $\mathcal{V}_{21}^\prime=\sum_{j=1}^{4}\mathcal{V}_{21,j}^\prime$.

  \par For $\mathcal{V}_{21,1}^\prime$, recall the definition of $U$ in (\ref{define U}),
  \begin{equation}\label{V_21,1'origin}
    \mathcal{V}_{21,1}^\prime=C\sum_{\sigma,\nu}
    \int_{\Lambda_L^2}(U\vartheta)^{\zeta^-}(x-y)
    a^*(h_{x,\sigma})a^*(h_{y,\nu})
    a(g_{y,\nu})a(\nabla_xg_{x,\sigma})dxdy+h.c.
  \end{equation}
  Using (\ref{theta norm}), (\ref{U L2}) and Lemma \ref{b^* bound by b general}, we can bound
  \begin{equation}\label{V_21,1'}
    \pm\mathcal{V}_{21,1}^\prime\leq C\tilde{\rho}_0^{1-\alpha_3}\mathcal{N}_{re}+C
    \tilde{\rho}_0^2L^3.
  \end{equation}
  Similarly, for $\mathcal{V}_{21,2}^\prime$ and $\mathcal{V}_{21,3}^\prime$, we bound them by
  \begin{equation}\label{V_21,2'}
    \pm\mathcal{V}_{21,2}^\prime\leq C\mathcal{K}_s+
    C\tilde{\rho}_0^{1-\alpha_3}\mathcal{N}_{re}+C
    \tilde{\rho}_0^2L^3,
  \end{equation}
  and
  \begin{equation}\label{V_21,3'}
    \pm\mathcal{V}_{21,3}^\prime\leq
    C\tilde{\rho}_0^{1-\alpha_3}\mathcal{N}_{re}+C
    \tilde{\rho}_0^2L^3.
  \end{equation}
  For $\mathcal{V}_{21,4}^\prime$, we just need to additionally notice that $\vert W_0\vert\leq C$, then
  \begin{equation}\label{V_21,4'}
    \pm\mathcal{V}_{21,4}^\prime\leq
    C(\tilde{\rho}_0)L^{-3}\big(\mathcal{N}_{re}+L^3\big).
  \end{equation}
  Thus we reach (\ref{qua V_21'}) using (\ref{split W})-(\ref{V_21,4'}).
\end{proof}

\par Lemma \ref{lemma qua control V_4} estimates the action of $e^B$ on $\mathcal{V}_4$.
\begin{lemma}\label{lemma qua control V_4}
  For any $N\in\mathbb{N}_{\geq0}$ and $\vert t\vert\leq 1$, assume further $2\alpha_3-3\alpha_2>0$, we have
  \begin{equation}\label{qua control V_4}
    e^{-tB}\mathcal{V}_4e^{tB}\leq C\mathcal{V}_4+C\tilde{\rho}_0^{\frac{5}{3}+\alpha_2}
    \mathcal{N}_{re}+C\tilde{\rho}_0^2L^3.
  \end{equation}
\end{lemma}
\begin{proof}
\par We have
\begin{equation}\label{[V_4,B]}
  [\mathcal{V}_4,B]=\Theta_m+\Theta_{d,1}+\Theta_{d,2}+\Theta_r,
\end{equation}
with
    \begin{align*}
    \Theta_m&=\frac{1}{2L^3}\sum_{k,p,q,\sigma,\nu}\Big(\sum_{l}\hat{v}_{k-l}\eta_l\Big)
    (a^*_{p-k,\sigma}a^*_{q+k,\nu}a_{q,\nu}a_{p,\sigma}+h.c.)\\
    &\quad\quad\times
    \chi_{p-k\notin B_F}\chi_{q+k\notin B_F}
    \chi_{p\in B_F}\chi_{q\in B_F}\\
    \Theta_{d,1}&=-\frac{1}{2L^3}\sum_{l,k,p,q,\sigma,\nu}\hat{v}_{k-l}\eta_l
    \phi^-(l)
    (a^*_{p-k,\sigma}a^*_{q+k,\nu}a_{q,\nu}a_{p,\sigma}+h.c.)\\
    &\quad\quad\times
    \chi_{p-k\notin B_F}\chi_{q+k\notin B_F}
    \chi_{p\in B_F}\chi_{q\in B_F}\\
     \Theta_{d,2}&=\frac{1}{2L^3}\sum_{l,k,p,q,\sigma,\nu}\hat{v}_{k-l}\eta_l
     \phi^+(l)
    (a^*_{p-k,\sigma}a^*_{q+k,\nu}a_{q,\nu}a_{p,\sigma}+h.c.)\\
    &\quad\quad\times
    \chi_{p-k\notin B_F}\chi_{q+k\notin B_F}
    \chi_{p\in B_F}\chi_{q\in B_F}
    (\chi_{p-l\notin B_F}\chi_{q+l\notin B_F}-1)\\
    \Theta_r&=-\frac{1}{L^3}\sum_{k,p,q,l,r,s,\sigma,\nu,\varpi}
    \hat{v}_{k}\eta_l
    (a^*_{p-k,\sigma}a^*_{q+k,\nu}a^*_{s+l,\varpi}a_{q,\nu}
    a_{s,\varpi}a_{r,\sigma}+h.c.)\\
    &\quad\quad\times
    \chi_{p-k,p\notin B_F}\chi_{q+k,q\notin B_F}
    \chi_{r\in B_F}
    \chi_{s+l\notin B_F}\chi_{s\in B_F}\delta_{p,r-l}\phi^+(l)
    \end{align*}
\par For $\Theta_m$, we have
\begin{equation}\label{Theta_m}
  \Theta_m=\frac{L^{\frac{3}{2}}}{2}\sum_{\sigma,\nu}\int_{\Lambda_L^2}
    v(x-y)\eta(x-y)a^*(h_{x,\sigma})a^*(h_{y,\nu})
    a(g_{y,\nu})a(g_{x,\sigma})dxdy+h.c.
\end{equation}
We can bound it directly using Cauchy-Schwartz inequality
\begin{equation}\label{Theta_m 2}
  \pm\Theta_m\leq C\mathcal{V}_4+C\tilde{\rho}_0^2L^3.
\end{equation}

\par For $\Theta_{d,1}$,
\begin{equation}\label{Theta_d,1}
  \begin{aligned}
  \Theta_{d,1}=&-\frac{L^{\frac{3}{2}}}{2}\sum_{\sigma,\nu}\int_{\Lambda_L^2}
    v(x-y)\eta^{\phi^-}(x-y)\\
    &\quad\quad\quad\quad\times a^*(h_{x,\sigma})a^*(h_{y,\nu})
    a(g_{y,\nu})a(g_{x,\sigma})dxdy+h.c.
  \end{aligned}
\end{equation}
We just need to additionally use the bound (\ref{bound eta-etatilde Linfty}) to reach
\begin{equation}\label{Theta d,1 bound}
  \pm\Theta_{d,1}\leq C\theta\mathcal{V}_4
  +C\theta^{-1}\tilde{\rho}_0^{\frac{8}{3}+4\alpha_3-6\alpha_2}L^3,
\end{equation}
for any $\theta>0$.

\par For $\Theta_{d,2}$, we notice that $(\chi_{p-l\notin B_F}\chi_{q+l\notin B_F}-1)\neq0$ means $\vert l\vert\leq C\tilde{\rho}_0^{\frac{1}{3}}$, which means $\phi^+(l)=0$, and therefore $\Theta_{d,2}=0$.

\par For $\Theta_{r}$, since $p=r-l$, if we allow $r-l=p\in B_F$, this implies
\begin{equation*}
  \vert l\vert\leq \vert r-l\vert+\vert r\vert\leq 2\tilde{\mu}^{\frac{1}{2}}
  \leq \tilde{\mu}^{\frac{1}{2}}\tilde{\rho}_0^{-\alpha_2},
\end{equation*}
which makes $\phi^+(l)=0$. Hence we have
\begin{equation}\label{Theta_r}
  \begin{aligned}
  \Theta_r=&-L^{\frac{3}{2}}\sum_{\sigma,\nu,\varpi}\int_{\Lambda_L^3}
  v(x_1-x_3)\eta_{\phi^+}(x_1-x_4)
  a^*(h_{x_1,\sigma})a^*(h_{x_3,\nu})a^*(h_{x_4,\varpi})\\
  &\quad\quad\quad\quad
  \times a(h_{x_3,\nu})a(g_{x_4,\varpi})a(g_{x_1,\sigma})dx_1dx_3dx_4+h.c.\\
  =&-L^{\frac{3}{2}}\sum_{\sigma,\nu,\varpi}\int_{\Lambda_L^2}
  v(x_1-x_3)
  a^*(h_{x_1,\sigma})a^*(h_{x_3,\nu})b^*_{x_1,\varpi}(\eta_{\phi^+})\\
  &\quad\quad\quad\quad
  \times a(g_{x_1,\sigma})a(h_{x_3,\nu})dx_1dx_3+h.c.
  \end{aligned}
\end{equation}
Thus using (\ref{ineq b 2}), we can bound
\begin{equation}\label{Theta_r bound}
  \pm\Theta_r\leq C\theta\mathcal{V}_4+C\theta^{-1}\tilde{\rho}_0^{\frac{5}{3}+\alpha_2}
  \mathcal{N}_{re},
\end{equation}
for some $\theta>0$.
\par Combining (\ref{Theta_m})-(\ref{Theta_r bound}) we reach (\ref{qua control V_4}) using Gronwall's inequality and Lemma \ref{lemma qua control N_re}.
\end{proof}

\par Lemma \ref{lemma qua conj V_0 1 22 23} collects the results of renormalization of operators under their a-priori estimates given in Section \ref{a-prior}.
\begin{lemma}\label{lemma qua conj V_0 1 22 23}
For any $N\in\mathbb{N}_{\geq0}$, and $\frac{1}{3}>\delta_1\geq\delta_2>0$, we have
\begin{equation}\label{qua conj V_0 1 22 23}
  \begin{aligned}
  e^{-B}\big(\mathcal{V}_0+\mathcal{V}_1+\mathcal{V}_{22}+\mathcal{V}_{23}\big)e^{B}
  =&\frac{\hat{v}_0}{2L^3}\mathbf{q}(\mathbf{q}-1)\bar{N}_0^2\\
  &+\frac{\hat{v}_0}{L^3}(\mathbf{q}-1)\bar{N}_0(N-\mathbf{q}\bar{N}_0)
  +\mathcal{E}_0,
  \end{aligned}
\end{equation}
where
\begin{equation}\label{qua bound E_0}
  \begin{aligned}
  \pm\mathcal{E}_0\leq&C\tilde{\rho}_0\mathcal{K}_s+C\tilde{\rho}_0^{1+\delta_2}\mathcal{N}_{re}
  +C\tilde{\rho}_0^{1+\delta_2}\tilde{\mathcal{N}}_{re}+C\tilde{\rho}_0^{1-\delta_2}
  \mathcal{N}_h[\delta_1]+C\tilde{\rho}_0\mathcal{N}_i[\delta_1]\\
  &+C\big(\tilde{\rho}_0^{\frac{8}{3}}+\tilde{\rho}_0^{\frac{8}{3}+\alpha_3-\delta_2}\big)L^3.
  \end{aligned}
\end{equation}
\end{lemma}
\begin{proof}
  \par Lemma \ref{lemma qua conj V_0 1 22 23} is a direct corollary of Lemmas \ref{lemma V_0}-\ref{lemma V_22,23}, and Lemmas \ref{lemma qua control N_re}-\ref{lemma qua cal [K,B]}.
\end{proof}

\par Lemma \ref{lemma qua conj V_3} calculates the renormalization of $\mathcal{V}_3$.
\begin{lemma}\label{lemma qua conj V_3}
For any $N\in\mathbb{N}_{\geq0}$,
\begin{equation}\label{qua conj V_3}
  e^{-B}\mathcal{V}_3e^B=\mathcal{V}_{3,l}+\mathcal{E}_{\mathcal{V}_3},
\end{equation}
where $\gamma^-$ has been defined in (\ref{cutoffgamma}),
\begin{equation}\label{qua V_3,l}
\mathcal{V}_{3,l}=\frac{1}{L^3}\sum_{k,p,q,\sigma,\nu}\hat{v}_k\gamma^-(q)
(a_{p-k,\sigma}^*a_{q+k,\nu}^*a_{q,\nu}a_{p,\sigma}+h.c.)
\chi_{p-k,q+k\notin B_F}\chi_{q\notin B_F}\chi_{p\in B_F},
\end{equation}
and
\begin{equation}\label{qua E_V_3}
  \begin{aligned}
  \pm\mathcal{E}_{\mathcal{V}_3}\leq&C\big(\tilde{\rho}_0^{\frac{4}{3}+\alpha_2
  -3\delta_3-\alpha_4}+\tilde{\rho}_0^{\frac{4}{3}+\frac{\alpha_2}{2}}\big)
  \mathcal{N}_{re}+C\tilde{\rho}_0^{1+\delta_2}\tilde{\mathcal{N}}_{re}
  +C\tilde{\rho}_0\mathcal{N}_i[\delta_1]\\
  &+C\tilde{\rho}_0^{1-\delta_2}\mathcal{N}_h[\delta_1]
  +C\tilde{\rho}_0^{\frac{2}{3}-\alpha_4}\mathcal{N}_h[-\delta_3]
  +C\tilde{\rho}_0^{\frac{1}{3}+\alpha_4}
  \mathcal{V}_4+C\tilde{\rho}_0^{\frac{7}{3}+\alpha_4}L^3,
  \end{aligned}
\end{equation}
for $\frac{1}{24}>\frac{1}{2}\delta_3>\frac{1}{4}\alpha_3>\alpha_2>2\alpha_4>0$
and $\frac{1}{3}>\delta_1\geq\delta_2>0$.
\end{lemma}
\begin{proof}
  \par Let $\mathcal{V}_{3,h}=\mathcal{V}_3-\mathcal{V}_{3,l}$, then
  \begin{equation}\label{V_3,h}
    \mathcal{V}_{3,h}=\sum_{\sigma,\nu}
    \int_{\Lambda_L^2}v(x-y)a^*(h_{x,\sigma})a^*(h_{y,\nu})
    a(\mathfrak{h}_{y,\nu}^{\gamma^+})a(g_{x,\sigma})dxdy+h.c.
  \end{equation}
  where we define for $z\in\Lambda_L\times\mathcal{S}_{\mathbf{q}}$,
  \begin{equation}\label{h^gamma^+}
    \mathfrak{h}_{y,\nu}^{\gamma^+}(z)=
    \sum_{k\notin B_F}\frac{e^{ikx}}{L^{\frac{3}{2}}}\gamma^+(k)f_{k,\nu}(z).
  \end{equation}
  By the definition of $\gamma^+$, we can easily bound
  \begin{equation}\label{V_3,h bound}
    \pm\mathcal{V}_{3,h}\leq C\tilde{\rho}_0^{\frac{1}{3}+\alpha_4}
    \mathcal{V}_4+C\tilde{\rho}_0^{\frac{2}{3}-\alpha_4}\mathcal{N}_{h}[-\delta_3].
  \end{equation}
  Therefore, using Lemma \ref{lemma qua control N_re} and Lemma \ref{lemma qua control V_4},
  \begin{equation}\label{V_3,h bound conj}
    \begin{aligned}
    \pm e^{-B}\mathcal{V}_{3,h}e^{B}\leq&
    C\big(\tilde{\rho}_0^{2+\alpha_2-\alpha_4}+\tilde{\rho}_0^{\frac{4}{3}
    +4\alpha_3-\alpha_4}\big)\mathcal{N}_{re}
    +C\tilde{\rho}_0^{\frac{2}{3}-\alpha_4}\mathcal{N}_h[-\delta_3]\\
    &+C\tilde{\rho}_0^{\frac{1}{3}+\alpha_4}\mathcal{V}_4
    +C\big(\tilde{\rho}_0^{\frac{7}{3}+\alpha_4}+
    \tilde{\rho}_0^{\frac{7}{3}+\alpha_3-\alpha_4}\big)L^3.
    \end{aligned}
  \end{equation}

\par For $\mathcal{V}_{3,l}$, we use
\begin{equation*}
  \begin{aligned}
    e^{-B}a^*_{p-k,\sigma}a^*_{q+k,\nu}e^B
    &=a^*_{p-k,\sigma}a^*_{q+k,\nu}+\int_{0}^{1}
    e^{-tB}[a^*_{p-k,\sigma}a^*_{q+k,\nu},B]e^{tB}dt\\
    e^{-B}a_{q,\nu}a_{p,\sigma}e^B
    &=a_{q,\nu}a_{p,\sigma}+\int_{0}^{1}
    e^{-tB}[a_{q,\nu}a_{p,\sigma},B]e^{tB}dt
  \end{aligned}
  \end{equation*}
to rewrite $e^{-B}\mathcal{V}_{3,l}e^B=V_{3,l}
+\mathcal{E}_{\mathcal{V}_{3,l},1}+\mathcal{E}_{\mathcal{V}_{3,l},2}$, where
\begin{equation}\label{split E_V_3}
  \begin{aligned}
    \mathcal{E}_{\mathcal{V}_{3,l},1}&=\frac{1}{L^3}\sum_{k,p,q,\sigma,\nu}
  \hat{v}_k\gamma^-(q) a^*_{p-k,\sigma}a^*_{q+k,\nu}\int_{0}^{1}e^{-tB}
  [a_{q,\nu}a_{p,\sigma},B]e^{tB}dt\\
  &\quad\quad\quad\quad\times
  \chi_{p-k,q+k\notin B_F}\chi_{p\in B^{\sigma}_F}\chi_{q\notin B_F}
  +h.c.\\
    \mathcal{E}_{\mathcal{V}_{3,l},2}&=\frac{1}{L^3}\sum_{k,p,q,\sigma,\nu}
  \hat{v}_k\gamma^-(q)\int_{0}^{1}e^{-tB}
  [a^*_{p-k,\sigma}a^*_{q+k,\nu},B]e^{(t-1)B}
  a_{q,\nu}a_{p,\sigma}e^Bdt\\
  &\quad\quad\quad\quad\times
  \chi_{p-k,q+k\notin B_F}\chi_{p\in B^{\sigma}_F}\chi_{q\notin B_F}
  +h.c.
  \end{aligned}
  \end{equation}

\par For $\mathcal{E}_{\mathcal{V}_{3,l},1}$, we let $\mathcal{E}_{\mathcal{V}_{3,l},1}
=\mathcal{E}_{\mathcal{V}_{3,l},11}+\mathcal{E}_{\mathcal{V}_{3,l},12}$ with
    \begin{align}
    \mathcal{E}_{\mathcal{V}_{3,l},11}&=\sum_{\substack{\sigma,\nu\\k,p,q}}
    \sum_{\substack{\tau,\varpi\\l,r,s}}\frac{\hat{v}_k}{L^3}
    \eta_l\delta_{q,r-l}\delta_{\nu,\tau}
    a^*_{p-k,\sigma}a^*_{q+k,\nu}\int_{0}^{1}e^{-tB}a^*_{s+l,\varpi}a_{s,\varpi}a_{r,\tau}
    a_{p,\sigma}e^{tB}dt\nonumber\\
    &\quad\quad\times\phi^+(l)\gamma^-(q)
    \chi_{p-k,q+k,q\notin B_F}\chi_{p\in B_F}
    \chi_{r-l,s+l\notin B_F}\chi_{r,s\in B_F}+h.c.\nonumber\\
    \mathcal{E}_{\mathcal{V}_{3,l},12}&=\sum_{\substack{\sigma,\nu\\k,p,q}}
    \sum_{\substack{\tau,\varpi\\l,r,s}}\frac{\hat{v}_k}{L^3}
    \eta_l\delta_{p,r}\delta_{\sigma,\tau}
    a^*_{p-k,\sigma}a^*_{q+k,\nu}\int_{0}^{1}e^{-tB}a^*_{s,\varpi}a_{q,\nu}a_{s+l,\varpi}
    a_{r-l,\tau}e^{tB}dt\nonumber\\
    &\quad\quad\times\phi^+(l)\gamma^-(q)
    \chi_{p-k,q+k,q\notin B_F}\chi_{p\in B_F}
    \chi_{r-l,s+l\notin B_F}\chi_{r,s\in B_F}+h.c.\label{E_V_3_1 split}
    \end{align}

  \par For $\mathcal{E}_{\mathcal{V}_{3,l},11}$, since $q=r-l$ while $r\in B_F$, $\phi^+(l)\neq0$ implies $\vert q\vert>\tilde{\mu}^{\frac{1}{2}}
  \tilde{\rho}_0^{-\alpha_2}$, which means $\phi^+(2q)\equiv1$. We let
  \begin{equation*}
    F_{\mathcal{V}_3}(x)=\frac{1}{L^3}\sum_{q\in(2\pi/L)\mathbb{Z}^3}\phi^+(2q)
    \gamma^-(q)e^{iqx},
  \end{equation*}
  then we have
  \begin{equation}\label{E_V_3,l,11}
    \begin{aligned}
    \mathcal{E}_{\mathcal{V}_{3,l},11}=&\sum_{\sigma,\nu,\varpi}
    \int_{\Lambda_L^4}\int_{0}^{1}v(x_1-x_3)L^{\frac{3}{2}}\eta_{\phi^+}(x_2-x_4)
    F_{\mathcal{V}_{3}}(x_2-x_3)
    a^*(h_{x_1,\sigma})a^*(h_{x_3,\nu})\\
    &\times e^{-tB}a^*(h_{x_4,\varpi})a(g_{x_4,\varpi})a(g_{x_2,\nu})a(g_{x_1,\sigma})
    e^{tB}dtdx_1dx_2dx_3dx_4+h.c.\\
    =&\sum_{\sigma,\nu,\varpi}
    \int_{\Lambda_L^3}\int_{0}^{1}v(x_1-x_3)L^{\frac{3}{2}}
    F_{\mathcal{V}_{3}}(x_2-x_3)
    a^*(h_{x_1,\sigma})a^*(h_{x_3,\nu})\\
    &\times e^{-tB}b^*_{x_2,\varpi}(\eta_{\phi^+})a(g_{x_2,\nu})a(g_{x_1,\sigma})
    e^{tB}dtdx_1dx_2dx_3+h.c
    \end{aligned}
  \end{equation}
By Lemma \ref{cutoff lemma} and Lemma \ref{b^* bound by b}, we have
\begin{equation}\label{E_V_3,j,11 bound}
  \pm\mathcal{E}_{\mathcal{V}_{3,l},11}\leq C\tilde{\rho}_0^{\frac{1}{3}+\alpha_4}
  \mathcal{V}_4+C\tilde{\rho}_0^{\frac{7}{3}+\alpha_2-\alpha_4}L^3.
\end{equation}

\par For $\mathcal{E}_{\mathcal{V}_{3,l},12}$, it can be written as
 \begin{equation}\label{typical example E_V_3,l,12}
    \begin{aligned}
    \mathcal{E}_{\mathcal{V}_{3,l},12}
    =&-\sum_{\sigma,\nu,\varpi}
    \int_{\Lambda_L^3}\int_{0}^{1}v(x_1-x_3)L^{\frac{3}{2}}
    \Big(\frac{1}{L^3}\sum_{q_1\in B_F}e^{iq_1(x_2-x_1)}\Big)
    a^*(h_{x_1,\sigma})a^*(h_{x_3,\nu})\\
    &\times e^{-tB}a(\mathfrak{h}^{\gamma^-}_{x_3,\nu})
    b^*_{x_2,\varpi}(\eta_{\phi^+})a(h_{x_2,\sigma})
    e^{tB}dtdx_1dx_2dx_3+h.c\\
    =&-\frac{1}{L^3}\int_{0}^{1}\sum_{\sigma,\nu,\varpi}\sum_{q_1\in B_F}
    \mathcal{A}_{q_1,\sigma,\nu}(t)\mathcal{B}_{q_1,\sigma,\varpi}(t)dt+h.c.
    \end{aligned}
  \end{equation}
  where
  \begin{equation}\label{E_V_3,l 12 AB}
    \begin{aligned}
    \mathcal{A}_{q_1,\sigma,\nu}(t)= &\int_{\Lambda_L^2}e^{-iq_1x_1}v(x_1-x_3)
    a^*(h_{x_1,\sigma})a^*(h_{x_3,\nu})e^{-tB}a(\mathfrak{h}^{\gamma^-}_{x_3,\nu})
    dx_1dx_3\\
    \mathcal{B}_{q_1,\sigma,\varpi}(t)=& \int_{\Lambda_L} e^{iq_1x_2}
    L^{\frac{3}{2}}b^*_{x_2,\varpi}(\eta_{\phi^+})a(h_{x_2,\sigma})
    e^{tB}dx_2
    \end{aligned}
  \end{equation}
 Using Cauchy-Schwartz inequality, we have for $\psi\in\mathcal{H}^{\wedge N}$,
 \begin{equation}\label{E_V_3,l 12 bound 1st step}
   \begin{aligned}
   \big\vert\langle\mathcal{E}_{\mathcal{V}_{3,l},12}\psi,\psi\rangle\big\vert
   \leq & C\int_{0}^{1}\Big(\sum_{\sigma,\nu,\varpi}\sum_{q_1}\frac{1}{L^3}
   \Vert\mathcal{A}_{q_1,\sigma,\nu}(t)^*\psi\Vert^2 \Big)^{\frac{1}{2}}\\
   &\quad\quad\quad\quad\times
   \Big(\sum_{\sigma,\nu,\varpi}\sum_{q_1}\frac{1}{L^3}
   \Vert\mathcal{B}_{q_1,\sigma,\varpi}(t)\psi\Vert^2 \Big)^{\frac{1}{2}}dt
   \end{aligned}
 \end{equation}
 By Lemma \ref{lemma qua control N_re}, we can bound
  \begin{equation}\label{E_V_3,l 12 bound}
    \pm \mathcal{E}_{\mathcal{V}_{3,l},12}\leq C\tilde{\rho}_0^{\frac{1}{3}+\alpha_4}
  \mathcal{V}_4
  +C\tilde{\rho}_0^{\frac{4}{3}+\alpha_2-\alpha_4-3\delta_3}\mathcal{N}_{re}
  +C\tilde{\rho}_0^{3+\alpha_3+\alpha_2-\alpha_4-3\delta_3}L^3.
  \end{equation}

  \par For $\mathcal{E}_{\mathcal{V}_{3,l},2}$, we have $\mathcal{E}_{\mathcal{V}_{3,l},2}
  =\sum_{j=1}^{3}\mathcal{E}_{\mathcal{V}_{3,l},1j}$, where
    \begin{align}
     \mathcal{E}_{\mathcal{V}_{3,l},21}&=\sum_{\substack{\sigma,\nu\\k,p,q}}
    \sum_{\substack{\tau,\varpi\\l,r,s}}\frac{\hat{v}_k}{L^3}\eta_l
    \delta_{p-k,r-l}\delta_{\sigma,\tau}\delta_{q+k,s+l}\delta_{\nu,\varpi}
    \int_{0}^{1}e^{-tB}a^*_{r,\tau}a^*_{s,\varpi}e^{(t-1)B}
    \nonumber\\
    &\quad\quad\times a_{q,\nu}a_{p,\sigma}e^{B}dt\phi^+(l)\gamma^-(q)
    \chi_{p-k,q+k,q\notin B_F}
    \chi_{r-l,s+l\notin B_F}\chi_{r,s\in B_F}+h.c. \nonumber\\
    \mathcal{E}_{\mathcal{V}_{3,l},22}&=\sum_{\substack{\sigma,\nu\\k,p,q}}
    \sum_{\substack{\tau,\varpi\\l,r,s}}\frac{\hat{v}_k}{L^3}\eta_l
    \delta_{p-k,r-l}\delta_{\sigma,\tau}
    \int_{0}^{1}e^{-tB}a^*_{r,\tau}a^*_{s,\varpi}a^*_{q+k,\nu}a_{s+l,\varpi}e^{(t-1)B}
     \nonumber\\
    &\quad\quad\times a_{q,\nu}a_{p,\sigma}e^{B}dt(-\phi^+(l)\gamma^-(q)
    \chi_{p-k,q+k,q\notin B_F}
    \chi_{r-l,s+l\notin B_F}\chi_{r,s\in B_F})+h.c. \nonumber\\
    \mathcal{E}_{\mathcal{V}_{3,l},23}&=\sum_{\substack{\sigma,\nu\\k,p,q}}
    \sum_{\substack{\tau,\varpi\\l,r,s}}\frac{\hat{v}_k}{L^3}\eta_l
    \delta_{q+k,s+l}\delta_{\nu,\varpi}
    \int_{0}^{1}e^{-tB}a^*_{r,\tau}a^*_{s,\varpi}a^*_{p-k,\sigma}a_{r-l,\tau}e^{(t-1)B}
     \nonumber\\
    &\quad\quad\times a_{q,\nu}a_{p,\sigma}e^{B}dt(-\phi^+(l)\gamma^-(q)
    \chi_{p-k,q+k,q\notin B_F}
    \chi_{r-l,s+l\notin B_F}\chi_{r,s\in B_F})+h.c.\label{E_V_3,2 split}
    \end{align}

  \par For $\mathcal{E}_{\mathcal{V}_{3,l},22}$, since $p-k=r-l$ while $r\in B_F$, if we allow $r-l\in B_F$, we infer $\phi^+(l)=0$. That is,
  \begin{equation}\label{E_V_3,l 22}
    \begin{aligned}
    \mathcal{E}_{\mathcal{V}_{3,l},22}=&-\sum_{\sigma,\nu,\varpi}\int_{0}^{1}
    \int_{\Lambda_L^2}L^{\frac{3}{2}}v(x_1-x_3)e^{-tB}a^*(h_{x_3,\nu})
    a^*(g_{x_1,\sigma})b_{x_1,\varpi}(\eta_{\phi^+})\\
    &\quad\quad\times e^{(t-1)B}a(\mathfrak{h}_{x_3,\nu}^{\gamma^-})a(g_{x_1,\sigma})e^B
    dx_1dx_3dt+h.c.
    \end{aligned}
  \end{equation}
  and it can be bounded by
  \begin{equation}\label{E_V_3,l 22 bound}
    \pm\mathcal{E}_{\mathcal{V}_{3,l},22}\leq C\tilde{\rho}_0^{\frac{4}{3}+\frac{\alpha_2}{2}}
    \mathcal{N}_{re}+C\tilde{\rho}_0^{3+\alpha_3+\frac{\alpha_2}{2}}L^3.
  \end{equation}
 Similarly for $\mathcal{E}_{\mathcal{V}_{3,l},23}$,
  \begin{equation}\label{E_V_3,l 23}
    \begin{aligned}
    \mathcal{E}_{\mathcal{V}_{3,l},23}=&\sum_{\sigma,\nu,\tau}\int_{0}^{1}
    \int_{\Lambda_L^2}L^{\frac{3}{2}}v(x_1-x_3)e^{-tB}a^*(h_{x_1,\sigma})
    a^*(g_{x_3,\nu})b_{x_3,\tau}(\eta_{\phi^+})\\
    &\quad\quad\times e^{(t-1)B}a(\mathfrak{h}_{x_3,\nu}^{\gamma^-})a(g_{x_1,\sigma})e^B
    dx_1dx_3dt+h.c.
    \end{aligned}
  \end{equation}
  and we have
  \begin{equation}\label{E_V_3,l 23 bound}
    \pm\mathcal{E}_{\mathcal{V}_{3,l},23}\leq C\tilde{\rho}_0^{\frac{4}{3}+\frac{\alpha_2}{2}}
    \mathcal{N}_{re}+C\tilde{\rho}_0^{3+\alpha_3+\frac{\alpha_2}{2}}L^3.
  \end{equation}

  \par For $\mathcal{E}_{\mathcal{V}_{3,l},21}$, since $p-k=r-l$ while $r\in B_F$, if $r-l\in B_F$, we deduce $\phi^+(l)=0$. Similarly, if $s-l\in B_F$, we also have $\phi^+(l)=0$. Thus
  \begin{equation}\label{E_V_3,l 21 re}
     \begin{aligned}
    \mathcal{E}_{\mathcal{V}_{3,l},21}&=\int_{0}^{1}
    \sum_{k,p,q,\sigma,\nu}\Big(\sum_{m}\frac{\hat{v}_{k-m}}{L^3}\eta_m\phi^+(m)\Big)
    e^{-tB}a^*_{p-k,\sigma}a^*_{q+k,\nu}e^{(t-1)B}a_{q,\nu}a_{p,\sigma}e^Bdt\\
    &\quad\quad\quad\quad\times \gamma^-(q)
    \chi_{p-k,q+k,p\in B_F}\chi_{q\notin B_F}+h.c.
    \end{aligned}
  \end{equation}
   We once again use the fact
  \begin{equation*}
    e^{(t-1)B}a_{q,\nu}a_{p,\sigma}e^{(1-t)B}
    =a_{q,\nu}a_{p,\sigma}+\int_{0}^{1-t}
    e^{-sB}[a_{q,\nu}a_{p,\sigma},B]e^{sB}ds
  \end{equation*}
  to rewrite $\mathcal{E}_{\mathcal{V}_{3,l},21}=\sum_{j=1}^{3}\mathcal{E}_{\mathcal{V}_{3,l},21j}$ with
  \begin{align}
    &\mathcal{E}_{\mathcal{V}_{3,l},211}=\int_{0}^{1}
    \sum_{k,p,q,\sigma,\nu}\Big(\sum_{m}\frac{\hat{v}_{k-m}}{L^3}\eta_m
    \phi^+(m)\Big)
    e^{-tB}a^*_{p-k,\sigma}a^*_{q+k,\nu}a_{q,\nu}a_{p,\sigma}e^{tB}dt\nonumber\\
    &\quad\quad\quad\quad\times\gamma^-(q)
    \chi_{p-k,q+k,p\in B_F}\chi_{q\notin B_F}+h.c.\nonumber\\
   &\mathcal{E}_{\mathcal{V}_{3,l},212}=\int_{0}^{1}\int_{0}^{1-t}
   \sum_{\substack{\sigma,\nu\\k,p,q}}
    \sum_{\substack{\tau,\varpi\\l,r,s}}\Big(\sum_{m}\frac{\hat{v}_{k-m}}{L^3}\eta_m
    \phi^+(m)\Big)
    \eta_l\delta_{q,r-l}\delta_{\nu,\tau}\nonumber\\
    &\quad\quad\times
    e^{-tB}a^*_{p-k,\sigma}a^*_{q+k,\nu}
    e^{-sB}a^*_{s+l,\varpi}a_{s,\varpi}a_{r,\tau}
    a_{p,\sigma}e^{(s+t)B}dsdt\nonumber\\
    &\quad\quad\times\gamma^-(q)\phi^+(l)
    \chi_{p-k,q+k,p\in B_F}\chi_{q\notin B_F}
    \chi_{r-l,s+l\notin B_F}\chi_{r,s\in B_F}+h.c.\nonumber\\
   & \mathcal{E}_{\mathcal{V}_{3,l},213}=\int_{0}^{1}\int_{0}^{1-t}
   \sum_{\substack{\sigma,\nu\\k,p,q}}
    \sum_{\substack{\tau,\varpi\\l,r,s}}\Big(\sum_{m}\frac{\hat{v}_{k-m}}{L^3}\eta_m
    \phi^+(m)\Big)
    \eta_l\delta_{p,r}\delta_{\sigma,\tau}\nonumber\\
    &\quad\quad\times
    e^{-tB}a^*_{p-k,\sigma}a^*_{q+k,\nu}e^{-sB}a^*_{s,\varpi}a_{q,\nu}a_{s+l,\varpi}
    a_{r-l,\tau}e^{(s+t)B}dsdt\nonumber\\
    &\quad\quad\times\gamma^-(q)\phi^+(l)
     \chi_{p-k,q+k,p\in B_F}\chi_{q\notin B_F}
    \chi_{r-l,s+l\notin B_F}\chi_{r,s\in B_F}+h.c.\label{E_V_3,211 split}
  \end{align}

  \par For $\mathcal{E}_{\mathcal{V}_{3,l},211}$, we have
\begin{equation}\label{E_V_3,l,211}
  \begin{aligned}
  \mathcal{E}_{\mathcal{V}_{3,l},211}=&
  \sum_{\sigma,\nu}\int_{0}^{1}\int_{\Lambda_L^2}v(x-y)L^{\frac{3}{2}}\eta_{\phi^+}(x-y)\\
  &\quad\quad\times e^{-tB}a(g_{x,\sigma})a^*(g_{y,\nu})a^*(g_{x,\sigma})a(\mathfrak{h}
  ^{\gamma^-}_{y,\nu})e^{tB}dxdydt+h.c.
  \end{aligned}
\end{equation}
Similar to the proof of Lemma \ref{lemma V_1}, we have
\begin{equation}\label{E_V_3,l,211 bound}
  \pm \mathcal{E}_{\mathcal{V}_{3,l},211}
  \leq C\tilde{\rho}_0\mathcal{N}_i[\delta_1]
  +C\tilde{\rho}_0^{1-\delta_2}\mathcal{N}_h[\delta_1]
  +C\tilde{\rho}_0^{1+\delta_2}\tilde{\mathcal{N}}_{re}
  +C\tilde{\rho}_0^{1+\delta_2}{\mathcal{N}}_{re}
  +C\tilde{\rho}_0^{\frac{8}{3}+\alpha_3-\delta_2}L^3,
\end{equation}
for $\frac{1}{3}>\delta_1\geq\delta_2>0$.

\par For $\mathcal{E}_{\mathcal{V}_{3,l},212}$, we have
\begin{equation}\label{E_V_3,l 212}
  \begin{aligned}
  \mathcal{E}_{\mathcal{V}_{3,l},212}=&\sum_{\sigma,\nu,\varpi}
  \int_{0}^{1}\int_{0}^{1-t}\int_{\Lambda_L^3}
  v(x_1-x_3)L^{\frac{3}{2}}\eta_{\phi^+}(x_1-x_3)L^{\frac{3}{2}}
  F_{\mathcal{V}_{3}}(x_2-x_3)\\
  &\quad\quad\times
  e^{-tB}a^*(g_{x_1,\sigma})a^*(g_{x_3,\nu})e^{-sB}b^*_{x_2,\varpi}(\eta_{\phi^+})
  \\
  &\quad\quad\times a(g_{x_2,\nu})a(g_{x_1,\sigma})dx_1dx_2dx_3dsdt+h.c.
  \end{aligned}
\end{equation}
We bound it directly by
\begin{equation}\label{E_V_3,l 212 bound}
  \pm\mathcal{E}_{\mathcal{V}_{3,l},212}\leq C\tilde{\rho}_0^{\frac{7}{3}+\frac{\alpha_2}
  {2}}L^3.
\end{equation}

\par For $\mathcal{E}_{\mathcal{V}_{3,l},213}$, we rewrite it by
\begin{equation}\label{E_V_3,l 213}
  \begin{aligned}
  \mathcal{E}_{\mathcal{V}_{3,l},213}=&-\sum_{\sigma,\nu,\varpi}
  \int_{0}^{1}\int_{0}^{1-t}\int_{\Lambda_L^3}
  v(x_1-x_3)L^{\frac{3}{2}}\eta_{\phi^+}(x_1-x_3)L^{\frac{3}{2}}
  \\
  &\quad\quad\times\Big(\frac{1}{L^3}\sum_{q_3\in B_F}e^{iq_3(x_2-x_1)}\Big)
  e^{-tB}a^*(g_{x_1,\sigma})a^*(g_{x_3,\nu})e^{-sB}a(\mathfrak{h}_{x_3,\nu}^{\gamma^-})\\
  &\quad\quad\times
  b_{x_2,\varpi}(\eta_{\phi^+})a(h_{x_2,\sigma})e^{(s-t)B}dx_1dx_2dx_3dsdt+h.c.\\
  =&-\frac{1}{L^3}\sum_{\sigma,\nu,\varpi}\sum_{q_3\in B_F}
  \int_{0}^{1}\int_{0}^{1-t}\mathcal{A}_{q_3,\sigma,\nu}(t,s)
  \mathcal{B}_{q_3,\sigma,\varpi}(t,s)dsdt+h.c.
  \end{aligned}
\end{equation}
where
\begin{equation}\label{E_V_3,l 213 AB}
  \begin{aligned}
  \mathcal{A}_{q_3,\sigma,\nu}(t,s)&=
  \int_{\Lambda_L^2}e^{-iq_3x_1}
  v(x_1-x_3)L^{\frac{3}{2}}\eta_{\phi^+}(x_1-x_3)\\
  &\quad\quad\quad\quad\quad\quad\times
   e^{-tB}a^*(g_{x_1,\sigma})a^*(g_{x_3,\nu})e^{-sB}a(\mathfrak{h}_{x_3,\nu}^{\gamma^-})dx_1dx_3\\
  \mathcal{B}_{q_3,\sigma,\varpi}(t,s)&=
   \int_{\Lambda_L}e^{iq_3x_2}L^{\frac{3}{2}}
   b_{x_2,\varpi}(\eta_{\phi^+})a(h_{x_2,\sigma})e^{(s-t)B}dx_2
  \end{aligned}
\end{equation}
Similar to the estimate of (\ref{typical example E_V_3,l,12}), we deduce
\begin{equation}\label{E_V_3,l 213 bound}
  \pm\mathcal{E}_{\mathcal{V}_{3,l},213}\leq C\tilde{\rho}_0^{\frac{4}{3}+\alpha_2-\alpha_4-3\delta_3}\mathcal{N}_{re}
  +C\big(\tilde{\rho}_0^{\frac{7}{3}+\alpha_4}
  +\tilde{\rho}_0^{3+\alpha_3+\alpha_2-\alpha_4-3\delta_3}\big)L^3.
\end{equation}
\par Combining all the above estimates we close the proof of Lemma \ref{lemma qua conj V_3}.
\end{proof}

\par The next two lemmas compute the energy contribution of the quadratic terms.
\begin{lemma}\label{lemma qua V_21' Omega}
For any $N\in\mathbb{N}_{\geq0}$,
\begin{equation}\label{qua V_21' Omega}
  \begin{aligned}
  &\int_{0}^{1}\int_{0}^{t}e^{-sB}[\mathcal{V}_{21}^\prime+\Omega,B]e^{sB}dsdt
  \\
  &=\frac{1}{L^3}\sum_{k,p,q,\sigma,\nu}W_k\eta_k\phi^+(k)\zeta^-(k)
  \chi_{p-k,q+k\notin B_F}\chi_{p,q\in B_F}\\
  &-\frac{1}{L^3}\sum_{k,p,q,\sigma}W_k\eta_{k+q-p}\zeta^-(k)
  \phi^+(k+q-p)\chi_{p-k,q+k\notin B_F}\chi_{p,q\in B_F}\chi_{p\neq q}\\
  &-\sum_{k,p,q,\sigma}\eta_k\eta_{k+q-p}\phi^+(k)\zeta^-(k)
  \phi^+(k+q-p)k(q-p)
  \chi_{p-k,q+k\notin B_F}\chi_{p,q\in B_F}\\
  &+\mathcal{E}_{(\mathcal{V}_{21}^\prime+\Omega)},
  \end{aligned}
\end{equation}
where
\begin{equation}\label{qua E V_21'Omega}
  \begin{aligned}
  \pm\mathcal{E}_{(\mathcal{V}_{21}^\prime+\Omega)}\leq&
  C\big(\tilde{\rho}_0^{\frac{4}{3}-\alpha_3}
  +\tilde{\rho}_0^{\frac{3}{2}+\frac{\alpha_2}{2}+4\alpha_3-\frac{3}{2}\alpha_5}\big)
  \mathcal{N}_{re}+C\tilde{\rho}_0^{\frac{4}{3}-\alpha_3}\tilde{\mathcal{N}}_{re}\\
  &+\tilde{\rho}_0^{\frac{5}{6}+\frac{\alpha_2}{2}-\frac{3}{2}\alpha_5}
  \mathcal{N}_h[-\alpha_2]+C\tilde{\rho}_0^{\frac{5}{2}+\frac{\alpha_2}{2}+\alpha_3
  -\frac{3}{2}\alpha_5}L^3.
  \end{aligned}
\end{equation}

for $\frac{1}{24}>\frac{1}{4}\alpha_3>\alpha_2>0$, $\beta_1=\frac{1}{3}+\alpha_5$ and $2\alpha_3>\alpha_5>2\alpha_4$.
\end{lemma}
\begin{proof}
  \par For $[\mathcal{V}_{21}^\prime,B]$, we have
   \begin{equation*}
    [\mathcal{V}_{21}^\prime,B]=\Xi_{1}+\Xi_2+\Xi_3+\Xi_4
  \end{equation*}
  with
  \begin{align}
   &\Xi_1=\sum_{\substack{\sigma,\nu\\k,p,q}}\sum_{l,r,s}\frac{W_k}{L^3}\eta_l
   \zeta^-(k)\phi^+(l)
    \delta_{p-k,r-l}\delta_{q+k,s+l}a^*_{p,\sigma}a^*_{q,\nu}a_{s,\nu}a_{r,\sigma}
    \nonumber\\
    &\quad\quad\quad\quad\times
    \chi_{p-k,q+k,r-l,s+l\notin B_F}\chi_{p,q,r,s\in B_F}+h.c.\nonumber\\
    &\Xi_2=\sum_{\substack{\sigma,\nu\\k,p,q}}\sum_{\tau,l,r,s}2\frac{W_k}{L^3}\eta_l
    \zeta^-(k)\phi^+(l)
    \delta_{p-k,s+l}a^*_{p,\sigma}a^*_{q,\nu}a^*_{r-l,\tau}
    a_{q+k,\nu}a_{s,\sigma}a_{r,\tau}\nonumber\\
    &\quad\quad\quad\quad\times
    \chi_{p-k,q+k,r-l,s+l\notin B_F}\chi_{p,q,r,s\in B_F}+h.c.\nonumber\\
    &\Xi_3=-\sum_{\substack{\sigma,\nu\\k,p,q}}\sum_{\varpi,l,r,s}2\frac{W_k}{L^3}\eta_l
    \zeta^-(k)\phi^+(l)
    \delta_{p,r}a^*_{r-l,\sigma}a^*_{s+l,\varpi}a_{s,\varpi}a^*_{q,\nu}
    a_{q+k,\nu}a_{p-k,\sigma}\nonumber\\
    &\quad\quad\quad\quad\times
    \chi_{p-k,q+k,r-l,s+l\notin B_F}\chi_{p,q,r,s\in B_F}+h.c.\nonumber\\
    &\Xi_4=\sum_{\substack{\sigma,\nu\\k,p,q}}\sum_{l,r,s}\frac{W_k}{L^3}\eta_l
    \zeta^-(k)\phi^+(l)
    \delta_{p,r}\delta_{q,s}a^*_{r-l,\sigma}a^*_{s+l,\nu}
    a_{q+k,\nu}a_{p-k,\sigma}\nonumber\\
    &\quad\quad\quad\quad\times
    \chi_{p-k,q+k,r-l,s+l\notin B_F}\chi_{p,q,r,s\in B_F}+h.c.
    \label{[V_21' B] XI}
  \end{align}

\par For $\Xi_1$, we can calculate it in analog to the analysis to the corresponding $\Xi_1$ in \cite[Lemma 7.8]{WJH}. In the thermodynamic limit, we need to use the bound of $W$ and $\eta$ given in Section \ref{scattering eqn sec}, and additionally make use of Lemma \ref{cutoff lemma}, in place of the estimates given in \cite[Section 4]{WJH}. We also point out that the cut-off $\phi^+$ used in the definition of $A$ in (\ref{define A}) ensures, for example, for $r,r-l\in B_F$, hence $\phi^+(l)=0$, Such localizaitons helps us to avoid the $\mathcal{N}_{re}^2$ type estimate, which in the thermodynamic limit has the order $L^6\gg L^3$. We the present the result directly and omit further calculations.
\begin{equation}\label{Xi_1}
  \begin{aligned}
  \Xi_1=&\frac{2}{L^3}\sum_{k,p,q,\sigma,\nu}W_k\eta_k\phi^+(k)\zeta^-(k)
  \chi_{p-k,q+k\notin B_F}\chi_{p,q\in B_F}\\
  &-\frac{2}{L^3}\sum_{k,p,q,\sigma}W_k\eta_{k+q-p}\zeta^-(k)
  \phi^+(k+q-p)\chi_{p-k,q+k\notin B_F}\chi_{p,q\in B_F}\chi_{p\neq q}\\
  &+\mathcal{E}_{\Xi_1},
  \end{aligned}
\end{equation}
with
\begin{equation}\label{E_Xi_1}
  \pm\mathcal{E}_{\Xi_1}\leq C\tilde{\rho}_0^{\frac{4}{3}-\alpha_3}\tilde{\mathcal{N}}_{re}.
\end{equation}

  \par For $\Xi_2$, since $p-k=s+l\in B_F$, then $\phi^+(l)=0$, therefore,
  \begin{equation}\label{Xi_2}
    \begin{aligned}
    \Xi_2=2\sum_{\sigma,\nu,\tau}\int_{\Lambda_L^3}W^{\zeta^-}(x_1-x_3)
    &L^{\frac{3}{2}}\eta_{\phi^+}(x_2-x_1)
    a^*(h_{x_2,\tau})a^*(g_{x_1,\sigma})a^*(g_{x_3,\nu})\\
    &\times a(g_{x_1,\sigma})a(g_{x_2,\tau})
    a(h_{x_3,\nu})dx_1dx_2dx_3+h.c.
    \end{aligned}
  \end{equation}
  It is straight-forward to bound
  \begin{equation}\label{XI_2 bound}
    \pm\Xi_2\leq C\tilde{\rho}_0^{\frac{4}{3}+2\alpha_3}\mathcal{N}_{re}.
  \end{equation}

  \par For $\Xi_3$, notice that $\phi^+(l)>0$ implies $r-l,s+l\in P_{F,-\alpha_2}$, while $\zeta^-(k)>0$ implies $\vert p-k\vert,\vert q+k\vert\leq 3\tilde{\mu}^{\frac{1}{2}}\tilde{\rho}_0^{-\beta_1}$. So, we can write $\Xi_3$ as
  \begin{equation}\label{Xi_3}
    \Xi_3=-\frac{2}{L^3}\sum_{\sigma,\nu,\varpi}\sum_{q_1\in B_F}\mathcal{A}_{q_1,\sigma,\varpi}
    \mathcal{B}_{q_1,\sigma,\nu}+h.c.,
  \end{equation}
  where
  \begin{equation*}
    \begin{aligned}
    \mathcal{A}_{q_1,\sigma,\varpi}=&\int_{\Lambda_L}e^{-iq_1x_2}L^{\frac{3}{2}}
    a^*(H_{x_2,\sigma}[-\alpha_2])b^*_{x_2,\varpi}(\eta_{\phi^+})dx_2\\
    \mathcal{B}_{q_1,\sigma,\nu}=&\int_{\Lambda_L^2}e^{iq_1x_1}W^{\zeta^-}(x_1-x_3)
    a^*(g_{x_3,\nu})a(\mathbf{h}_{x_3,\nu})a(\mathbf{h}_{x_1,\sigma})dx_1dx_3
    \end{aligned}
  \end{equation*}
  and we have temporarily defined
  \begin{equation*}
  \mathbf{h}_{x,\sigma}(z)=\sum_{k_F<\vert k\vert\leq 3\tilde{\mu}^{\frac{1}{2}}\tilde{\rho}
  _0^{-\beta_1}}\frac{e^{ikx}}{L^{\frac{3}{2}}}f_{k,\sigma}(z).
\end{equation*}
Using $\mathbf{h}_{x_1,\sigma}=L_{x_1,\sigma}[-\alpha_2]
+(\mathbf{h}_{x_1,\sigma}-L_{x_1,\sigma}[-\alpha_2])$, we can split $\Xi_3=\Xi_{31}+\Xi_{32}$. Similar to the bound of (\ref{typical example E_V_3,l,12}), we have
\begin{equation}\label{Xi_31 and 32 bound}
  \begin{aligned}
  \pm\Xi_{31}\leq& C\tilde{\rho}_0^{\frac{11}{6}+\frac{3}{2}\alpha_5-\frac{5}{2}
    \alpha_2}\mathcal{N}_{re}+C\tilde{\rho}_0^{\frac{5}{6}+\frac{\alpha_2}{2}-\frac{3}{2}
    \alpha_5}\mathcal{N}_h[-\alpha_2]\\
    \pm\Xi_{32}\leq& C\tilde{\rho}_0^{\frac{5}{6}+\frac{\alpha_2}{2}-\frac{3}{2}
    \alpha_5}\mathcal{N}_h[-\alpha_2]
  \end{aligned}
\end{equation}
and therefore,
  \begin{equation}\label{XI_3 bound}
    \pm\Xi_3\leq C\tilde{\rho}_0^{\frac{11}{6}+\frac{3}{2}\alpha_5-\frac{5}{2}
    \alpha_2}\mathcal{N}_{re}+C\tilde{\rho}_0^{\frac{5}{6}+\frac{\alpha_2}{2}-\frac{3}{2}
    \alpha_5}\mathcal{N}_h[-\alpha_2].
  \end{equation}

  \par For $\Xi_4$, we combine the proof of (\ref{ineq b 2}) with the bound of $\Xi_3$ in (\ref{Xi_3}), i.e. we use
  \begin{equation*}
    1=\frac{1}{\vert s+l\vert^2}\big(\vert s\vert^2+2s\cdot l+\vert l\vert^2\big).
  \end{equation*}
  Referring to the proof of (\ref{ineq b 2}) for details, using the method in bounding $\Xi_3$, we can bound $\Xi_4$ similarly
  \begin{equation}\label{XI_4 bound}
    \pm\Xi_4\leq C\tilde{\rho}_0^{\frac{11}{6}+\frac{3}{2}\alpha_5-\frac{5}{2}
    \alpha_2}\mathcal{N}_{re}+C\tilde{\rho}_0^{\frac{5}{6}+\frac{\alpha_2}{2}-\frac{3}{2}
    \alpha_5}\mathcal{N}_h[-\alpha_2].
  \end{equation}

  \par For $[\Omega,B]$, we also have
  \begin{equation*}
    [\Omega,B]=\sum_{j=1}^{4}\tilde{\Xi}_j,
  \end{equation*}
  with each $\tilde{\Xi}_j$ defined by replacing $W_kL^{-3}\zeta^-(k)$ by
  $\eta_kk(q-p)\phi^+(k)\zeta^(k)$. The bounds to $\tilde{\Xi}_j$ for $j=2,3,4$ is totally analogous to the bounds to $\Xi_j$ above. On the other hand, the idea of calculation to $\tilde{\Xi}_1$ can be again found in \cite[Lemma 7.8]{WJH}, as that also have $\Xi_1$ and $\tilde{\Xi}_1$ with the same relation. Notice due to the cut-off in the thermodynamic limit, we have
  \begin{equation}\label{cutoff effect}
  \begin{aligned}
    &\sum_{k,p,q,\sigma,\nu}\eta^2_kk(q-p)\phi^+(k)\zeta^-(k)\chi_{p-k,q+k\notin B_F}
    \chi_{p,q\in B_F}\\
    &=\sum_{k,p,q,\sigma,\nu}\eta^2_kk(q-p)\phi^+(k)\zeta^-(k)
    \chi_{p,q\in B_F}\\
    &=-\sum_{k,p,q,\sigma,\nu}\eta^2_kk(q-p)\phi^+(k)\zeta^-(k)
    \chi_{p,q\in B_F}=0.
  \end{aligned}
  \end{equation}
  \par Combining above analysis we attain the proof of Lemma \ref{lemma qua V_21' Omega}.
\end{proof}

\begin{lemma}\label{lemma qua V_21}
For any $N\in\mathbb{N}_{\geq0}$,
  \begin{equation}\label{qua V_21}
    \begin{aligned}
    \int_{0}^{1}\int_{t}^{1}e^{-sB}[\mathcal{V}_{21},B]e^{sB}dsdt&=
    \frac{1}{2L^3}\sum_{k\in(2\pi/L)\mathbb{Z}^3}\hat{v}_k\eta_k
    \mathbf{q}(\mathbf{q}-1)\bar{N}_0^2\\
    &-\frac{1}{L^3}\sum_{k\in(2\pi/L)\mathbb{Z}^3}\hat{v}_k\eta_k(\mathbf{q}-1)\bar{N}_0
    \tilde{\mathcal{N}}_{re}+\mathcal{E}_{\mathcal{V}_{21}},
    \end{aligned}
  \end{equation}
  where
  \begin{equation}\label{qua E_V_21}
    \begin{aligned}
    \pm\mathcal{E}_{\mathcal{V}_{21}}\leq&
    C\tilde{\rho}_0^{1+\delta_1}\mathcal{N}_{re}+C\tilde{\rho}_0^{1+\delta_1}
    \tilde{\mathcal{N}}_{re}+C\tilde{\rho}_0\mathcal{N}_h[\delta_1]
    +C\tilde{\rho}_0\mathcal{N}_i[\delta_1]\\
    &+C\tilde{\rho}_0^{\frac{1}{3}+\alpha_4}\mathcal{V}_4+C\tilde{\rho}_0^
    {\frac{7}{3}+\alpha_4}L^3,
    \end{aligned}
  \end{equation}
  for $\frac{1}{24}>\frac{1}{4}\alpha_3>\alpha_2>2\alpha_4>0$ and $0<\delta_1<\frac{1}{3}$.
\end{lemma}
\begin{proof}
  \par For $[\mathcal{V}_{21},B]$, we have
   \begin{equation*}
    [\mathcal{V}_{21},B]=\hat{\Xi}_{1}+\hat{\Xi}_2+\hat{\Xi}_3+\hat{\Xi}_4
  \end{equation*}
  with
  \begin{align}
   &\hat{\Xi}_1=\sum_{\substack{\sigma,\nu\\k,p,q}}\sum_{l,r,s}\frac{v_k}{2L^3}\eta_l
   \phi^+(l)
    \delta_{p-k,r-l}\delta_{q+k,s+l}a^*_{p,\sigma}a^*_{q,\nu}a_{s,\nu}a_{r,\sigma}
    \nonumber\\
    &\quad\quad\quad\quad\times
    \chi_{p-k,q+k,r-l,s+l\notin B_F}\chi_{p,q,r,s\in B_F}+h.c.\nonumber\\
    &\hat{\Xi}_2=\sum_{\substack{\sigma,\nu\\k,p,q}}\sum_{\tau,l,r,s}\frac{v_k}{L^3}\eta_l
    \phi^+(l)
    \delta_{p-k,s+l}a^*_{p,\sigma}a^*_{q,\nu}a^*_{r-l,\tau}
    a_{q+k,\nu}a_{s,\sigma}a_{r,\tau}\nonumber\\
    &\quad\quad\quad\quad\times
    \chi_{p-k,q+k,r-l,s+l\notin B_F}\chi_{p,q,r,s\in B_F}+h.c.\nonumber\\
    &\hat{\Xi}_3=-\sum_{\substack{\sigma,\nu\\k,p,q}}\sum_{\varpi,l,r,s}\frac{v_k}{L^3}\eta_l
    \phi^+(l)
    \delta_{p,r}a^*_{r-l,\sigma}a^*_{s+l,\varpi}a_{s,\varpi}a^*_{q,\nu}
    a_{q+k,\nu}a_{p-k,\sigma}\nonumber\\
    &\quad\quad\quad\quad\times
    \chi_{p-k,q+k,r-l,s+l\notin B_F}\chi_{p,q,r,s\in B_F}+h.c.\nonumber\\
    &\hat{\Xi}_4=\sum_{\substack{\sigma,\nu\\k,p,q}}\sum_{l,r,s}\frac{v_k}{2L^3}\eta_l
    \phi^+(l)
    \delta_{p,r}\delta_{q,s}a^*_{r-l,\sigma}a^*_{s+l,\nu}
    a_{q+k,\nu}a_{p-k,\sigma}\nonumber\\
    &\quad\quad\quad\quad\times
    \chi_{p-k,q+k,r-l,s+l\notin B_F}\chi_{p,q,r,s\in B_F}+h.c.
    \label{[V_21 B] XI}
  \end{align}

  \par For $\hat{\Xi}_1$, we can follow the calculation of the corresponding $\hat{\Xi}_1$ in \cite[Lemma 7.9]{WJH}, despite here we use estimates in Section \ref{scattering eqn sec}, instead of \cite[Section 4]{WJH}. We thus omit further tedious computation and write out the result
  \begin{equation}\label{Xi_1hat}
    \begin{aligned}
    \hat{\Xi}_1=&\frac{1}{2L^3}\sum_{k\in(2\pi/L)\mathbb{Z}^3}\hat{v}_k\eta_k
    \mathbf{q}(\mathbf{q}-1)\bar{N}_0^2
    -\frac{1}{L^3}\sum_{k\in(2\pi/L)\mathbb{Z}^3}\hat{v}_k\eta_k(\mathbf{q}-1)\bar{N}_0
    \tilde{\mathcal{N}}_{re}+\mathcal{E}_{\hat{\Xi}_1},
    \end{aligned}
  \end{equation}
  with
  \begin{equation}\label{E_hatXi_1}
    \pm\mathcal{E}_{\hat{\Xi}_1}\leq C\tilde{\rho}_0^{1+\delta_1}\tilde{\mathcal{N}}_{re}
    +C\tilde{\rho}_0\mathcal{N}_i[\delta]+C\tilde{\rho}_0^{\frac{7}{3}+2\alpha_3-3\alpha_2}L^3.
  \end{equation}

  \par For $\hat{\Xi}_2$, its bound is similar to (\ref{Xi_2}), due to the existence of cut-off $\phi^+$:
   \begin{equation}\label{hatXi_2}
    \begin{aligned}
    \hat{\Xi}_2=\sum_{\sigma,\nu,\tau}\int_{\Lambda_L^3}v(x_1-x_3)
    &L^{\frac{3}{2}}\eta_{\phi^+}(x_2-x_1)
    a^*(h_{x_2,\tau})a^*(g_{x_1,\sigma})a^*(g_{x_3,\nu})\\
    &\times a(g_{x_1,\sigma})a(g_{x_2,\tau})
    a(h_{x_3,\nu})dx_1dx_2dx_3+h.c.
    \end{aligned}
  \end{equation}
  It is straight-forward to bound
  \begin{equation}\label{XI_2 hat bound}
    \pm\hat{\Xi}_2\leq C\tilde{\rho}_0^{\frac{4}{3}+2\alpha_3}\mathcal{N}_{re}.
  \end{equation}

  \par For $\hat{\Xi}_3$, its bound is similar to (\ref{Xi_3}),
  \begin{equation}\label{hatXi_3}
    \Xi_3=-\frac{1}{L^3}\sum_{\sigma,\nu,\varpi}\sum_{q_1\in B_F}\mathcal{A}_{q_1,\sigma,\varpi}
    \mathcal{B}_{q_1,\sigma,\nu},
  \end{equation}
  where
  \begin{equation*}
    \begin{aligned}
    \mathcal{A}_{q_1,\sigma,\varpi}=&\int_{\Lambda_L}e^{-iq_1x_2}L^{\frac{3}{2}}
    a^*(h_{x_2,\sigma})b^*_{x_2,\varpi}(\eta_{\phi^+})dx_2\\
    \mathcal{B}_{q_1,\sigma,\nu}=&\int_{\Lambda_L^2}e^{iq_1x_1}v(x_1-x_3)
    a^*(g_{x_3,\nu})a({h}_{x_3,\nu})a({h}_{x_1,\sigma})dx_1dx_3
    \end{aligned}
  \end{equation*}
  Similar to (\ref{typical example E_V_3,l,12}), we can use Cauchy-Schwartz inequality to obtain
  \begin{equation}\label{XI_3 hat bound}
    \pm\hat{\Xi}_3\leq C\tilde{\rho}_0^{\frac{4}{3}+\alpha_2-\alpha_4}
    \mathcal{N}_{re}+C\tilde{\rho}_0^{\frac{1}{3}+\alpha_4}\mathcal{V}_4.
  \end{equation}

  \par For $\hat{\Xi}_4$, the idea is similar to (\ref{XI_4 bound}), which means we combine the proof of (\ref{ineq b 2}) with the bound of $\hat{\Xi}_3$. Therefore,
  \begin{equation}\label{XI_4 hat bound}
    \pm\hat{\Xi}_4\leq C\tilde{\rho}_0^{\frac{4}{3}+\alpha_2-\alpha_4}
    \mathcal{N}_{re}+C\tilde{\rho}_0^{\frac{1}{3}+\alpha_4}\mathcal{V}_4.
  \end{equation}
\par Combining the above analysis we finish the proof of Lemma \ref{lemma qua V_21}.
\end{proof}

\par Lemma \ref{lemma qua Gamma} estimates the operator $\Gamma$ defined in (\ref{define Gamma}).
\begin{lemma}\label{lemma qua Gamma}
For any $N\in\mathbb{N}_{\geq0}$,
\begin{equation}\label{qua Gamma}
  \begin{aligned}
  \pm\int_{0}^{1}e^{-tB}\Gamma e^{tB}dt\leq& C\tilde{\rho}_0^{1+\delta_1}\mathcal{N}_{re}
  +C\tilde{\rho}_0\mathcal{N}_h[\delta_1]+C
  \tilde{\rho}_0^{1+\alpha_5-4\alpha_3-\alpha_4}\mathcal{N}_h[-\beta_1]\\
  &+C\big(\tilde{\rho}_0^{\frac{1}{3}+2\alpha_3}+
  \tilde{\rho}_0^{\frac{1}{3}+\alpha_5-\alpha_4}\big)\mathcal{K}_s
  +C\tilde{\rho}_0^{\frac{1}{3}+\alpha_4}\mathcal{V}_4\\
  &+C\big(\tilde{\rho}_0^{\frac{7}{3}+\alpha_4}
  +\tilde{\rho}_0^{2+4\alpha_3-7\alpha_2}
  +\tilde{\rho}_0^{\frac{8}{3}+\alpha_5-3\alpha_3-\alpha_4}\big)L^3,
  \end{aligned}
\end{equation}
for $\frac{1}{24}>\frac{1}{4}\alpha_3>\alpha_2>2\alpha_4>0$, $0<\delta_1<\frac{1}{3}$, $\beta_1=\frac{1}{3}+\alpha_5$ and $2\alpha_3>\alpha_5>2\alpha_4$.
\end{lemma}
\begin{proof}
  \par By (\ref{Wdiscrete asymptotic energy pde on the torus}), (\ref{define V_21' and Omega}), (\ref{define Gamma}), Lemma \ref{lemma qua cal [K,B]} and Lemma \ref{lemma qua control V_4} we have
  \begin{equation}\label{Gamma cal}
    \Gamma=\Omega_h-\Omega_{m,l}+\Theta_{d,1}+\Theta_{d,2}+\Theta_r+\mathcal{V}_{21,h}^\prime,
  \end{equation}
  where $\Omega_h$ and $\Omega_{m,l}$ are bounded in (\ref{Omega m,l and Omega h bound}), $\Theta_{d,1}$ is bounded in (\ref{Theta d,1 bound}), $\Theta_{d,2}=0$, and $\mathcal{V}_{21,h}^\prime$ is bounded in (\ref{qua V_21,h'}). Then using Lemma \ref{lemma qua control N_re}, estimates (\ref{qua control K}) and (\ref{qua control V_4}) we obtain (\ref{qua Gamma}).
\end{proof}

\vspace{1em}

\begin{proof}[Proof of Proposition \ref{qua prop}]
  (\ref{define G_N}) together with Lemmas \ref{lemma qua conj V_0 1 22 23}-\ref{lemma qua Gamma} yield Proposition \ref{qua prop}.
\end{proof}

\section{Cubic Renormalization}\label{cub}
\
\par In this section, we analyze the excitation Hamiltonian $\mathcal{J}_N$ defined in (\ref{J_N}), and prove Proposition \ref{cub prop}. We define
\begin{equation}\label{define B'}
  B^\prime=A^\prime-{A^\prime}^*
\end{equation}
where
\begin{equation}\label{define A'}
  A^\prime=\sum_{k,p,q,\sigma,\nu}
  \eta_k\phi^+(k)\zeta^-(k)a^*_{p-k,\sigma}a^*_{q+k,\nu}a_{q,\nu}a_{p,\sigma}\chi_{p-k,q+k\notin B_F}\chi_{q\in A_{F,\delta_4}}
  \chi_{p\in B_F},
\end{equation}
and as defined in (\ref{partition 3D space}),
\begin{equation}\label{A^nu_F,kappa}
  A_{F,\delta_4}=\{k\in(2\pi/L)\mathbb{Z}^3,\,k_F<\vert k\vert\leq k_F+\tilde{\mu}^{\frac{1}{2}}\tilde{\rho}_0^{\delta_4}\}
\end{equation}
for $\frac{1}{3}>\delta_4\geq\delta_1$ and $\delta_4>2\alpha_3+2\alpha_4$.

\par To deal with the length scale $L\to\infty$, we apply the frequency localization. We let
\begin{equation}\label{divide V_4}
  \mathcal{V}_4=\mathcal{V}_{4,4h}+\mathcal{V}_{4,2h}^{(1)}+\mathcal{V}_{4,r},
\end{equation}
with
\begin{equation}\label{divide V_4 detailed}
  \begin{aligned}
  \mathcal{V}_{4,4h}&=\frac{1}{2L^3}\sum_{k,p,q,\sigma,\nu}
  \hat{v}_ka^*_{p-k,\sigma}a^*_{q+k,\nu}a_{q,\nu}a_{p,\sigma}\chi_{p,p-k\in P_{F,\delta_4}}\chi_{q,q+k\in P_{F,\delta_4}}\\
  \mathcal{V}_{4,2h}^{(1)}&=\frac{1}{2L^3}\sum_{k,p,q,\sigma,\nu}
  \hat{v}_ka^*_{p-k,\sigma}a^*_{q+k,\nu}a_{q,\nu}a_{p,\sigma}\chi_{p,p-k\in P_{F,\delta_4}}\chi_{q,q+k\in A_{F,\delta_4}}+h.c.
  \end{aligned}
\end{equation}
where $P_{F,\delta_4}$ is also defined in (\ref{partition 3D space}). We also let
\begin{equation}\label{divide V_3,l}
  \mathcal{V}_{3,l}=\mathcal{V}_{3,L}+\mathcal{V}_{3,R},
\end{equation}
with
\begin{equation}\label{V_3,L}
  \mathcal{V}_{3,L}=\frac{1}{L^3}\sum_{k,p,q,\sigma,\nu}\hat{v}_k
(a_{p-k,\sigma}^*a_{q+k,\nu}^*a_{q,\nu}a_{p,\sigma}+h.c.)
\chi_{p-k,q+k\in P_{F,\delta_4}}\chi_{q\in A_{F,\delta_4}}\chi_{p\in B_F}
\end{equation}
Using (\ref{G_N qua prop}) and Newton-Leibniz formula, we rewrite $\mathcal{J}_N$ by
\begin{equation}\label{define J_N}
  \begin{aligned}
  &\mathcal{J}_N\coloneqq e^{-B^\prime}\mathcal{G}_Ne^{B^\prime}\\
  =&C_{\mathcal{G}_N}+\mathcal{K}
  +\mathcal{V}_{4,4h}+\mathcal{V}_{4,2h}^{(1)}
  +e^{-B^\prime}\big(Q_{\mathcal{G}_N}+\mathcal{V}_{21}^\prime
  +\Omega+\mathcal{V}_{3,R}+\mathcal{V}_{4,r}+
  \mathcal{E}_{\mathcal{G}_N}\big)e^{B^\prime}\\
  &+\int_{0}^{1}e^{-tB\prime}\Gamma^\prime e^{tB^\prime}dt
  +\int_{0}^{1}\int_{t}^{1}e^{-sB^\prime}[\mathcal{V}_{3,L},B^\prime]e^{sB^\prime}dsdt.
  \end{aligned}
\end{equation}
with
\begin{equation}\label{define Gamma'}
  \Gamma^\prime=[\mathcal{K}+\mathcal{V}_{4,4h}
  +\mathcal{V}_{4,2h}^{(1)},B^\prime]+\mathcal{V}_{3,L}.
\end{equation}

\par To prove Proposition \ref{cub prop}, we are going to analyze each term on the right-hand side of (\ref{define J_N}), in Lemmas \ref{lemma cub misc}-\ref{lemma cub Gamma'}. But before we can rigorously calculate each term of $\mathcal{J}_N$, we first need to bound the actions of the cubic renormalization on some special operators, such as $\mathcal{N}_{re}, \mathcal{K}_s$ and $\mathcal{V}_4$. These results are collected respectively in Lemmas \ref{lemma cub control N_re}-\ref{lemma cub control V_4}.

\par Lemma \ref{lemma cub control N_re} controls the action of $e^{B^\prime}$ on several particle number operators.
\begin{lemma}\label{lemma cub control N_re}
For any $N\in\mathbb{N}_{\geq0}$ and $\vert t\vert\leq 1$,
\begin{align}
 e^{-tB^\prime}\mathcal{N}_{re}e^{tB^\prime}&\leq C\mathcal{N}_{re},
 \label{cub control N_re1}\\
 e^{-tB^\prime}\tilde{\mathcal{N}}_{re}e^{tB^\prime}&\leq C\tilde{\mathcal{N}}_{re}+ C\mathcal{N}_{re},
 \label{cub control N_retilde1}\\
 \pm\big(e^{-tB^\prime}\mathcal{N}_{re}e^{tB^\prime}-\mathcal{N}_{re}\big)&\leq
 C\tilde{\rho}_0^{\frac{1}{3}+\frac{\alpha_3}{2}}\mathcal{N}_{re},
 \label{cub control diff N_re}\\
  \pm\big(e^{-tB^\prime}\tilde{\mathcal{N}}_{re}e^{tB^\prime}-\tilde{\mathcal{N}}_{re}\big)&\leq
 C\tilde{\rho}_0^{\frac{1}{3}+\frac{\alpha_3}{2}}\mathcal{N}_{re}.
 \label{cub control diff N_retilde}
\end{align}
Moreover, for $0<\delta\leq\delta_4$,
\begin{equation}\label{cub control N_i[delta]}
  e^{-tB^\prime}\mathcal{N}_i[\delta]e^{tB^\prime}\leq C\mathcal{N}_i[\delta]
  +C\mathcal{N}_h[-\alpha_2]+C\tilde{\rho}_0^{\frac{5}{3}+\delta_4+\alpha_3}L^3,
\end{equation}
and for any $\delta\leq\delta_4$,
\begin{equation}\label{cub control N_h[delta]1}
  e^{-tB^\prime}\mathcal{N}_h[\delta]e^{tB^\prime}\leq C\mathcal{N}_h[\delta]+
  C\tilde{\rho}_0^{\frac{2}{3}+\delta_4+4\alpha_3}
  \mathcal{N}_h[-\alpha_2]+C\tilde{\rho}_0^{\frac{5}{3}+\delta_4+\alpha_3}L^3,
\end{equation}
\end{lemma}
\begin{proof}
  \par We have $[\mathcal{N}_{re},A^\prime]=A^\prime$ and therefore
  $[\mathcal{N}_{re},B^\prime]=A^\prime+A^{\prime*}$. We can write $A^\prime$ by
  \begin{equation}\label{A'}
    \begin{aligned}
    A^\prime=&\sum_{\sigma,\nu}L^{\frac{3}{2}}\int_{\Lambda_L^2}
    \eta_{\phi^+}^{\zeta^-}(x-y)a^*(h_{x,\sigma})a^*(h_{y,\nu})a(L_{y,\nu}[\delta_4])
    a(g_{x,\sigma})dxdy\\
    =&\sum_{\sigma,\nu}L^{\frac{3}{2}}\int_{\Lambda_L}
    a^*(h_{y,\nu})b^*_{y,\sigma}(\eta_{\phi^+}^{\zeta^-})a(L_{y,\nu}[\delta_4])dy
    \end{aligned}
  \end{equation}
  By Lemma \ref{b^* bound by b} and Cauchy-Schwartz inequality, for $\psi\in\mathcal{H}^{\wedge}$
  \begin{equation}\label{bound A' 1st}
    \begin{aligned}
    \vert\langle A^\prime\psi,\psi\rangle\vert&\leq CL^{\frac{3}{2}}\sum_{\sigma,\nu}
    \int_{\Lambda_L}\Vert a(h_{y,\nu})\psi\Vert\big(
    \Vert a(L_{y,\nu}[\delta_4]) b_{y,\sigma}(\eta_{\phi^+}^{\zeta^-})
    \psi\Vert\\
    &\quad\quad\quad\quad\quad\quad\quad\quad+\tilde{\rho}_0^{\frac{1}{2}}
    \Vert\eta_{\phi^+}^{\zeta^-}\Vert_2\Vert a(L_{y,\nu}[\delta_4])\psi\Vert \big)dy
    \end{aligned}
  \end{equation}
  By Cauchy-Schwartz inequality, we have
  \begin{equation}\label{A' bound}
    \pm A^\prime\leq
    C\tilde{\rho}_0^{\frac{1}{3}+\frac{\alpha_3}{2}}\mathcal{N}_{re}.
  \end{equation}
 (\ref{cub control N_re1}) follows by Gronwall's inequality.

  \par On the other hand, we have $[\tilde{\mathcal{N}}_{re},A^\prime]=\mathcal{A}^\prime$, as well, thus (\ref{cub control N_retilde1}) follows easily from (\ref{A' bound}) and (\ref{cub control N_re1}). Moreover, (\ref{cub control diff N_re}) and (\ref{cub control diff N_retilde}) are also direct consequences of (\ref{A' bound}) and (\ref{cub control N_re1}).

  \par For inequality (\ref{cub control N_h[delta]1}), for $\delta\leq\delta_4$, we calculate directly
  \begin{equation}\label{com N_hdelta A'}
  \begin{aligned}
    &[{\mathcal{N}}_{h}[\delta],A^\prime]\\
    =&
    \sum_{k,p,q,\sigma,\nu}
  \eta_k\phi^+(k)\zeta^-(k)a^*_{p-k,\sigma}a^*_{q+k,\nu}a_{q,\nu}a_{p,\sigma}
  \chi_{p-k\in P_{F,\delta}}\chi_{q+k\notin B_F}\chi_{q\in A_{F,\delta_4}}\chi_{p\in B_F}\\
   &+\sum_{k,p,q,\sigma,\nu}
  \eta_k\phi^+(k)\zeta^-(k)a^*_{p-k,\sigma}a^*_{q+k,\nu}a_{q,\nu}a_{p,\sigma}
  \chi_{p-k\notin B_F}\chi_{q+k\in P_{F,\delta}}\chi_{q\in A_{F,\delta_4}}\chi_{p\in B_F}
  \end{aligned}
  \end{equation}
  Notice that $\phi^+(k)>0$ makes $p-k,q+k\in P_{F,-\alpha_2}$. For $-\alpha_2\leq\delta\leq\delta_4$, since $\mathcal{N}_{h}[-\alpha_2]\leq\mathcal{N}_{h}[\delta]$, we can bound $[{\mathcal{N}}_{h}[\delta],A^\prime]$ using Lemma \ref{b^* bound by b general}:
  \begin{equation}\label{N_hA' -alpha_2}
    \pm[{\mathcal{N}}_{h}[\delta],A^\prime]\leq  C\mathcal{N}_h[\delta]+
    C\tilde{\rho}_0^{\frac{5}{3}+\delta_4+\alpha_3}L^3.
  \end{equation}
  While for $\delta<-\alpha_2$, we can similarly bound
  \begin{equation}\label{N_hA' }
    \pm[{\mathcal{N}}_{h}[\delta],A^\prime]\leq  C\mathcal{N}_h[\delta]+
  C\tilde{\rho}_0^{\frac{2}{3}+\delta_4+4\alpha_3}
  \mathcal{N}_h[-\alpha_2]+C\tilde{\rho}_0^{\frac{5}{3}+\delta_4+\alpha_3}L^3.
  \end{equation}
 (\ref{cub control N_h[delta]1}) follows by Gronwall's inequality.

  \par For inequality (\ref{cub control N_i[delta]}) for $0<\delta\leq\delta_4$, we have
  \begin{equation}\label{com N_idelta A'}
  \begin{aligned}
    &[{\mathcal{N}}_{i}[\delta],A^\prime]\\
    =&
    \sum_{k,p,q,\sigma,\nu}
  \eta_k\phi^+(k)\zeta^-(k)a^*_{p-k,\sigma}a^*_{q+k,\nu}a_{q,\nu}a_{p,\sigma}
  \chi_{p-k,q+k\notin B_F}\chi_{q\in A_{F,\delta_4}}\chi_{p\in \underline{B}_{F,\delta}}
  \end{aligned}
  \end{equation}
  Also notice that $\phi^+(k)>0$ implies $p-k,q+k\in P_{F,-\alpha_2}$, so we have
   \begin{equation}\label{[N_i B']}
   \begin{aligned}
     [{\mathcal{N}}_{i}[\delta],A^\prime]=\sum_{\sigma,\nu}L^{\frac{3}{2}}\int_{\Lambda_L^2}
    \eta_{\phi^+}^{\zeta^-}&(x-y)a^*(H_{x,\sigma}[-\alpha_2])a^*(H_{y,\nu}[-\alpha_2])\\
    &\times a(L_{y,\nu}[\delta_4])
    a(I_{x,\sigma}[\delta])dxdy.
   \end{aligned}
   \end{equation}
   For $0<\delta\leq\delta_4$, since $\mathcal{N}_{h}[-\alpha_2]\leq\mathcal{N}_{h}[\delta]$, similar to (\ref{N_hA' -alpha_2}), we have
  \begin{equation}\label{N_idelta A'}
    \pm[{\mathcal{N}}_{i}[\delta],A^\prime]\leq
  C\mathcal{N}_h[-\alpha_2]+C\tilde{\rho}_0^{\frac{5}{3}+\delta_4+\alpha_3}L^3.
  \end{equation}
  Hence (\ref{cub control N_i[delta]}) follows from (\ref{cub control N_h[delta]1}) and Newton-Leibniz formula.
\end{proof}

\begin{lemma}\label{lem cub conrtol K}
  For any $N\in\mathbb{N}_{\geq0}$ and $\vert t\vert\leq 1$, and for $\frac{1}{24}>\frac{1}{2}\delta_3>\frac{1}{4}\alpha_3>\alpha_2>2\alpha_4>0$, $\frac{1}{3}>\delta_4\geq\delta_1\geq\delta_2>\frac{1}{12}$, $\delta_4>2\alpha_3+2\alpha_4$, $\beta_1=\frac{1}{3}+\alpha_5$ and $2\alpha_3>\alpha_5>2\alpha_4$, we have
  \begin{align}
  e^{-tB^\prime}\mathcal{K}_se^{tB^\prime}&\leq
  \mathcal{K}_s+\tilde{\rho}_0^{\frac{2}{3}}
  \mathcal{N}_h[-\alpha_2]+\tilde{\rho}_0^{2+\delta_4}L^3,\label{cub control K_s}\\
  e^{-tB^\prime}\mathcal{K}_h[-\alpha_2]e^{tB^\prime}&\leq
  \mathcal{K}_h[-\alpha_2]+\tilde{\rho}_0^{\frac{4}{3}
  +\delta_4+2\alpha_3}
  \mathcal{N}_h[-\alpha_2]+\tilde{\rho}_0^{2+\delta_4}L^3.\label{cub control K_h[-alpha2]}
  \end{align}
  Here $\mathcal{K}_h[-\alpha_2]$ is defined in (\ref{K_h[delta]}). We also have
\begin{equation}\label{cub com [K,B']}
  [\mathcal{K},B^\prime]=\Omega_m^\prime+\Omega_r^\prime,
\end{equation}
where
\begin{equation}\label{cub Omega_m'}
  \Omega_m^\prime=\sum_{k,p,q,\sigma,\nu}2\vert k\vert^2\eta_k
  (a_{p-k,\sigma}^*a_{q+k,\nu}^*a_{q,\nu}a_{p,\sigma}+h.c.)
  \chi_{p-k,q+k\notin B_F}\chi_{q\in A_{F,\delta_4}}\chi_{p\in B_F},
\end{equation}
and
\begin{equation}\label{cub Omega_r'}
\begin{aligned}
\pm\Omega_r^\prime\leq& C\tilde{\rho}_0^{1+\delta_1}\tilde{\mathcal{N}}_{re}
+C\tilde{\rho}_0\mathcal{N}_h[\delta_1]+C\tilde{\rho}_0^{\frac{2}{3}
+\delta_4+\alpha_3-\alpha_4}\mathcal{N}_h[-\alpha_2]\\
&+C\big(\tilde{\rho}_0^{1+\delta_4-2\alpha_3+\alpha_5-\alpha_4}
+\tilde{\rho}_0^{1+\frac{\delta_4}{2}+\alpha_3}\big)\mathcal{N}_h[-\beta_1]\\
&+C\big(\tilde{\rho}_0^{\delta_4+\alpha_3-\alpha_4}+\tilde{\rho}_0^
{\frac{1}{3}+\delta_4-2\alpha_3+\alpha_5-\alpha_4}\big)\mathcal{K}_h[-\alpha_2]\\
&+C\big(\tilde{\rho}_0^{\frac{7}{3}+\alpha_4}+\tilde{\rho}_0^
{2+\delta_4+4\alpha_3-7\alpha_2}\big)L^3.
\end{aligned}
\end{equation}
\end{lemma}
\begin{proof}
  \par By a direct calculation of $[\mathcal{K},B^\prime]$, we reach (\ref{cub com [K,B']}) with
  \begin{equation}\label{Omega_r'}
    \Omega_r^\prime=\Omega^\prime-\Omega_{d,1}^\prime-\Omega_{d,2}^\prime,
  \end{equation}
  where
  \begin{equation}\label{Omega_r'detailed}
    \begin{aligned}
    \Omega^\prime=&\sum_{k,p,q,\sigma,\nu}
  2 k(q-p)\eta_k\phi^+(k)\zeta^-(k)(a^*_{p-k,\sigma}a^*_{q+k,\nu}a_{q,\nu}a_{p,\sigma}+h.c.)\\
  &\quad\quad\quad\quad
  \times\chi_{p-k,q+k\notin B_F}\chi_{q\in A_{F,\delta_4}}
  \chi_{p\in B_F}\\
  \Omega^\prime_{d,1}=&\sum_{k,p,q,\sigma,\nu}
  2\vert k\vert^2\eta_k\phi^-(k)
  (a^*_{p-k,\sigma}a^*_{q+k,\nu}a_{q,\nu}a_{p,\sigma}+h.c.)\\
  &\quad\quad\quad\quad
  \times\chi_{p-k,q+k\notin B_F}\chi_{q\in A_{F,\delta_4}}
  \chi_{p\in B_F}\\
  \Omega^\prime_{d,2}=&\sum_{k,p,q,\sigma,\nu}
  2 \vert k\vert^2\eta_k\zeta^+(k)
  (a^*_{p-k,\sigma}a^*_{q+k,\nu}a_{q,\nu}a_{p,\sigma}+h.c.)\\
  &\quad\quad\quad\quad
  \times\chi_{p-k,q+k\notin B_F}\chi_{q\in A_{F,\delta_4}}
  \chi_{p\in B_F}
    \end{aligned}
  \end{equation}
  First, we have, for some universal complex constant $C$, that
  \begin{equation}\label{Omega0'}
    \begin{aligned}
    \Omega_m^\prime-\Omega_{d,1}^\prime-\Omega_{d,2}^\prime
    =CL^{\frac{3}{2}}\sum_{\sigma,\nu}\int_{\Lambda_L^2}
    &\Delta_x\big(\eta_{\phi^+}^{\zeta^-}\big)(x-y)
    a^*(H_{x,\sigma}[-\alpha_2])a^*(H_{y,\nu}[-\alpha_2])\\
    &\times a(L_{y,\nu}[\delta_4])
    a(g_{x,\sigma})dxdy+h.c.
    \end{aligned}
  \end{equation}
  Here, we use the fact that $\phi^+(k)>0$ implies $p-k,q+k\in P_{F,-\alpha_2}$. Similar to the bound of (\ref{Omegam}), we have
  \begin{equation}\label{Omega0' bound}
    \pm(\Omega_m^\prime-\Omega_{d,1}^\prime-\Omega_{d,2}^\prime)\leq C\mathcal{K}_s+C\tilde{\rho}_0^{\frac{2}{3}}
    \mathcal{N}_h[-\alpha_2]+C\tilde{\rho}_0^{2+\delta_4}L^3.
  \end{equation}
  For $\Omega^\prime$, we have
  \begin{equation}\label{Omega'}
    \begin{aligned}
    \Omega^\prime
    =&CL^{\frac{3}{2}}\sum_{\sigma,\nu}\int_{\Lambda_L^2}
    \nabla\big(\eta_{\phi^+}^{\zeta^-}\big)(x-y)
    a^*(H_{x,\sigma}[-\alpha_2])a^*(H_{y,\nu}[-\alpha_2])\\
    &\quad\quad\times a(\nabla_y L_{y,\nu}[\delta_4])
    a(g_{x,\sigma})dxdy+h.c.\\
    +&CL^{\frac{3}{2}}\sum_{\sigma,\nu}\int_{\Lambda_L^2}
    \nabla\big(\eta_{\phi^+}^{\zeta^-}\big)(x-y)
    a^*(H_{x,\sigma}[-\alpha_2])a^*(H_{y,\nu}[-\alpha_2])\\
    &\quad\quad
    \times a( L_{y,\nu}[\delta_4])
    a(\nabla_xg_{x,\sigma})dxdy+h.c.
    \end{aligned}
  \end{equation}
  Similar to the bound of $\Omega$ in (\ref{Omega1st}) and using Lemma \ref{b^* bound by b general}, we have
  \begin{equation}\label{Omega'bound1}
    \pm\Omega^\prime\leq C\mathcal{K}_s+C\tilde{\rho}_0^{\frac{2}{3}}
    \mathcal{N}_h[-\alpha_2]+C\tilde{\rho}_0^{\frac{7}{3}+\delta_4+\alpha_3}L^3.
  \end{equation}
  With $\mathcal{K}_h[-\alpha_2]$ define in (\ref{K_h[delta]}), we also have
  \begin{equation}\label{Omega'bound2}
    \pm\Omega^\prime \leq C\tilde{\rho}_0^{\delta_4+\alpha_3-\alpha_4}
    \mathcal{K}_h[-\alpha_2]+C\tilde{\rho}_0^{\frac{2}{3}+\delta_4+\alpha_3-\alpha_4}
    \mathcal{N}_h[-\alpha_2]+C\tilde{\rho}_0^{\frac{7}{3}+\alpha_4}L^3,
  \end{equation}
   Combining (\ref{Omega0' bound}) and (\ref{Omega'bound1}), we reach
  \begin{equation}\label{[K,B']bound}
   \pm[\mathcal{K},B^\prime]\leq C\mathcal{K}_s+C\tilde{\rho}_0^{\frac{2}{3}}
    \mathcal{N}_h[-\alpha_2]+C\tilde{\rho}_0^{2+\delta_4}L^3,
  \end{equation}
which together with Lemma \ref{lemma cub control N_re} and Gronwall's inequality yield (\ref{cub control K_s}), since we have $[\mathcal{K},B^\prime]=[\mathcal{K}_s,B^\prime]$, which is from (\ref{K_s and K}).
  \par For $\Omega_{d,1}^\prime$, we can bound it in the way completely analogous to the bound of $\Omega_{m,l}$ in (\ref{Omega_m,l}):
  \begin{equation}\label{Omega d1'bound}
    \pm\Omega_{d,1}^\prime\leq C\tilde{\rho}_0^{1+\delta_1}\tilde{\mathcal{N}}_{re}
+C\tilde{\rho}_0\mathcal{N}_h[\delta_1]+C\tilde{\rho}_0^{2+\delta_4+4\alpha_3-7\alpha_2}L^3.
  \end{equation}
  \par For $\Omega_{d,2}^\prime$, since $\zeta^+(k)>0$ implies $p-k,q+k\in P_{F,-\beta_1}$, we can write it by
  \begin{equation}\label{Omegad2'}
  \begin{aligned}
    \Omega_{d,2}^\prime=CL^{\frac{3}{2}}\sum_{\sigma,\nu}\int_{\Lambda_L^2}
    \Delta_x\big(\eta_{\zeta^+}\big)(x-y)&
    a^*(H_{x,\sigma}[-\beta_1])a^*(H_{y,\nu}[-\beta_1])\\
    &\times
    a(L_{y,\nu}[\delta_4])
    a(g_{x,\sigma})dxdy+h.c.
  \end{aligned}
  \end{equation}
  Using the estimate (\ref{special grad eta L2}) and the fact that $\mathcal{K}_h[-\beta_1]\leq\mathcal{K}_h[-\alpha_2]$, we can bound it similar to (\ref{Omega0'}), that
  \begin{equation}\label{Omega d2'bound}
  \begin{aligned}
     \pm\Omega_{d,2}^\prime\leq &
+C\big(\tilde{\rho}_0^{1+\delta_4-2\alpha_3+\alpha_5-\alpha_4}
+\tilde{\rho}_0^{1+\frac{\delta_4}{2}+\alpha_3}\big)\mathcal{N}_h[-\beta_1]\\
&+  C\big(\tilde{\rho}_0^{\frac{1}{3}+\delta_4-2\alpha_3+\alpha_5-\alpha_4}
+\tilde{\rho}_0^{\frac{1}{3}+\frac{\delta_4}{2}+\alpha_3}\big)\mathcal{K}_h[-\alpha_2]
+C\tilde{\rho}_0^{\frac{7}{3}+\alpha_4}L^3.
\end{aligned}
  \end{equation}
  Therefore, combining (\ref{Omega_r'}), (\ref{Omega'bound2}), (\ref{Omega d1'bound}) and (\ref{Omega d2'bound}) we reach (\ref{cub Omega_r'}).

  \par For inequality (\ref{cub control K_h[-alpha2]}), we have
  \begin{equation}\label{[K_h,B']}
  \begin{aligned}
    [\mathcal{K}_h[-\alpha_2],B^\prime]=
    &\sum_{\substack{k,p,q\\\sigma,\nu}}
  2 \big(\vert p-k\vert^2\chi_{p-k\in P_{F,-\alpha_2}}
  +\vert q+k\vert^2\chi_{q+k\in P_{F,-\alpha_2}}\big)
  \eta_k\phi^+(k)\zeta^-(k)\\
  &\quad\quad
  \times(a^*_{p-k,\sigma}a^*_{q+k,\nu}a_{q,\nu}a_{p,\sigma}+h.c.)
  \chi_{p-k,q+k\notin B_F}\chi_{q\in A_{F,\delta_4}}
  \chi_{p\in B_F}
  \end{aligned}
  \end{equation}
  Again we notice that $\phi^+(k)>0$ implies $p-k,q+k\in P_{F,-\alpha_2}$, and we can use
  \begin{align*}
    \vert p-k\vert^2=(p-k)p-(p-k)k,\quad
    \vert q+k\vert^2=(q+k)q+(q+k)k.
  \end{align*}
  to rewrite $[\mathcal{K}_h[-\alpha_2],B^\prime]$. Then switching to the position space, we can use Lemma \ref{b^* bound by b general} to bound it by
  \begin{equation}\label{[K_h,B'] bound}
    \pm[\mathcal{K}_h[-\alpha_2],B^\prime]\leq C\mathcal{K}_h[-\alpha_2]
    +C\tilde{\rho}_0^{\frac{4}{3}+2\alpha+3+\delta_4}\mathcal{N}_h[-\alpha_2]+
    C\tilde{\rho}_0^{2+\delta_4}L^3.
  \end{equation}
  Combining (\ref{[K_h,B'] bound}) with Lemma \ref{lemma cub control N_re} and Gronwall's inequality, we reach (\ref{cub control K_h[-alpha2]}).
\end{proof}

\begin{lemma}\label{lemma cub control V_4}
  For any $N\in\mathbb{N}_{\geq0}$ and $\vert t\vert\leq 1$, and for $\frac{1}{24}>\frac{1}{2}\delta_3>\frac{1}{4}\alpha_3>\alpha_2>2\alpha_4>0$, $\frac{1}{3}>\delta_4\geq\delta_1\geq\delta_2>\frac{1}{12}$, $\delta_4>2\alpha_3+2\alpha_4$, $\beta_1=\frac{1}{3}+\alpha_5$ and $2\alpha_3>\alpha_5>2\alpha_4$, we have
  \begin{equation}\label{cub control V_4,4h}
     e^{-tB^\prime}\mathcal{V}_{4,4h}e^{tB^\prime}\leq
     \mathcal{V}_{4,4h}+\tilde{\rho}_0\mathcal{N}_{re}.
  \end{equation}
  We also have
  \begin{align}
    \pm\mathcal{V}_{4,2h}^{(1)}&\leq C\tilde{\rho}_0^{\alpha_4}
    \mathcal{V}_{4,4h}+\tilde{\rho}_0^{2+2\delta_4-\alpha_4}L^3,\label{cub control V_4,2h}\\
    \pm\mathcal{V}_{4,r}&\leq C\tilde{\rho}_0\mathcal{N}_h[{\delta_4}]
    +C\tilde{\rho}_0^{2+2\delta_4}L^3.\label{cub control V_4r}
  \end{align}
  Moreover,
  \begin{equation}\label{cub com [V_4,4h,B']}
    [\mathcal{V}_{4,4h},B^\prime]=\Theta_m^\prime+\Theta_r^\prime,
  \end{equation}
  where
  \begin{equation}\label{cub Theta_m'}
  \begin{aligned}
    \Theta_m^\prime=\sum_{k,p,q,\sigma,\nu}\Big(\frac{1}{L^3}\sum_{l}
    \hat{v}_{k-l}\eta_l\Big)&(a_{p-k,\sigma}^*a_{q+k,\nu}^*a_{q,\nu}a_{p,\sigma}+h.c.)\\
  &\times\chi_{p-k,q+k\notin B_F}\chi_{q\in A_{F,\delta_4}}\chi_{p\in B_F},
  \end{aligned}
  \end{equation}
  and
  \begin{equation}\label{cub Theta_r'}
    \begin{aligned}
    \pm\Theta_r^\prime\leq&C\big(\tilde{\rho}_0^{\frac{5}{3}+3\alpha_3-6\alpha_2-\delta_4}
    +\tilde{\rho}_0^{\frac{7}{3}-4\alpha_3+2\alpha_5-\delta_4}+\tilde{\rho}_0^
    {\frac{5}{3}+\alpha_2}\big)\mathcal{N}_{re}\\
    &+C\tilde{\rho}_0^{1+\delta_4}\tilde{\mathcal{N}}_{re}+C\tilde{\rho}_0\mathcal{N}_h[\delta_4]
    +C\tilde{\rho}_0^{\delta_4}\mathcal{V}_{4,4h}+C\tilde{\rho}_0^{2+2\delta_4}L^3.
    \end{aligned}
  \end{equation}
  Furthermore,
  \begin{equation}\label{cub com V_4 2h 1}
    \begin{aligned}
    \pm[\mathcal{V}_{4,2h}^{(1)},B^\prime]\leq&C\tilde{\rho}_0^{1+\delta_4}\tilde{\mathcal{N}}_
    {re}+C\tilde{\rho}_0^{\frac{5}{3}+\alpha_2+\delta_4}\mathcal{N}_h[\delta_4]
    +C\tilde{\rho}_0^{\frac{5}{3}+\alpha_2}\mathcal{N}_h[-\alpha_2]\\
    &+C\tilde{\rho}_0^{\delta_4}\mathcal{V}_{4,4h}+C\tilde{\rho}_0^{2+2\delta_4}L^3.
    \end{aligned}
  \end{equation}
  Therefore,
  \begin{equation}\label{cub control V_4}
    \begin{aligned}
    e^{-tB^\prime}\mathcal{V}_{4}e^{tB^\prime}\leq&C\mathcal{V}_{4,4h}
    +C\tilde{\rho}_0\mathcal{N}_{re}+C\tilde{\rho}_0^{1+\delta_4}\tilde{\mathcal{N}}_{re}
    \\
    &+C\tilde{\rho}_0^{\frac{5}{3}+\alpha_2}\mathcal{N}_h[-\alpha_2]
    +C\tilde{\rho}_0^{2+2\delta_4-\alpha_4}L^3.
    \end{aligned}
  \end{equation}
\end{lemma}
\begin{proof}
  \par From (\ref{divide V_4}), we rewrite $\mathcal{V}_4$ by
  \begin{equation}\label{V_4split}
    \mathcal{V}_{4}=\mathcal{V}_{4,4h}+\mathcal{V}_{4,3h}+\mathcal{V}_{4,2h}^{(1)}
    +\mathcal{V}_{4,2h}^{(2)}+\mathcal{V}_{4,1h}+\mathcal{V}_{4,0h},
  \end{equation}
  where $\mathcal{V}_{4,4h}$ and $\mathcal{V}_{4,2h}^{(1)}$ have been defined in (\ref{divide V_4 detailed}), and $\mathcal{V}_{4,r}=\mathcal{V}_{4,3h}
    +\mathcal{V}_{4,2h}^{(2)}+\mathcal{V}_{4,1h}$. Here $\mathcal{V}_{4,3h}$ means the $3$ indices of its $4$ creation and annihilation operators are inside of $P_{F,\delta_4}$, and similarly, $\mathcal{V}_{4,1h}$ means $1$ of its $4$ is inside of $P_{F,\delta_4}$, $\mathcal{V}_{4,0h}$ means all of its $4$ are inside of $A_{F,\delta_4}$, while $\mathcal{V}_{4,2h}^{(2)}$ means $2$ of its $4$ are inside of $P_{F,\delta_4}$, except for $\mathcal{V}_{4,2h}^{(1)}$, which is special. It is then trivial to reach the bounds (\ref{cub control V_4,2h}) and (\ref{cub control V_4r}).

  \par We the calculate $[\mathcal{V}_{4,4h},B^\prime]$. We have
  \begin{equation}\label{[V_4,4h,B']}
    [\mathcal{V}_{4,4h},B^\prime]=\Theta_m^\prime+\Theta_{d,0}^\prime-\Theta_{d,1}^\prime
    -\Theta_{d,2}^\prime+\Theta_{d,3}^\prime+\Theta_{r,1}^\prime
    +\Theta_{r,2}^\prime,
  \end{equation}
  with
  \begin{align}
  \Theta_m^\prime+\Theta_{d,0}^\prime=&\sum_{k,p,q,\sigma,\nu}
 \Big(\frac{1}{L^3}\sum_{l}\hat{v}_{k-l}\eta_l\Big)
 (a^*_{p-k,\sigma}a^*_{q+k,\nu}a_{q,\nu}a_{p,\sigma}+h.c.)\nonumber\\
  &\quad\quad\quad\quad
  \times\chi_{p-k,q+k\in P_{F,\delta_4}}\chi_{q\in A_{F,\delta_4}}
  \chi_{p\in B_F}\nonumber\\
  \Theta_{d,1}^\prime&=\sum_{k,p,q,\sigma,\nu}
 \Big(\frac{1}{L^3}\sum_{l}\hat{v}_{k-l}\eta_l\phi^-(l)\Big)
 (a^*_{p-k,\sigma}a^*_{q+k,\nu}a_{q,\nu}a_{p,\sigma}+h.c.)\nonumber\\
  &\quad\quad\quad\quad
  \times\chi_{p-k,q+k\in P_{F,\delta_4}}\chi_{q\in A_{F,\delta_4}}
  \chi_{p\in B_F}\nonumber\\
  \Theta_{d,2}^\prime&=\sum_{k,p,q,\sigma,\nu}
 \Big(\frac{1}{L^3}\sum_{l}\hat{v}_{k-l}\eta_l\zeta^+(l)\Big)
 (a^*_{p-k,\sigma}a^*_{q+k,\nu}a_{q,\nu}a_{p,\sigma}+h.c.)\nonumber\\
  &\quad\quad\quad\quad
  \times\chi_{p-k,q+k\in P_{F,\delta_4}}\chi_{q\in A_{F,\delta_4}}
  \chi_{p\in B_F}\nonumber\\
   \Theta_{d,3}^\prime&=\frac{1}{L^3}\sum_{l,k,p,q,\sigma,\nu}
 \hat{v}_{k-l}\eta_l\phi^+(l)\zeta^-(l)
 (a^*_{p-k,\sigma}a^*_{q+k,\nu}a_{q,\nu}a_{p,\sigma}+h.c.)\nonumber\\
  &\quad\quad\quad\quad
  \times\chi_{p-k,q+k\in P_{F,\delta_4}}\chi_{q\in A_{F,\delta_4}}
  \chi_{p\in B_F}(\chi_{p-l\notin B_F}\chi_{q+l\notin B_F}-1)\nonumber\\
  \Theta_{r,1}^\prime&=\frac{1}{L^3}\sum_{\substack{k,p,q\\\sigma,\nu}}\sum_{\varpi,l,r,s}
  -\hat{v}_k\eta_l\delta_{p,r-l}(a^*_{p-k,\sigma}a^*_{q+k,\nu}a^*_{s+l,\varpi}
  a_{q,\nu}a_{s,\varpi}a_{r,\sigma}+h.c.)\nonumber\\
  &\quad\quad\quad\quad\times\phi^+(l)\zeta^-(l)
  \chi_{p-k,q+k,p,q\in P_{F,\delta_4}}
  \chi_{r-l,s+l\notin B_F}\chi_{s\in A_{F,\delta_4}}\chi_{r\in B_F}\nonumber\\
  \Theta_{r,2}^\prime&=\frac{1}{L^3}\sum_{\substack{k,p,q\\\sigma,\nu}}\sum_{\tau,l,r,s}
  \hat{v}_k\eta_l\delta_{p,s+l}(a^*_{p-k,\sigma}a^*_{q+k,\nu}a^*_{r-l,\tau}
  a_{q,\nu}a_{s,\sigma}a_{r,\tau}+h.c.)\nonumber\\
  &\quad\quad\quad\quad\times\phi^+(l)\zeta^-(l)
  \chi_{p-k,q+k,p,q\in P_{F,\delta_4}}
  \chi_{r-l,s+l\notin B_F}\chi_{s\in A_{F,\delta_4}}\chi_{r\in B_F}
  \label{cal [V_4,B'] detailed}
  \end{align}
  \par For $\Theta_m^\prime+\Theta_{d,0}^\prime$, we can write it by
  \begin{equation}\label{Theta_m'+Theta_d,0'}
  \begin{aligned}
    \Theta_m^\prime+\Theta_{d,0}^\prime=\sum_{\sigma,\nu}\int_{\Lambda_L^2}
    v(x-y)L^{\frac{3}{2}}&\eta(x-y)a^*(H_{x,\sigma}[\delta_4])a^*(H_{y,\nu}[\delta_4])\\
    &\times a(L_{y,\nu}[\delta_4])
    a(g_{x,\sigma})dxdy+h.c.
  \end{aligned}
  \end{equation}
  Using Cauchy-Schwartz inequality, we have
  \begin{equation}\label{Theta_m' bound}
    \pm(\Theta_m^\prime+\Theta_{d,0}^\prime)\leq C\mathcal{V}_{4,4h}+C\tilde{\rho}_0\mathcal{N}_{re}.
  \end{equation}
  \par For $\Theta_{d,0}^\prime$ alone, we can use (\ref{cub Theta_m'}) and (\ref{Theta_m'+Theta_d,0'}) to bound it by
  \begin{equation}\label{Theta_d,0'bound}
    \pm\Theta_{d,0}^\prime\leq C\tilde{\rho}_0^{1+\delta_4}\tilde{\mathcal{N}}_{re}
    +C\tilde{\rho}_0\mathcal{N}_h[\delta_4]
    +C\tilde{\rho}_0^{2+2\delta_4}L^3.
  \end{equation}
  \par For $\Theta_{d,1}^\prime$, we bound it similar to (\ref{Theta_m'+Theta_d,0'}), but using estimate (\ref{bound eta-etatilde Linfty}):
  \begin{equation}\label{Theta_d,1'bound}
    \pm\Theta_{d,1}^\prime\leq C\tilde{\rho}_0^{\delta_4}\mathcal{V}_{4,4h}
    +C\tilde{\rho}_0^{\frac{5}{3}+4\alpha_3-6\alpha_2-\delta_4}\mathcal{N}_{re}.
  \end{equation}
 \par For $\Theta_{d,2}^\prime$, we bound it similar to (\ref{Theta_m'+Theta_d,0'}), but using estimate (\ref{special eta Linf}):
  \begin{equation}\label{Theta_d,2'bound}
    \pm\Theta_{d,2}^\prime\leq C\tilde{\rho}_0^{\delta_4}\mathcal{V}_{4,4h}
    +C\tilde{\rho}_0^{\frac{7}{3}-4\alpha_3+2\alpha_5-\delta_4}\mathcal{N}_{re}.
  \end{equation}
  \par For $\Theta_{d,3}^\prime$, $(\chi_{p-l\notin B_F}\chi_{q+l\notin B_F}-1)\neq0$ implies
  $\phi^+(l)=0$, therefore $\Theta_{d,3}^\prime=0$.
  \par For $\Theta_{r,1}^\prime$, if $r-l=p\in B_F\cup A_{F,\delta_4}$, then $\phi^+(l)=0$. Therefore,
  \begin{equation}\label{Theta_r,1'}
    \begin{aligned}
    \Theta_{r,1}^\prime=-\sum_{\sigma,\nu,\varpi}\int_{\Lambda_3}
    v(x_1-x_3)&L^{\frac{3}{2}}\eta_{\phi^+}^{\zeta^-}(x_1-x_4)
    a^*(H_{x_1,\sigma}[\delta_4])a^*(H_{x_3,\nu}[\delta_4])a^*(h_{x_4,\varpi})\\
    &\times a(L_{x_4,\varpi}[\delta_4])a(g_{x_1,\sigma})a(H_{x_3,\nu}[\delta_4])
    dx_1dx_3dx_4+h.c.\\
    =-\sum_{\sigma,\nu,\varpi}\int_{\Lambda_2}
    v(x_1-x_3)&L^{\frac{3}{2}}a^*(H_{x_1,\sigma}[\delta_4])a^*(H_{x_3,\nu}[\delta_4])\\
    &\times b^*_{x_1,\varpi}(\eta_{\phi^+}^{\zeta^-})a(g_{x_1,\sigma})a(H_{x_3,\nu}[\delta_4])
    dx_1dx_3+h.c.
    \end{aligned}
  \end{equation}
  Using (\ref{ineq b 2 general}), we can bound it by
  \begin{equation}\label{Theta_r,1'bound}
    \pm\Theta_{r,1}^\prime\leq C\tilde{\rho}_0^{\delta_4}\mathcal{V}_{4,4h}
    +C\tilde{\rho}_0^{\frac{5}{3}+\alpha_2}\mathcal{N}_{re}.
  \end{equation}
  $\Theta_{r,2}^\prime$ can be bounded similarly:
  \begin{equation}\label{Theta_r,2'bound}
    \pm\Theta_{r,2}^\prime\leq C\tilde{\rho}_0^{\delta_4}\mathcal{V}_{4,4h}
    +C\tilde{\rho}_0^{\frac{5}{3}+\alpha_2}\mathcal{N}_{re}.
  \end{equation}
Therefore (\ref{cub control V_4,4h}), (\ref{cub com [V_4,4h,B']}) and (\ref{cub Theta_r'}) follow from (\ref{[V_4,4h,B']})-(\ref{Theta_r,2'bound}), Lemma \ref{lemma cub control N_re} and Gronwall's inequality.

  \par For $[\mathcal{V}_{4,2h}^{(1)},B^\prime]$, we have $[\mathcal{V}_{4,2h}^{(1)},B^\prime]=\sum_{j=1}^{4}\hat{\Theta}_{r,j}^\prime$
  \begin{align}
  \hat{\Theta}_{r,1}^\prime&=\frac{1}{L^3}\sum_{\substack{k,p,q\\\sigma,\nu}}\sum_{l,r,s}
  -\hat{v}_k\eta_l\phi^+(l)\zeta^-(l)\delta_{p,r-l}\delta_{q,s+l}(a^*_{p-k,\sigma}a^*_{q+k,\nu}
  a_{s,\nu}a_{r,\sigma}+h.c.)\nonumber\\
  &\quad\quad\quad\quad\times\chi_{p-k,q+k\in A_{F,\delta_4}}
  \chi_{p-k,q+k,p,q\in P_{F,\delta_4}}
  \chi_{r-l,s+l\notin B_F}\chi_{s\in A_{F,\delta_4}}\chi_{r\in B_F}\nonumber\\
  \hat{\Theta}_{r,2}^\prime&=\frac{1}{L^3}\sum_{\substack{k,p,q\\\sigma,\nu}}\sum_{\varpi,l,r,s}
  -\hat{v}_k\eta_l\phi^+(l)\zeta^-(l)\delta_{p,r-l}(a^*_{p-k,\sigma}a^*_{q+k,\nu}a^*_{s+l,\varpi}
  a_{q,\nu}a_{s,\varpi}a_{r,\sigma}+h.c.)\nonumber\\
  &\quad\quad\quad\quad\times\chi_{p-k,q+k\in A_{F,\delta_4}}
  \chi_{p-k,q+k,p,q\in P_{F,\delta_4}}
  \chi_{r-l,s+l\notin B_F}\chi_{s\in A_{F,\delta_4}}\chi_{r\in B_F}\nonumber\\
  \hat{\Theta}_{r,3}^\prime&=\frac{1}{L^3}\sum_{\substack{k,p,q\\\sigma,\nu}}\sum_{\tau,l,r,s}
  \hat{v}_k\eta_l\phi^+(l)\zeta^-(l)\delta_{p,s+l}(a^*_{p-k,\sigma}a^*_{q+k,\nu}a^*_{r-l,\tau}
  a_{q,\nu}a_{s,\sigma}a_{r,\tau}+h.c.)\nonumber\\
  &\quad\quad\quad\quad\times\chi_{p-k,q+k\in A_{F,\delta_4}}
  \chi_{p,q\in P_{F,\delta_4}}
  \chi_{r-l,s+l\notin B_F}\chi_{s\in A_{F,\delta_4}}\chi_{r\in B_F}\nonumber\\
  \hat{\Theta}_{r,4}^\prime&=\frac{1}{L^3}\sum_{\substack{k,p,q\\\sigma,\nu}}\sum_{\tau,l,r,s}
  -\hat{v}_k\eta_l\phi^+(l)\zeta^-(l)\delta_{p-k,s}(a^*_{r-l,\tau}a^*_{s+l,\sigma}a^*_{q+k,\nu}
  a_{q,\nu}a_{p,\sigma}a_{r,\tau}+h.c.)\nonumber\\
  &\quad\quad\quad\quad\times\chi_{p-k,q+k\in A_{F,\delta_4}}
  \chi_{p,q\in P_{F,\delta_4}}
  \chi_{r-l,s+l\notin B_F}\chi_{s\in A_{F,\delta_4}}\chi_{r\in B_F}
  \label{cal [V_4,2h,B'] detailed}
  \end{align}
  \par For $\hat{\Theta}_{r,1}^\prime$, if $r-l=p\in B_F\cup A_{F,\delta_4}$, then $\phi^+(l)=0$, and if we allow $s+l=q\in B_F\cup A_{F,\delta_4}$, then also $\phi^+(l)=0$. Therefore
  \begin{equation}\label{hat Theta_r1'}
  \begin{aligned}
    \hat{\Theta}_{r,1}^\prime=\sum_{\sigma,\nu}\int_{\Lambda_L^2}
    v(x-y)L^{\frac{3}{2}}&\eta_{\phi^+}^{\zeta^-}(x-y)
    a^*(L_{x,\sigma}[\delta_4])a^*(L_{y,\nu}[\delta_4])\\
    &\times a(L_{y,\nu}[\delta_4])
    a(g_{x,\sigma})dxdy+h.c.
  \end{aligned}
  \end{equation}
  \begin{equation}\label{hat Theta_r,1'bound}
    \pm\hat{\Theta}_{r,1}^\prime\leq C\tilde{\rho}_0^{1+\delta_4}\tilde{\mathcal{N}}_{re}
    +C\tilde{\rho}_0^{2+2\delta_4}L^3.
  \end{equation}
  \par For $\hat{\Theta}_{r,2}^\prime$ and $\hat{\Theta}_{r,3}^\prime$, with the above observation, similar to the bound of $\Theta_{r,1}^\prime$ in (\ref{Theta_r,1'}), we have
  \begin{equation}\label{hat Theta_r,2'bound}
    \pm\hat{\Theta}_{r,2}^\prime,\pm\hat{\Theta}_{r,3}^\prime\leq
    C\tilde{\rho}_0^{\frac{5}{3}+\alpha_2+\delta_4}{\mathcal{N}}_{h}[\delta_4]
    +C\tilde{\rho}_0^{2+2\delta_4}L^3.
  \end{equation}
  \par For $\hat{\Theta}_{r,4}^\prime$, we first notice that $\phi^+(l)>0$ implies $r-l,s+l\in P_{F,-\alpha_2}$, then we can write $\hat{\Theta}_{r,4}^\prime$ by
  \begin{equation}\label{hatTheta_r4'}
    \hat{\Theta}_{r,4}^\prime=-\frac{1}{L^3}\sum_{\sigma,\nu,\tau}\sum_{p_1\in A_{F,\delta_4}}
    \mathcal{A}_{p_1,\sigma,\tau}\mathcal{B}_{p_1,\sigma,\nu}+h.c.,
  \end{equation}
  where
  \begin{align*}
  \mathcal{A}_{p_1,\sigma,\tau}=&\int_{\Lambda_L}e^{ip_1x_4}L^{\frac{3}{2}}
  a^*(H_{x_4,\sigma}[-\alpha_2])
  b_{x_4,\tau}^*(\eta_{\phi^+}^{\zeta^-})dx_4\\
  \mathcal{B}_{p_1,\sigma,\nu}=&\int_{\Lambda_L^2}e^{-ip_1x_1}v(x_1-x_3)
  a^*(L_{x_3,\nu}[\delta_4])a(H_{x_3,\nu}[\delta_4])a(H_{x_1,\sigma}[\delta_4])dx_1dx_3
  \end{align*}
  Using (\ref{ineq b 2 general}), we can then bound it similar to (\ref{typical example E_V_3,l,12}):
  \begin{equation}\label{hat Theta_r,4'bound}
    \pm\hat{\Theta}_{r,4}^\prime\leq
    C\tilde{\rho}_0^{\frac{5}{3}+\alpha_2}{\mathcal{N}}_{h}[-\alpha_2]
    +C\tilde{\rho}_0^{\delta_4}\mathcal{V}_{4,4h}.
  \end{equation}
  Then we reach (\ref{cub com V_4 2h 1}) by combining (\ref{cal [V_4,2h,B'] detailed})-(\ref{hat Theta_r,4'bound}).

  \par (\ref{cub control V_4}) follows from (\ref{cub control V_4,4h})-(\ref{cub com V_4 2h 1}), (\ref{V_4split}) and Lemma \ref{lemma cub control N_re}.
\end{proof}

\begin{lemma}\label{lemma cub misc}
For any $N\in\mathbb{N}_{\geq0}$, and for $\frac{1}{24}>\frac{1}{2}\delta_3>\frac{1}{4}\alpha_3>\alpha_2>2\alpha_4>0$, $\frac{1}{3}>\delta_4\geq\delta_1\geq\delta_2>\frac{1}{12}$, $\delta_4>2\alpha_3+2\alpha_4$, $\beta_1=\frac{1}{3}+\alpha_5$ and $2\alpha_3>\alpha_5>2\alpha_4$, we have
 \begin{equation}\label{cub Q_GN}
   e^{-B^\prime}Q_{\mathcal{G}_N}e^{B^\prime}={Q}_{\mathcal{G}_N}
   +\mathcal{E}_{{Q}_{\mathcal{G}_N}},
 \end{equation}
 where
 \begin{equation}\label{cub bound E QGN}
   \pm\mathcal{E}_{{Q}_{\mathcal{G}_N}}\leq
   C\tilde{\rho}_0^{\frac{4}{3}+\frac{\alpha_3}{2}}\mathcal{N}_{re}.
 \end{equation}
 We also have
 \begin{equation}\label{cub V_3,R}
   \begin{aligned}
   \pm e^{-B^\prime}\mathcal{V}_{3,R}e^{B^\prime}\leq&
   C\tilde{\rho}_0^{1+\delta_2}\mathcal{N}_{re}+C\tilde{\rho}_0^{1+\delta_4}
   \tilde{\mathcal{N}}_{re}+C\tilde{\rho}_0^{1-\delta_2}\mathcal{N}_h[\delta_4]\\
   &+C\tilde{\rho}_0^{\delta_2}\mathcal{V}_{4,4h}
   +C\big(\tilde{\rho}_0^{2+2\delta_4}+\tilde{\rho}_0^{\frac{8}{3}
   +\delta_4-\delta_2+\alpha_3}\big)L^3,
   \end{aligned}
 \end{equation}
 and
 \begin{equation}\label{cub V_4,r}
   \pm e^{-B^\prime}\mathcal{V}_{4,r}e^{B^\prime}\leq
   C\tilde{\rho}_0\mathcal{N}_h[\delta_4]
   +C\tilde{\rho}_0^{2+2\delta_4}L^3.
 \end{equation}
 Moreover
 \begin{equation}\label{cub control E}
  \begin{aligned}
  &\pm e^{-B^\prime}\mathcal{E}_{\mathcal{G}_N}e^{B^\prime}\\
  \leq&
  C\big(\tilde{\rho}_0^{1+\delta_2}+\tilde{\rho}_0^{\frac{4}{3}+\alpha_2-3\delta_3-\alpha_4}
  +\tilde{\rho}_0^{\frac{3}{2}+\frac{\alpha_2}{2}+4\alpha_3-\frac{3}{2}\alpha_5}\big)
  \mathcal{N}_{re}+C\big(\tilde{\rho}_0^{1+\delta_2}
  +\tilde{\rho}_0^{\frac{4}{3}-\alpha_3}\big)\tilde{\mathcal{N}}_{re}\\
  &+C\tilde{\rho}_0\mathcal{N}_i[\delta_1]+C\tilde{\rho}_0^{1-\delta_2}\mathcal{N}_h[\delta_1]
  +C\tilde{\rho}_0^{\frac{2}{3}-\alpha_4}\mathcal{N}_h[-\delta_3]\\
  &+C\tilde{\rho}_0^{\frac{5}{6}+\frac{\alpha_2}{2}-\frac{3}{2}\alpha_5}
  \mathcal{N}_h[-\alpha_2]
  +C\big(\tilde{\rho}_0+\tilde{\rho}_0^{1+\alpha_5
  -4\alpha_3-\alpha_4}\big)\mathcal{N}_h[-\beta_1]\\
  &+C\big(\tilde{\rho}_0^{\frac{1}{3}+2\alpha_3}+\tilde{\rho}_0^{\frac{1}{3}+
  \alpha_5-\alpha_4}\big)\mathcal{K}_s+C\tilde{\rho}_0^{\frac{1}{3}+
  \alpha_4}\mathcal{V}_{4,4h}\\
  &+C\big(\tilde{\rho}_0^{\frac{7}{3}+\alpha_4}+\tilde{\rho}_0^{2+4\alpha_3-7\alpha_2}
  +\tilde{\rho}_0^{\frac{5}{2}+\frac{\alpha_2}{2}+\alpha_3-\frac{3}{2}\alpha_5}
  +\tilde{\rho}_0^{\frac{8}{3}+\alpha_5-3\alpha_3-\alpha_4}
  +\tilde{\rho}_0^{\frac{7}{3}+2\delta_4}\big)L^3.
  \end{aligned}
\end{equation}
\end{lemma}
\begin{proof}
\par (\ref{cub Q_GN}) and (\ref{cub bound E QGN}) follow directly from (\ref{qua prop C and Q}) and (\ref{cub control diff N_retilde}). (\ref{cub V_4,r}) follows directly from (\ref{cub control N_h[delta]1}) and (\ref{cub control V_4r}). (\ref{cub control E}) follows from (\ref{qua prop E}) and Lemmas \ref{lemma cub control N_re}-\ref{lemma cub control V_4}.
  \par For (\ref{cub V_3,R}), by the definitions (\ref{cutoffgamma}), (\ref{qua V_3,l}), (\ref{divide V_3,l}) and (\ref{V_3,L}), we can write $\mathcal{V}_{3,R}=\mathcal{V}_{3,R1}
  +\mathcal{V}_{3,R2}$, where
    \begin{equation}\label{V_3,R 1 2}
    \begin{aligned}
    \mathcal{V}_{3,R1}=&\frac{1}{L^3}\sum_{k,p,q,\sigma,\nu}\hat{v}_k
(a_{p-k,\sigma}^*a_{q+k,\nu}^*a_{q,\nu}a_{p,\sigma}+h.c.)\\
&\quad\quad\quad\quad\times
(\chi_{p-k,q+k\in B_F}-\chi_{p-k,q+k\in P_{F,\delta_4}})
\chi_{q\in A_{F,\delta_4}}\chi_{p\in B_F}\\
\mathcal{V}_{3,R2}=&\frac{1}{L^3}\sum_{k,p,q,\sigma,\nu}\hat{v}_k\gamma^-(q)
(a_{p-k,\sigma}^*a_{q+k,\nu}^*a_{q,\nu}a_{p,\sigma}+h.c.)\\
&\quad\quad\quad\quad\times
(\chi_{p-k,q+k\in B_F}-\chi_{p-k,q+k\in P_{F,\delta_4}})
\chi_{q\in P_{F,\delta_4}}\chi_{p\in B_F}
    \end{aligned}
  \end{equation}
  Simply using Cauchy-Schwartz inequality, we have
  \begin{equation}\label{V_3,R1bound}
    \pm\mathcal{V}_{3,R1}\leq
    C\tilde{\rho}_0^{1+\delta_4}
   \tilde{\mathcal{N}}_{re}+C\tilde{\rho}_0\mathcal{N}_h[\delta_4]
   +C\tilde{\rho}_0^{2+2\delta_4}L^3,
  \end{equation}
  and
  \begin{equation}\label{V_3,R2bound}
  \begin{aligned}
    \pm\mathcal{V}_{3,R1}\leq&
    C\tilde{\rho}_0^{1+\delta_4}\mathcal{N}_{re}+C\tilde{\rho}_0^{1+\delta_4}
   \tilde{\mathcal{N}}_{re}+C\tilde{\rho}_0^{1-\delta_2}\mathcal{N}_h[\delta_4]\\
   &+C\tilde{\rho}_0^{\delta_2}\mathcal{V}_{4,4h}
   +C\tilde{\rho}_0^{2+2\delta_4}L^3,
   \end{aligned}
  \end{equation}
  which together yield (\ref{cub V_3,R}).
\end{proof}

\begin{lemma}\label{lemma cub V_21' Omega}
 For any $N\in\mathbb{N}_{\geq0}$, and for $\frac{1}{24}>\frac{1}{2}\delta_3>\frac{1}{4}\alpha_3>\alpha_2>2\alpha_4>0$, $\frac{1}{3}>\delta_4\geq\delta_1\geq\delta_2>\frac{1}{12}$, $\delta_4>2\alpha_3+2\alpha_4$, $\beta_1=\frac{1}{3}+\alpha_5$ and $2\alpha_3>\alpha_5>2\alpha_4$, we have
  \begin{equation}\label{cub V_21' Omega}
    e^{-B^\prime}(V^\prime_{21}+\Omega)e^{B^\prime}
    =V^\prime_{21}+\Omega+\mathcal{E}_{V^\prime_{21}}+\mathcal{E}_{\Omega},
  \end{equation}
  where
  \begin{equation}\label{cub bound E V_21' Omega}
    \begin{aligned}
    \pm(\mathcal{E}_{V^\prime_{21}}+\mathcal{E}_{\Omega})
    \leq&C\tilde{\rho}_0^{\frac{4}{3}+3\alpha_3+\alpha_4}\mathcal{N}_{re}
    +C\big(\tilde{\rho}_0^{\frac{4}{3}+\alpha_2-\alpha_4-3\alpha_3}+
    \tilde{\rho}_0^{\frac{5}{6}+\frac{\alpha_2}{2}
    +\frac{\delta_4}{2}-\frac{3}{2}\alpha_5}\big)\mathcal{N}_h[-\alpha_2]\\
    &+C\tilde{\rho}_0^{\frac{7}{3}+\alpha_4}L^3.
    \end{aligned}
  \end{equation}
\end{lemma}
\begin{proof}
  \par For $\mathcal{V}_{21}^\prime$, we have $[\mathcal{V}_{21}^\prime,B^\prime]=\sum_{j=1}^{5}\Pi_j$, where,
  \begin{align}
  &\Pi_1=\frac{1}{L^3}\sum_{\substack{\sigma,\nu\\k,p,q}}\sum_{l,r,s}2W_k\eta_l
    \delta_{p-k,r-l}\delta_{q+k,s+l}a^*_{p,\sigma}a^*_{q,\nu}a_{s,\nu}a_{r,\sigma}\nonumber\\
    &\quad\quad\times\zeta^-(k)\phi^+(l)\zeta^-(l)
    \chi_{p-k,q+k,r-l,s+l\notin B_F}\chi_{p,q,r\in B_F}
    \chi_{s\in A_{F,\delta_4}}+h.c.\nonumber\\
    &\Pi_2=-\frac{1}{L^3}\sum_{\substack{\sigma,\nu\\k,p,q}}\sum_{\varpi,l,r,s}2W_k\eta_l
    \delta_{p-k,r-l}a^*_{p,\sigma}a^*_{q,\nu}a^*_{s+l,\varpi}
    a_{q+k,\nu}a_{s,\varpi}a_{r,\sigma}\nonumber\\
    &\quad\quad\times\zeta^-(k)\phi^+(l)\zeta^-(l)
    \chi_{p-k,q+k,r-l,s+l\notin B_F}\chi_{p,q,r\in B_F}
    \chi_{s\in A_{F,\delta_4}}+h.c.\nonumber\\
    &\Pi_3=-\frac{1}{L^3}\sum_{\substack{\sigma,\nu\\k,p,q}}\sum_{\tau,l,r,s}2W_k\eta_l
    \delta_{q+k,s+l}a^*_{p,\sigma}a^*_{q,\nu}a^*_{r-l,\tau}
    a_{p-k,\sigma}a_{s,\nu}a_{r,\tau}\nonumber\\
    &\quad\quad\times\zeta^-(k)\phi^+(l)\zeta^-(l)
    \chi_{p-k,q+k,r-l,s+l\notin B_F}\chi_{p,q,r\in B_F}
    \chi_{s\in A_{F,\delta_4}}+h.c.\nonumber\\
    &\Pi_4=\frac{1}{L^3}\sum_{\substack{\sigma,\nu\\k,p,q}}\sum_{\tau,l,r,s}2W_k\eta_l
    \delta_{q+k,s}a^*_{r-l,\tau}a^*_{s+l,\nu}a^*_{p-k,\sigma}a_{r,\tau}
    a_{q,\nu}a_{p,\sigma}\nonumber\\
    &\quad\quad\times\zeta^-(k)\phi^+(l)\zeta^-(l)
    \chi_{p-k,q+k,r-l,s+l\notin B_F}\chi_{p,q,r\in B_F}
    \chi_{s\in A_{F,\delta_4}}
    +h.c.\nonumber\\
    &\Pi_5=\frac{1}{L^3}\sum_{\substack{\sigma,\nu\\k,p,q}}\sum_{\varpi,l,r,s}2W_k\eta_l
    \delta_{p,r}a^*_{r-l,\sigma}a^*_{s+l,\varpi}a^*_{q,\nu}
    a_{s,\varpi}a_{q+k,\nu}a_{p-k,\sigma}\nonumber\\
    &\quad\quad\times\zeta^-(k)\phi^+(l)\zeta^-(l)
    \chi_{p-k,q+k,r-l,s+l\notin B_F}\chi_{p,q,r\in B_F}
    \chi_{s\in A_{F,\delta_4}}+h.c.\label{[V_21',B']}
  \end{align}
  Similarly, $[\Omega,B^\prime]=\sum_{j=1}^{5}\tilde{\Pi}_j$, with the coefficients $W_k\zeta^-(k)$ in each $\Pi_j$ replaced by $\eta_kk(p-q)\phi^+(k)\zeta^-(k)$.
  \par For $\Pi_1$, notice that
  \begin{equation}\label{notice}
    \phi^+(l)\chi_{r-l\notin B_F}\chi_{r\in B_F}=\phi^+(l)\chi_{r\in B_F},\quad
    \phi^+(l)\chi_{s+l\notin B_F}\chi_{s\in A_{F,\delta_4}}=\phi^+(l)
    \chi_{s\in A_{F,\delta_4}}.
  \end{equation}
  Therefore,
  \begin{equation}\label{Pi_1}
    \begin{aligned}
    \Pi_1=2\sum_{\sigma,\nu}\int_{\Lambda_L^2}W^{\zeta^-}(x-y)
    L^{\frac{3}{2}}&\eta_{\phi^+}^{\zeta^-}(x-y)a^*(g_{x,\sigma})a^*(g_{y,\nu})\\
    &\times
    a(L_{y,\nu}[\delta_4])a(g_{x,\sigma})dxdy+h.c.
    \end{aligned}
  \end{equation}
  We can bound it trivially by
  \begin{equation}\label{Pi_1 bound}
    \pm\Pi_1\leq CL^{\frac{3}{2}}\Vert W^{\zeta^-}\Vert_2\Vert\eta_{\phi^+}^{\zeta^-}\Vert_2L^3
    \tilde{\rho}_0^{2+\frac{\delta_4}{2}}  \leq C\tilde{\rho}_0^{\frac{7}{3}+\frac{\delta_4}{2}-\alpha_3}L^3.
  \end{equation}
  Similarly, for $\tilde{\Pi}_1$, we have for some universal constant $C$:
   \begin{equation}\label{tildePi_1}
    \begin{aligned}
    \tilde{\Pi}_1=C\sum_{\sigma,\nu}\int_{\Lambda_L^2}L^{\frac{3}{2}}
    \nabla_x\eta_{\phi^+}^{\zeta^-}(x-y)
    L^{\frac{3}{2}}&\eta_{\phi^+}^{\zeta^-}(x-y)a^*(\nabla_x g_{x,\sigma})a^*(g_{y,\nu})\\
    &\times
    a(L_{y,\nu}[\delta_4])a(g_{x,\sigma})dxdy+h.c.
    \end{aligned}
  \end{equation}
  Integrating by parts yields
  \begin{equation}\label{tildePi_1 2nd}
    \begin{aligned}
    \tilde{\Pi}_1=-\tilde{\Pi}_1-
    C\sum_{\sigma,\nu}\int_{\Lambda_L^2}
    L^{3}&\big\vert\eta_{\phi^+}^{\zeta^-}(x-y)\big\vert^2
    a^*(\Delta_x g_{x,\sigma})a^*(g_{y,\nu})\\
    &\times
    a(L_{y,\nu}[\delta_4])a(g_{x,\sigma})dxdy+h.c.\\
    -C\sum_{\sigma,\nu}\int_{\Lambda_L^2}L^{3}&\big\vert\eta_{\phi^+}^{\zeta^-}(x-y)\big\vert^2
    a^*(\nabla_x g_{x,\sigma})a^*(g_{y,\nu})\\
    &\times
    a(L_{y,\nu}[\delta_4])a(\nabla_xg_{x,\sigma})dxdy+h.c.
    \end{aligned}
  \end{equation}
  Therefore
  \begin{equation}\label{Pi_1 tilde bound}
    \pm\tilde{\Pi}_1\leq CL^{3}\Vert\eta_{\phi^+}^{\zeta^-}\Vert_2^2L^3
    \tilde{\rho}_0^{\frac{8}{3}+\frac{\delta_4}{2}}  \leq C\tilde{\rho}_0^{\frac{7}{3}+\frac{\delta_4}{2}+\alpha_3}L^3.
  \end{equation}
  \par For $\Pi_2$, by observation (\ref{notice}), we have
  \begin{equation}\label{Pi_2}
    \begin{aligned}
    \Pi_2=-2\sum_{\sigma,\nu,\varpi}\int_{\Lambda_L^3}&W^{\zeta^-}(x_1-x_3)
    L^{\frac{3}{2}}\eta_{\phi^+}^{\zeta^-}(x_1-x_4)a^*(g_{x_1,\sigma})a^*(g_{x_3,\nu})
    a^*(h_{x_4,\varpi})\\
    &\times a(h_{x_3,\nu})a(L_{x_4,\varpi}[\delta_4])a(g_{x_2,\sigma})dx_1dx_3dx_4+h.c.
    \end{aligned}
  \end{equation}
  We can bound it by
  \begin{equation}\label{Pi_2 bound}
    \pm\Pi_2\leq C\tilde{\rho}_0^{2+\frac{\delta_4}{2}}L^{\frac{3}{2}}
    \Vert W^{\zeta^-}\Vert_1\Vert\eta_{\phi^+}^{\zeta^-}\Vert_1\mathcal{N}_{re}
    \leq C\tilde{\rho}_0^{\frac{4}{3}+\frac{\delta_4}{2}+2\alpha_3}\mathcal{N}_{re}.
  \end{equation}
   $\Pi_3$ can be bounded in the same way:
  \begin{equation}\label{Pi_3 bound}
    \pm\Pi_3\leq C\tilde{\rho}_0^{\frac{4}{3}+\frac{\delta_4}{2}+2\alpha_3}\mathcal{N}_{re}.
  \end{equation}
  The bounds of $\tilde{\Pi}_2$ and $\tilde{\Pi}_3$ is completely the same, with even smaller upper bounds, we thus omit further details.
  \par For $\Pi_4$, we again use the fact that $\phi^+(l)>0$ implies $r-l,s+l\in P_{F,\delta_4}$, then we have
  \begin{equation}\label{Pi_4}
    \Pi_4=-\frac{2}{L^3}\sum_{\sigma,\nu,\tau}\sum_{p_2\in A_{F,\delta_4}}
    \mathcal{A}_{p_2,\nu,\tau}\mathcal{B}_{p_2,\sigma,\nu}+h.c.
  \end{equation}
  where
  \begin{align*}
    \mathcal{A}_{p_2,\nu,\tau}&=\int_{\Lambda_L}L^{\frac{3}{2}}e^{-ip_2x_4}
    a^*(H_{x_4,\nu}[-\alpha_2])b^*_{x_4,\tau}(\eta_{\phi^+}^{\zeta^-})dx_4\\
    \mathcal{B}_{p_2,\sigma,\nu}&=\int_{\Lambda_L}e^{ip_2x_3}a(g_{x_3,\nu})b
    _{x_3,\sigma}^*(W^{\zeta^-})dx_3
  \end{align*}
  Similar to (\ref{typical example E_V_3,l,12}), using Lemma \ref{b^* bound by b} deduces
  \begin{equation}\label{Pi_4 bound}
    \pm\Pi_4\leq C\tilde{\rho}_0^{\frac{4}{3}+\alpha_4+3\alpha_3}\mathcal{N}_{re}
    +C\tilde{\rho}_0^{\frac{4}{3}+\alpha_2-\alpha_4-3\alpha_3}
    \mathcal{N}_{h}[-\alpha_2]
    +C\tilde{\rho}_0^{\frac{7}{3}+\alpha_4}L^3,
  \end{equation}
  where we have used the fact that, for $\psi\in\mathcal{H}^{\wedge N}$,
  \begin{equation}\label{factb^2}
  \begin{aligned}
    \int_{\Lambda_L}\Vert b_{x_3,\sigma}(W^{\zeta^-})\psi\Vert^2 dx_3&\leq
    C\tilde{\rho}_0\int_{\Lambda_L^3}
    \vert W^{\zeta^-}(y_1-x_3)\vert\vert W^{\zeta^-}(y_2-x_3)\vert\\
    &\quad\quad\times
    \Vert a(h_{y_1,\sigma})\psi\Vert\Vert a(h_{y_2,\sigma})\psi\Vert dy_1dy_2dx_3\\
    &\leq C\tilde{\rho}_0\langle\mathcal{N}_{re}\psi,\psi\rangle.
  \end{aligned}
  \end{equation}
  $\tilde{\Pi}_4$ can be bounded similarly, with an even smaller upper bound, we thus omit further details.

  \par For $\Pi_5$, we make use of the fact that $\phi^+(l)>0$ implies $r-l,s+l\in P_{F,\delta_4}$. Moreover, we notice that $\zeta^-(k)>0$ implies $\vert p-k\vert,\vert q+k\vert\leq
  3\tilde{\mu}^{\frac{1}{2}}\tilde{\rho}_0^{-\beta_1}$. Thus we temporarily define
  \begin{equation*}
  \mathbf{h}_{x,\sigma}(z)=\sum_{k_F<\vert k\vert\leq 3\tilde{\mu}^{\frac{1}{2}}\tilde{\rho}
  _0^{-\beta_1}}\frac{e^{ikx}}{L^{\frac{3}{2}}}f_{k,\sigma}(z),
\end{equation*}
then $\Pi_5$ can be rewritten by
\begin{equation}\label{Pi_5}
  \Pi_5=-\frac{2}{L^3}\sum_{\sigma,\nu,\varpi}\sum_{q_1\in B_F}
  \mathcal{A}_{q_1,\sigma,\varpi}\mathcal{B}_{q_1,\sigma,\nu}+h.c.
\end{equation}
where (here the operator $b^*_{x_2,\varpi}(\eta_{\phi^+}^{\zeta^-})$ is defined in the form of (\ref{define btilde}), different from the one in the estimate of $\Pi_4$),
\begin{align*}
    \mathcal{A}_{q_1,\sigma,\varpi}&=\int_{\Lambda_L}L^{\frac{3}{2}}e^{-iq_1x_2}
    a^*(H_{x_2,\sigma}[-\alpha_2])b^*_{x_2,\varpi}(\eta_{\phi^+}^{\zeta^-})dx_4\\
    \mathcal{B}_{q_1,\sigma,\nu}&=\int_{\Lambda_L^2}e^{iq_1x_1}
    W^{\zeta^-}(x_1-x_3)a^*(g_{x_3,\nu})a(\mathbf{h}_{x_3,\nu})a(\mathbf{h}_{x_1,\sigma})dx_1dx_3
  \end{align*}
  Using $\mathbf{h}_{x_1,\sigma}=L_{x_1,\sigma}[-\alpha_2]+(\mathbf{h}_{x_1,\sigma}-
  L_{x_1,\sigma}[-\alpha_2])$, we further rewrite $\Pi_5=\Pi_{51}+\Pi_{52}$. By Lemma \ref{b^* bound by b general}, we can bound them analogously to (\ref{typical example E_V_3,l,12}):
  \begin{align*}
    \pm\Pi_{51}&\leq C\tilde{\rho}_0^{\frac{11}{6}+
    \frac{\delta_4}{2}+3\alpha_5-\frac{5}{2}\alpha_2}\mathcal{N}_{re}
    +C\tilde{\rho}_0^{\frac{5}{6}+\frac{\delta_4}{2}+\frac{\alpha_2}{2}-\frac{3}{2}\alpha_5}
    \mathcal{N}_{h}[-\alpha_2],\\
    \pm\Pi_{52}&\leq C\tilde{\rho}_0^{\frac{5}{6}+\frac{\delta_4}{2}+\frac{\alpha_2}{2}-\frac{3}{2}\alpha_5}
    \mathcal{N}_{h}[-\alpha_2].
  \end{align*}
  Thus
  \begin{equation}\label{Pi_5 bound}
    \pm\Pi_5\leq C\tilde{\rho}_0^{\frac{11}{6}+
    \frac{\delta_4}{2}+3\alpha_5-\frac{5}{2}\alpha_2}\mathcal{N}_{re}
    +C\tilde{\rho}_0^{\frac{5}{6}+\frac{\delta_4}{2}+\frac{\alpha_2}{2}-\frac{3}{2}\alpha_5}
    \mathcal{N}_{h}[-\alpha_2].
  \end{equation}
  $\tilde{\Pi}_5$ can be bounded similarly, with an even smaller upper bound, we thus omit further details.
  \par With all the estimates above, and Newton-Leibniz formula, and the help of Lemma \ref{lemma cub control N_re}, we reach (\ref{cub V_21' Omega}) and (\ref{cub bound E V_21' Omega}).
\end{proof}

\begin{lemma}\label{lemma cub core}
For any $N\in\mathbb{N}_{\geq0}$, and for $\frac{1}{24}>\frac{1}{2}\delta_3>\frac{1}{4}\alpha_3>\alpha_2>2\alpha_4>0$, $\frac{1}{3}>\delta_4\geq\delta_1\geq\delta_2>\frac{1}{12}$, $\delta_4>2\alpha_3+2\alpha_4$, $\beta_1=\frac{1}{3}+\alpha_5$ and $2\alpha_3>\alpha_5>2\alpha_4$, we have
  \begin{equation}\label{cub core}
    \int_{0}^{1}\int_{t}^{1}e^{-sB^\prime}[\mathcal{V}_{3,L},B^\prime]e^{sB^\prime}dsdt
    =\frac{1}{L^3}\sum_{k}\hat{v}_k\eta_k(\mathbf{q}-1)\bar{N}_0\mathcal{N}_{re}
    +\mathcal{E}_{[\mathcal{V}_{3,L},B^\prime]},
  \end{equation}
  where
  \begin{equation}\label{cub bound Ecore}
    \begin{aligned}
    \pm\mathcal{E}_{[\mathcal{V}_{3,L},B^\prime]}
    \leq& C\tilde{\rho}_0^{1+\delta_4}\mathcal{N}_{re}
    +C\tilde{\rho}_0^{1+\delta_4}\tilde{\mathcal{N}}_{re}
    +C\tilde{\rho}_0\mathcal{N}_h[\delta_4]\\
    &+C\tilde{\rho}_0^{\frac{1}{3}}\mathcal{V}_{4,4h}
    +C\tilde{\rho}_0^{\frac{7}{3}+\delta_4}L^3.
    \end{aligned}
  \end{equation}
\end{lemma}
\begin{proof}
  \par We divide $[\mathcal{V}_{3,L},B^\prime]=\sum_{j=1}^{13}\Upsilon_{j}$, that is, $13$ parts. To avoid occupying too much space, we only write out each part when it is that part's turn. For $\Upsilon_1$:
  \begin{equation}\label{Upsilon 1}
    \begin{aligned}
    \Upsilon_1&=\frac{1}{L^3}\sum_{\substack{\sigma,\nu\\k,p,q}}\sum_{l,r,s}
    \hat{v}_k\eta_l\delta_{p-k,r-l}\delta_{q+k,s+l}
    a^*_{p,\sigma}a^*_{q,\nu}a_{s,\nu}a_{r,\sigma}\\
    &\quad\quad\times
    \phi^+(l)\zeta^-(l)\chi_{p-k,q+k\in P_{F,\delta_4}}\chi_{r-l,s+l\notin B_F}
    \chi_{p,r\in B_F}\chi_{q,s\in A_{F,\delta_4}}+h.c.
    \end{aligned}
  \end{equation}
  The calculation to $\Upsilon_1$ follows exactly from the process in analyzing the corresponding $\Upsilon_1$ in \cite[Lemma 8.7]{WJH}. We just need to use the corresponding estimates of $\eta$ and $v$ in the thermodynamic limit, collecting in Section \ref{scattering eqn sec}. Using $1=\chi_{k=l}+\chi_{k\neq l}$, we can rewrite $\Upsilon_1=\Upsilon_{1,1}+\Upsilon_{1,2}$. For $\Upsilon_{1,1}$, we have
  \begin{align}
    \Upsilon_{1,1}=&\frac{2}{L^3}\sum_{\substack{\sigma,\nu\\k,p,q}}
    \hat{v}_k\eta_k\phi^+(k)\zeta^-(k)a_{p,\sigma}^*a_{p,\sigma}
    a_{q,\nu}^*a_{q,\nu}\chi_{p\in B_F}\chi_{q\in A_{F,\delta_4}}\chi_{p-k,q+k\in P_{F,\delta_4}}\nonumber\\
    =&\frac{2}{L^3}\sum_{\substack{\sigma,\nu\\k,p,q}}
    \hat{v}_k\eta_k\phi^+(k)\zeta^-(k)a_{p,\sigma}^*a_{p,\sigma}
    a_{q,\nu}^*a_{q,\nu}\chi_{p\in B_F}\chi_{q\in A_{F,\delta_4}}
  \end{align}
  By (\ref{est of eta_0}) and (\ref{eta cutoff sum}), we have
  \begin{equation}\label{etav cutoff est}
    \Big\vert\sum_{k}\hat{v}_k\eta_k\phi^-(k)\Big\vert\leq C\tilde{\rho}_0^
    {\frac{1}{3}+2\alpha_3-3\alpha_2},\quad
    \Big\vert\sum_{k}\hat{v}_k\eta_k\zeta^+(k)\Big\vert\leq C\tilde{\rho}_0^
    {\frac{2}{3}-2\alpha_3+\alpha_5}.
  \end{equation}
  Also notice the fact that $L^{-3}\tilde{N}_{re}\mathcal{N}_h[\delta_4]\leq
  C\tilde{\rho}_0^{1+\delta_4}\tilde{\mathcal{N}}_{re}$, we calculate
  \begin{equation}\label{Upsilon 1bound}
    \begin{aligned}
    \Upsilon_{1,1}=\frac{2}{L^3}\sum_{k}\hat{v}_k\eta_k\mathbf{q}\bar{N}_0\mathcal{N}_{re}
    +\mathcal{E}_{\Upsilon_{1,1}},
    \end{aligned}
  \end{equation}
  where
  \begin{equation}\label{E_Upsilon 1}
    \pm\mathcal{E}_{\Upsilon_{1,1}}\leq C\tilde{\rho}_0^{1+\delta_4}\tilde{\mathcal{N}}_{re}+
    C\tilde{\rho}_0\mathcal{N}_h[\delta_4]+
    C\big(\tilde{\rho}_0^{\frac{7}{3}+2\alpha_3+\delta_4-3\alpha_2}
    +\tilde{\rho}_0^{\frac{8}{3}+\alpha_5+\delta_4-2\alpha_3}\big)L^3.
  \end{equation}
For $\Upsilon_{1,2}$, we have
\begin{align}
    \Upsilon_{1,2}&=-\frac{1}{L^3}\sum_{\substack{\sigma,\nu\\k,p,q}}\sum_{l,r,s}
    \hat{v}_k\eta_l\delta_{p-k,r-l}\delta_{q+k,s+l}
    a_{r,\sigma}a^*_{q,\nu}a_{s,\nu}a^*_{p,\sigma}\chi_{k\neq l}\nonumber\\
    &\quad\quad\times
    \phi^+(l)\zeta^-(l)\chi_{p-k,q+k\in P_{F,\delta_4}}\chi_{r-l,s+l\notin B_F}
    \chi_{p,r\in B_F}\chi_{q,s\in A_{F,\delta_4}}+h.c.\nonumber\\
    &=-\frac{1}{L^3}\sum_{\substack{\sigma,\nu\\k,p,q}}\sum_{l,r,s}
    \hat{v}_k\eta_l\delta_{p-k,r-l}\delta_{q+k,s+l}
    a_{r,\sigma}a^*_{q,\nu}a_{s,\nu}a^*_{p,\sigma}(1-\chi_{k=l})\nonumber\\
    &\quad\quad\times
    \phi^+(l)\zeta^-(l)
    \chi_{p,r\in B_F}\chi_{q,s\in A_{F,\delta_4}}+h.c\label{Upsilon 12}
    \end{align}
    and the following trivial bound
    \begin{equation}\label{Upsilon 12 bound}
      \pm\Upsilon_{1,2}\leq C\tilde{\rho}_0^{1+\delta_4}\tilde{\mathcal{N}}_{re}.
    \end{equation}
  \par For $\Upsilon_2$:
  \begin{equation}\label{Upsilon 2}
    \begin{aligned}
    \Upsilon_2&=-\frac{1}{L^3}\sum_{\substack{\sigma,\nu\\k,p,q}}\sum_{l,r,s}
  \hat{v}_k\eta_l\delta_{p-k,s+l}\delta_{q+k,r-l}
    a^*_{p,\sigma}a^*_{q,\nu}a_{s,\sigma}a_{r,\nu}\\
    &\quad\quad\times
    \phi^+(l)\zeta^-(l)\chi_{p-k,q+k\in P_{F,\delta_4}}\chi_{r-l,s+l\notin B_F}
    \chi_{p,r\in B_F}\chi_{q,s\in A_{F,\delta_4}}+h.c.
    \end{aligned}
  \end{equation}
  Our computation here follows exactly from the corresponding $\Upsilon_2$ in \cite[Lemma 8.7]{WJH}, and is analogous to the estimate of $\Upsilon_1$ above. Notice in the thermodynamic limit, we also have an estimate corresponding to \cite[(8.85)]{WJH}, by (\ref{eta sum}):
  \begin{equation}\label{sumeta diff v}
    \sum_{k}\vert \hat{v}_{k+p-q}-\hat{v}_k\vert\vert\eta_k\vert\leq C\tilde{\rho}_0^{\frac{1}{3}},
  \end{equation}
  for $p\in B_F$ and $q\in A_{F,\delta_4}$. Thus we write out the result directly:
  \begin{equation}\label{Upsilon 2bound}
    \begin{aligned}
    \Upsilon_2=-\frac{2}{L^3}\sum_{k}\hat{v}_k\eta_k\bar{N}_0\mathcal{N}_{re}
    +\mathcal{E}_{\Upsilon_2},
    \end{aligned}
  \end{equation}
  where
  \begin{equation}\label{E_Upsilon 2}
    \pm\mathcal{E}_{\Upsilon_2}\leq C\tilde{\rho}_0^{1+\delta_4}\tilde{\mathcal{N}}_{re}+
    C\tilde{\rho}_0\mathcal{N}_h[\delta_4]+
    C\tilde{\rho}_0^{\frac{7}{3}+\delta_4}L^3.
  \end{equation}

  \par For $\Upsilon_3$ and $\Upsilon_4$:
  \begin{equation}\label{Upsilon 3 4}
    \begin{aligned}
    &\Upsilon_3=\frac{1}{L^3}\sum_{\substack{\sigma,\nu\\k,p,q}}\sum_{\varpi,l,r,s}
  \hat{v}_k\eta_l\delta_{q,r-l}
  a^*_{p-k,\sigma}a^*_{q+k,\nu}a^*_{s+l,\varpi}a_{p,\sigma}a_{s,\varpi}a_{r,\nu}\\
  &\quad\quad\times
  \phi^+(l)\zeta^-(l)\chi_{p-k,q+k\in P_{F,\delta_4}}\chi_{r-l,s+l\notin B_F}
    \chi_{p,r\in B_F}\chi_{q,s\in A_{F,\delta_4}}+h.c.\\
  &\Upsilon_4=-\frac{1}{L^3}\sum_{\substack{\sigma,\nu\\k,p,q}}\sum_{\tau,l,r,s}
  \hat{v}_k\eta_l\delta_{q,s+l}
  a^*_{p-k,\sigma}a^*_{q+k,\nu}a^*_{r-l,\tau}a_{p,\sigma}a_{s,\nu}a_{r,\tau}\\
  &\quad\quad\times
 \phi^+(l)\zeta^-(l)\chi_{p-k,q+k\in P_{F,\delta_4}}\chi_{r-l,s+l\notin B_F}
    \chi_{p,r\in B_F}\chi_{q,s\in A_{F,\delta_4}}+h.c.
    \end{aligned}
  \end{equation}
  Since $r-l\in A_{F,\delta_4}$ or $s+l\in A_{F,\delta_4}$ implies $\phi^+(l)=0$, we have $\Upsilon_3=\Upsilon_4=0$.

  \par For $\Upsilon_5$ and $\Upsilon_6$:
  \begin{equation}\label{Upsilon 5 6}
    \begin{aligned}
    &\Upsilon_5=\frac{1}{L^3}\sum_{\substack{\sigma,\nu\\k,p,q}}\sum_{\tau,l,r,s}
  \hat{v}_k\eta_l\delta_{q+k,s}
  a^*_{r-l,\tau}a^*_{s+l,\nu}a^*_{p-k,\sigma}a_{r,\tau}a_{q,\nu}a_{p,\sigma}\\
  &\quad\quad\times
  \phi^+(l)\zeta^-(l)\chi_{p-k,q+k\in P_{F,\delta_4}}\chi_{r-l,s+l\notin B_F}
    \chi_{p,r\in B_F}\chi_{q,s\in A_{F,\delta_4}}+h.c.\\
  &\Upsilon_6=-\frac{1}{L^3}\sum_{\substack{\sigma,\nu\\k,p,q}}\sum_{\tau,l,r,s}
  \hat{v}_k\eta_l\delta_{p-k,s}
  a^*_{r-l,\tau}a^*_{s+l,\sigma}a^*_{q+k,\nu}a_{r,\tau}a_{q,\nu}a_{p,\sigma}\\
  &\quad\quad\times
  \phi^+(l)\zeta^-(l)\chi_{p-k,q+k\in P_{F,\delta_4}}\chi_{r-l,s+l\notin B_F}
    \chi_{p,r\in B_F}\chi_{q,s\in A_{F,\delta_4}}+h.c.
    \end{aligned}
  \end{equation}
  Since ${p-k,q+k\in P_{F,\delta_4}}$ while ${s\in A_{F,\delta_4}}$, we have
  $\Upsilon_5=\Upsilon_6=0$.

  \par For $\Upsilon_7\text{\textendash}\Upsilon_{10}$, we have
  \begin{align}
  &\Upsilon_7=\frac{1}{L^3}\sum_{\substack{\sigma,\nu\\k,p,q}}\sum_{\varpi,l,r,s}
  \hat{v}_k\eta_l\delta_{q+k,r-l}
  a^*_{p,\sigma}a^*_{q,\nu}a^*_{s+l,\varpi}a_{p-k,\sigma}a_{s,\varpi}a_{r,\nu}\nonumber\\
  &\quad\quad\times
  \phi^+(l)\zeta^-(l)\chi_{p-k,q+k\in P_{F,\delta_4}}\chi_{r-l,s+l\notin B_F}
    \chi_{p,r\in B_F}\chi_{q,s\in A_{F,\delta_4}}+h.c.\nonumber\\
  &\Upsilon_8=-\frac{1}{L^3}\sum_{\substack{\sigma,\nu\\k,p,q}}\sum_{\varpi,l,r,s}
  \hat{v}_k\eta_l\delta_{p-k,r-l}
  a^*_{p,\sigma}a^*_{q,\nu}a^*_{s+l,\varpi}a_{q+k,\nu}a_{s,\varpi}a_{r,\sigma}\nonumber\\
  &\quad\quad\times
  \phi^+(l)\zeta^-(l)\chi_{p-k,q+k\in P_{F,\delta_4}}\chi_{r-l,s+l\notin B_F}
    \chi_{p,r\in B_F}\chi_{q,s\in A_{F,\delta_4}}+h.c.\nonumber\\
  &\Upsilon_9=\frac{1}{L^3}\sum_{\substack{\sigma,\nu\\k,p,q}}\sum_{\tau,l,r,s}
  \hat{v}_k\eta_l\delta_{p-k,s+l}
  a^*_{p,\sigma}a^*_{q,\nu}a^*_{r-l,\tau}a_{q+k,\nu}a_{s,\sigma}a_{r,\tau}\nonumber\\
  &\quad\quad\times
  \phi^+(l)\zeta^-(l)\chi_{p-k,q+k\in P_{F,\delta_4}}\chi_{r-l,s+l\notin B_F}
    \chi_{p,r\in B_F}\chi_{q,s\in A_{F,\delta_4}}+h.c.\nonumber\\
  &\Upsilon_{10}=-\frac{1}{L^3}\sum_{\substack{\sigma,\nu\\k,p,q}}\sum_{\tau,l,r,s}
  \hat{v}_k\eta_l\delta_{q+k,s+l}
  a^*_{p,\sigma}a^*_{q,\nu}a^*_{r-l,\tau}a_{p-k,\sigma}a_{s,\nu}a_{r,\tau}\nonumber\\
  &\quad\quad\times
  \phi^+(l)\zeta^-(l)\chi_{p-k,q+k\in P_{F,\delta_4}}\chi_{r-l,s+l\notin B_F}
    \chi_{p,r\in B_F}\chi_{q,s\in A_{F,\delta_4}}+h.c.\label{Upsilon 7-10}
  \end{align}
  Notice that
   \begin{equation}\label{notice2}
    \phi^+(l)\chi_{r-l\notin A_{F,\delta_4}}\chi_{r\in B_F}=\phi^+(l)\chi_{r\in B_F},\quad
    \phi^+(l)\chi_{s+l\notin A_{F,\delta_4}}\chi_{s\in A_{F,\delta_4}}=\phi^+(l)
    \chi_{s\in A_{F,\delta_4}},
  \end{equation}
  we can write $\Upsilon_7$ by
  \begin{equation}\label{Upsilon7}
    \begin{aligned}
    \Upsilon_7=\sum_{\sigma,\nu,\varpi}\int_{\Lambda_L^3}v(x_1-x_3)&L^{\frac{3}{2}}
    \eta_{\phi^+}^{\zeta^-}(x_3-x_4)a^*(g_{x_1,\sigma})a^*(L_{x_3,\nu}[\delta_4])
    a^*(h_{x_4,\varpi})\\
    &\times a(H_{x_1,\sigma}[\delta_4])a(L_{x_4,\varpi}[\delta_4])a(g_{x_3,\nu})
    dx_1dx_3dx_4+h.c.
    \end{aligned}
  \end{equation}
  and we can bound it by
   \begin{equation}\label{Upsilon 7 bound}
    \pm\Upsilon_7\leq C\tilde{\rho}_0^{\frac{4}{3}+\delta_4+2\alpha_3}\mathcal{N}_{re}.
  \end{equation}
  For $j=7,8,9,10$, their bounds are all similar:
  \begin{equation}\label{Upsilon 7-10 bound}
    \pm\Upsilon_j\leq C\tilde{\rho}_0^{\frac{4}{3}+\delta_4+2\alpha_3}\mathcal{N}_{re}.
  \end{equation}
  \par For $\Upsilon_{11}$:
  \begin{equation}\label{Upsilon11}
    \begin{aligned}
    \Upsilon_{11}&=-\frac{1}{L^3}\sum_{\substack{\sigma,\nu\\k,p,q}}\sum_{\varpi,l,r,s}
  \hat{v}_k\eta_l\delta_{p,r}
  a^*_{r-l,\sigma}a^*_{s+l,\varpi}a_{s,\varpi}a^*_{q,\nu}a_{q+k,\nu}a_{p-k,\sigma}\\
  &\quad\quad\times
  \phi^+(l)\zeta^-(l)\chi_{p-k,q+k\in P_{F,\delta_4}}\chi_{r-l,s+l\notin B_F}
    \chi_{p,r\in B_F}\chi_{q,s\in A_{F,\delta_4}}+h.c.
    \end{aligned}
  \end{equation}
   Recall that $\phi^+(l)>0$ implies $r-l,s+l\in P_{F,-\alpha_2}$, we have
   \begin{equation}\label{Upsilon 11}
     \Upsilon_{11}=-\frac{1}{L^3}\sum_{\sigma,\nu,\varpi}\sum_{q_1\in B_F}
     \mathcal{A}_{q_1,\sigma,\varpi}\mathcal{B}_{q_1,\sigma,\nu}+h.c.
   \end{equation}
   where
   \begin{align*}
     \mathcal{A}_{q_1,\sigma,\varpi}=&\int_{\Lambda_L}L^{\frac{3}{2}}
     e^{-iq_1x_2}a^*(H_{x_2,\sigma}[-\alpha_2])b^*_{x_2,\varpi}(\eta_{\phi^+}^
     {\zeta^-})dx_2\\
     \mathcal{B}_{q_1,\sigma,\nu}=&\int_{\Lambda_L^2}e^{iq_1x_1}v(x_1-x_3)
     a^*(L_{x_3,\nu}[\delta_4])a(H_{x_3,\nu}[\delta_4])a(H_{x_1,\sigma}[\delta_4])
   \end{align*}
   Similar to (\ref{typical example E_V_3,l,12}), we have
  \begin{equation}\label{Upsilon 11bound}
    \pm\Upsilon_{11}\leq C\tilde{\rho}_0^{\frac{1}{3}}\mathcal{V}_{4,4h}+
    C\tilde{\rho}_0^{\frac{4}{3}+\alpha_2+2\delta_4}\mathcal{N}_h[-\alpha_2].
  \end{equation}
  \par For $\Upsilon_{12}$:
  \begin{equation}\label{Upsilon12}
    \begin{aligned}
    \Upsilon_{12}&=-\frac{1}{L^3}\sum_{\substack{\sigma,\nu\\k,p,q}}\sum_{\tau,l,r,s}
  \hat{v}_k\eta_l\delta_{q,s}
  a^*_{r-l,\tau}a^*_{s+l,\nu}a_{r,\tau}a^*_{p,\sigma}a_{q+k,\nu}a_{p-k,\sigma}\\
  &\quad\quad\times
  \phi^+(l)\zeta^-(l)\chi_{p-k,q+k\in P_{F,\delta_4}}\chi_{r-l,s+l\notin B_F}
    \chi_{p,r\in B_F}\chi_{q,s\in A_{F,\delta_4}}+h.c.
    \end{aligned}
  \end{equation}
  The idea is analogous to $\Upsilon_{11}$, we thus have
  \begin{equation}\label{Upsilon 12bound}
    \pm\Upsilon_{12}\leq C\tilde{\rho}_0^{\frac{1}{3}}\mathcal{V}_{4,4h}+
    C\tilde{\rho}_0^{\frac{4}{3}+\alpha_2}\mathcal{N}_h[-\alpha_2].
  \end{equation}
  \par For $\Upsilon_{13}$:
  \begin{equation}\label{Upsilon13}
    \begin{aligned}
    \Upsilon_{13}&=\frac{1}{L^3}\sum_{\substack{\sigma,\nu\\k,p,q}}\sum_{l,r,s}
  \hat{v}_k\eta_l\delta_{p,r}\delta_{q,s}
  a^*_{r-l,\sigma}a^*_{s+l,\nu}a_{q+k,\nu}a_{p-k,\sigma}\\
  &\quad\quad\times
  \phi^+(l)\zeta^-(l)\chi_{p-k,q+k\in P_{F,\delta_4}}\chi_{r-l,s+l\notin B_F}
    \chi_{p,r\in B_F}\chi_{q,s\in A_{F,\delta_4}}+h.c.
    \end{aligned}
  \end{equation}
  We combine the proof of (\ref{ineq b 2}) with the bound of $\Upsilon_{11}$ in (\ref{Upsilon 11}), i.e. we use
  \begin{equation*}
     1=\frac{1}{\vert s+l\vert^2}\big(\vert s\vert^2+2s\cdot l+\vert l\vert^2\big),
  \end{equation*}
  and therefore we have
  \begin{equation}\label{Upsilon 13bound}
    \pm\Upsilon_{13}\leq C\tilde{\rho}_0^{\frac{1}{3}}\mathcal{V}_{4,4h}+
    C\tilde{\rho}_0^{\frac{4}{3}+\alpha_2+2\delta_4}\mathcal{N}_h[-\alpha_2].
  \end{equation}
  \par Combining all the estimates above, together with Lemma \ref{lemma cub control N_re} and \ref{lemma cub control V_4}, we reach the final result.
\end{proof}

\begin{lemma}\label{lemma cub Gamma'}
 For any $N\in\mathbb{N}_{\geq0}$, and for $\frac{1}{24}>\frac{1}{2}\delta_3>\frac{1}{4}\alpha_3>\alpha_2>2\alpha_4>0$, $\frac{1}{3}>\delta_4\geq\delta_1\geq\delta_2>\frac{1}{12}$, $\delta_4>2\alpha_3+2\alpha_4$, $\beta_1=\frac{1}{3}+\alpha_5$ and $2\alpha_3>\alpha_5>2\alpha_4$, we have
  \begin{equation}\label{cub Gamma'}
    \begin{aligned}
    &\pm\int_{0}^{1}e^{-tB\prime}\Gamma^\prime e^{tB^\prime}dt\\
    \leq& C\tilde{\rho}_0^{1+\delta_2}\mathcal{N}_{re}
    +C\tilde{\rho}_0^{1+\delta_1}\tilde{\mathcal{N}}_{re}
    +C\tilde{\rho}_0^{1-\delta_2-3\alpha_2}\mathcal{N}_h[\delta_4]
    +C\tilde{\rho}_0^{1-3\alpha_2}\mathcal{N}_i[\delta_4]\\
    &+C\big(\tilde{\rho}_0^{\frac{5}{6}+\frac{\alpha_2}{2}-\frac{3}{2}\alpha_5
    }+\tilde{\rho}_0^{\frac{2}{3}+\delta_4-\alpha_3-\alpha_4}\big)\mathcal{N}_h[-\alpha_2]
    +C\tilde{\rho}_0^{1-4\alpha_3+\delta_4+\alpha_5-\alpha_4}\mathcal{N}_h[-\beta_1]\\
    &+C\tilde{\rho}_0^{\delta_4-\alpha_3-\alpha_4}\mathcal{K}_h[-\alpha_2]
    +C\tilde{\rho}_0^{\delta_4}\mathcal{V}_{4,4h}+C\big(\tilde{\rho}_0^{\frac{7}{3}+\alpha_4}
    +\tilde{\rho}_0^{2+2\delta_4}
    +\tilde{\rho}_0^{2+\delta_4+4\alpha_3-7\alpha_2}\big)L^3.
    \end{aligned}
  \end{equation}
\end{lemma}
\begin{proof}
  \par By (\ref{Wdiscrete asymptotic energy pde on the torus}), (\ref{V_3,L}), (\ref{define Gamma'}), (\ref{cub com [K,B']}), (\ref{cub com [V_4,4h,B']}) and (\ref{V_3,R 1 2}), we have
  \begin{equation}\label{Gamma'}
    \Gamma^\prime=[\mathcal{V}_{4,2h}^{(1)},B^\prime]+\Omega_r^\prime
    +\Theta^\prime_r+\mathcal{V}_{3,R1}+\mathcal{V}_{3,l}^\prime
    +\mathcal{V}_{3,h}^\prime+\mathcal{V}_{3,m}^\prime,
  \end{equation}
  where
  \begin{align}
    \mathcal{V}_{3,l}^\prime=&\frac{2}{L^3}\sum_{k,p,q,\sigma,\nu}{W}_k\phi^-(k)
(a_{p-k,\sigma}^*a_{q+k,\nu}^*a_{q,\nu}a_{p,\sigma}+h.c.)\nonumber\\
&\quad\quad\times
\chi_{p-k,q+k\notin B_{F}}\chi_{q\in A_{F,\delta_4}}\chi_{p\in B_F}\nonumber\\
\mathcal{V}_{3,m}^\prime=&\frac{2}{L^3}\sum_{k,p,q,\sigma,\nu}{W}_k\phi^+(k)\zeta^-(k)
(a_{p-k,\sigma}^*a_{q+k,\nu}^*a_{q,\nu}a_{p,\sigma}+h.c.)\nonumber\\
&\quad\quad\times
\chi_{p-k,q+k\notin B_{F}}\chi_{q\in A_{F,\delta_4}}\chi_{p\in B_F}\nonumber\\
\mathcal{V}_{3,h}^\prime=&\frac{2}{L^3}\sum_{k,p,q,\sigma,\nu}{W}_k\zeta^+(k)
(a_{p-k,\sigma}^*a_{q+k,\nu}^*a_{q,\nu}a_{p,\sigma}+h.c.)\nonumber\\
&\quad\quad\times
\chi_{p-k,q+k\notin B_{F}}\chi_{q\in A_{F,\delta_4}}\chi_{p\in B_F}\label{V_3'}
  \end{align}
  Notice that the bounds to $[\mathcal{V}_{4,2h}^{(1)},B^\prime]$ can be found in (\ref{cub com V_4 2h 1}), to $\Omega_r^\prime$ can be found in (\ref{cub Omega_r'}), to $\Theta_r^\prime$ can be found in (\ref{cub Theta_r'}), to $\mathcal{V}_{3,R1}$ can be found in (\ref{V_3,R1bound}).
  \par For $\mathcal{V}_{3,l}^\prime$, notice that $\phi^-(k)>0$ implies $\vert p-k\vert,\vert q+k\vert\leq 3\tilde{\mu}^{\frac{1}{2}}\tilde{\rho}_0^{-\alpha_2}$. If we temporarily define
  \begin{equation*}
  \mathbf{h}_{x,\sigma}(z)=\sum_{k_F<\vert k\vert\leq 3\tilde{\mu}^{\frac{1}{2}}\tilde{\rho}
  _0^{-\alpha_2}}\frac{e^{ikx}}{L^{\frac{3}{2}}}f_{k,\sigma}(z),
\end{equation*}
then
\begin{equation}\label{V_3 1'}
\begin{aligned}
  \mathcal{V}_{3,l}^\prime=2\sum_{\sigma,\nu}\int_{\Lambda_L^2}&W^{\phi^-}(x-y)
  a^*(\mathbf{h}_{x,\sigma})a^*(\mathbf{h}_{y,\nu})\\
  &\times a(L_{y,\nu}[\delta_4])a(g_{x,\sigma})dxdy+h.c.
\end{aligned}
\end{equation}
Using
\begin{align*}
  \mathbf{h}_{x,\sigma}g_{x,\sigma}=&
  H_{x,\sigma}[\delta_4]g_{x,\sigma}+(\mathbf{h}_{x,\sigma}-
  H_{x,\sigma}[\delta_4])S_{x,\sigma}[\delta_4]\\
  &+(\mathbf{h}_{x,\sigma}-
  H_{x,\sigma}[\delta_4])I_{x,\sigma}[\delta_4],
\end{align*}
we rewrite $\mathcal{V}_{3,l}^\prime=\sum_{j=1}^{3}\mathcal{V}_{3,lj}^\prime$. By Cauchy-Schwartz inequality, we can bound them easily by
\begin{align*}
  \pm\mathcal{V}_{3,l1}^\prime&\leq C\tilde{\rho}_0^{1+\delta_2}\mathcal{N}_{re}
    +C\tilde{\rho}_0^{1-\delta_2-3\alpha_2}
    \mathcal{N}_h[\delta_4]\\
    \pm\mathcal{V}_{3,l2}^\prime&\leq C\tilde{\rho}_0^{1+\delta_4}\mathcal{N}_{re}
    +C\tilde{\rho}_0^{2+2\delta_4}L^3\\
    \pm\mathcal{V}_{3,l3}^\prime&\leq
    C\tilde{\rho}_0^{1-3\alpha_2}\mathcal{N}_i[\delta_4]+C\tilde{\rho}_0^{2+2\delta_4}L^3
\end{align*}
Thus
  \begin{equation}\label{V_3,l'bound}
    \pm\mathcal{V}_{3,l}^\prime\leq C\tilde{\rho}_0^{1+\delta_2}\mathcal{N}_{re}
    +C\tilde{\rho}_0^{1-3\alpha_2}\mathcal{N}_i[\delta_4]+C\tilde{\rho}_0^{1-\delta_2-3\alpha_2}
    \mathcal{N}_h[\delta_4]+C\tilde{\rho}_0^{2+2\delta_4}L^3.
  \end{equation}
   \par For $\mathcal{V}_{3,h}^\prime$, notice that $\zeta^+(k)>0$ implies $p-k,q+k\in P_{F,-\beta_1}$, hence
   \begin{equation}\label{V_3 h'}
\begin{aligned}
  \mathcal{V}_{3,h}^\prime=2\sum_{\sigma,\nu}\int_{\Lambda_L^2}&W_{\zeta^+}(x-y)
  a^*(H_{x,\sigma}[-\beta_1])a^*(H_{y,\nu}[-\beta_1])\\
  &\times a(L_{y,\nu}[\delta_4])a(g_{x,\sigma})dxdy+h.c.
\end{aligned}
\end{equation}
Using (\ref{special W L2 norm}), we have
  \begin{equation}\label{V_3,h'bound}
    \pm\mathcal{V}_{3,h}^\prime\leq C\big(\tilde{\rho}_0^{1+\frac{\delta_4}{2}}
    +C\tilde{\rho}_0^{1-4\alpha_3+\delta_4+\alpha_5-\alpha_4}\big)
    \mathcal{N}_h[-\beta_1]+C\tilde{\rho}_0^{\frac{7}{3}+\alpha_4}L^3.
  \end{equation}
  \par For $e^{-tB^\prime}\mathcal{V}_{3,m}^\prime e^{tB^\prime}$, we first put
  \begin{equation}\label{V_3,m' 1st}
    e^{-tB^\prime}\mathcal{V}_{3,m}^\prime e^{tB^\prime}=\mathcal{V}_{3,m}^\prime
    +\int_{0}^{t}e^{-sB^\prime}[\mathcal{V}_{3,m}^\prime,B^\prime] e^{sB^\prime}ds.
  \end{equation}
  For $\mathcal{V}_{3,m}^\prime$, we use \begin{equation}\label{split W cub}
  \begin{aligned}
    W_k=\Big(\frac{1}{L^{\frac{3}{2}}}
    \sum_{m\neq k}\frac{W_{k-m}(k-m)}{\vert k-m\vert^2}\vartheta_m\Big)
    \big(p-(p-k)-m\big)+\frac{1}{L^{\frac{3}{2}}}W_0\vartheta_k,
  \end{aligned}
  \end{equation}
  to rewrite it by $\mathcal{V}_{3,m}^\prime=\sum_{j=1}^{4}\mathcal{V}_{3,m,j}^\prime$.
   For $\mathcal{V}_{3,m,1}^\prime$, Notice that $\phi^+(k)>0$ implies $p-k,q+k\in P_{F,-\alpha_2}$, and therefore
   \begin{equation}\label{V_3 m 1'}
\begin{aligned}
  \mathcal{V}_{3,m,1}^\prime=C\sum_{\sigma,\nu}\int_{\Lambda_L^2}&(U\vartheta)_{\phi^+}^
  {\zeta^-}(x-y)
  a^*(H_{x,\sigma}[-\alpha_2])a^*(H_{y,\nu}[-\alpha_2])\\
  &\times a(L_{y,\nu}[\delta_4])a(\nabla_xg_{x,\sigma})dxdy+h.c.
\end{aligned}
\end{equation}
Using (\ref{theta norm}), (\ref{U L2}) and Lemma \ref{b^* bound by b general}, we can bound
  \begin{equation}\label{V_3,m1'}
  \pm\mathcal{V}_{3,m,1}^\prime\leq C\tilde{\rho}_0^{\frac{2}{3}+\delta_4-\alpha_3-\alpha_4}
  \mathcal{N}_h[-\alpha_2]+C\tilde{\rho}_0^{\frac{7}{3}+\alpha_4}L^3.
  \end{equation}
   Using (\ref{theta norm}), (\ref{U L2}) and Lemma \ref{b^* bound by b general}, we can also bound $\mathcal{V}_{3,m,2}^\prime$ similarly by
  \begin{equation}\label{V_3,m2'}
  \pm\mathcal{V}_{3,m,2}^\prime\leq C\tilde{\rho}_0^{1+\frac{\delta_4}{2}-\frac{\alpha_3}{2}}
  \mathcal{N}_h[-\alpha_2]+
  C\tilde{\rho}_0^{\delta_4-\alpha_3-\alpha_4}\mathcal{K}_h[-\alpha_2]
  +C\tilde{\rho}_0^{\frac{7}{3}+\alpha_4}L^3.
  \end{equation}
   Similarly by (\ref{theta norm}), (\ref{U L2}) and Lemma \ref{b^* bound by b general}, we can bound $\mathcal{V}_{3,m,3}^\prime$ by
  \begin{equation}\label{V_3,m3'}
 \pm\mathcal{V}_{3,m,3}^\prime\leq C\tilde{\rho}_0^{\frac{2}{3}+\delta_4-\alpha_3-\alpha_4}
  \mathcal{N}_h[-\alpha_2]+C\tilde{\rho}_0^{\frac{7}{3}+\alpha_4}L^3.
  \end{equation}
   and we can bound $\mathcal{V}_{3,m,4}^\prime$ by
  \begin{equation}\label{V_3,m4'}
 \pm\mathcal{V}_{3,m,4}^\prime\leq C(\tilde{\rho}_0)L^{-3}
 \big(\mathcal{N}_h[-\alpha_2]+L^3\big).
  \end{equation}
  \par For $[\mathcal{V}_{3,m}^\prime,B^\prime]$, we have $[\mathcal{V}_{3,m}^\prime,B^\prime]=\sum_{j=1}^{13}\hat{\Upsilon}_{j}$. Here each $\hat{\Upsilon}_j$ corresponds to $\Upsilon_j$ defined in the proof of Lemma \ref{lemma cub core}, with $\hat{v}_k$ in each $\Upsilon_j$ replaced by $2W_k\phi^+(k)\zeta^-(k)$. So thet can be handled similarly. For $\hat{\Upsilon}_1$, using $1=\chi_{k=l}+\chi_{k\neq l}$, we can split it by $\hat{\Upsilon}_1=\hat{\Upsilon}_{1,1}+\hat{\Upsilon}_{1,2}$. Via
  \begin{align*}
  \sum_{k}\vert W_k\vert\vert \eta_k\vert\leq L^{\frac{3}{2}}\Vert W\Vert_2\Vert\eta\Vert_2
  \leq C\tilde{\rho}_0^{\frac{1}{3}-\alpha_3},
  \end{align*}
  we can bound $\hat{\Upsilon}_{1,1}$ by
  \begin{equation}\label{hatUpsilon11bound}
    \pm\hat{\Upsilon}_{1,1}\leq C\tilde{\rho}_0^{\frac{7}{3}+\delta_4-\alpha_3}L^3.
  \end{equation}
  Similar to the bound of $\Upsilon_{12}$ in (\ref{Upsilon 12 bound}), we can bound  $\hat{\Upsilon}_{1,2}$ by
  \begin{equation}\label{hatUpsilon1,2 bound}
    \pm\hat{\Upsilon}_{1,2}\leq C\tilde{\rho}_0^{1+\delta_4}\tilde{\mathcal{N}}_{re}.
  \end{equation}
  For $\hat{\Upsilon}_{2}$, the proof is analogous to $\hat{\Upsilon}_{1}$, therefore,
   \begin{equation}\label{hatUpsilon2 bound}
    \pm\hat{\Upsilon}_{2}\leq C\tilde{\rho}_0^{1+\delta_4}\tilde{\mathcal{N}}_{re}
    +C\tilde{\rho}_0^{\frac{7}{3}+\delta_4-\alpha_3}L^3.
  \end{equation}
  \par We find $\hat{\Upsilon}_j=0$ for $j=3,4,5,6$. Since for $j=3,4$, $r-l\in A_{F,\delta_4}$ or $s+l\in A_{F,\delta_4}$ implies $\phi^+(l)=0$, and for $j=5,6$,  $p-k\in A_{F,\delta_4}$ or $q+k\in A_{F,\delta_4}$ implies $\phi^+(k)=0$.
  \par For $j=7,8,9,10$, the proof is totally analogous to $\Upsilon_j$ in (\ref{Upsilon 7-10 bound}):
  \begin{equation}\label{hatUpsilon7-10 bound}
    \pm\hat{\Upsilon}_j\leq C\tilde{\rho}_0^{\frac{4}{3}+\delta_4+2\alpha_3}\mathcal{N}_{re}.
  \end{equation}
  \par For $j=11,12,13$, the proof is analogous to $\Upsilon_j$ in (\ref{Upsilon 11})-(\ref{Upsilon 13bound}). We need to additionally use the fact that $\zeta^-(k)>0$ implies $\vert p-k\vert,\vert q+k\vert\leq 3\tilde{\mu}^{\frac{1}{2}}\tilde{\rho}_0^{-\beta_1}$, and thus need the following notation
  \begin{equation*}
  \mathbf{h}_{x,\sigma}(z)=\sum_{k_F<\vert k\vert\leq 3\tilde{\mu}^{\frac{1}{2}}\tilde{\rho}
  _0^{-\beta_1}}\frac{e^{ikx}}{L^{\frac{3}{2}}}f_{k,\sigma}(z).
\end{equation*}
The same process gives out the bounds directly:
  \begin{equation}\label{hatUpsilon11-13 bound}
    \pm\hat{\Upsilon}_j\leq C\tilde{\rho}_0^{\frac{5}{6}+\frac{\alpha_2}{2}
    -\frac{3}{2}\alpha_5}\mathcal{N}_{h}[-\alpha_2].
  \end{equation}
  \par Collect all the results above, together with Lemmas \ref{lemma cub control N_re}-\ref{lemma cub control V_4}, we reach (\ref{cub Gamma'}).
\end{proof}

\vspace{1em}

\begin{proof}[Proof of Proposition \ref{cub prop}]
\par (\ref{define J_N}) together with Lemmas \ref{lemma cub misc}-\ref{lemma cub Gamma'} yield Proposition \ref{cub prop}.
\end{proof}

\section{Bogoliubov Transformation}\label{bog}
\
\par In this section, we analyze the excitation Hamiltonian $\mathcal{Z}_N$ defined in (\ref{Z_N}), and prove Proposition \ref{bog prop}. The Bogoliubov transformation is defined through
\begin{equation}\label{define B tilde}
  \tilde{B}=\frac{1}{2}(\tilde{A}-\tilde{A}^*)
\end{equation}
with
\begin{equation}\label{define A tilde}
  \tilde{A}=\sum_{k,p,q,\sigma,\nu}
  \xi_{k,q,p}^{\nu,\sigma}a^*_{p-k,\sigma}a^*_{q+k,\nu}a_{q,\nu}a_{p,\sigma}\chi_{p-k,q+k\notin B_F}\chi_{p,q\in B_F}.
\end{equation}
Recall in Section \ref{Bog coeff sec}, we choose
\begin{equation}\label{define xi_k,q,p,nu,sigma}
  \xi_{k,q,p}^{\nu,\sigma}=\frac{-\big(L^{-3}W_k\tilde{\zeta}^-(k)
  +\eta_kk(q-p)\phi^+(k)\tilde{\zeta}^-(k)\big)}
  {\frac{1}{2}\big(\vert q+k\vert^2+\vert p-k\vert^2-\vert q\vert^2-\vert p\vert^2\big)
  +\epsilon_0}
  \chi_{p-k,q+k\notin B_F}\chi_{p,q\in B_F}.
\end{equation}
and for simplicity, we let $\xi_{k,q,p}^{\nu,\sigma}=\xi_{k,q,p}^{\nu,\sigma,1}+\xi_{k,q,p}^{\nu,\sigma,2}$, where
\begin{equation}\label{xi1and 2}
  \begin{aligned}
  &\xi_{k,q,p}^{\nu,\sigma,1}=\frac{-L^{-3}W_k\tilde{\zeta}^-(k)}
  {\frac{1}{2}\big(\vert q+k\vert^2+\vert p-k\vert^2-\vert q\vert^2-\vert p\vert^2\big)
  +\epsilon_0}
  \chi_{p-k,q+k\notin B_F}\chi_{p,q\in B_F}\\
  &\xi_{k,q,p}^{\nu,\sigma,2}=\frac{-\eta_kk(q-p)\phi^+(k)\tilde{\zeta}^-(k)}
  {\frac{1}{2}\big(\vert q+k\vert^2+\vert p-k\vert^2-\vert q\vert^2-\vert p\vert^2\big)
  +\epsilon_0}
  \chi_{p-k,q+k\notin B_F}\chi_{p,q\in B_F}.
  \end{aligned}
\end{equation}
Here $\epsilon_0$ is a small but positive (when $L$ tends to infinity) gap to avoid logarithmic growth in $L$. We let $\epsilon_0=\tilde{\rho}_0^2$.

\par To deal with the length scale $L$, which tends to infinity in the thermodynamic limit, we continue applying frequency localization:
\begin{equation}\label{define V_21' and Omega tilde}
  \begin{aligned}
  \tilde{\mathcal{V}}_{21}^\prime&=\frac{1}{L^3}\sum_{k,p,q,\sigma,\nu}
  W_k\tilde{\zeta}^{-}(k)(a^*_{p-k,\sigma}a^*_{q+k,\nu}a_{q,\nu}a_{p,\sigma}+h.c.)
  \chi_{p-k,q+k\notin
  B_F}\chi_{p,q\in B_F}\\
  \tilde{\Omega}&=\sum_{k,p,q,\sigma,\nu}
  \eta_k\phi^+(k)\tilde{\zeta}^{-}(k)k(q-p)
  \\
  &\quad\quad\quad\quad
  \times(a^*_{p-k,\sigma}a^*_{q+k,\nu}a_{q,\nu}a_{p,\sigma}+h.c.)\chi_{p-k,q+k\notin B_F}
  \chi_{p,q\in B_F}\\
  R&=\mathcal{V}_{21}^\prime+\Omega-\tilde{\mathcal{V}}_{21}^\prime-\tilde{\Omega}
  \end{aligned}
\end{equation}
Using (\ref{J_N cub prop}) and Newton-Leibniz formula,  we rewrite $\mathcal{Z}_N$ by
\begin{equation}\label{define Z_N}
  \begin{aligned}
  \mathcal{Z}_N\coloneqq e^{-\tilde{B}}\mathcal{J}_Ne^{\tilde{B}}=
  &C_{\mathcal{G}_N}
  +\mathcal{K}+e^{-\tilde{B}}\mathcal{V}_{4,4h}e^{\tilde{B}}
  +e^{-\tilde{B}}(\mathcal{E}_{\mathcal{J}_N}+R)e^{\tilde{B}}
  +\int_{0}^{1}e^{-t\tilde{B}}\tilde{\Gamma}e^{t\tilde{B}}dt\\
  &+\int_{0}^{1}\int_{t}^{1}e^{-s\tilde{B}}[\tilde{\mathcal{V}}_{21}^\prime
  +\tilde{\Omega},\tilde{B}]e^{s\tilde{B}}dsdt
  \end{aligned}
\end{equation}
with
\begin{equation}\label{define Gamme tilde}
  \tilde{\Gamma}=[\mathcal{K},\tilde{B}]+\tilde{\mathcal{V}}_{21}^\prime
  +\tilde{\Omega}.
\end{equation}

\par To prove Proposition \ref{bog prop}, we are going to analyze each term on the right-hand side of (\ref{define Z_N}), in Lemmas \ref{lemma bog control K}-\ref{cal com bog lemma}. Before we rigorously calculate each term of $\mathcal{Z}_N$, we first need to bound the action of the Bogoliubov transformation on some special operators, such as $\mathcal{N}_{re}, \mathcal{K}_s$. These results are collected respectively in Lemma \ref{lemma bog control N_re}-\ref{lemma bog control K}. Also, we need to analyze some properties of the $t-$version quadratic operators defined in Section \ref{other} in Lemmas \ref{property htgt}-\ref{lemma bog htgt}.

\begin{lemma}\label{lemma bog control N_re}
For any $N\in\mathbb{N}_{\geq0}$ and $\vert t\vert\leq 1$,
\begin{align}
 e^{-t\tilde{B}}\mathcal{N}_{re}e^{t\tilde{B}}&\leq C\mathcal{N}_{re}+C\tilde{\rho}_0^{\frac{5}{3}-\alpha_3-\varepsilon}L^3,
 \label{bog control N_re}\\
 e^{-t\tilde{B}}\tilde{\mathcal{N}}_{re}e^{t\tilde{B}}&\leq C\tilde{\mathcal{N}}_{re}+ C\mathcal{N}_{re}+C\tilde{\rho}_0^{\frac{5}{3}-\alpha_3-\varepsilon}L^3,
 \label{bog control N_retilde}
\end{align}
For $0<\alpha_3<\frac{1}{6}$ and some universal but arbitrarily small constant $\varepsilon>0$. Moreover, for $\delta>0$,
\begin{align}
 e^{-t\tilde{B}}\mathcal{N}_h[\delta]e^{t\tilde{B}}&\leq
  C\mathcal{N}_h[\delta]
  +C\tilde{\rho}_0^{\frac{2}{3}-2\varepsilon}\mathcal{N}_{re}
  +C\tilde{\rho}_0^{\frac{5}{3}-\alpha_3-\varepsilon}L^3,\label{bog control N_h[delta]}\\
  e^{-t\tilde{B}}\mathcal{N}_i[\delta]e^{t\tilde{B}}&\leq C\mathcal{N}_i[\delta]
  +C\mathcal{N}_h[\delta]+C\tilde{\rho}_0^{\frac{2}{3}-2\varepsilon}\mathcal{N}_{re}
  +C\tilde{\rho}_0^{\frac{5}{3}-\alpha_3-\varepsilon}L^3.\label{bog control N_i[delta]}
\end{align}
for $\delta<0$,
\begin{equation}\label{Bog control N_h[delta]delta<0}
  e^{-t\tilde{B}}\mathcal{N}_h[\delta]e^{t\tilde{B}}\leq C\mathcal{N}_h[\delta]+C\tilde{\rho}_0^{\frac{2}{3}-2\varepsilon}\mathcal{N}_{re}
  +C\tilde{\rho}_0^{\frac{5}{3}-\alpha_3-2\delta}L^3.
\end{equation}
\end{lemma}
\begin{proof}
  \par Calculating directly, we have $[\mathcal{N}_{re},\tilde{B}]=\tilde{A}+\tilde{A}^*$. For $\tilde{A}=\tilde{A}_1+\tilde{A}_2$, where we rewrite $\tilde{A}$ using (\ref{xi1and 2}), we have
  \begin{equation}\label{tildeA}
  \begin{aligned}
    \tilde{A}_1=&2\sum_{\sigma,\nu}\int_{0}^{\infty}\int_{\Lambda_L^2}e^{-2\epsilon_0 t}
    W^{\tilde{\zeta}^-}(x-y)a^*(h_{x,\sigma}^t)a^*(h^t_{y,\nu})
    a(g_{y,\nu}^t)a(g_{x,\sigma}^t)dxdydt\\
    &=2\sum_{\sigma,\nu}\int_{0}^{\infty}\int_{\Lambda_L}e^{-2\epsilon_0 t}
     a^*(h_{x,\sigma}^t)a(g_{x,\sigma}^t)c^{t*}_{x,\nu}(W^{\tilde{\zeta}^-})dxdt.
    \end{aligned}
  \end{equation}
  Therefore, for any $\psi\in\mathcal{H}^{\wedge N}$, using Cauchy-Schwartz inequality and Lemma \ref{lemma bog htgt}, we have
  \begin{align*}
   \vert\langle\tilde{A}_1\psi,\psi\rangle\vert&\leq C
   \sum_{\sigma,\nu}\int_{0}^{\infty}\int_{\Lambda_L}e^{-2\epsilon_0 t}\Vert
   a^*(g_{x,\sigma}^t)a(h_{x,\sigma}^t)\psi\Vert\\
   &\times \Big(\Vert c^{t}_{x,\nu}(W^{\tilde{\zeta}^-})\psi\Vert+\sum_{p\notin B_F,q\in B_F}\frac{\vert W_{p-q}\vert^2}{L^6}e^{-2(p^2-q^2)}\Vert\psi\Vert^2
    \Big).
  \end{align*}
  By  Lemma \ref{property htgt} and Lemma \ref{est sum htgt},
  \begin{align*}
    \sum_{\sigma,\nu}\int_{0}^{\infty}\int_{\Lambda_L^2}e^{-2\epsilon_0 t}
    \vert W^{\tilde{\zeta}^-}(x-y)\vert \Vert
   a^*(g_{x,\sigma}^t)a(h_{x,\sigma}^t)\psi\Vert^2dxdydt\leq C\tilde{\rho}_0^{\frac{1}{3}}
   \vert\ln\tilde{\rho}_0\vert\mathcal{N}_{re},
  \end{align*}
  hence we have, by Cauchy-Schwartz inequality and Lemma \ref{est xi/ p^2 lemma}
  \begin{equation}\label{Atilde1bound}
    \pm\tilde{A}_1\leq C\mathcal{N}_{re}+C\tilde{\rho}_0^{\frac{5}{3}-\alpha_3}\vert\ln\tilde{\rho}_0\vert L^3.
  \end{equation}
  Notice that $\vert\ln\tilde{\rho}_0\vert\leq C\tilde{\rho}_0^{-\varepsilon}$ for some arbitrarily small constant $\varepsilon>0$. $\tilde{A}_2$ can be bounded similarly, with an even smaller upper bound, one can notice this fact by comparing the difference between (\ref{est Wk/p^2}) and (\ref{est etakk/p^2}). Thus (\ref{bog control N_re}) follows by Gronwall's inequality.

  \par For $\tilde{\mathcal{N}}_{re}$, we have $[\tilde{\mathcal{N}}_{re},\tilde{B}]
  =[\mathcal{N}_{re},\tilde{B}]$, thus (\ref{bog control N_retilde}) follows directly from Newton-Leibniz formula, (\ref{bog control N_re}) and (\ref{Atilde1bound}).

  \par For $\mathcal{N}_h[\delta]$, we have
  \begin{equation}\label{com N_h bog}
    [\mathcal{N}_h[\delta],\tilde{B}]=\sum_{k,p,q,\sigma,\nu}
  \xi_{k,q,p}^{\nu,\sigma}a^*_{p-k,\sigma}a^*_{q+k,\nu}a_{q,\nu}a_{p,\sigma}
  \chi_{p-k\in P_{F,\delta}}\chi_{q+k\notin B_F}\chi_{p,q\in B_F}+h.c.
  \end{equation}
  The bound of $[\mathcal{N}_h[\delta],\tilde{B}]$ is analogous to the bound of $[\mathcal{N}_{re},\tilde{B}]$, only need to notice that here we use (\ref{property htgt N_h and i}) in Lemma \ref{property htgt}, and (\ref{est sum d}) in Lemma \ref{est sum htgt}. We have, for $\delta\geq0$,
  \begin{equation}\label{N_h bog bound delta>=0}
    \pm[\mathcal{N}_h[\delta],\tilde{B}]\leq
    C\mathcal{N}_{h}[\delta]+
    C\tilde{\rho}_0^{\frac{2}{3}-2\varepsilon}\mathcal{N}_{re}
    +C\tilde{\rho}_0^{\frac{5}{3}-\alpha_3-\varepsilon} L^3,
  \end{equation}
  and for $\delta<0$,
  \begin{equation}\label{N_h bog bound delta<0}
    \pm[\mathcal{N}_h[\delta],\tilde{B}]\leq
    C\mathcal{N}_{h}[\delta]+
    C\tilde{\rho}_0^{\frac{2}{3}-\varepsilon-2\delta}\mathcal{N}_{re}
    +C\tilde{\rho}_0^{\frac{5}{3}-\alpha_3-2\delta} L^3.
  \end{equation}
  Then (\ref{bog control N_h[delta]}) and (\ref{Bog control N_h[delta]delta<0}) follow from Gronwall's inequality and (\ref{bog control N_re}).
  \par For $\mathcal{N}_i[\delta]$, we have
  \begin{equation}\label{com N_i bog}
    [\mathcal{N}_i[\delta],\tilde{B}]=\sum_{k,p,q,\sigma,\nu}
  \xi_{k,q,p}^{\nu,\sigma}a^*_{p-k,\sigma}a^*_{q+k,\nu}a_{q,\nu}a_{p,\sigma}
  \chi_{p-k,q+k\notin B_F}\chi_{p\in \underline{B}_{F,\delta}}\chi_{q\in B_F}+h.c.
  \end{equation}
  We use (\ref{xi1and 2}) to rewrite $[\mathcal{N}_i[\delta],\tilde{B}]=\mathrm{I}+\mathrm{II}$. As discussed above, we only need to bound $\mathrm{I}$. Using $\chi_{p-k\notin B_F}=\chi_{p-k\in P_{F,\delta}}+\chi_{p-k\in A_{F,\delta}}$, we can furthermore divide $\mathrm{I}=\mathrm{I}_1+\mathrm{I}_2$. For $\mathrm{I}_1$, we can bound it in the way of (\ref{N_h bog bound delta>=0})
  \begin{equation}\label{N_i bog I_1}
    \pm\mathrm{I}_1\leq C\mathcal{N}_{h}[\delta]+
    C\tilde{\rho}_0^{\frac{2}{3}-2\varepsilon}\mathcal{N}_{re}
    +C\tilde{\rho}_0^{\frac{5}{3}-\alpha_3-\varepsilon} L^3.
  \end{equation}
  For $\mathrm{I}_2$, we use Lemma \ref{property htgt} and Lemma \ref{est sum htgt} to attain
  \begin{align*}
    \sum_{\sigma,\nu}\int_{0}^{\infty}\int_{\Lambda_L^2}e^{-2\epsilon_0 t}
    \vert W^{\tilde{\zeta}^-}(x-y)\vert \Vert
   a^*(I_{x,\sigma}^t[\delta])a(L_{x,\sigma}^t[\delta])\psi\Vert^2dxdydt\leq C\tilde{\rho}_0^{\frac{1}{3}}
   \vert\ln\tilde{\rho}_0\vert\mathcal{N}_{i}[\delta],
  \end{align*}
  then similar to the bound of $\tilde{A}_1$ in (\ref{tildeA}), we have
  \begin{equation}\label{N_i bog I_2}
    \pm\mathrm{I}_1\leq C\mathcal{N}_{i}[\delta]+
    C\tilde{\rho}_0^{\frac{2}{3}-2\varepsilon}\mathcal{N}_{re}
    +C\tilde{\rho}_0^{\frac{5}{3}-\alpha_3-\varepsilon} L^3.
  \end{equation}
 Therefore, we have
  \begin{equation}\label{N_i bog bound delta>=0}
    \pm[\mathcal{N}_i[\delta],\tilde{B}]\leq C\mathcal{N}_i[\delta]+
    C\mathcal{N}_{h}[\delta]+
    C\tilde{\rho}_0^{\frac{2}{3}-2\varepsilon}\mathcal{N}_{re}
    +C\tilde{\rho}_0^{\frac{5}{3}-\alpha_3-\varepsilon} L^3.
  \end{equation}
  Then (\ref{bog control N_i[delta]}) follows from Gronwall's inequality, (\ref{bog control N_re}) and (\ref{bog control N_h[delta]}).
\end{proof}

\begin{lemma}\label{lemma bog control K}
  For $0<8\varepsilon<4\alpha_2<\alpha_3<\frac{1}{6}$, and any $N\in\mathbb{N}_{\geq0}$ and $\vert t\vert\leq 1$, we have
  \begin{align}
  e^{-t\tilde{B}}\mathcal{K}_se^{t\tilde{B}}\leq&C\mathcal{K}_s
  +C\tilde{\rho}_0^{1-\alpha_3}\mathcal{N}_{re}+C\tilde{\rho}_0^2L^3,\label{bog con K_s}\\
  e^{-t\tilde{B}}\mathcal{K}_h[-\alpha_2]e^{t\tilde{B}}\leq&
  C\mathcal{K}_h[-\alpha_2]+C\tilde{\rho}_0^{\frac{2}{3}}\mathcal{N}_h[-\alpha_2]
  +C\tilde{\rho}_0^{\frac{4}{3}+2\alpha_2-\varepsilon}\mathcal{N}_{re}
  +C\tilde{\rho}_0^{\frac{7}{3}-\alpha_3-
    \varepsilon}L^3.\label{bog con K_h}
  \end{align}
  Moreover,
  \begin{equation}\label{bog con Gammatilde}
    \pm\tilde{\Gamma}\leq C\epsilon_0\mathcal{N}_{re}+C\epsilon_0\tilde{\rho}_0^
    {\frac{5}{3}-\alpha_3-\varepsilon}L^3.
  \end{equation}
  Since we have set $\epsilon_0=\tilde{\rho}_0^2$, we have
  \begin{equation}\label{bog con Gammatilde 2}
    \pm\int_{0}^{1}e^{-t\tilde{B}}\tilde{\Gamma}e^{t\tilde{B}}dt\leq
    C\tilde{\rho}_0^2\mathcal{N}_{re}+C\tilde{\rho}_0^
    {\frac{11}{3}-\alpha_3-\varepsilon}L^3.
  \end{equation}
\end{lemma}
\begin{proof}
  \par For (\ref{bog con K_s}), we have
  \begin{equation}\label{com bog K_s}
    [\mathcal{K}_s,\tilde{B}]=[\mathcal{K},\tilde{B}]=
    -\tilde{\mathcal{V}}_{21}^\prime-\tilde{\Omega}-\epsilon_0
    (\tilde{A}+\tilde{A}^*).
  \end{equation}
  $\tilde{V}_{21}^\prime$ can be bounded in the same way as (\ref{qua V_21'}), $\tilde{\Omega}$ can be bounded in the same way as (\ref{est qua Omega}), and $\tilde{A}$ has been bounded in (\ref{Atilde1bound}). Therefore, for $\epsilon_0=\tilde{\rho}_0^2$
  \begin{equation}\label{bound com bog K_s}
    \pm[\mathcal{K}_s,\tilde{B}]\leq C\mathcal{K}_s
    +C\tilde{\rho}_0^{1-\alpha_3}\mathcal{N}_{re}+C\tilde{\rho}_0^{2}L^3.
  \end{equation}
  Hence (\ref{bog con K_s}) follows from Gronwall's inequality and (\ref{bog control N_re}).

\par For $\tilde{\Gamma}$, from (\ref{define Gamme tilde}) and (\ref{com bog K_s}), we have
\begin{equation}\label{Gamma  tilde}
  \tilde{\Gamma}=-\epsilon_0(\tilde{A}+\tilde{A}^*).
\end{equation}
Therefore, (\ref{bog con Gammatilde}) and (\ref{bog con Gammatilde 2}) follow from (\ref{Atilde1bound}) and Lemma \ref{lemma bog control N_re}.

   \par For $\mathcal{K}_h[-\alpha_2]$, we have
  \begin{equation}\label{com bog K_h-alpha_2}
  \begin{aligned}
    &[\mathcal{K}_h[-\alpha_2],\tilde{A}]\\
    =&\sum_{k,p,q,\sigma,\nu}
  \xi_{k,q,p}^{\nu,\sigma}\vert p-k\vert^2
  a^*_{p-k,\sigma}a^*_{q+k,\nu}a_{q,\nu}a_{p,\sigma}
  \chi_{p-k\in P_{F,-\alpha_2}}\chi_{q+k\notin B_F}\chi_{p,q\in B_F}\\
  &+\sum_{k,p,q,\sigma,\nu}
  \xi_{k,q,p}^{\nu,\sigma}\vert q+k\vert^2
  a^*_{p-k,\sigma}a^*_{q+k,\nu}a_{q,\nu}a_{p,\sigma}
  \chi_{q+k\in P_{F,-\alpha_2}}\chi_{p-k\notin B_F}\chi_{p,q\in B_F}.
  \end{aligned}
  \end{equation}
  From definitions (\ref{define xi_k,q,p,nu,sigma}) and (\ref{xi1and 2}), we use $1=\phi^-(k)+\phi^+(k)$, and notice the fact that $\phi^+(k)\chi_{p,q\in B_F}\neq0$ implies $p-k, q+k\in P_{F,-\alpha_2}$, we rewrite
  \begin{equation}\label{com bog K_h-alpha_2 split}
    [\mathcal{K}_h[-\alpha_2],\tilde{A}]=\mathrm{I}+\mathrm{II}+\mathrm{III},
  \end{equation}
  with
  \begin{align}
  \mathrm{I}&=2\sum_{k,p,q,\sigma,\nu}
  \xi_{k,q,p}^{\nu,\sigma}\vert p-k\vert^2\phi^-(k)
  a^*_{p-k,\sigma}a^*_{q+k,\nu}a_{q,\nu}a_{p,\sigma}
  \chi_{p-k\in P_{F,-\alpha_2}}\chi_{q+k\notin B_F}\chi_{p,q\in B_F}\nonumber\\
  \mathrm{II}&=\sum_{k,p,q,\sigma,\nu}
  \xi_{k,q,p}^{\nu,\sigma}\big(\vert p\vert^2+\vert q\vert^2-2\epsilon_0\big)\phi^+(k)
  \nonumber\\
  &\quad\quad\times a^*_{p-k,\sigma}a^*_{q+k,\nu}a_{q,\nu}a_{p,\sigma}
  \chi_{p-k,q+k\in P_{F,-\alpha_2}}\chi_{p,q\in B_F}\nonumber\\
  \mathrm{III}&=-2\sum_{k,p,q,\sigma,\nu}
  \big(L^{-3}W_k+\eta_kk(q-p)\big)\phi^+(k)\tilde{\zeta}^-(k)\nonumber\\
  &\quad\quad\times
  a^*_{p-k,\sigma}a^*_{q+k,\nu}a_{q,\nu}a_{p,\sigma}
  \chi_{p-k,q+k\in P_{F,-\alpha_2}}\chi_{p,q\in B_F}\label{com bog K_h-alpha_2 I II III}
  \end{align}
  \par For $\mathrm{I}$, by definitions (\ref{define xi_k,q,p,nu,sigma}) and (\ref{xi1and 2}), and notice that $\phi^-(k)\neq0$ implies $\vert p-k\vert, \vert q+k\vert\leq 3\tilde{\mu}^{\frac{1}{2}}\tilde{\rho}_0^{-\alpha_2}$, we temporarily make the following notations:
  \begin{align*}
    &\mathbf{h}_{x,\sigma}^t=\sum_{\vert k\vert\leq 3\tilde{\mu}^{\frac{1}{2}}\tilde{\rho}_0^{-\alpha_2}
    }\frac{e^{ikx}}{L^{\frac{3}{2}}}e^{-k^2t}f_{k,\sigma}(z)\chi_{k\in P_{F,-\alpha_2}},\\
    &\tilde{\mathbf{h}}_{x,\sigma}^t=
    \sum_{\vert k\vert\leq 3\tilde{\mu}^{\frac{1}{2}}\tilde{\rho}_0^{-\alpha_2}
    }\frac{e^{ikx}}{L^{\frac{3}{2}}}e^{-k^2t}f_{k,\sigma}(z)\chi_{k\notin B_F},
  \end{align*}
  then
  \begin{equation}\label{I K_h -alpha_2}
    \begin{aligned}
    \mathrm{I}=C\sum_{\sigma,\nu}\int_{0}^{\infty}\int_{\Lambda_L^2}W^{\phi^-}(x-y)
    &e^{-2\epsilon_0 t}a^*(\Delta_x\mathbf{h}_{x,\sigma}^t)a^*(\tilde{\mathbf{h}}_{y,\nu}^t)\\
    &\times a(g_{y,\nu}^t)a(g_{x,\sigma}^t)dxdydt
    \end{aligned}
  \end{equation}
  By Cauchy-Schwartz inequality, for any $\psi\in\mathcal{H}^{\wedge N}$,
  \begin{align*}
    &\vert\langle\mathrm{I}\psi,\psi\rangle\vert\\
    &\leq
    \Big(\sum_{\sigma,\nu}\int_{0}^{\infty}\int_{\Lambda_L^2}\vert W^{\phi^-}(x-y)\vert
    e^{-2\epsilon_0 t}e^{2k_F^2t}\Vert a^*(g_{y,\nu}^t)\Vert^2
    \Vert a(\tilde{\mathbf{h}}_{y,\nu}^t)\Vert^2\Vert a^*(\Delta_x\mathbf{h}_{x,\sigma}^t)
    \psi\Vert^2 \Big)^{\frac{1}{2}}\\
    &\times (\sum_{\sigma,\nu}\int_{0}^{\infty}\int_{\Lambda_L^2}\vert W^{\phi^-}(x-y)\vert
    e^{-2\epsilon_0 t}e^{-2k_F^2t}\Vert a(g_{x,\sigma}^t)\Vert^2
    \Vert\psi\Vert^2 \Big)^{\frac{1}{2}}
  \end{align*}
  Since
  \begin{equation*}
    \Vert a^*(g_{y,\nu}^t)\Vert^2
    \Vert a(\tilde{\mathbf{h}}_{y,\nu}^t)\Vert^2\leq C\tilde{\rho}_0e^{2k_F^2t}
    \sum_{\vert k\vert\leq 3\tilde{\mu}^{\frac{1}{2}}\tilde{\rho}_0^{-\alpha_2}}\frac{e^{-2p^2t}}{L^3}
  \end{equation*}
  and
  \begin{align*}
   \int_{\Lambda_L}\Vert a^*(\Delta_x\mathbf{h}_{x,\sigma}^t)\psi\Vert^2dx
   \leq C\tilde{\rho}_0^{\frac{2}{3}-2\alpha_2}e^{-2k_F^2t}\langle\mathcal{K}_h[-\alpha_2]
   \psi,\psi\rangle.
  \end{align*}
  Therefore, estimating similarly to (\ref{est sum A_F d <=0}), We have
  \begin{equation*}
  \begin{aligned}
   \sum_{\sigma,\nu}\int_{0}^{\infty}\int_{\Lambda_L^2}&\vert W^{\phi^-}(x-y)\vert
    e^{-2\epsilon_0 t}e^{2k_F^2t}\Vert a^*(g_{y,\nu}^t)\Vert^2
    \Vert a(\tilde{\mathbf{h}}_{y,\nu}^t)\Vert^2\Vert a^*(\Delta_x\mathbf{h}_{x,\sigma}^t)
    \psi\Vert^2\\
    &\leq C\tilde{\rho}_0^{2-3\alpha_2}\mathcal{K}_h[-\alpha_2].
    \end{aligned}
  \end{equation*}
  On the other hand, using (\ref{property gt norm}) and (\ref{est sum}), it is easy to bound
  \begin{align*}
    \sum_{\sigma,\nu}\int_{0}^{\infty}\int_{\Lambda_L^2}\vert W^{\phi^-}(x-y)\vert
    e^{-2\epsilon_0 t}e^{-2k_F^2t}\Vert a(g_{x,\sigma}^t)\Vert^2
    \Vert\psi\Vert^2\leq C\tilde{\rho}_0^{\frac{1}{3}-\varepsilon}L^3.
  \end{align*}
  Therefore,
  \begin{equation}\label{I bound bog}
    \pm\mathrm{I}\leq C\mathcal{K}_{h}[-\alpha_2]+C\tilde{\rho}_0^{\frac{7}{3}-3\alpha_2-
    \varepsilon}L^3.
  \end{equation}
  \par For $\mathrm{II}$, it can be bounded analogously to $\tilde{A}$:
  \begin{equation}\label{II bound bog}
    \pm\mathrm{II}\leq C\tilde{\rho}_0^{\frac{2}{3}}
    \mathcal{N}_{h}[-\alpha_2]+C\tilde{\rho}_0^{\frac{7}{3}-\alpha_3-
    \varepsilon}L^3.
  \end{equation}
  \par For $\mathrm{III}$, the analysis is analogous to the process in proving (\ref{est qua Omega}) and (\ref{qua V_21'}), we omit further tedious details and write out the result directly:
  \begin{equation}\label{III bound bog}
    \pm\mathrm{III}\leq C\mathcal{K}_h[-\alpha_2]+C\tilde{\rho}_0^{\frac{2}{3}}
    \mathcal{N}_{h}[-\alpha_2]+C\tilde{\rho}_0^{\frac{7}{3}-\alpha_3}L^3.
  \end{equation}
Therefore, (\ref{bog con K_h}) follows by Gronwall's inequality and Lemma \ref{lemma bog control N_re}.

\end{proof}

\begin{lemma}\label{lemma bog control V_4,4h}
For $\frac{1}{24}>\frac{1}{2}\delta_3>\frac{1}{4}\alpha_3>\alpha_2>\alpha_5>2\alpha_4>2\varepsilon >0$, $\frac{1}{3}>\delta_4\geq\delta_1\geq\delta_2>\frac{1}{12}$, $\delta_4>\alpha_6+\varepsilon>2\alpha_3+2\alpha_4+3\varepsilon$, $\beta_1=\frac{1}{3}+\alpha_5$, assume further that $\alpha_1>\delta_4$, we have
  \begin{equation}\label{bog control V_4,4h}
    \tr e^{-\tilde{B}}\mathcal{V}_{4,4h}e^{\tilde{B}}\tilde{G}_0\leq C
    \tilde{\rho}_0^{\frac{8}{3}-\alpha_6-\alpha_3}L^3.
  \end{equation}
\end{lemma}
\begin{proof}
  We will show that for any $N\in\mathbb{N}_{\geq0}$, we have
  \begin{equation}\label{bog bound V_4,4h}
    e^{-\tilde{B}}\mathcal{V}_{4,4h}e^{\tilde{B}}\leq C\mathcal{V}_{4,4H}
    +C\tilde{\rho}_0\mathcal{N}_h[\delta_4]+C\tilde{\rho}_0^
    {\frac{5}{3}-2\varepsilon}\mathcal{N}_{re}+C\tilde{\rho}_0^{\frac{8}{3}
    -\alpha_6-\alpha_3}L^3,
  \end{equation}
  where
  \begin{equation}\label{V_4,4H}
    \mathcal{V}_{4,4H}=\frac{1}{2L^3}\sum_{k,p,q,\sigma,\nu}
    \hat{v}_ka^*_{p-k,\sigma}a^*_{q+k,\nu}a_{q,\nu}a_{p,\sigma}
    \chi_{\vert p-k\vert,\vert q+k\vert,\vert p\vert,\vert q\vert>4k_F}.
  \end{equation}
  Now recall that $\tilde{G}_0=(\tr e^{-\beta(\mathcal{K}-\tilde{\mu}\mathcal{N})})^{-1}
  e^{-\beta(\mathcal{K}-\tilde{\mu}\mathcal{N})}$, since $\tilde{G}_0$ is quasi-free, translation-invariant, and preserves particle number, by Wick's theorem and the exact form of $\tilde{G}_0$, we have
  \begin{equation}\label{111}
    \frac{1}{L^3}\tr\mathcal{V}_{4,4H}\tilde{G}_0
    \leq\Big(\frac{C}{L^3}\tr\mathcal{N}_{>4k_F}\tilde{G}_0\Big)^2
    \leq \big(C\tilde{\rho}_0^{1+2\alpha_1}e^{-\tilde{\rho}_0^{-\alpha_1}}+C(\tilde{\rho}_0)L^{-1}
    \big)^2,
  \end{equation}
  where we let
  \begin{equation}\label{N>4kF}
    \mathcal{N}_{>4k_F}=\sum_{\vert k\vert>4k_F,\sigma}a_{k,\sigma}^*a_{k,\sigma},
  \end{equation}
  and similar to (\ref{aproribound number of particle outside}), we have
  \begin{equation}\label{aproribound N>4k_F}
    \frac{1}{L^3}\tr\mathcal{N}_{>4k_F}\tilde{G}_0\leq C\tilde{\rho}_0^{1+2\alpha_1}e^{-\tilde{\rho}_0^{-\alpha_1}}+C(\tilde{\rho}_0)L^{-1}.
  \end{equation}
  Also,
  \begin{equation}\label{333}
    \mathcal{N}_{re}=\mathcal{N}_l[\delta_4]+\mathcal{N}_h[\delta_4]
    \leq \mathcal{N}_h[\delta_4]+C\tilde{\rho}_0^{1+\delta_4}L^3.
  \end{equation}
  and by (\ref{aproribound number of particle outside}), we have
  \begin{equation}\label{222}
    \frac{1}{L^3}\tr\mathcal{N}_h[\delta_4]\tilde{G}_0
    \leq C\tilde{\rho}_0^{1+2\alpha_1-\delta_4}
    e^{-\tilde{\rho}_0^{-\alpha_1+\delta_4}}+C(\tilde{\rho}_0)L^{-1}.
  \end{equation}
  Combining (\ref{bog bound V_4,4h}), (\ref{111}), (\ref{333}) and (\ref{222}), we reach (\ref{bog control V_4,4h}).

  \par For the proof of (\ref{bog bound V_4,4h}), we first let in this proof
  \begin{align*}
    \mathbf{h}_{x,\sigma}(z)=\sum_{\vert k\vert>4k_F}\frac{e^{ikx}}{L^{\frac{3}{2}}}
    f_{k,\sigma}(z),\quad
    \mathbf{g}_{x,\sigma}(z)=\sum_{\vert k\vert\leq4k_F}\frac{e^{ikx}}{L^{\frac{3}{2}}}
    f_{k,\sigma}(z)\chi_{k\in P_{F,\delta_4}}.
  \end{align*}
  Similar to the proof of estimates (\ref{cub control V_4r}), we can divide according to how many indices are in far high frequency,
  \begin{align}\label{diff V_4 4H}
   \mathcal{V}_{4,4h}-\mathcal{V}_{4,4H}=\mathcal{V}_{4,3H}+\mathcal{V}_{4,2H}+\mathcal{V}_{4,1H}+
   \mathcal{V}_{4,0H}.
  \end{align}
  Their estimates are in this case trivial, for example
  \begin{equation}\label{V_4 1H}
    \mathcal{V}_{4,1H}=\sum_{\sigma,\nu}\int_{\Lambda_L^2}v(x-y)
    a^*(\mathbf{h}_{x,\sigma})a^*(\mathbf{g}_{y,\nu})
    a(\mathbf{g}_{y,\nu})a(\mathbf{g}_{x,\sigma})dxdy+h.c.
  \end{equation}
  By Cauchy-Schwartz inequality, and notice that $\mathcal{N}_{>4k_F}\leq \mathcal{N}_h[\delta_4]$,
  \begin{equation}\label{V_4 1H bound}
    \pm\mathcal{V}_{4,1H}\leq C\tilde{\rho}_0
    \mathcal{N}_h[\delta_4].
  \end{equation}
  Thus we present the bound to (\ref{diff V_4 4H}) directly:
  \begin{equation}\label{444}
    \pm\big(\mathcal{V}_{4,4h}-\mathcal{V}_{4,4H}\big)\leq \mathcal{V}_{4,4H}+C\tilde{\rho}_0
    \mathcal{N}_h[\delta_4].
  \end{equation}
  On the other hand, we have
  \begin{equation}\label{com V_4 eBtilde}
    [\mathcal{V}_{4,4H},\tilde{B}]=\tilde{\Theta}_1+\tilde{\Theta}_2
  \end{equation}
  with
  \begin{equation}\label{Thetatilde}
   \begin{aligned}
    \tilde{\Theta}_1&=\frac{1}{2L^3}\sum_{k,p,q,\sigma,\nu}\sum_{l,r,s}
    \hat{v}_{k}\xi_{l,s,r}^{\nu,\sigma}
    (a^*_{p-k,\sigma}a^*_{q+k,\nu}a_{s,\nu}a_{r,\sigma}+h.c.)\\
    &\quad\quad\times
    \chi_{\vert p-k\vert,\vert q+k\vert,\vert p\vert,\vert q\vert>4k_F}
    \chi_{r-l,s+l\notin B_F}
    \chi_{r,s\in B_F}
    \delta_{p,r-l}\delta_{q,s+l}
    \\
    \tilde{\Theta}_2&=-\frac{1}{L^3}\sum_{k,p,q,l,r,s,\sigma,\nu,\varpi}
    \hat{v}_{k}\xi_{l,s,r}^{\varpi,\sigma}
    (a^*_{p-k,\sigma}a^*_{q+k,\nu}a^*_{s+l,\varpi}a_{q,\nu}
    a_{s,\varpi}a_{r,\sigma}+h.c.)\\
    &\quad\quad\times
    \chi_{\vert p-k\vert,\vert q+k\vert,\vert p\vert,\vert q\vert>4k_F}
    \chi_{r-l,s+l\notin B_F}
    \chi_{r,s\in B_F}\delta_{p,r-l}
    \end{aligned}
  \end{equation}
  \par For $\tilde{\Theta}_1$, we use (\ref{xi1and 2}) to divide it by $\tilde{\Theta}_1=
  \tilde{\Theta}_{11}+\tilde{\Theta}_{12}$. We only need to present the proof of the bound to $\tilde{\Theta}_{11}$, as the bound to $\tilde{\Theta}_{12}$ can be proved similarly, with an even smaller upper bound. First notice that by definition (\ref{xi1and 2}), in $\tilde{\Theta}_{11}$, $\tilde{\zeta}^-(l)\neq 0$ implies $\vert r-l\vert,\vert s+l\vert\leq 3\tilde{\mu}^{\frac{1}{2}}\tilde{\rho}_0^{-\alpha_6}$, moreover, we let
  \begin{align*}
    \gamma^t_{x_1,x_3,\sigma}(z)=\sum_{q_3\in B_F}\Big(\frac{1}{L^3}\sum_
    {4k_F<\vert p_3\vert\leq 3\tilde{\mu}^{\frac{1}{2}}\tilde{\rho}_0^{-\alpha_6}}
    &W_{q_3-p_3}\tilde{\zeta}^-(q_3-p_3)e^{-ip_3(x_3-x_1)}e^{-p_3^2t}\Big)\\
    &\times
    \frac{e^{iq_3x_3}}{L^{\frac{3}{2}}}e^{q_3^2t}f_{q_3,\sigma}(z),
  \end{align*}
  and
  \begin{align*}
    G(x,t)=\frac{1}{L^3}\sum_{4k_F<\vert p_4\vert\leq 3\tilde{\mu}^{\frac{1}{2}}\tilde{\rho}_0^{-\alpha_6}}e^{ip_4x}e^{-p_4^2t},
  \end{align*}
  we then write
  \begin{equation}\label{Thetatilde_11}
    \begin{aligned}
    \tilde{\Theta}_{11}=C\sum_{\sigma,\nu}\int_{0}^{\infty}\int_{\Lambda_L^3}e^{-2\epsilon_0 t}
    v(x_1-x_2)&G(x_3-x_2,t)a^*(\mathbf{h}_{x_1,\sigma})a^*(\mathbf{h}_{x_2,\nu})\\
    &\times a(g_{x_3,\nu}^t)
    a(\gamma^t_{x_1,x_3,\sigma})dx_1dx_2dx_3+h.c.
    \end{aligned}
  \end{equation}
  By Cauchy-Schwartz inequality, for any $\psi\in\mathcal{H}^{\wedge N}$, we attain
  \begin{equation*}
    \vert\langle\tilde{\Theta}_{11}\psi,\psi\rangle\vert\leq CS_1^{\frac{1}{2}}S_2^{\frac{1}{2}},
  \end{equation*}
  where by (\ref{property gt norm}) and (\ref{est sum A_F d <=0}), we have
  \begin{align*}
    S_1&=\sum_{\sigma,\nu}\int_{0}^{\infty}\int_{\Lambda_L^3}e^{-2\epsilon_0 t}
    v(x_1-x_2)\vert G(x_3-x_2,t)\vert^2\Vert a^*(g_{x_3,\nu}^t)
     a(\mathbf{h}_{x_2,\nu})a(\mathbf{h}_{x_1,\sigma})\psi\Vert^2\\
     &\leq \int_{0}^{\infty}\frac{\tilde{\rho}_0}{L^3}\sum_{4k_F<\vert p_4\vert\leq 3\tilde{\mu}^{\frac{1}{2}}\tilde{\rho}_0^{-\alpha_6}}e^{-2(p_4^2-k_F^2+\epsilon_0)t}
     dt\langle\mathcal{V}_{4,4H}\psi,\psi\rangle\\
     &\leq \tilde{\rho}_0^{\frac{4}{3}-\alpha_6}\langle\mathcal{V}_{4,4H}\psi,\psi\rangle,
  \end{align*}
  and by (\ref{est Wk/p^2}), we deduce
  \begin{align*}
    S_2&=\sum_{\sigma,\nu}\int_{0}^{\infty}\int_{\Lambda_L^3}e^{-2\epsilon_0 t}
    v(x_1-x_2)\Vert a(\gamma^t_{x_1,x_3,\sigma})\psi\Vert^2\\
    &\leq L^3\int_{0}^{\infty}\sum_{\substack{q\in B_F\\ p\notin B_F}}\frac{\vert W_{q-p}\vert}{L^6}e^{-2(p^2-k_F^2+\epsilon_0)t}dt\leq C\tilde{\rho}_0^{\frac{4}{3}-\alpha_3}L^3.
  \end{align*}
  Therefore,
  \begin{equation}\label{tildeTheta_1bound}
    \pm\tilde{\Theta}_1\leq C\mathcal{V}_{4,4H}
    +C\tilde{\rho}_0^{\frac{8}{3}-\alpha_6-\alpha_3}L^3.
  \end{equation}
  \par For $\tilde{\Theta}_2$, we also rewrite it using (\ref{xi1and 2}) by $\tilde{\Theta}_2=\tilde{\Theta}_{21}+\tilde{\Theta}_{22}$. We are going to only bound $\tilde{\Theta}_{21}$, since $\tilde{\Theta}_{22}$ is even smaller. We define the smooth radially symmetric cut-off function $\zeta^-_0\in C^{\infty}(\mathbb{R}^3;[0,1])$:
  \begin{equation}\label{cutoff zeta0}
    \zeta_0^-(p)=\left\{
  \begin{aligned}
  &1,\quad \vert p\vert<2k_F\\
  &0,\quad \vert p\vert>3k_F
  \end{aligned}
  \right.
  \end{equation}
  and $\zeta_0^+=1-\zeta_0^-$. We notice the following facts
  \begin{equation}\label{facts Thetatilde2}
    \begin{aligned}
    &\tilde{\zeta}^-(l)\chi_{r\in B_F}\neq0\Rightarrow \vert r-l\vert< 3\tilde{\mu}^{\frac{1}{2}}\tilde{\rho}_0^{-\alpha_6}\Rightarrow
    \tilde{\zeta}^-(\frac{r-l}{2})\equiv 1\\
    &\delta_{p,r-l}\chi_{\vert p\vert>4k_F}\neq0\Rightarrow \zeta_0^+(r-l)\equiv1\\
    &\delta_{p,r-l}\chi_{\vert p\vert>4k_F}\chi_{r,s\in B_F}\neq0\Rightarrow
    \vert s+l\vert \geq2k_F\\
    &\tilde{\zeta}^-(l)\chi_{s\in B_F}\neq0\Rightarrow \vert s+l\vert\leq 3\tilde{\mu}^
    {\frac{1}{2}}\tilde{\rho}_0^{-\alpha_6}
    \end{aligned}
  \end{equation}
  Therefore, we let
  \begin{align*}
    E^t(x)=\frac{1}{L^3}\sum_{k\in (2\pi/L)\mathbb{Z}^3}e^{-ikx}e^{-k^2t}
    \zeta_0^+(k)\tilde{\zeta}^-(k),
  \end{align*}
  and
  \begin{align*}
    \mathbf{j}_{x,\sigma}^t(z)=\sum_{2k_F\leq \vert k\vert\leq 3\tilde{\mu}^
    {\frac{1}{2}}\tilde{\rho}_0^{-\alpha_6}}\frac{e^{ikx}}{L^{\frac{3}{2}}}
    e^{-k^2t}f_{k,\sigma}(z).
  \end{align*}
  From Lemma \ref{cutoff lemma}, we know that $\Vert E^t\Vert_1\leq C$. We rewrite $\tilde{\Theta}_{21}$ by
  \begin{equation}\label{tildeTheta_21}
    \begin{aligned}
    \tilde{\Theta}_{21}=&C\sum_{\sigma,\nu,\varpi}\int_{0}^{\infty}\int_{\Lambda_L^4}
    e^{-2\epsilon_0 t}
    v(x_1-x_3)W^{\tilde{\zeta}^-}(x_2-x_4)E^t(x_1-x_2)
    a^*(\mathbf{h}_{x_1,\sigma})a^*(\mathbf{h}_{x_3,\nu})
    \\
    &\times a^*(\mathbf{j}_{x_4,\varpi}^t)
    a(\mathbf{h}_{x_3,\nu})a(g_{x_4,\varpi}^t)a(g_{x_2,\sigma}^t)+h.c.
    \end{aligned}
  \end{equation}
  By Cauchy-Schwartz inequality, we have for $\psi\in\mathcal{H}^{\wedge N}$,
  \begin{equation*}
    \vert\langle\tilde{\Theta}_{21}\psi,\psi\rangle\vert\leq C\int_{0}^{\infty}
    e^{-2\epsilon_0 t}(S^t_1)^{\frac{1}{2}}(S^t_2)^{\frac{1}{2}}dt,
  \end{equation*}
  where
  \begin{align*}
     S^t_1&=\sum_{\sigma,\nu,\varpi}\int_{\Lambda_L^4}
      v(x_1-x_3)\vert W^{\tilde{\zeta}^-}(x_2-x_4)\vert \vert E^t(x_1-x_2)\vert
      \Vert a(\mathbf{h}_{x_3,\nu})a(\mathbf{h}_{x_1,\sigma})\psi\Vert^2\\
      &\leq C\langle\mathcal{V}_{4,4H}\psi,\psi\rangle,
  \end{align*}
  and
  \begin{align*}
   S^t_2&=\sum_{\sigma,\nu,\varpi}\int_{\Lambda_L^4}
      v(x_1-x_3)\vert W^{\tilde{\zeta}^-}(x_2-x_4)\vert \vert E^t(x_1-x_2)\vert
      \Vert a^*(\mathbf{j}_{x_4,\varpi}^t)\Vert^2\\
      &\quad\quad\quad\quad\quad\quad\times
    \Vert a(g_{x_4,\varpi}^t)\Vert^2\Vert a(g_{x_2,\sigma}^t)\Vert^2
    \Vert a(\mathbf{h}_{x_3,\nu})\psi\Vert^2\\
    &\leq C\tilde{\rho}_0e^{2k_F^2t}\Big(\sum_{2k_F\leq \vert p\vert\leq 3\tilde{\mu}^
    {\frac{1}{2}}\tilde{\rho}_0^{-\alpha_6}}\frac{e^{-2p^2t}}{L^3}\Big)
    \Big(\sum_{q\in B_F}\frac{e^{2q^2t}}{L^3}\Big)
    \langle\mathcal{N}_{>4k_F}\psi,\psi\rangle.
  \end{align*}
  On the other hand, similar to (\ref{est sum A_F d <=0}), we have
  \begin{align*}
    \sum_{2k_F\leq \vert p\vert\leq 3\tilde{\mu}^
    {\frac{1}{2}}\tilde{\rho}_0^{-\alpha_6}}\int_{0}^{\infty}
    e^{-2(p^2-2k_F^2+\epsilon_0)t}dt\leq \sum_{2k_F\leq \vert p\vert\leq 3\tilde{\mu}^
    {\frac{1}{2}}\tilde{\rho}_0^{-\alpha_6}}
    \frac{C}{p^2-2k_F^2+\epsilon_0}\leq C\tilde{\rho}_0^{\frac{1}{3}-\alpha_6}
  \end{align*}
  Therefore using Cauchy-Schwartz inequality and (\ref{est sum}), we have
  \begin{equation}\label{tildeTheta_2bound}
    \pm\tilde{\Theta}_2\leq C\mathcal{V}_{4,4H}
    +C\tilde{\rho}_0^{\frac{5}{3}-\alpha_6-\varepsilon}\mathcal{N}_{\geq 4k_F}.
  \end{equation}
  Combining (\ref{tildeTheta_1bound}), (\ref{tildeTheta_2bound}) and the fact that $\mathcal{N}_{>4k_F}\leq \mathcal{N}_h[\delta_4]$, together with Lemma \ref{lemma bog control N_re} and Gronwall's inequality, we obtain
  \begin{equation}\label{bog con V_4 4H}
    e^{-\tilde{B}}\mathcal{V}_{4,4H}e^{\tilde{B}}\leq C\mathcal{V}_{4,4H}
    +C\tilde{\rho}_0^{\frac{5}{3}-\alpha_6-\varepsilon}
    \mathcal{N}_h[\delta_4]+C\tilde{\rho}_0^
    {\frac{7}{3}-\alpha_6-3\varepsilon}\mathcal{N}_{re}+C\tilde{\rho}_0^{\frac{8}{3}
    -\alpha_6-\alpha_3}L^3.
  \end{equation}
  Combining (\ref{bog con V_4 4H}) with (\ref{444}) and Lemma \ref{lemma bog control N_re}, we reach (\ref{bog bound V_4,4h}).
\end{proof}

\begin{lemma}\label{lemma bog control error}
For any $N\in\mathbb{N}_{\geq0}$, and $\frac{1}{24}>\frac{1}{2}\delta_3>\frac{1}{4}\alpha_3>\alpha_2>\alpha_5>2\alpha_4>2\varepsilon >0$, $\frac{1}{3}>\delta_4\geq\delta_1\geq\delta_2>\frac{1}{12}$, $\delta_4>\alpha_6+\varepsilon>2\alpha_3+2\alpha_4+3\varepsilon$, $\beta_1=\frac{1}{3}+\alpha_5$, we have

   \begin{align}
   \pm& e^{-\tilde{B}}(\mathcal{E}_{\mathcal{J}_N}+R)e^{\tilde{B}}\nonumber\\
   &\leq
   C\big(\tilde{\rho}_0^{1+\delta_2}+\tilde{\rho}_0^{\frac{4}{3}+\alpha_2-3\delta_3-\alpha_4}
   +\tilde{\rho}_0^{\frac{5}{3}+\alpha_5-4\alpha_3-\alpha_4-2\varepsilon}
   +\tilde{\rho}_0^{\frac{4}{3}-4\alpha_3+\alpha_6-\alpha_4-2\varepsilon}\big)
   \mathcal{N}_{re}\nonumber\\
  & +C\big(\tilde{\rho}_0^{1+\delta_2}
  +\tilde{\rho}_0^{\frac{4}{3}-\alpha_3}\big)\tilde{\mathcal{N}}_{re}
  +C\tilde{\rho}_0^{1-3\alpha_2}\mathcal{N}_i[\delta_4]
  +C\tilde{\rho}_0^{1-\delta_2-3\alpha_2}\mathcal{N}_h[\delta_4]\nonumber\\
  & +C\tilde{\rho}_0^{\frac{2}{3}-\alpha_4}\mathcal{N}_h[-\delta_3]
  +C\big(\tilde{\rho}_0^{\frac{2}{3}+\delta_4-\alpha_2-\alpha_3}
  +\tilde{\rho}_0^{\frac{4}{3}+\alpha_2-\alpha_4-3\alpha_3}+
  \tilde{\rho}_0^{\frac{5}{6}+\frac{\alpha_2}{2}-\frac{3}{2}\alpha_5}\big)
  \mathcal{N}_h[-\alpha_2]\nonumber\\
  &+C\tilde{\rho}_0^{\frac{2}{3}-4\alpha_3+\alpha_6-\alpha_4}\mathcal{N}_h
    [-\alpha_6]+C\tilde{\rho}_0^{1+\alpha_5
  -4\alpha_3-\alpha_4}\mathcal{N}_h[-\beta_1]\nonumber\\
  &+C\big(\tilde{\rho}_0^{\frac{1}{3}+2\alpha_3}+\tilde{\rho}_0^{\frac{1}{3}+
  \alpha_5-\alpha_4}\big)\mathcal{K}_s
  +C\big(\tilde{\rho}_0^{\delta_4-\alpha_3-\alpha_4}+\tilde{\rho}_0^
  {\alpha_6-\alpha_4}\big)\mathcal{K}_h[-\alpha_2]
  \nonumber\\
  &+C\big(\tilde{\rho}_0^{2+4\alpha_3-7\alpha_2}
  +\tilde{\rho}_0^{\frac{5}{2}+\frac{\alpha_2}{2}+\alpha_3-\frac{3}{2}\alpha_5}
  +\tilde{\rho}_0^{\frac{8}{3}+\alpha_5-3\alpha_3-\alpha_4}\big)L^3\nonumber\\
  &+C\big(\tilde{\rho}_0^{\frac{7}{3}+\alpha_4}+\tilde{\rho}_0^{2+2\delta_4-\alpha_4}+
  \tilde{\rho}_0^{\frac{8}{3}-\alpha_3-\delta_2-3\alpha_2-\varepsilon}
  +\tilde{\rho}_0^{\frac{7}{3}+2\delta_3-\alpha_3-\alpha_4}\big)L^3\nonumber\\
  &+Ce^{-\tilde{B}}\tilde{\rho}_0^{\alpha_4}\mathcal{V}_{4,4h}e^{\tilde{B}}.\label{bog control error}
   \end{align}

\end{lemma}
\begin{proof}
  \par We let $R=R_1+R_2$, where $R_1=\mathcal{V}_{21}^\prime-\tilde{\mathcal{V}}_{21}^\prime$. Notice that $\tilde{\zeta}^{+}\neq0$ implies $p-k,q+k\in P_{F,-\alpha_6}$. We have,
  \begin{equation}\label{R1}
    \begin{aligned}
    R_1=\sum_{\sigma,\nu}\int_{\Lambda_L^2}&W_{\tilde{\zeta^+}}^{\zeta^-}(x-y)
    a^*(H_{x,\sigma}[-\alpha_6])a^*(H_{y,\nu}[-\alpha_6])\\
    &\times a(g_{y,\nu})a(g_{x,\sigma})dxdy+h.c.
    \end{aligned}
  \end{equation}
  Using (\ref{special zeta tilde eta W L2 norm}) and Lemma \ref{b^* bound by b general}, we can easily bound
  \begin{equation}\label{R1bound}
    \pm R_1\leq C\tilde{\rho}_0^{\frac{2}{3}-4\alpha_3+\alpha_6-\alpha_4}\mathcal{N}_h
    [-\alpha_6]+C\tilde{\rho}_0^{\frac{7}{3}+\alpha_4}L^3.
  \end{equation}
  For $R_2$, we can bound it similar to $R_1$, with an additional step of integration by parts and using the fact $\alpha_6>2\alpha_3>\alpha_2$, which implies $\mathcal{K}_h[-\alpha_6]
  \leq\mathcal{K}_h[-\alpha_2]$. That is,
  \begin{equation}\label{R2bound}
    \pm R_2\leq C\tilde{\rho}_0^{\frac{2}{3}+\alpha_6-\alpha_4}\mathcal{N}_h
    [-\alpha_6]+C\tilde{\rho}_0^{\alpha_6-\alpha_4}\mathcal{K}_h
    [-\alpha_2]
    +C\tilde{\rho}_0^{\frac{7}{3}+\alpha_4}L^3.
  \end{equation}
  Then (\ref{bog control   error}) follows from (\ref{cub Error}), (\ref{R1bound}), (\ref{R2bound}) and Lemmas \ref{lemma bog control N_re}-\ref{lemma bog control K}.
\end{proof}

\begin{lemma}\label{cal com bog lemma}
For any $N\in\mathbb{N}_{\geq0}$, and $\frac{1}{24}>\frac{1}{2}\delta_3>\frac{1}{4}\alpha_3>\alpha_2>\alpha_5>2\alpha_4>2\varepsilon >0$, $\frac{1}{3}>\delta_4\geq\delta_1\geq\delta_2>\frac{1}{12}$, $\delta_4>\alpha_6+\varepsilon>2\alpha_3+2\alpha_4+3\varepsilon$, $\beta_1=\frac{1}{3}+\alpha_5$, we have
  \begin{equation}\label{cal com bog}
  \begin{aligned}
    &\int_{0}^{1}\int_{t}^{1}e^{-s\tilde{B}}[\tilde{\mathcal{V}}_{21}^\prime
  +\tilde{\Omega},\tilde{B}]e^{s\tilde{B}}dsdt\\
  &=\sum_{k,p,q,\sigma,\nu}\big(\frac{W_k}{L^3}\tilde{\zeta}^-(k)
  +\eta_kk(q-p)\phi^+(k)\tilde{\zeta}^-(k)\big)\xi_{k,q,p}^{\nu,\sigma}
  \chi_{p-k,q+k\notin B_F}
  \chi_{p,q\in B_F}\\
  &-\sum_{k,p,q,\sigma}\big(\frac{W_k}{L^3}\tilde{\zeta}^-(k)
  +\eta_kk(q-p)\phi^+(k)\tilde{\zeta}^-(k)\big)\xi_{(k+q-p),p,q}^{\sigma,\sigma}
  \chi_{p-k,q+k\notin B_F}\chi_{p,q\in B_F}\chi_{p\neq q}\\
  &+\mathcal{E}_{[\mathcal{V}_{21}^\prime
  +\Omega,\tilde{B}]}
  \end{aligned}
  \end{equation}
  where
  \begin{equation}\label{bound E com bog}
    \begin{aligned}
    \pm\mathcal{E}_{[\mathcal{V}_{21}^\prime
  +\Omega,\tilde{B}]}\leq& C\big(\tilde{\rho}_0^{1+\delta_2}+
  \tilde{\rho}_0^{\frac{4}{3}-\alpha_3}
  +\tilde{\rho}_0^{\frac{7}{3}-4\alpha_6-\delta_2-3\varepsilon}\big)\mathcal{N}_{re}
  +C\tilde{\rho}_0^{\frac{4}{3}-\alpha_3}\tilde{\mathcal{N}}_{re}\\
  &+C\tilde{\rho}_0^{\frac{5}{3}-4\alpha_6-\delta_2-\varepsilon}
  \mathcal{N}_h[\delta_4]\\
  &+C\big(\tilde{\rho}_0^{\frac{8}{3}-\alpha_3+\delta_2-\varepsilon}
  +\tilde{\rho}_0^{\frac{9}{3}-2\alpha_3-\varepsilon}
  +\tilde{\rho}_0^{\frac{10}{3}-4\alpha_6-\alpha_3-\delta_2-\varepsilon}\big)L^3.
   \end{aligned}
  \end{equation}
  for $\delta_4>\alpha_6>2\alpha_3$
\end{lemma}
\begin{proof}
  \par We calculate $[\tilde{\mathcal{V}}_{21}^\prime,\tilde{B}]$, and $[\tilde{\Omega},\tilde{B}]$ follows similarly. For $\tilde{\mathcal{V}}_{21}^\prime$, we have
   \begin{equation}\label{Phi}
   [\tilde{\mathcal{V}}_{21}^\prime,\tilde{B}]=\sum_{j=1}^{4}\Phi_j
  \end{equation}
  with
  \begin{align}
  &\Phi_1=\frac{1}{L^3}\sum_{\substack{\sigma,\nu\\k,p,q}}\sum_{l,r,s}W_k\xi_{l,s,r}^{\nu,\tau}
    \delta_{p-k,r-l}\delta_{q+k,s+l}a^*_{p,\sigma}a^*_{q,\nu}a_{s,\nu}a_{r,\sigma}\nonumber\\
    &\quad\quad\quad\quad\times\tilde{\zeta}^-(k)
    \chi_{p-k,q+k,r-l,s+l\notin B_F}\chi_{p,q,r,s\in B_F}+h.c.\nonumber\\
    &\Phi_2=\frac{1}{L^3}\sum_{\substack{\sigma,\nu\\k,p,q}}
    \sum_{\tau,l,r,s}2W_k\xi_{l,s,r}^{\sigma,\tau}
    \delta_{p-k,s+l}a^*_{p,\sigma}a^*_{q,\nu}a^*_{r-l,\tau}
    a_{q+k,\nu}a_{s,\sigma}a_{r,\tau}\nonumber\\
    &\quad\quad\quad\quad\times\tilde{\zeta}^-(k)
    \chi_{p-k,q+k,r-l,s+l\notin B_F}\chi_{p,q,r,s\in B_F}+h.c.\nonumber\\
    &\Phi_3=\frac{1}{L^3}\sum_{\substack{\sigma,\nu\\k,p,q}}\sum_{l,r,s}W_k\xi_{l,s,r}^{\nu,\sigma}
    \delta_{p,r}\delta_{q,s}a^*_{r-l,\sigma}a^*_{s+l,\nu}
    a_{q+k,\nu}a_{p-k,\sigma}\nonumber\\
    &\quad\quad\quad\quad\times\tilde{\zeta}^-(k)
    \chi_{p-k,q+k,r-l,s+l\notin B_F}\chi_{p,q,r,s\in B_F}+h.c.\nonumber\\
    &\Phi_4=\frac{1}{L^3}\sum_{\substack{\sigma,\nu\\k,p,q}}\sum_{\varpi,l,r,s}
    -2W_k\xi_{l,s,r}^{\varpi,\sigma}
    \delta_{p,r}a^*_{r-l,\sigma}a^*_{s+l,\varpi}a_{s,\varpi}a^*_{q,\nu}
    a_{q+k,\nu}a_{p-k,\sigma}\nonumber\\
    &\quad\quad\quad\quad\times\tilde{\zeta}^-(k)
    \chi_{p-k,q+k,r-l,s+l\notin B_F}\chi_{p,q,r,s\in B_F}+h.c.\label{Phi detailed}
  \end{align}
  Similarly, we have $[\Omega,\tilde{B}]=\sum_{j=1}^{4}\tilde{\Phi}_j$, with $L^{-3}W_k\tilde{\zeta}^-(k)$ in $\Phi_j$ being replaced by $\eta_kk(q-p)\phi^+(k)\tilde{\zeta}^-(k)$. The analysis to $\tilde{\Phi}_j$ is analogous to $\Phi_j$, we thus omit this part.

  \par For $\Phi_1$, the process is analogous to the $\Phi_1$ in \cite[Lemma 9.5]{WJH}. We sketch the main estimate. We use $1=\chi_{k=l}+\chi_{k\neq l}$ to decompose $\Phi_1=\Phi_{11}+\Phi_{12}$. For $\Phi_{11}$,
  \begin{equation}\label{Phi_11 result}
    \begin{aligned}
    \Phi_{11}=2\sum_{k,p,q,\sigma,\nu}\frac{W_k}{L^3}\tilde{\zeta}^-(k)\xi_{k,q,p}^{\nu,\sigma}
  \chi_{p-k,q+k\notin B_F}\chi_{p,q\in B_F}+\mathcal{E}_{\Phi_{11}},
    \end{aligned}
  \end{equation}
  with
  \begin{equation}\label{E_Phi_11 bound}
    \pm\mathcal{E}_{\Phi_{11}}\leq C\tilde{\rho}_0^{\frac{4}{3}-\alpha_3}\tilde{\mathcal{N}}
    _{re}+C(\tilde{\rho}_0).
  \end{equation}
  Here we use the following estimate deduced from (\ref{est Wk/p^2}), (\ref{est etakk/p^2}) and the fact that $\vert W_k\vert\leq C$, and for $p,q\in B_F$, $L^3\vert\eta_kk(q-p)\vert\leq C\tilde{\rho}_0^{\alpha_3}$:
  \begin{equation}\label{Wxi}
    \sum_{k,q}\Big(\frac{\vert W_k\vert}{L^3}
    +\vert\eta_kk(q-p\vert)\Big)\vert\xi_{k,q,p}^{\nu,\sigma}\vert\leq  C\tilde{\rho}_0^{\frac{4}{3}-\alpha_3}.
  \end{equation}
   \par For $\Phi_{12}$, we follow the steps in \cite[Lemma 9.5]{WJH} to reach
   \begin{equation}\label{Phi_12 result}
     \begin{aligned}
     \Phi_{12}=-2\sum_{k,p,q,\sigma}\frac{W_k}{L^3}\tilde{\zeta}^-(k)
     \xi_{(k+q-p),p,q}^{\sigma,\sigma}
  \chi_{p-k,q+k\notin B_F}\chi_{p,q\in B_F}\chi_{p\neq q}+\mathcal{E}_{\Phi_{12}},
     \end{aligned}
   \end{equation}
   with
  \begin{equation}\label{E_Phi_12 bound}
    \pm\mathcal{E}_{\Phi_{12}}\leq C\tilde{\rho}_0^{\frac{4}{3}-\alpha_3}\tilde{\mathcal{N}}
    _{re}.
  \end{equation}
  The only difference is the bound to $\Phi_{1231}$ (see the correspondence in \cite[(9.45)]{WJH}):
  \begin{equation}\label{Phi1231}
  \begin{aligned}
    \Phi_{1231}=\sum_{\substack{k,p,q\\\sigma,\nu}}\sum_{l,r,s}&\frac{W_k\tilde{\zeta}^-(k)}{L^3}
    \xi_{l,s,r}^{\nu,\sigma}a_{r,\sigma}a_{s,\nu}a_{q,\nu}^*a_{p,\sigma}^*\\
    &\times \chi_{p-k,q+k,r-l,s+l\notin B_F}\chi_{p,q,r,s\in B_F}\delta_{p-k,r-l}
    \delta_{q+k,s+l}+h.c.
  \end{aligned}
  \end{equation}
   We again use (\ref{xi1and 2}) to split $\Phi_{1231}=\Phi_{1231}^{(1)}+
  \Phi_{1231}^{(2)}$, and present the estimate to $\Phi_{1231}^{(1)}$. $\Phi_{1231}^{(2)}$ can be bounded similarly, with a smaller upper bound. We write
  \begin{equation}\label{Phi1231 1}
    \begin{aligned}
    \Phi_{1231}^{(1)}&=C\sum_{\sigma,\nu}\int_{0}^{\infty}\int_{\Lambda_L^2}
    e^{-2\epsilon_0 t}a(g_{y,\sigma}^t)a(\gamma_{x,y,\nu}^t)a^*(g_{x,\nu}^t)
    a^*(\tilde{\gamma}_{x,y,\sigma}^t)dxdydt+h.c.
    \end{aligned}
  \end{equation}
  where we let, in this proof
  \begin{align*}
    \gamma_{x,y,\nu}^t(z)=\sum_{q\in B_F}
    \Big(\frac{1}{L^3}\sum_{p\notin B_F}W_{p-q}\tilde{\zeta}^-(p-q)e^{-ip(y-x)}e^{-p^2t}\Big)
    \frac{e^{iqy}}{L^{\frac{3}{2}}}e^{q^2t}f_{q,\nu}(z),
  \end{align*}
  and
  \begin{align*}
   \tilde{\gamma}_{x,y,\sigma}^t(z)=\sum_{q\in B_F}
    \Big(\frac{1}{L^3}\sum_{p\notin B_F}W_{p-q}\tilde{\zeta}^-(p-q)e^{ip(y-x)}e^{-p^2t}\Big)
    \frac{e^{iqy}}{L^{\frac{3}{2}}}f_{q,\sigma}(z).
  \end{align*}
  Then for any $\psi\in\mathcal{H}^{\wedge N}$,
  \begin{align*}
    \vert\langle\Phi_{1231}^{(1)}\psi,\psi,\rangle\vert&\leq C
    \Big(\sum_{\sigma,\nu}\int_{0}^{\infty}\int_{\Lambda_L^2}
    e^{-2\epsilon_0 t}\Vert a^*(\gamma_{x,y,\nu}^t)\psi\Vert^2\Big)^{\frac{1}{2}}\\
    &\times\Big(\sum_{\sigma,\nu}\int_{0}^{\infty}\int_{\Lambda_L^2}
    e^{-2\epsilon_0 t}\Vert a(g_{y,\sigma}^t)a^*(g_{x,\nu}^t)
    a^*(\tilde{\gamma}_{x,y,\sigma}^t)\psi\Vert^2\Big)^{\frac{1}{2}}
  \end{align*}
  Similar to the bound of $\tilde{\Theta}_{11}$ in (\ref{Thetatilde_11}), we can bound
  \begin{equation}\label{Phi1231 1 bound}
    \pm\Phi_{1231}^{(1)}\leq C\tilde{\rho}_0^{\frac{4}{3}-\alpha_3}\tilde{\mathcal{N}}_{re}.
  \end{equation}

  \par For $\Phi_2$, we use
  \begin{equation*}
    \begin{aligned}
    &a^*_{p,\sigma}a^*_{q,\nu}a^*_{r-l,\tau}
    a_{q+k,\nu}a_{s,\sigma}a_{r,\tau}\\
    =&a^*_{r-l,\tau}\big(
    a_{s,\sigma}a_{r,\tau}a^*_{p,\sigma}a^*_{q,\nu}
    +\delta_{q,s}\delta_{\sigma,\nu}a_{p,\sigma}^*a_{r,\tau}
    -\delta_{p,r}\delta_{\sigma,\tau}a_{s,\sigma}a^*_{q,\nu}\\
    &+\delta_{p,s}a_{r,\tau}a^*_{q,\nu}+\delta_{q,r}\delta_{\nu,\tau}a_{s,\sigma}a^*_{p,\sigma}
    -\delta_{p,s}\delta_{q,r}\delta_{\nu,\tau}\big)a_{q+k,\nu}
    \end{aligned}
  \end{equation*}
  to rewrite $\Phi_{2}=\sum_{j=1}^{6}\Phi_{2j}$. For each $\Phi_{2j}$, we also use (\ref{xi1and 2}) to rewrite them as $\Phi_{2j}=\Phi_{2j}^{(1)}+\Phi_{2j}^{(2)}$. We present only the estimate to $\Phi_{2j}^{(1)}$, since $\Phi_{2j}^{(2)}$ are even smaller. For $\Phi_{21}^{(1)}$, we have
  \begin{equation}\label{Phi 21}
    \begin{aligned}
    \Phi_{21}^{(1)}=\frac{C}{L^3}\sum_{\sigma,\nu\,\tau}\sum_{p_1\notin B_F}
    \int_{0}^{\infty}e^{-2\epsilon_0 t}e^{-p_1^2t}
    \mathcal{A}_{p_1,\sigma,\nu}(t)\mathcal{B}_{p_1,\sigma,\tau}dt+h.c.
    \end{aligned}
  \end{equation}
  where
  \begin{align*}
    \mathcal{A}_{p_1,\sigma,\nu}(t)=&\int_{\Lambda_L^2}e^{ip_1x_4}
    W^{\tilde{\zeta}^-}(x_2-x_4)
    a^*(h_{x_2,\tau}^t)
    a(g_{x_2,\tau}^t)a(g_{x_4,\sigma}^t)dx_2dx_4\\
    \mathcal{B}_{p_1,\sigma,\tau}=&\int_{\Lambda_L^2}e^{-ip_1x_1}
    W^{\tilde{\zeta}^-}(x_1-x_3)
    a^*(g_{x_1,\sigma})a^*(g_{x_3,\nu})a(h_{x_3,\nu})dx_1dx_3
  \end{align*}
  Then for any $\psi\in\mathcal{H}^{\wedge N}$, similar to (\ref{typical example E_V_3,l,12}), with additional help of Lemma \ref{property htgt} and Lemma \ref{est sum htgt}, we have
  \begin{align*}
    \vert\langle\Phi_{21}^{(1)}\psi,\psi\rangle\vert&\leq \int_{0}^{\infty}e^{-2\epsilon_0 t}e^{-k_F^2t}\Big(\frac{1}{L^3}\sum_{\substack{\sigma,\nu,\tau\\ p_1}}
    \Vert\mathcal{A}^*_{p_1,\sigma,\nu}(t)\psi\Vert^2\Big)^{\frac{1}{2}}
    \Big(\frac{1}{L^3}\sum_{\substack{\sigma,\nu,\tau\\ p_1}}
    \Vert\mathcal{B}_{p_1,\sigma,\tau}\psi\Vert^2\Big)^{\frac{1}{2}}\\
    &\leq C\tilde{\rho}_0^{\frac{4}{3}-\varepsilon}\langle{\mathcal{N}}_{re}\psi,
    \psi\rangle.
  \end{align*}
  Thus,
  \begin{equation}\label{Phi_21 bound}
    \pm\Phi_{21}\leq C\tilde{\rho}_0^{\frac{4}{3}-\varepsilon}{\mathcal{N}}_{re}.
  \end{equation}
  $\Phi_{24}$ can be bounded in a way  similar to $\Phi_{22}$, and also satisfies the estimate (\ref{Phi_21 bound}).
  \par For $\Phi_{22}$, we have
  \begin{equation}\label{Phi 22}
    \begin{aligned}
    \Phi_{22}^{(1)}=C\sum_{\sigma,\tau}\int_{0}^{\infty}\int_{\Lambda_L^3}
    &e^{-2\epsilon_0 t}W^{\tilde{\zeta}^-}(x_2-x_4)M^t(x_4-x_3)
    a^*(h^t_{x_2,\tau})a^*(\gamma^t_{x_3,x_4,\sigma})\\
    &\times a(g^t_{x_2,\tau})a(h^t_{x_3,\sigma})dx_2dx_3dx_4dt+h.c.
    \end{aligned}
  \end{equation}
  where we let, in this proof
  \begin{equation*}
    M^t(x)=\frac{1}{L^3}\sum_{q\in B_F}e^{-iqx}e^{q^2t}.
  \end{equation*}
  Then for any $\psi\in\mathcal{H}^{\wedge N}$, by Lemma \ref{est xi/ p^2 lemma} and Lemma \ref{est sum htgt},
  \begin{align*}
    \vert\langle\Phi_{22}^{(1)}\psi,\psi\rangle\vert&\leq C
     \Big(\sum_{\sigma,\tau}\int_{0}^{\infty}\int_{\Lambda_L^3}
    e^{-2\epsilon_0 t}\vert W^{\tilde{\zeta}^-}(x_2-x_4)\vert
    \vert M^t(x_4-x_3)\vert^2\Vert a(h^t_{x_2,\tau})\psi\Vert^2\Big)^{\frac{1}{2}}\\
    &\times \Big(\sum_{\sigma,\tau}\int_{0}^{\infty}\int_{\Lambda_L^3}
    e^{-2\epsilon_0 t}\vert W^{\tilde{\zeta}^-}(x_2-x_4)\Vert a^*(\gamma^t_{x_3,x_4,\sigma})
    a(g^t_{x_2,\tau})a(h^t_{x_3,\sigma})\psi\Vert^2\Big)^{\frac{1}{2}}\\
    &\leq C\tilde{\rho}_0^{\frac{4}{3}-\frac{\varepsilon}{2}
    -\frac{\alpha_3}{2}}\langle{\mathcal{N}}_{re}\psi,\psi\rangle.
  \end{align*}
  Therefore
  \begin{equation}\label{Phi_22 bound}
    \pm\Phi_{22}\leq C\tilde{\rho}_0^{\frac{4}{3}-\frac{\varepsilon}{2}
    -\frac{\alpha_3}{2}}{\mathcal{N}}_{re}.
  \end{equation}
  $\Phi_{23}$ and $\Phi_{25}$ can be bounded in a similar way like $\Phi_{22}$, and also satisfy the estimate (\ref{Phi_22 bound}). For $\Phi_{26}$, we can bound it easily by Lemma \ref{est sum htgt}, and the fact that $\vert W_k\vert\leq C$, and for $p,q\in B_F$, $L^3\vert\eta_kk(q-p)\vert\leq C\tilde{\rho}_0^{\alpha_3}$:
  \begin{equation}\label{Phi_26 bound}
    \pm\Phi_{26}\leq C\tilde{\rho}_0^{\frac{4}{3}-\varepsilon}{\mathcal{N}}_{re}.
  \end{equation}
  Therefore
  \begin{equation}\label{Phi_2 bound}
    \pm\Phi_{2}\leq C\tilde{\rho}_0^{\frac{4}{3}-\frac{\varepsilon}{2}
    -\frac{\alpha_3}{2}}{\mathcal{N}}_{re}.
  \end{equation}

  \par For $\Phi_3$, we again use (\ref{xi1and 2}) to rewrite it by $\Phi_{3}=\Phi_{3}^{(1)}+\Phi_{3}^{(2)}$. As usual, we present the estimate to $\Phi_{3}^{(1)}$ only. Notice that $\tilde{\zeta}^-(l)\chi_{r,s\in B_F}\neq0$ implies $\vert r-l\vert,\vert s+l\vert \leq 3\tilde{\mu}\tilde{\rho}_0^{-\alpha_6}$, similarly, we have $\vert p-k\vert,\vert q+k\vert \leq 3\tilde{\mu}^{\frac{1}{2}}\tilde{\rho}_0^{-\alpha_6}$
  Then similar to $\Phi_{1231}^{(1)}$ as aforementioned, we write $\Phi_{3}^{(1)}$ by
  \begin{equation}\label{Phi 3 1}
    \Phi_{3}^{(1)}=C\sum_{\sigma,\nu}\int_{0}^{\infty}\int_{\Lambda_L^2}e^{-2\epsilon_0 t}
    a^*(\mathbf{h}_{y,\sigma}^t)a^*(\omega^t_{x,y,\nu})a(\mathbf{h}_{x,\nu})
    a(\tilde{\omega}^t_{x,y,\sigma})dxdydt+h.c.
  \end{equation}
  where we let
  \begin{align*}
    \mathbf{h}_{y,\sigma}^t(z)=\sum_{k_F<\vert k\vert<3\tilde{\mu}^{\frac{1}{2}}\tilde{\rho}_0^{-\alpha_6}}\frac{e^{iky}}{L^{\frac{3}{2}}}
    f_{k,\sigma}(z)e^{-k^2t},\quad
    \mathbf{h}_{x,\nu}(z)=\sum_{k_F<\vert k\vert<3\tilde{\mu}^{\frac{1}{2}}\tilde{\rho}_0^{-\alpha_6}}\frac{e^{ikx}}{L^{\frac{3}{2}}}
    f_{k,\nu}(z)
  \end{align*}
  and
  \begin{align*}
   \omega^t_{x,y,\nu}(z)=\sum_{k_F<\vert p\vert<3\tilde{\mu}^{\frac{1}{2}}\tilde{\rho}_0^{-\alpha_6}}\Big(
   \frac{1}{L^3}\sum_{q\in B_F}W_{p-q}\tilde{\zeta}^-(p-q)e^{-iq(y-x)}e^{q^2t}\Big)
   \frac{e^{ipy}}{L^{\frac{3}{2}}}e^{-p^2t}f_{p,\nu},
  \end{align*}
  and
  \begin{align*}
   \tilde{\omega}^t_{x,y,\nu}(z)=\sum_{k_F<\vert p\vert<3\tilde{\mu}^{\frac{1}{2}}\tilde{\rho}_0^{-\alpha_6}}\Big(
   \frac{1}{L^3}\sum_{q\in B_F}W_{p-q}\tilde{\zeta}^-(p-q)e^{iq(y-x)}e^{q^2t}\Big)
   \frac{e^{ipy}}{L^{\frac{3}{2}}}f_{p,\nu}.
  \end{align*}
  Using $\mathbf{h}_{x,\nu}=a(L_{x,\nu}[\delta_4])+(\mathbf{h}_{x,\nu}-a(L_{x,\nu}[\delta_4]))$, we furthermore divide it by $\Phi_{3}^{(1)}=\Phi_{31}^{(1)}+\Phi_{32}^{(1)}$. We can bound them in a way similar to $\Phi_{1231}^{(1)}$ in (\ref{Phi1231 1}), but we will additionally use the estimate (\ref{est sum A_F d <=0}). That is, for $\Phi_{32}^{(1)}$,
  \begin{equation}\label{Phi 32 bound}
    \pm\Phi_{32}^{(1)}\leq C\tilde{\rho}_0^{1+\delta_2}{\mathcal{N}}_{re}
    +C\tilde{\rho}_0^{\frac{5}{3}-4\alpha_6-\delta_2-\varepsilon}\mathcal{N}_{h}[\delta_4].
  \end{equation}
  For $\Phi_{31}^{(1)}$,
   \begin{equation}\label{Phi 31 bound}
    \pm\Phi_{31}^{(1)}\leq C\tilde{\rho}_0^{\frac{4}{3}+\frac{\delta_4}{2}-\frac{\alpha_6}{2}-\frac{\varepsilon}{2}}
    {\mathcal{N}}_{re}.
  \end{equation}
  Therefore,
   \begin{equation}\label{Phi_3 bound}
    \pm\Phi_{3}\leq C\tilde{\rho}_0^{1+\delta_2}{\mathcal{N}}_{re}
    +C\tilde{\rho}_0^{\frac{5}{3}-4\alpha_6-\delta_2-\varepsilon}\mathcal{N}_{h}[\delta_4].
  \end{equation}
  \par For $\Phi_4$, we again use (\ref{xi1and 2}) to rewrite it by $\Phi_{4}=\Phi_{4}^{(1)}+\Phi_{4}^{(2)}$. As usual, we present the estimate to $\Phi_{4}^{(1)}$. Using the observations and notations above, we can write
  \begin{equation}\label{Phi_4 1}
    \Phi_4^{(1)}=\frac{C}{L^3}\sum_{\sigma,\nu,\varpi}\sum_{q_1\in B_F}\int_{0}^{\infty}
    e^{-2\epsilon_0 t}e^{q_1^2t}\mathcal{A}_{q_1,\sigma,\varpi}(t)
    \mathcal{B}_{q_1,\sigma,\nu}+h.c.
  \end{equation}
  where
  \begin{align*}
    \mathcal{A}_{q_1,\sigma,\varpi}(t)&=\int_{\Lambda_L^2}e^{-iq_1x_2}
    W^{\tilde{\zeta}^-}(x_2-x_4)
    a^*(\mathbf{h}^t_{x_2,\sigma})a^*(\mathbf{h}^t_{x_4,\varpi})
    a(\mathbf{g}_{x_4,\varpi}^t)dx_2dx_4\\
    \mathcal{B}_{q_1,\sigma,\nu}&=\int_{\Lambda_L^2}e^{iq_1x_1}
    W^{\tilde{\zeta}^-}(x_1-x_3)a^*(g_{x_3,\nu})a(h_{x_3,\nu})a(h_{x_1,\sigma})dx_1dx_3
  \end{align*}
  Using $\mathbf{h}_{x_1,\sigma}=a(L_{x_1,\sigma}[\delta_4])+
  (\mathbf{h}_{x_1,\sigma}-a(L_{x_1,\sigma}[\delta_4]))$, we furthermore divide it by $\Phi_{4}^{(1)}=\Phi_{41}^{(1)}+\Phi_{42}^{(1)}$. Similar to the bound of $\Phi_{3}^{(1)}$, and combining the idea in proving $\Phi_{21}^{(1)}$, we have
  \begin{equation}\label{Phi_4 bound}
    \pm\Phi_{4}\leq C\tilde{\rho}_0^{1+\delta_2}{\mathcal{N}}_{re}
    +C\tilde{\rho}_0^{\frac{5}{3}-4\alpha_6-\delta_2-\varepsilon}\mathcal{N}_{h}[\delta_4].
  \end{equation}

  \par Collecting all the results above, using Newton-Leibniz formula and Lemma \ref{lemma bog control N_re}, we prove Lemma \ref{cal com bog lemma}.
\end{proof}

\vspace{1em}

\begin{proof}[Proof of Proposition \ref{bog prop}]
\par (\ref{define Z_N}) together with Lemmas \ref{lemma bog control K}-\ref{cal com bog lemma} yield Proposition \ref{bog prop}.
\end{proof}

\section*{Acknowledgement}
X.Chen is partially supported by DMS-2406620, J.Wu and Z.Zhang are partially supported by the National Key R\&D Program of China under Grant 2023YFA1008801 and NSF of China under Grant 12288101. We sincerely appreciate the helpful discussion with B.Schlein, who gave us many kind advices and comments.

\bibliographystyle{abbrv}
\bibliography{refPressure}

\begin{thebibliography}{10}

\bibitem{workshopfermi}
Mini-workshop: {M}athematics of many-body fermionic systems.
\newblock {\em Oberwolfach Rep.}, 20(4):2809--2849, 2023.
\newblock Abstracts from the mini-workshop held October 29--November 4, 2023,
  Organized by Nikolai Leopold, Phan Th\`anh Nam and Chiara Saffirio.

\bibitem{soviet}
A.~A. Abrikosov and I.~M. Khalatnikov.
\newblock Concerning a model for a non-ideal {F}ermi gas.
\newblock {\em Soviet Phys. JETP}.

\bibitem{GenH_F}
V.~Bach, E.~H. Lieb, and J.~P. Solovej.
\newblock Generalized {H}artree-{F}ock theory and the {H}ubbard model.
\newblock {\em Journal of Statistical Physics}, 76:3--89.

\bibitem{meanfieldfermi2019}
N.~Benedikter, P.~T. Nam, M.~Porta, B.~Schlein, and R.~Seiringer.
\newblock Optimal upper bound for the correlation energy of a {F}ermi gas in
  the mean-field regime.
\newblock {\em Comm. Math. Phys.}, 374(3):2097--2150, 2020.

\bibitem{meanfieldfermi2021}
N.~Benedikter, P.~T. Nam, M.~Porta, B.~Schlein, and R.~Seiringer.
\newblock Correlation energy of a weakly interacting {F}ermi gas.
\newblock {\em Invent. Math.}, 225(3):885--979, 2021.

\bibitem{Seiringerfermi}
N.~Benedikter, M.~Porta, B.~Schlein, and R.~Seiringer.
\newblock Correlation energy of a weakly interacting {F}ermi gas with large
  interaction potential.
\newblock {\em Arch. Ration. Mech. Anal.}, 247(4):Paper No. 65, 57, 2023.

\bibitem{BEC2018}
C.~Boccato, C.~Brennecke, S.~Cenatiempo, and B.~Schlein.
\newblock Complete {B}ose-{E}instein condensation in the {G}ross-{P}itaevskii
  regime.
\newblock {\em Comm. Math. Phys.}, 359(3):975--1026, 2018.

\bibitem{2018Bogoliubov}
C.~Boccato, C.~Brennecke, S.~Cenatiempo, and B.~Schlein.
\newblock Bogoliubov theory in the {G}ross-{P}itaevskii limit.
\newblock {\em Acta Math.}, 222(2):219--335, 2019.

\bibitem{boccatoBrenCena2020optimal}
C.~Boccato, C.~Brennecke, S.~Cenatiempo, and B.~Schlein.
\newblock Optimal rate for {B}ose-{E}instein condensation in the
  {G}ross-{P}itaevskii regime.
\newblock {\em Comm. Math. Phys.}, 376(2):1311--1395, 2020.

\bibitem{me}
X.~Chen, J.~Wu, and Z.~Zhang.
\newblock The second order 2{D} behaviors of a 3{D} bose gases in the
  {G}ross-{P}itaevskii regime, 2024.
\newblock arXiv: 2401.15540, 142pp.

\bibitem{WJH}
X.~Chen, J.~Wu, and Z.~Zhang.
\newblock The second order {H}uang-{Y}ang formula to the 3{D} {F}ermi gas: the
  {G}ross-{P}itaevskii regime, 2024.
\newblock arXiv: 2410.16620, 122pp.

\bibitem{Phanfermimeanfield}
M.~R. Christiansen, C.~Hainzl, and P.~T. Nam.
\newblock The random phase approximation for interacting {F}ermi gases in the
  mean-field regime.
\newblock {\em Forum Math. Pi}, 11:Paper No. e32, 131, 2023.

\bibitem{christiansen2024correlationenergyelectrongas}
M.~R. Christiansen, C.~Hainzl, and P.~T. Nam.
\newblock The correlation energy of the electron gas in the mean-field regime,
  2024.
\newblock arXiv: 2405.01386, 60pp.

\bibitem{dingle1973asymptotic}
R.~Dingle.
\newblock {\em Asymptotic Expansions: Their Derivation and Interpretation}.
\newblock Academic Press, 1973.

\bibitem{Dyson1957}
F.~J. Dyson.
\newblock Ground-state energy of a hard-sphere gas.
\newblock {\em Phys. Rev.}, 106:20--26, Apr 1957.

\bibitem{dy2006}
L.~Erd\H{o}s, B.~Schlein, and H.-T. Yau.
\newblock Derivation of the {G}ross-{P}itaevskii hierarchy for the dynamics of
  {B}ose-{E}instein condensate.
\newblock {\em Comm. Pure Appl. Math.}, 59(12):1659--1741, 2006.

\bibitem{dilutefermiBog}
M.~Falconi, E.~L. Giacomelli, C.~Hainzl, and M.~Porta.
\newblock The dilute {F}ermi gas via {B}ogoliubov theory.
\newblock {\em Ann. Henri Poincar\'e}, 22(7):2283--2353, 2021.

\bibitem{bose2ndthermo1}
S.~Fournais and J.~P. Solovej.
\newblock The energy of dilute {B}ose gases.
\newblock {\em Ann. of Math. (2)}, 192(3):893--976, 2020.

\bibitem{bose2ndthermo2}
S.~Fournais and J.~P. Solovej.
\newblock The energy of dilute {B}ose gases {II}: the general case.
\newblock {\em Invent. Math.}, 232(2):863--994, 2023.

\bibitem{wickformula}
M.~Gaudin.
\newblock Une démonstration simplifiée du théorème de wick en mécanique
  statistique.
\newblock {\em Nuclear Physics}, 15:89--91, 1960.

\bibitem{fermiupper}
E.~L. Giacomelli.
\newblock An optimal upper bound for the dilute {F}ermi gas in three
  dimensions.
\newblock {\em J. Funct. Anal.}, 285(8):Paper No. 110073, 73, 2023.

\bibitem{giacomelli2024}
E.~L. Giacomelli.
\newblock An optimal lower bound for the low density {F}ermi gas in three
  dimensions, 2024.
\newblock arXiv: 2410.08904, 38pp.

\bibitem{2024huangyangformulalowdensityfermi}
E.~L. Giacomelli, C.~Hainzl, P.~T. Nam, and R.~Seiringer.
\newblock The {H}uang-{Y}ang formula for the low-density {F}ermi gas: upper
  bound, 2024.
\newblock arXiv: 2409.17914, 48pp.

\bibitem{giacomelli2025huangyangconjecturelowdensityfermi}
E.~L. Giacomelli, C.~Hainzl, P.~T. Nam, and R.~Seiringer.
\newblock The {H}uang-{Y}ang conjecture for the low-density {F}ermi gas, 2025.
\newblock arXiv: 2505.22340, 65pp.

\bibitem{haberberger2023}
F.~Haberberger, C.~Hainzl, P.~T. Nam, R.~Seiringer, and A.~Triay.
\newblock The free energy of dilute {B}ose gases at low temperatures, 2024.
\newblock arXiv: 2304.02405, 68pp.

\bibitem{haberberger2024}
F.~Haberberger, C.~Hainzl, B.~Schlein, and A.~Triay.
\newblock Upper bound for the free energy of dilute {B}ose gases at low
  temperature, 2024.
\newblock arXiv: 2405.03378, 54pp.

\bibitem{hainzlSchleinTriay2022bogoliubov}
C.~Hainzl, B.~Schlein, and A.~Triay.
\newblock Bogoliubov theory in the {G}ross-{P}itaevskii limit: a simplified
  approach.
\newblock {\em Forum Math. Sigma}, 10:Paper No. e90, 39, 2022.

\bibitem{HeathBrown+1999+883+892}
D.~Heath-Brown.
\newblock {\em Lattice points in the sphere}, pages 883--892.
\newblock De Gruyter, Berlin, Boston, 1999.

\bibitem{huang2008statistical}
K.~Huang.
\newblock {\em Statistical Mechanics, 2nd Ed}.
\newblock Wiley India Pvt. Limited, 2008.

\bibitem{huangyang}
K.~Huang and C.~N. Yang.
\newblock Quantum-mechanical many-body problem with hard-sphere interaction.
\newblock {\em Phys. Rev.}, 105:767--775, Feb 1957.

\bibitem{polarizedupper}
A.~B. Lauritsen and R.~Seiringer.
\newblock Ground state energy of the dilute spin-polarized {F}ermi gas: upper
  bound via cluster expansion.
\newblock {\em J. Funct. Anal.}, 286(7):Paper No. 110320, 98, 2024.

\bibitem{spin1lower}
A.~B. Lauritsen and R.~Seiringer.
\newblock Pressure of a dilute spin-polarized {F}ermi gas: lower bound.
\newblock {\em Forum Math. Sigma}, 12:Paper No. e78, 37, 2024.

\bibitem{spin1upper}
A.~B. Lauritsen and R.~Seiringer.
\newblock Pressure of a dilute spin-polarized {F}ermi gas: {U}pper bound, 2024.
\newblock arXiv: 2407.05990, 40pp.

\bibitem{LHY106}
T.~D. Lee, K.~Huang, and C.~N. Yang.
\newblock Eigenvalues and eigenfunctions of a {B}ose system of hard spheres and
  its low-temperature properties.
\newblock {\em Phys. Rev.}, 106:1135--1145, Jun 1957.

\bibitem{TDLEECNYANG112}
T.~D. Lee and C.~N. Yang.
\newblock Low-temperature behavior of a dilute {B}ose system of hard spheres.
  i. equilibrium properties.
\newblock {\em Phys. Rev.}, 112:1419--1429, Dec 1958.

\bibitem{spectrum}
M.~Lewin, P.~T. Nam, S.~Serfaty, and J.~P. Solovej.
\newblock Bogoliubov spectrum of interacting {B}ose gases.
\newblock {\em Comm. Pure Appl. Math.}, 68(3):413--471, 2015.

\bibitem{BEC2002}
E.~H. Lieb and R.~Seiringer.
\newblock Proof of {B}ose-{E}instein condensation for dilute trapped gases.
\newblock {\em Phys. Rev. Lett.}, 88:170409, Apr 2002.

\bibitem{2005fermiphy}
E.~H. Lieb, R.~Seiringer, and J.~P. Solovej.
\newblock Ground-state energy of the low-density {F}ermi gas.
\newblock {\em Phys. Rev. A}, 71:053605, May 2005.

\bibitem{lieb2005mathematics}
E.~H. Lieb, R.~Seiringer, J.~P. Solovej, and J.~Yngvason.
\newblock {\em {The mathematics of the {B}ose gas and its condensation}},
  volume~34.
\newblock Springer Science \& Business Media, 2005.

\bibitem{LiebYng1998}
E.~H. Lieb and J.~Yngvason.
\newblock Ground state energy of the low density {B}ose gas.
\newblock {\em Phys. Rev. Lett.}, 80:2504--2507, Mar 1998.

\bibitem{NIST}
F.~W.~J. Olver, , D.~W. Lozier, R.~F. Boisvert, and C.~W. Clark.
\newblock {\em The {NIST} Handbook of Mathematical Functions}.
\newblock Cambridge Univ. Press, 2010.

\bibitem{FermithermoTpositive}
R.~Seiringer.
\newblock The thermodynamic pressure of a dilute {F}ermi gas.
\newblock {\em Comm. Math. Phys.}, 261(3):729--757, 2006.

\bibitem{meanfield2011seiR}
R.~Seiringer.
\newblock The excitation spectrum for weakly interacting bosons.
\newblock {\em Comm. Math. Phys.}, 306(2):565--578, 2011.

\bibitem{4pages}
W.~Thirring.
\newblock Bounds on the entropy in terms of one particle distributions.
\newblock {\em Letters in Mathematical Physics}, 4:67--70, 1980.

\end{thebibliography}
\end{document}